%% file: main.tex
\documentclass[
11pt,                          
english                        
]{book}

\usepackage[english]{babel}    
\usepackage{amsmath}           
\usepackage[utf8]{inputenc}    
\usepackage[T1]{fontenc}       
\usepackage{longtable}         
\usepackage{exscale}           
\usepackage[final]{graphicx}   
\usepackage[sort]{cite}        
\usepackage{array}             
\usepackage{wasysym}           
\usepackage[a4paper]{geometry} 
\usepackage{xspace}            
\usepackage{tikz-cd}              
\usepackage{chairx}            
\usepackage[expansion=false    
           ]{microtype}        
\usepackage[nottoc]{tocbibind} 
\usepackage[backref=page,      
           final=true,         
           pdfpagelabels       
           ]{hyperref}         
\usepackage[toc,page]{appendix}
\usepackage{epigraph}

\graphicspath{{../tikz/}}

\usetikzlibrary{matrix}
\usetikzlibrary{arrows}
\usetikzlibrary{patterns}
\usetikzlibrary{decorations.pathreplacing}
\usetikzlibrary{decorations.text}

\theorembodyfont{\itshape}
\theoremstyle{nonumberplain}
\newtheorem{maintheorem}{\textbf{Main Theorem}}

\begin{document}

%
\thispagestyle{empty}
\begin{center}
{\huge \textbf{Star Products that can not be induced by Drinfel\textquoteright d Twists}}\\
  \vspace{1cm}
  {\large \textbf{Thomas Weber}\\}
  {\large Würzburg, June 06, 2016\\}
  \vspace{1.5cm}
  \includegraphics[width=6cm]{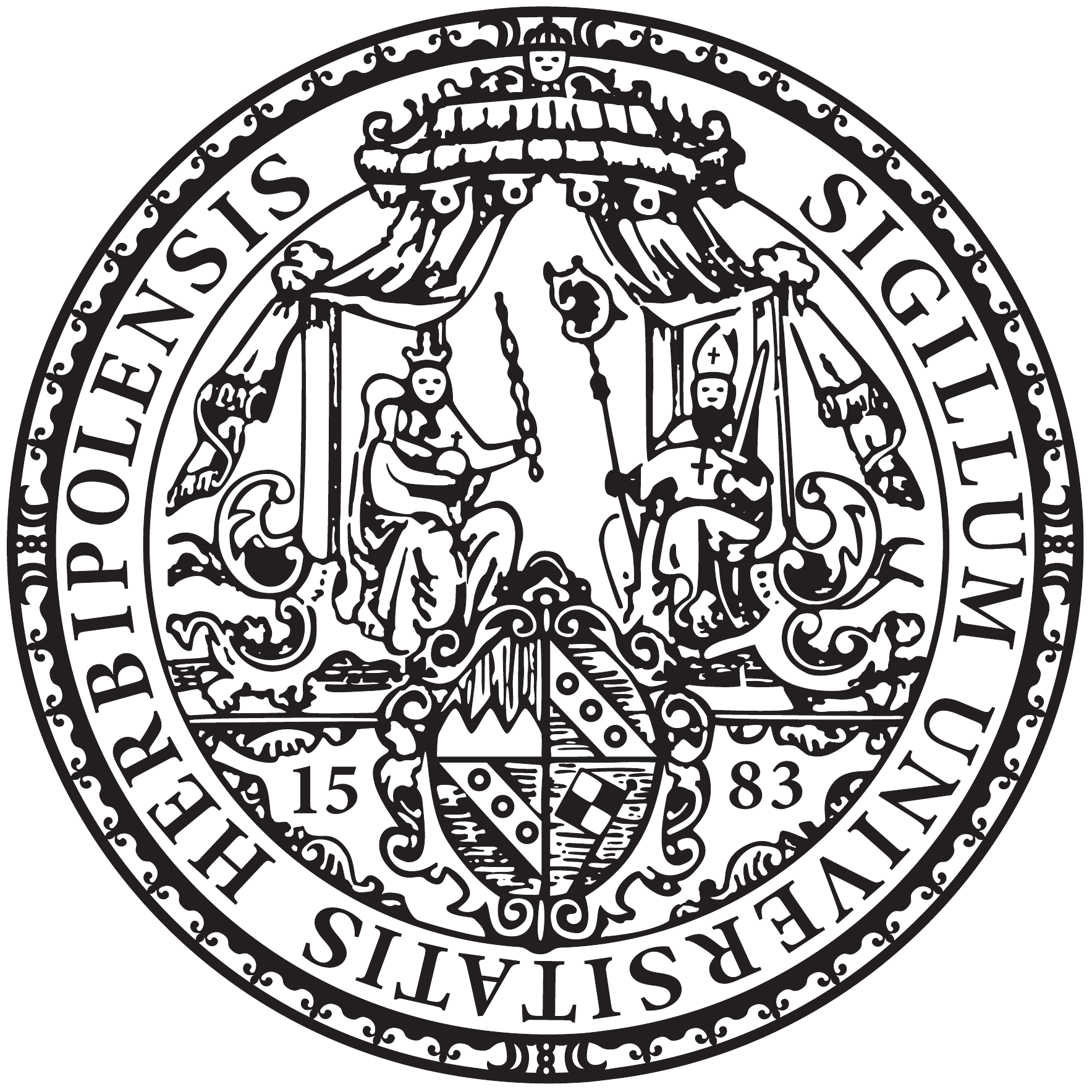}\\
  \vspace{0.5cm}
    {\textbf{Master Thesis}}\\
    \vspace{0.25cm}
    {in the Study Program Mathematical Physics, Master of Science}\\
    {with}\\
    {Prof. Dr. Stefan Waldmann}\\
    {and}\\
    {Dr. Chiara Esposito}\\
    \vspace{0.3cm}
  {Chair X (Mathematical Physics)\\
  Department of Mathematics\\
  Julius-Maximilians-University, Würzburg\\}
    \vspace{\fill} 
  \newpage
  
  \null\vfill
  {\large Die vorliegende Arbeit wurde im Zeitraum vom 07.12.2015 bis zum 06.06.2016 am Lehrstuhl für Mathematische
  Physik der Julius-Maximilians-Universität Würzburg unter der Leitung von Dr. Chiara Esposito und Prof. Dr. Stefan
  Waldmann angefertigt.}
  \vfill

\end{center}
\newpage

\section*{Summary of this Thesis}

We give obstructions to
the existence of Drinfel\textquoteright d twists on connected compact $ 2 $-dimensional symplectic manifolds.
In fact, only the $ 2 $-torus permits a
twist star product that deforms the symplectic structure. The main observation in the line of arguments is
that a twist star product forces such a symplectic manifold to be a homogeneous space. This immediately excludes
the higher pretzel surfaces $ \mathrm{T}(g) $ since they can not be structured as homogeneous spaces: by a theorem
due to G. D. Mostow the Euler characteristic $ \chi $ of a compact homogeneous space is non-negative, while
$ \chi(\mathrm{T}(g))=2-2g $. The $ 2 $-sphere $ \mathbb{S}^{2} $ is a bit more involved. A. L. Onishchik classified all
connected Lie groups that act transitively and effectively on $ \mathbb{S}^{2} $ up to equivalence. In particular, they
are semisimple Lie groups. We prove that a twist star product on $ \mathbb{S}^{2} $ induces a transitive effective
action of a connected Lie group on $ \mathbb{S}^{2} $ that is not semisimple, to produce an obstruction
also in this situation. It is the Etingof-Schiffmann subgroup of the $ r $-matrix corresponding to the twist.

\section*{Summary of this Thesis (in German)}

In dieser Masterarbeit geben wir Obstruktionen für die Existenz von Drinfel\textquoteright d Twists auf
zusammenhängenden, kompakten, $ 2 $-dimensionalen symplektischen Mannigfaltigkeiten. Einzig der $ 2 $-Torus
besitzt ein Twist-Sternprodukt, welches die symplektische Struktur deformiert. Das Hauptargument der
Beweisführung ist dabei, dass solch eine symplektische Mannigfaltigkeit mit einem Twist-Sternprodukt bereits
ein homogener Raum ist. Somit können wir sofort die höheren Brezelflächen $ \mathrm{T}(g) $ ausschließen,
da diese keine homogenen Räume sind: Nach einem Satz von G. D. Mostow ist die Euler-Charakteristik $ \chi $ eines
kompakten homogenen Raumes nicht-negativ, während $ \chi(\mathrm{T}(g))=2-2g $. Die Argumentation für die $ 2 $-Spähre
ist etwas aufwendiger. A. L. Onishchik klassifizierte bis auf Äquivalenz alle zusammenhängenden Lie-Gruppen, welche
transitiv und effektiv auf $ \mathbb{S}^{2} $ wirken. Diese Lie-Gruppen sind insbesondere halbeinfach. Wir beweisen,
dass ein Twist-Sternprodukt auf $ \mathbb{S}^{2} $ die transitive und effektive Lie-Gruppen-Wirkung einer
zusammenhängenden nicht-halbeinfachen Lie-Gruppe induziert, um auch in diesem Fall eine Obstruktion zu erhalten.
Es ist die Etingof-Schiffmann-Untergruppe der $ r $-Matrix, welche dem Twist zugehört.

\cleardoublepage

%
\begin{flushright}
        \vspace*{3cm}
        {\large
                
        }
\end{flushright}
\thispagestyle{empty}

%
\tableofcontents
\thispagestyle{empty}



%
\include{Introduction}

%

\include{Chapter1}

%

\include{Chapter2}

%

\include{Chapter3}

%

\include{Chapter4}

%

\include{Chapter5}

%

\include{Appendix}



\bibliographystyle{plain}
\bibliography{Thomas}


\include{LastPage}

\end{document}

%% file: Introduction.tex
\begin{tiny}

\end{tiny}\chapter{Introduction}

\section*{Deformation Quantization}

\epigraph{\textit{The miracle of the appropriateness of the language of mathematics for the formulation of the laws of
physics is
a wonderful gift which we neither understand nor deserve. We should be grateful for it and hope that it will remain
valid
in future research and that it will extend, for better or for worse, to our pleasure, even though perhaps also to our
bafflement, to wide branches of learning.}}{---\textup{Eugene Wigner}, The Unreasonable Effectiveness of Mathematics in
the Natural Sciences}
A fundamental strategy in applied science is approximating physical systems by mathematical formalisms. The advantage
of this idea is that mathematical implications can lead to predictions on the corresponding physical system.
However, one has to be careful: a mathematical theory is just an approach of nature. If it is in conflict with a
reproducible measurement, the mathematical framework has to be adapted. At its best by a more general theory which
includes
the aspects of the old theory that are conform to the physical results. One might think of \textit{quantum mechanics}
extending \textit{classical mechanics}. Even if the laws of classical mechanics seem more nearby in daily life,
quantum mechanics is said to be the best description of nature today. This statement is based on measurements on very
small scales. But there is an open question: how does the classical situation appear as a limit of the
quantum mechanical system? Remark that this is a pure mathematical question, since we already know that both, the
classical and the quantum world, exist. There are several approaches of \textit{quantization}, e.g. operator
formalism on Hilbert spaces (consider \cite{dirac1981principles,neumann2013mathematische}) or path integral
quantization (see \cite{Feynman,peskin1995introduction}). The one we are interested in is the so-called
\textit{deformation quantization}. To explain its essence we give a concrete example following
\cite[Section~9.1.1]{waldmann2014difgeolecturenotes}.

We want to discuss an approach to describe a particle of mass $ m $ influenced by a force field
$ F\colon\mathbb{R}^{3}\rightarrow\mathbb{R}^{3} $. The position
of this particle at the time $ t\in\mathbb{R} $ is represented by a vector $ q(t)\in\mathbb{R}^{3} $ as well as the
momentum $ p(t)=m\dot{q}(t) $ of the particle. According to Newton\textquoteright s second law one has
\begin{align}
F(q(t))=m\ddot{q}(t),
\end{align}
for any time $ t\in\mathbb{R} $, where $ \ddot{q}(t) $ denotes the acceleration of the particle at the time $ t $.
This second order differential equation is equivalent to the two first order differential equations
\begin{align}
\dot{q}(t)=\frac{1}{m}p(t)\text{  and  }\dot{p}(t)=F(q(t)).
\end{align}
Remark that we made a change of coordinates from $ q(t)\in\mathbb{R}^{3} $ to $ (q(t),p(t))\in\mathbb{R}^{6} $.
If we assume $ F $ to be conservative there is a potential $ V\colon\mathbb{R}^{3}\rightarrow\mathbb{R} $ such that
\begin{align}
F=-\nabla V,
\end{align}
where $ \nabla V\colon\mathbb{R}^{3}\rightarrow\mathbb{R}^{3} $ denotes the gradient of $ V $.
Then one can define the \textit{Hamilton function}
\begin{align}
H\colon\mathbb{R}^{6}\ni(q,p)\mapsto\frac{p^{2}}{2m}+V(q)\in\mathbb{R}
\end{align}
of the physical system that leads to another equivalent formulation of Newton\textquoteright s second law:
\textit{Hamilton\textquoteright s equations of motion}
\begin{align}
\dot{q}(t)=\frac{\partial H}{\partial p}(q(t),p(t))\text{  and  }\dot{p}(t)=-\frac{\partial H}{\partial q}(q(t),p(t)).
\end{align}
By introducing the anti-symmetric matrix
\begin{align}
\Omega=\begin{pmatrix}
0 & \mathbb{1} \\
-\mathbb{1} & 0
\end{pmatrix}\in M_{6\times 6}(\mathbb{R})
\end{align}
and combined coordinates $ x=(q,p) $, we are able to reduce Hamilton\textquoteright s equations of motion to a single
first order vector-valued differential equation
\begin{align}
\dot{x}(t)=\Omega((\nabla H)(x(t))).
\end{align}
The map
\begin{align}
X_{H}=\Omega\nabla H\colon\mathbb{R}^{6}\rightarrow\mathbb{R}^{6}
\end{align}
is said to be the \textit{Hamiltonian vector field} of $ H $. A solution $ x\colon\mathbb{R}\rightarrow\mathbb{R}^{6} $
of the corresponding flow equation
\begin{align}
\dot{x}(t)=X_{H}(x(t))
\end{align}
describes the position and momentum of the particle. By defining the \textit{Poisson bracket}
\begin{align}
\left\lbrace f,g\right\rbrace=\langle\nabla f,\Omega\nabla g\rangle
\end{align}
of two smooth real-valued functions $ f,g\in\Cinfty(\mathbb{R}^{6}) $ on $ \mathbb{R}^{6} $, we can state an easy
condition for a function to be a \textit{constant of motion} of the corresponding Hamilton system: the derivative of
$ f $ along a solution $ x $ of Hamilton\textquoteright s equations of motion reads
\begin{align*}
\frac{\mathrm{d}}{\mathrm{d}t}f(x(t))=\langle\nabla f(x(t)),\dot{x}(t)\rangle
=\langle\nabla f(x(t)),\Omega((\nabla H)(x(t)))\rangle=\left\lbrace f,H\right\rbrace(x(t)).
\end{align*}
Thus $ f $ is a constant of motion if and only if $ \left\lbrace f,H\right\rbrace(x(t))=0 $ for all $ t\in\mathbb{R} $.
It is easy to check that $ \left\lbrace\cdot,\cdot\right\rbrace $ is a bilinear anti-symmetric map that
assigns two smooth functions on $ \mathbb{R}^{6} $ another smooth function on $ \mathbb{R}^{6} $. Moreover,
$ \left\lbrace\cdot,\cdot\right\rbrace $ satisfies the \textit{Leibniz rule}
\begin{align}
\left\lbrace f,gh\right\rbrace=\left\lbrace f,g\right\rbrace h+g\left\lbrace f,h\right\rbrace
\end{align}
and the \textit{Jacobi identity}
\begin{align}
\left\lbrace\left\lbrace f,g\right\rbrace,h\right\rbrace+\left\lbrace\left\lbrace h,f\right\rbrace,g\right\rbrace
+\left\lbrace\left\lbrace g,h\right\rbrace,f\right\rbrace=0,
\end{align}
for all $ f,g,h\in\Cinfty(\mathbb{R}^{6}) $. The space of smooth real-valued
functions on $ \mathbb{R}^{6} $ equipped with the pointwise product of functions is said to be the
\textit{algebra of classical observables}. In
deformation quantization one does not change the space of variables by passing to the quantum mechanical system.
Instead of this, the multiplication is changed: one deforms the associative and commutative multiplication of the
classical observables by adding perturbation terms, such that the result is still
associative but not commutative any more. This has to be done in a way such that the classical limit gives back the
pointwise product. One also wants the new product to satisfy the \textit{correspondence principle}: the first
perturbation term of the commutator with respect to this product has to result in the Poisson bracket of the
classical system. Returning to our example, we define the \textit{Weyl-Moyal star product}
\begin{align}
(f\star g)(q,p)=\sum_{m,n=0}^{\infty}\Bigg(\frac{i\hbar}{2}\Bigg)^{m+n}\frac{(-1)^{m}}{m!n!}
\Bigg(\frac{\partial^{m}}{\partial p^{m}}\frac{\partial^{n}}{\partial q^{n}}f(q,p)\Bigg)
\Bigg(\frac{\partial^{n}}{\partial p^{n}}\frac{\partial^{m}}{\partial q^{m}}g(q,p)\Bigg).
\end{align}
This associative product is a formal power series
\begin{align}
f\star g=\sum_{n=0}^{\infty}(i\hbar)^{n}C_{n}(f,g),
\end{align}
where $ C_{n} $ are bidifferential operators and $ \hbar $ denotes \textit{Planck\textquoteright s constant},
the formal parameter of this series.
This is the first step of deformation quantization: we obtained a noncommutative algebra
$ (\Cinfty(M)[[\hbar]],\star) $ of quantum observables. The product is not commutative and we realize that
\begin{align}
f\star g=fg+\frac{i\hbar}{2}\left\lbrace f,g\right\rbrace\mod{\hbar^{2}}.
\end{align}
Thus, the classical limit $ \lim_{\hbar\rightarrow 0}f\star g=fg $ gives back the pointwise product.
Moreover, by defining the $ \star $-commutator $ \left[ f,g\right]_{\star}=f\star g-g\star f $ one obtains
the correspondence principle
\begin{align}
\lim_{\hbar\rightarrow 0}\frac{1}{i\hbar}\left[ f,g\right]_{\star}=\left\lbrace f,g\right\rbrace.
\end{align}
This is an a posteriori motivation to choose $ \hbar $ as the formal parameter of this formal power series.
The zeroth order of $ \star $ represents the classical world while quantum effects first appear in order of
Planck\textquoteright s constant, i.e. in the first order of $ \star $. Higher orders of $ \hbar $ are not
physically measurable. In a second and last step one has to consider \textit{states}, i.e.
$ \mathbb{C}[[\hbar]] $-linear positive functionals
\begin{align}
\omega\colon\Cinfty(M)[[\hbar]]\rightarrow\mathbb{C}[[\hbar]],
\end{align}
and \textit{representations} of $ (\Cinfty(M)[[\hbar]],\star) $ on Hilbert spaces which lead to the notion of
$ * $-algebras. One approach is the $ GNS $ construction (see \cite{WaldmannRepresentation}).

In a general setting, the classical observables are encoded by the algebra of smooth functions on a Poisson manifold
(see Section \ref{secLieBiPoLie} for a definition). Then a star product (see Definition \ref{def3thomas}) is a
deformation of the pointwise product of these functions, such that the correspondence principle leads to the
Poisson bracket of the manifold. Thus, if one is interested in quantization of a classical system, one option is
to consider star products on Poisson manifolds. The notion of star products goes back to the work of F. Bayen
et al. \cite{BAYEN197861}. S. Gutt \cite{Gutt1983} and V. G. Drinfel\textquoteright d
\cite{DrinfeldOnConstant} discovered star products as deformations of linear Poisson structures on the dual of Lie
algebras. Later B. V. Fedosov proved that there is a star product on any symplectic manifold (c.f. \cite{fedosov1994})
extending the work of M. de Wilde and P. Lecomte (c.f. \cite{deWilde1983}). The contribution of Fedosov
is remarkable, since he gives a recursive formula of a star product, while constructive proofs are rare
in deformation quantization. Ideas of his proof are also used beyond symplectic geometry, e.g. on Poisson
manifolds (c.f. \cite{FelderQuantization}) and Kähler manifolds (c.f. \cite{WaldmannBordemann}).
By a result of M. Kontsevich (c.f. \cite{kontsevich1997deformation})
there is always a star product on a Poisson manifold. A different proof by D. E. Tamarkin appeared later
(c.f. \cite{TamarkinKontsevich}). A physically more relevant approach is to consider convergent star products.
This can be looked up e.g. in \cite{Bieliavsky2002Convergent,SolovievConvStar}.
Consider also \cite{Dito1990} for a star product approach to quantum field theory.
The papers \cite{SternheimerTwentyYears} and \cite{WaldmannRecentDevelopments}
overview the development of deformation quantizations and provide further references.

\section*{Drinfel\textquoteright d Twists}

We are interested in a more specific situation, again motivated
form physics: if there is a quantization of the symmetries of a classical system, the system itself has to be
quantized in a compatible way. This led to the notion of \textit{twists} in the 1980's. It was
V. G. Drinfel\textquoteright d who developed the corresponding theory
(see e.g. \cite{DrinfeldOnConstant,drinfeldquantumgroups}). Symmetries are encoded by formal power series of
universal enveloping algebras that act on the algebra of classical observables. The twist is an element of the
tensor product of the symmetries. As a result there is a way to change the Hopf algebra structure of the symmetries
such that the new structure is noncocommutative. Moreover, one can deform the product of the classical
observables and obtains a noncommutative product $ \star $ in terms of the twist. If $ \star $ is a star product, we
received
a deformation quantization of the classical system via a twist. Thus the advantage is that any module algebra of
the symmetries is deformed in an appropriate way. There has been many research on twist deformation:
in \cite{DrinfeldQuasiHopf} Drinfel\textquoteright d proves that twists are
in one-to-one correspondence with left-invariant star products on formal power series of smooth functions on
Lie groups. Concrete examples and formulas can be found in
\cite{BieliavskyUDF} and \cite{Giaquintobialgebraactions}, while \cite{schenkeldrintwistdef} gives an extension
of twists to arbitrary
connections. Via formality G. Halbout provides a quantization of twists on Lie bialgebras to obtain quantum twists
(c.f. \cite{HalboutFormality}). Furthermore, twists enable deformation quantization e.g. for actions of Kählerian Lie
groups (c.f. \cite{bieliavsky2015deformation}), $ 3 $-dimensional solvable Lie groups (c.f. \cite{Bieliavsky2005})
and the Heisenberg supergroup (c.f. \cite{BieliavskyHeisenberg}).
In quantum field theory one uses twists to deform the Poincaré group and implement the framework
of noncommutative geometry (see \cite{QFTDrinfeldTwist}). 
There are even twist approaches in string theory
(consider \cite{AsakawaTwistStringTheory,AsakawaTwistStringTheory2}).

Drinfel\textquoteright d introduced the twist element
\begin{align}
\mathcal{F}\in(\mathcal{U}(\mathfrak{g})\tensor\mathcal{U}(\mathfrak{g}))[[\hbar]]
\end{align}
on the formal power series of the tensor product of the universal enveloping algebra
(see Definition \ref{DefTwistUEA}). By construction it deforms the Hopf algebra structure of
$ \mathcal{U}(\mathfrak{g})[[\hbar]] $. In particular, it induces a deformed coproduct and a deformed antipode.
Assume that the formal power series $ \Cinfty(M)[[\hbar]] $ of smooth functions of a Poisson manifold $ (M,\pi) $
are a left module algebra of $ \mathcal{U}(\mathfrak{g})[[\hbar]] $ (see Definition \ref{DefModAlg}) via an action
$ \rhd $. The product on $ \Cinfty(M)[[\hbar]] $ is the pointwise multiplication $ \cdot $ of functions. Then a star
product $ \star $ on $ M $ is said to be \textit{induced by the twist} $ \mathcal{F} $ or a \textit{twist star product}
if for every $ f,g\in\Cinfty(M)[[\hbar]] $ one has
\begin{align}\label{IntrTwistStar}
f\star g=m(\mathcal{F}^{-1}\rhd(f\tensor g)),
\end{align}
where $ m\colon\Cinfty(M)[[\hbar]]\tensor\Cinfty(M)[[\hbar]]\ni(f\tensor g)\mapsto f\cdot g\in\Cinfty(M)[[\hbar]] $.
For a general left Hopf algebra module algebra $ (\algebra{A},\cdot) $ the \textit{twist product} defined in Eq.
(\ref{IntrTwistStar}) provides a deformation of a product $ \cdot $. But one deforms in a very controlled way: the new
product $ \star $ is associative and has the same unit. Moreover, $ (\algebra{A},\star) $ is a left Hopf algebra
module algebra of the deformed Hopf algebra. This is a very desired situation, since we do not only deform the
symmetries
of a system, but the system itself in a compatible way. Thus it is not surprising that Drinfel\textquoteright d
twists are frequently used tools in deformation quantization.
For example, there are many twist approaches on the so-called
$ q $-deformed sphere (consider \cite{WatamuraFuzzy,SteinackerFuzzyII,KurkcuogluFuzzy,SteinackerFuzzyI}).
Unfortunately, there are not many classifications of
twist elements around (c.f. \cite{HalboutFormality}) and the existence can often not be assured, besides there
are not many examples.
Moreover, it is not known if twist products deform the Poisson bracket of the Poisson manifold, i.e. if
the correspondence principle is valid. This means that even if there are twist products it is not clear if they are
physically relevant.
The aim of this thesis is to show that deformation quantization via a twist is not possible in many cases.

\begin{maintheorem}
There is no twist star product on the $ 2 $-sphere and the pretzel surfaces of genus $ g>1 $ deforming a
symplectic structure.
\end{maintheorem}

Furthermore, the theorem can be applied to any symplectic foliation of a Poisson manifold.
This is no classification result, but we show the constraints of twist deformation theory.
Consequently, for many algebras of classical observables one has to consider a different quantization approach.
This theorem will also appear in a preprint as a collaboration with Pierre Bieliavsky, Chiara Esposito
and Stefan Waldmann \cite{NoTwistPaper}.

\section*{Outlook}

There are many possibilities to generalize and expand the obstructions given in this thesis. We prove in
Corollary \ref{CompactTwistHomSp} that any connected compact symplectic manifold, which inherits a twist
star product, has to be a homogeneous space. Thus there are no twist star products on the higher pretzel surfaces
(see Theorem \ref{ThmPretzel}), but the contradiction applies to any connected compact symplectic homogeneous space.
In particular, one can consider connected compact symplectic manifolds with negative Euler characteristic to
obtain new counterexamples immediately. Also the argumentation in the proof of Theorem \ref{LastTheorem} can be
adapted. There we use the classification of effective transitive actions on the $ 2 $-sphere (see
Theorem \ref{ThmAllSphereActions}). The existence of a twist star product on $ \mathbb{S}^{2} $ induces a transitive
effective action of a connected non-semisimple Lie group on the $ 2 $-sphere. This is a contradiction to the
classification that only involves semisimple Lie groups. Thus the same argumentation holds for connected
compact symplectic manifolds that only admit transitive effective actions of semisimple Lie groups. The condition
of being compact can be dropped if one assures the integrability of the involved Lie algebra actions. Maybe
some candidates to consider are the complex projective spaces $ \mathbb{CP}^{n} $ and the Lagrangian Grassmannian.
Furthermore, this thesis could be helpful to give a classification of Drinfel\textquoteright d twists, since it
shows how obstructions apply on topological and geometrical level.

\section*{Organization}

This master thesis is structured as follows: Chapter \ref{chapTransAct} starts with an introduction to Lie group
actions. In particular, orbits and stabilizers of Lie group actions are
treated. We also recap some essential results in the theory of differential geometry, but only refer to their proofs.
The study of homogeneous spaces leads naturally to transitive actions and vice versa. We prove this
correspondence in full detail, since this is the most important characterization of homogeneous spaces for our
purpose. Also three possibilities on how the $ 2 $-sphere can be structured as a homogeneous space are given. It was
already mentioned that compact homogeneous spaces have non-negative Euler characteristic. We give a proof assuming
slightly stronger conditions by using group-invariant de Rham cohomology and Lie group representation theory.
This chapter ends by summarizing the classification of connected Lie groups that act transitively and effectively
on the $ 2 $-sphere. The following chapter is more algebraic. After studying some basic Lie algebra representation
theory, we consider Lie bialgebras and mention Poisson-Lie groups as their global counterpart. In particular, the dual
character of Lie bialgebras is pointed out. In fact, we are interested in a more specific situation: some Lie bialgebra
structures are cocycles, i.e. they are obtained as the Chevalley-Eilenberg differential of an element
$ r\in\mathfrak{g}\wedge\mathfrak{g} $. This naturally leads to $ r $-matrices and the classical Yang-Baxter
equation. While until now most of the theory was basic, the last section of this chapter is
quite specific: we define the Etingof-Schiffmann subalgebra as the Lie subalgebra, in which a $ r $-matrix is
non-degenerate.
The Etingof-Schiffmann subalgebra is never semisimple, which is an essential argument. The first section of Chapter
\ref{ChapTwistDef} is again very algebraic. We examine the notion of Drinfel\textquoteright d twists in great detail.
First we focus on twists on arbitrary Hopf algebras and in the second section on twists on formal power series of the
universal enveloping algebra. As mentioned in the introduction, the twist induces deformation in a more
categorial frame: not only the Hopf algebra but every module algebra can be deformed. The deformed algebra
multiplication is said to be a twist star product in the context of formal power series and star products.
As a last statement, we view a twist on formal power series as a quantization of a $ r $-matrix. In Chapter
\ref{rmatrixhomogeneous} we prove one of the main results of the thesis by connecting all previous notions and results.
Assuming the Poisson bivector of a connected compact symplectic manifold to be the image of a $ r $-matrix under a
Lie algebra action, we construct a transitive Lie group action on this manifold. This Lie group is the
Etingof-Schiffmann subgroup
corresponding to the $ r $-matrix. Thus this special Poisson bivector is an indicator for homogeneous spaces.
In particular, this situation occurs if there exists a twist star product. Thus we prove that a connected
compact symplectic manifold, that can be equipped with a twist star product, is a homogeneous space and there is
a $ r $-matrix which is non-degenerate in a Lie algebra, whose Lie group acts transitively on the homogeneous space.
In the last chapter we use this result to produce obstructions to twist star products. The higher pretzel
surfaces are not homogeneous spaces and the sphere only admits semisimple transitive effective Lie group actions,
which stands in contradiction to the action of the Etingof-Schiffmann subgroup. Remark that there are star products on 
all these symplectic manifolds according to the Fedosov construction. Thus there are indeed star products that can
not be induced by Drinfel\textquoteright d twists. The two appendices are short introductions to Hopf algebras
and semisimple Lie algebras, respectively.

\section*{Thanks to...}

I want to thank my advisors, Chiara Esposito and Stefan Waldmann, for their support and patience. They invested
a lot of time in discussions that promoted the thesis. Not to mention the nice $ \hbar $bqs
regardlessly of the weather.
I am grateful to Pierre Bieliavsky, since the idea that led to this thesis is based on a discussion of him with
Chiara Esposito and Stefan Waldmann in Oberwolfach \cite{PierreDiscussion}.
I want to thank Jonas Schnitzer and Thorsten Reichert for giving me their assistance and advice.
Finally, I am thankful to my family and Verónica for their encouragement that strengthens me.

%% file: Chapter1.tex
\chapter{Transitive Actions and Homogeneous Spaces}\label{chapTransAct}

Homogeneous spaces are one of the central concepts of this thesis. A first and very simple definition of a homogeneous
space is the following: if $ H $ is a closed subgroup of a Lie group $ G $, the set of left cosets of $ H $ in $ G $,
i.e. $ \left\lbrace\left\lbrace gh~|~h\in H\right\rbrace~|~g\in G\right\rbrace $, is said to be a homogeneous space
$ G/H $. This also works for topological groups. But with this point of view we do not see the actual strength of
homogeneous spaces. The original motivation to consider them is their nature to encode symmetries of spaces.
To make this idea a bit more concrete one needs the notion of Lie group actions. They are smooth functions that
assign to any point of a Lie group $ G $ a diffeomorphism on a smooth manifold $ M $ in a way that respects the
left multiplication on $ G $. Also, one wants the identity on $ G $ to induce the identity map on $ M $. There a
various types of Lie group actions with many nice and useful properties and we discuss the ones that are interesting
for our purpose in Section \ref{SecGroupAction}. For any Lie group action there are two fundamental structures, one
on the Lie group and one on the manifold. The structure on the Lie group is the set of elements which induce
diffeomorphisms that reproduce a given point $ x $ on $ M $ and it is said to be the stabilizer of $ x $. We can
prove that it is always a Lie subgroup of $ G $. On the other hand the interesting structures on the manifold
are the orbits of the action. Intuitively, everyone has an idea of what an orbit should be: one might think of the
way a satellite circuits a planet or electrons an atomic nucleus in the Bohr model. There is a way to phrase these
ideas also in a mathematical definition: for a fixed point $ x $ on the manifold $ M $ the orbit of $ x $ is the set of
points on $ M $ that are reached by applying $ x $ to every diffeomorphism that is produced by elements of $ G $
via the action. Unfortunately, not every orbit is a smooth submanifold of $ M $, but we prove that this is true
if one stays close to the identity in $ G $ or if the Lie group is compact. A very interesting situation is obtained
if there is only one orbit, i.e. if all elements of $ x $ produce the same orbit which coincides with $ M $ in this
case. Such actions are said to be transitive and they connect homogeneous spaces to Lie group actions. It is the task
of Section \ref{sectionhomspace} to establish this connection. First of all, it is easy to conclude that the left
cosets of a closed subgroup $ H $ of the Lie group $ G $ are exactly the orbits of the right multiplication of
elements of $ H $ in $ G $. Thus one could define a homogeneous space $ G/H $ as the orbit space of this
particular action.
One benefit of this definition is that one is able to build the structure of a smooth manifold
on the homogeneous space as a simple consequence of the ``Free and Proper Action Theorem''.
But there is more: we get a whole classification of homogeneous spaces by transitive actions. On one hand
there is a natural transitive action of $ G $ on any homogeneous space $ G/H $. Surprisingly, the converse
statement is also true: for any transitive action of a Lie group $ G $ on a smooth manifold $ M $ there
is a diffeomorphism between $ M $ and the homogeneous space $ G/G_{x} $, where $ G_{x} $ denotes the stabilizer
of any point $ x $ of $ M $. Thus we are able to identify $ M $ with the orbit space $ G/G_{x} $, i.e. any point of
$ M $ is viewed as a left coset or orbit of an element in $ G $ with respect to $ G_{x} $. This also explains why
one calls $ G $ the symmetries of $ M $ which was one of our initial statements.
The last perception is fundamental for the proceeding of this thesis. Instead of searching
for coset or orbit spaces we are looking for transitive actions, that are easy to handle. As an example, we
give three different ways to structure the $ 2 $-sphere as a homogeneous space via different Lie groups that
act transitively on it. The last two sections of this chapter are then more specific classifications and indicators of
homogeneous spaces with topological arguments. Section \ref{SecEulerPos} is all about to prove that connected
compact homogeneous spaces have non-negative Euler characteristic. The Euler characteristic is a
topological invariant of a manifold, namely the alternating series of the dimensions of its cohomology groups.
It is an integer that is known for many manifolds. In particular, the connected compact pretzel surfaces
$ \mathrm{T}(g) $ of genus $ g\in\mathbb{N}_{0} $ have Euler characteristic $ 2(1-g) $. It follows that
$ \mathrm{T}(g) $ is not a homogeneous space
if $ g>1 $ and we experience the first time a very typical spirit of this thesis: we found an
obstruction, i.e. there is no Lie group that acts transitively on those manifolds. In Chapter \ref{ChapObstruction}
we expand these obstructions to twist star products. But first we continue with Section \ref{sectiononish}
and classify all transitive actions that act on the $ 2 $-sphere with some additional properties. In fact,
these are the three Lie group actions that we discussed in Section \ref{sectionhomspace} before. In particular,
they are semisimple Lie groups.

\section{Lie Group Actions}\label{SecGroupAction}

We have to assume that the reader has a basic understanding of Lie theory. The symmetries of our systems are
represented by Lie groups and we often investigate their infinitesimal information, i.e. their Lie algebras.
For a introduction to these two fundamental objects and their connection we refer to
\cite[Chapter~III]{bourbaki1998lie}, \cite[Chapter~1]{duistermaat2012lie}, \cite[Chapter~4]{kosmann2009groups}
or \cite[Chapter~4]{procesi2007lie}. Most of the time we argue with a very geometric point of view, thus the proofs
often inherit tools known from differential geometry. An efficient approach to this huge topic is given in
\cite{lee2003introduction} while we also refer to \cite{lee2009manifolds} and \cite{michortopics}.
The notion of Lie group actions is also very common, but nevertheless we present the basic definitions and results
in this section since these ideas are very central in this thesis. However, we proceed in a colloquial style to
accelerate the process, but also to carve out the important concepts. The major source is
\cite[Part~I~Chapter~1]{onishchik1993lie}, but the elementary theory of Lie group actions can also be found in
\cite[Lecture~1~and~2]{hsiang2000lectures}, \cite[Chapter~2]{kirillov2008introduction},
\cite[Chapter~9]{lee2003introduction} and \cite[Chapter~2]{michortopics}.

In the following $ M $ always denotes a smooth manifold and $ G $ a Lie group with identity element $ e\in G $. A
\textit{(left) action} of $ G $ on $ M $ is a smooth map
\begin{align}\label{action}
\Phi\colon G\times M\rightarrow M
\end{align}
that satisfies $ \Phi(g,\Phi(h,x))=\Phi(gh,x) $ and $ \Phi(e,x)=x $, for all $ g,h\in G $ and $ x\in M $. We further
denote for any $ g\in G $ the induced diffeomorphism on $ M $ by
\begin{align}
\Phi_{g}\colon M\ni x\mapsto\Phi_{g}(x)=\Phi(g,x)\in M.
\end{align}
This is clearly a diffeomorphism since there is the inverse map $ \Phi_{g^{-1}} $ and the smoothness of both maps
follows from the smoothness of $ \Phi $. Thus one can consider
\begin{align}
\tilde{\Phi}\colon G\ni g\mapsto\Phi_{g}\in\text{Diff}(M),
\end{align}
which maps an element of the Lie group to a diffeomorphism of $ M $. By definition $ \Phi_{e}=\text{id}_{M} $ is the
identity map on $ M $, which means $ e\in\ker\tilde{\Phi} $. If we have $ \ker\tilde{\Phi}=\left\lbrace
e\right\rbrace $ the action $ \Phi $ is said to be \textit{effective} or \textit{faithful}. Actions are often
demanded to be effective since one wants different group elements to induce different diffeomorphisms and this is
exactly what effective actions do. If $ \Phi $ is effective and $ g,h\in G $ are two different elements there is an
element $ x\in M $ such that $ \Phi_{g}(x)\neq \Phi_{h}(x) $, since otherwise $ \Phi_{h^{-1}g}=\text{id}_{M} $ while
$ h^{-1}g\neq e $, which is a contradiction to the effectiveness of $ \Phi $. In particular, this situation is
fulfilled
if for two different elements $ g,h\in G $ there is no $ x\in M $ such that $ \Phi_{g}(x)=\Phi_{h}(x) $. This
kind of action is said to be \textit{free}. In general $ \ker
\tilde{\Phi}\subseteq G $ is a Lie subgroup. This holds because $ \tilde{\Phi} $ is a group homomorphism by the
action property of $ \Phi $. If $ \ker\tilde{\Phi} $ is a discrete subgroup we call the action
$ \Phi $ \textit{locally effective}. Thus by definition any effective action is locally effective but not vice versa.
Two more notions induced by
actions can be obtained by considering for any $ x\in M $ the smooth map
\begin{align}
\Phi_{x}\colon G\ni g\mapsto\Phi_{x}(g)=\Phi(g,x)\in M.
\end{align}
We call the image
\begin{align}
G\cdot x=\left\lbrace\Phi_{x}(g)\in M~|~g\in G\right\rbrace
\end{align}
of $ G $ under $ \Phi_{x} $ the \textit{orbit} of $ x $ and the preimage
\begin{align}
G_{x}=\left\lbrace g\in G~|~\Phi_{x}(g)=x\right\rbrace
\end{align}
of $ x $ through $ \Phi_{x} $ the \textit{stabilizer} or \textit{isotropy group} of $ x $. Orbits are interesting
objects because they build a partition of $ M $, i.e. $ M $ can be regarded as the disjoint union of the orbits of
any action $ \Phi $ on $ M $. To prove this one first recognizes that the orbits of two different points $ x,y\in M $
are either equal or disjoint, since if $ (G\cdot x)\cap(G\cdot y)\neq\emptyset $ there is
an $ z\in(G\cdot x)\cap(G\cdot y) $, i.e. there are $ g,h\in G $ such that $ \Phi_{x}(g)=z=\Phi_{y}(h) $. Then for
any $ m\in G $ the point $ \Phi_{y}(m) $ of the orbit $ G\cdot y $ is an element of $ G\cdot x $ since
\begin{align*}
\Phi_{x}(mh^{-1}g)&=\Phi_{mh^{-1}g}(x)=\Phi_{mh^{-1}}(\Phi(g)(x))=\Phi_{mh^{-1}}(\Phi(h)(y))\\
&=\Phi_{mh^{-1}h}(y)=\Phi_{y}(m).
\end{align*}
In the same fashion $ G\cdot x $ is a subset of $ G\cdot y $ which implies $ G\cdot x=G\cdot y $. Also any point
$ x\in M $ lies inside one orbit, namely in its own since $ \Phi_{x}(e)=x $. Thus we achieved the partition of $ M $
as it was stated above. This gives the possibility to define an equivalence relation on $ M $ by declaring
$ x\sim y $ for $ x,y\in M $ if and only if $ x $ and $ y $ lie in the same orbit. We are interested in the
situation
in which there is only one orbit because such actions are very much related to homogeneous spaces. We call them
\textit{transitive} actions. Equivalently, one can say that an action is transitive if for each pair $ (x,y)\in M
\times M $ there is an element $ g\in G $ such that $ \Phi(g,x)=y $. One missing piece to connect transitive actions
to homogeneous spaces is the second object we defined along with orbits, the stabilizer. In the next section we
see that in the case of a transitive action all stabilizers are conjugate to each other and the quotient space of the
Lie group and any stabilizer is isomorphic to $ M $. This procedure is said to structure $ M $ as a homogeneous
space via a transitive action. Moreover, it turns out that the concepts of homogeneous spaces and transitive actions
are equivalent.

Before we study this connection we want to prove some results that already apply to arbitrary actions. The question
is what kind of substructures orbits and stabilizers are. For this we consider the map $ \Phi_{x} $ for a $ x\in M $
and prove that is has constant rank. Remember that the \textit{rank} $ \text{rank}_{p}f $ of a smooth map
$ f\colon M\rightarrow N $ between smooth manifolds at a point $ p\in M $ is the dimension of the image of the
derivative $ T_{p}f\colon T_{p}M\rightarrow T_{f(p)}N $ of $ f $ at $ p $, i.e. $ \text{rank}_{p}f=\dim\text{im}T_{p}
f $. If the rank does not depend on the point $ p $ which was chosen, the \textit{Constant Rank Theorem} (c.f.
\cite[Theorem~8.8]{lee2003introduction}) states that each level set of $ f $ is a closed
embedded submanifold of $ M $ of codimension $ \text{rank}_{p}f $. In this case we denote the rank of $ f $ by
$ \text{rank}f $ since the point we choose is not relevant. A further consequence of this theorem is that
\begin{align}\label{kernelformula}
T_{p}f^{-1}(\left\lbrace q\right\rbrace)=\ker T_{p}f
\end{align}
holds for all $ q\in N $ (c.f. \cite[Lemma~8.15]{lee2003introduction}).
Also recall that $ f $ is a \textit{local diffeomorphism} if for every $ p\in M $ there is an open
neighbourhood $ U\subseteq M $ of $ p $ such that $ \left.f\right|_{U}\colon U\rightarrow f(U) $ is a diffeomorphism.
As a consequence of the \textit{Inverse Function Theorem} (c.f \cite[Theorem~7.10]{lee2003introduction})
this definition is equivalent to the requirement that $ T_{p}f\colon T_{p}M\rightarrow T_{f(p)}N $ is an isomorphism
for any $ p\in M $.

\begin{theorem}\label{thmsubmfd}
Let $ \Phi\colon G\times M\rightarrow M $ be an action and $ x\in M $ an arbitrary point. Then the following
statements hold:
\begin{compactenum}
\item The map $ \Phi_{x} $ has constant rank $ \operatorname{rank}(\Phi_{x})
=k\in\left\lbrace 0,\ldots,\dim M\right\rbrace $.
\item The stabilizer $ G_{x}\subseteq G $ is a Lie subgroup of codimension $ k $ and
\begin{align}
T_{e}G_{x}=\ker T_{e}\Phi_{x}.
\end{align}
\item There is a neighbourhood $ U\subseteq G $ of the identity such that the set $ U\cdot x\subseteq M $ is a
$ k $-dimensional submanifold and
\begin{align}
T_{x}(U\cdot x)=T_{e}\Phi_{x}(T_{e}G).
\end{align}
\item If the orbit $ G\cdot x\subseteq G $ is a submanifold, then $ \dim(G\cdot x)=k $.
\item If $ G $ is compact, then $ G\cdot x\subseteq M $ is a submanifold of dimension $ k $.
\end{compactenum}
\end{theorem}

\begin{proof}
The proof is inspired by \cite[Theorem~I.2.1~and~Theorem~I.2.3]{onishchik1993lie}.
Fix a point $ x\in M $. Because $ \Phi_{x} $ is smooth we can differentiate it at $ g\in G $ and obtain a linear map
$ T_{g}\Phi_{x}\colon
T_{g}G\rightarrow T_{\Phi_{x}(g)}M $ between vector spaces and a natural number $ k\in\left\lbrace
0,\ldots,\dim M\right\rbrace $ such that
\begin{align}
\text{rank}_{g}\Phi_{x}=\dim\text{im}(T_{g}\Phi_{x})=k.
\end{align}
We show that $ k $
is independent of the element $ g\in G $. Let $ h\in G $ be another element. By the action property one has
\begin{align}\label{eq01t}
\Phi_{h}(\Phi_{x}(g))=\Phi(h,\Phi(g,x))=\Phi(hg,x)=\Phi_{x}(L_{h}(g)),
\end{align}
where we denoted the left multiplication with the element $ h $ on $ G $ by $ L_{h}\colon G\ni g\mapsto hg\in G $.
Since all the involved mappings are smooth, differentiating Eq. (\ref{eq01t}) gives by the chain rule
\begin{align}\label{eq02t}
T_{\Phi_{x}(g)}\Phi_{h}~T_{g}\Phi_{x}=T_{L_{h}(g)}\Phi_{x}~T_{g}L_{h}=T_{hg}\Phi_{x}~T_{g}L_{h}.
\end{align}
Since $ L_{h} $ is a diffeomorphism it is also a local diffeomorphism, i.e. we have $ \text{im}T_{g}L_{h}=T_{hg}G $
and since $ \Phi_{h} $ is a diffeomorphism we get by the same argument
\begin{align}
\dim\text{im}(T_{\Phi_{x}(g)}\left.\Phi_{h}\right|_{\text{im}T_{g}\Phi_{x}})=\dim(\text{im}T_{g}\Phi_{x})=k.
\end{align}
These two observations together with Eq. (\ref{eq02t}) imply that
$ \text{rank}_{hg}\Phi_{x}=\dim\text{im}T_{hg}\Phi_{x}
=\dim\text{im}T_{g}\Phi_{x}=\text{rank}_{g}\Phi_{x} $. Thus if we set $ h=g'g^{-1} $ we see that the rank of
$ \Phi_{x} $ at $ g $ coincides with the rank of $ \Phi_{x} $ at any other point $ g'\in G $. Then the rank is indeed
constant and the first claim is proved. As argumented before the Constant Rank Theorem implies then that the preimage
$ \Phi_{x}^{-1}(\left\lbrace x\right\rbrace)=G_{x} $ is an embedded submanifold of $ G $ of codimension $ k $ and
$ T_{e}G_{x}=T_{e}\Phi_{x}^{-1}(\left\lbrace x\right\rbrace)=\ker T_{e}\Phi_{x} $. Because the stabilizer is a closed
subgroup, as the preimage of a closed subset $ \left\lbrace x\right\rbrace\subseteq M $ under a smooth function, it is
also a Lie subgroup of $ G $ (c.f. \cite[Theorem~20.10]{lee2003introduction}).
This proves the second claim. Now since $ \Phi_{x} $ is of constant rank, the third
statement is just another consequence of the Constant Rank Theorem (c.f. \cite[Theorem~7.13]{lee2003introduction}).
Then the next statement follows since there is a countable cover of $ G\cdot x $ of $ k $-dimensional
submanifolds according to iii.). Finally,
according to iii.) there is a neighbourhood $ U\subseteq G $ of the identity such that
$ U\cdot x\subseteq M $ is a submanifold. The set $ C=G\setminus(UG_{x}) $ fulfils $ (U\cdot x)\cap(C\cdot x)=\emptyset
$ and $ G\cdot x=(U\cdot x)\cup(C\cdot x) $. To prove the first property assume there is an element
$ x'\in(U\cdot x)\cap(C\cdot x) $. Then, on one hand there is a $ g\in U $ such that $ \Phi_{x}(g)=x' $ and on the
other hand there is a $ h\in C $ such that $ \Phi_{x}(h)=x' $. But this means $ \Phi_{g}(x)=\Phi_{x}(g)=\Phi_{x}(h)
=\Phi_{h}(x) $ which implies $ \Phi_{g^{-1}h}(x)=x $, i.e. $ g^{-1}h\in G_{x} $. Thus $ h\in gG_{x}\subseteq UG_{x} $
which gives the contradiction. For the second property take $ x'\in G\cdot x $ and assume that $ x'\notin U\cdot x $,
i.e. there is a $ g\in G $ such that $ x'=\Phi_{x}(g) $, but there is no $ u\in U $ such that $ x'=\Phi_{x}(u) $. This
is
exactly the condition $ x'\in C\cdot x $ with $ C=G\setminus(UG_{x}) $. Since $ U $ is open and $ G_{x} $ closed the
set $ UG_{x}\subseteq G $ is open. That means $ C $ is closed and for this compact as a subset of a compact set $ G $.
We already mentioned that $ \Phi_{x} $ is continuous, so it is well known that $ C\cdot x=\Phi_{x}(C)\subseteq M $ is
still compact. We proved that the intersection of $ G\cdot x $ with the open set $ M\setminus C\cdot x\subseteq M $ is
a submanifold, i.e. $ (G\cdot x)\cap(M\setminus C\cdot x)=U\cdot x $ is a submanifold. This construction is
independent of the point $ x\in M $. This concludes the proof.
\end{proof}
Remark that $ G\cdot x $ is not always a submanifold of $ M $. One might think of the \textit{dense
winding of the torus} (c.f. \cite[page~14]{onishchik1993lie}).

We want to stress that there is also a notion of \textit{right actions} which works completely
analog to left actions. One calls $ \Phi\colon G\times M\rightarrow M $ a right action of $ G $ on $ M $ if for every
$ x\in M $ and $ g,h\in G $ one has $ \Phi(gh,x)=\Phi(h,\Phi(g,x)) $ and $ \Phi(e,x)=x $. We adapt the same notations
and results we developed for left actions. It is intuitive to denote the orbit of an element $ x\in M $ in the
case of a right action by $ x\cdot G $. In the following sections we need both kinds of actions but left actions more
frequently. For this reason, we refer to left actions just as actions.

To conclude this section we mention the correspondence of Lie group actions and Lie algebra actions. A (left)
\textit{Lie algebra action} of a Lie algebra $ \mathfrak{g} $ on $ M $ is an anti-homomorphism
\begin{align}
\phi\colon\mathfrak{g}\rightarrow\Gamma^{\infty}(TM)
\end{align}
of Lie algebras, while a right Lie algebra action is a homomorphism. Consider a Lie group action
$ \Phi\colon G\times M\rightarrow M $ and fix a $ x\in M $. One can check that
\begin{align}\label{LieAlgebraAction}
\phi|_{x}=T_{e}\Phi_{x}\colon\mathfrak{g}\rightarrow T_{x}M
\end{align}
determines a Lie algebra action $ \phi $ of the Lie algebra $ \mathfrak{g} $ corresponding to $ G $ on $ M $
(see \cite[Section~6.2]{michortopics}). Then one calls $ \phi $ the Lie algebra action corresponding to $ \Phi $.
A Lie group action $ \Phi $ is said to be \textit{locally transitive} if for any $ x\in M $ the map defined in Eq.
(\ref{LieAlgebraAction}) is surjective.
For $ \xi\in\mathfrak{g} $ one calls $ \xi_{M} $ defined for any $ x\in M $ by
\begin{align}
\xi_{M}(x)=\left.\frac{\mathrm{d}}{\mathrm{d}t}\right|_{t=0}\Phi_{\exp(t\xi)}(x)=T_{e}\Phi_{x}\xi\in T_{x}M
\end{align}
the \textit{fundamental vector field} of $ \xi $ on $ M $. Conversely, if the flow of every vector field
$ \phi(\xi)\in\Gamma^{\infty}(TM) $ of a Lie algebra action $ \phi $ of $ \mathfrak{g} $ on $ M $ is complete
there is a Lie group action $ \Phi $ of the connected Lie group corresponding to $ \mathfrak{g} $ on $ M $ such
that Eq. (\ref{LieAlgebraAction}) holds. This is a famous theorem by R. Palais (c.f. \cite[Theorem~6.5]{michortopics}).
In this case the Lie algebra action $ \phi $ is said to \textit{integrate} to the Lie group action $ \Phi $.

\section{Homogeneous Spaces}\label{sectionhomspace}

We start this section with the definition of a homogeneous spaces, an essential object of this thesis. After
clarifying the
smooth structure on homogeneous spaces we discover that there is always a natural transitive action on them and
the corresponding stabilizers can be easily calculated. This can also be done by using an infinitesimal point of view.
What follows is the most important classification of homogeneous spaces: they are induced by transitive actions, which
help us to detect a manifold $ M $ as a homogeneous space. Of course this
is also a way to find obstructions to homogeneous structures on manifolds. Relevant sources are
\cite[Chapter~8]{baker2002matrix}, \cite[Chapter~II]{helgason2001differential},
\cite[Chapter~2]{kirillov2008introduction}, and \cite[Chapter~2]{michortopics} and as in the section before we are
very close to \cite[Part~I~Chapter~1]{onishchik1993lie}.

Let $ H\subseteq G $ be a closed subgroup of the Lie group $ G $ and consider the right multiplication
$ R\colon H\times
G\ni(h,g)\mapsto gh\in G $. It is known that a closed subgroup of a Lie group is indeed a Lie subgroup (c.f.
\cite[Theorem~20.10]{lee2003introduction}), thus $ R $ is a right action of $ H $ on $ G $. The orbits of this right
action are the left cosets of $ H $ in $ G $, i.e. $ g\cdot H=\left\lbrace g'\in G~|~\exists h\in H
\text{ such that }g'=gh\right\rbrace $. Surprisingly, there is always a smooth structure on these set of left cosets.
One can prove this claim directly (c.f. \cite[Theorem~I.3.1]{onishchik1993lie}), but we want to introduce a
nice tool such that the proof becomes quite trivial: the \textit{Free and Proper Action Theorem}
(c.f. \cite[Corollary~6.5.1]{rudolph2012differential}). Remember that a (left or right) action
$ \Phi\colon G\times M\rightarrow M $ is said to be \textit{proper} if the extended map $ \overline{\Phi}\colon
G\times M\ni(g,x)\mapsto(\Phi(g,x),x)\in M\times M $ is proper, i.e. if the preimage of all compact subsets of
$ M\times M $ under $ \overline{\Phi} $ is again compact in $ G\times M $.

\begin{theorem}[Free and Proper Action Theorem]
Let $ \Phi\colon G\times M\rightarrow M $ be a (left or right) action that is free and proper. Then there is a unique
way to structure the set of orbits $ M/G $ as a smooth manifold such that the natural projection
\begin{align}
\operatorname{pr}\colon M\rightarrow M/G
\end{align}
is a smooth surjective submersion.
\end{theorem}
Now it is easy to prove that the set $ G/H $ of orbits of $ R $ has the structure of a smooth manifold. One just
has to check that $ R $ is free and proper.

\begin{theorem}\label{hommfd}
Let $ H $ be a closed subgroup of $ G $. Then there is a unique way to structure the orbit space $ G/H $ as a
smooth manifold such that the map
\begin{align}\label{eq03t}
\operatorname{pr}\colon G\ni g\mapsto g\cdot H\in G/H
\end{align}
is a smooth surjective submersion.
\end{theorem}

\begin{proof}
We already mentioned that the right multiplication $ R\colon H\times G\ni(h,g)\mapsto gh\in G $ is a right action of
the Lie subgroup $ H $ on $ G $ and the corresponding orbits are the left cosets of $ H $ in $ G $. The action $ R $
is free, since for any $ h_{1},h_{2}\in H $ and $ g\in G $ the equation $ gh_{1}=R(h_{1},g)=R(h_{2},g)=gh_{2} $ implies
$ h_{1}=h_{2} $. To prove that $ R $ is proper consider $ \overline{R}\colon H\times G
\ni(h,g)\mapsto(gh,g)\in G\times G $ and an arbitrary compact subset $ K\subseteq G\times G $. If we also consider
the usual right multiplication $ r\colon G\times G\ni(g_{1},g_{2})\mapsto g_{2}g_{1}\in G $ on $ G $ and
$ \overline{r}\colon G\times G\ni(g_{1},g_{2})\mapsto(g_{2}g_{1},g_{2})\in G\times G $ one has
\begin{align}\label{eq06t}
\overline{R}^{-1}(K)=(H\times G)\cap\overline{r}^{-1}(K).
\end{align}
Now $ \overline{r} $ is a proper map since it is a diffeomorphism and for this
$ \overline{r}^{-1}(K)\subseteq G\times G $ is
compact. By assumption $ H\subseteq G $ is closed and so is $ H\times G\subseteq G\times G $. Then Eq. (\ref{eq06t})
implies that $ \overline{R}^{-1}(K) $ is compact. This concludes the proof.
\end{proof}
After discovering the smooth structure on $ G/H $ we introduce the following

\begin{definition}[Homogeneous Space]
A homogeneous space is the set $ G/H $ of left cosets of a closed subgroup $ H $ of a Lie group $ G $ endowed
with the unique smooth structure that exists according to Theorem \ref{hommfd}.
\end{definition}

Let us discuss homogeneous spaces in detail. First we want to see that there is a natural action of
$ G $ on $ G/H $ and that is action is transitive. The second step will be a kind of converse statement, namely
that any transitive action on a manifold makes it into a homogeneous space. To achieve it, we need some preparation.

The statement that Eq. (\ref{eq03t}) defines a surjective submersion is very beneficial and we will use this in some
proofs to get smoothness of maps. To do so we refer to the

\begin{theorem}[Surjective Submersion Theorem]\label{ThmSurjSubm}
Let $ M,N $ and $ P $ be smooth manifolds and $ \pi\colon M\rightarrow N $ a smooth surjective submersion. Then
a map $ F\colon N\rightarrow P $ is smooth if and only if $ F\circ\pi $ is smooth.
\begin{equation}
\begin{tikzpicture}
  \matrix (m) [matrix of math nodes,row sep=3em,column sep=0.5em,minimum width=2em]
  {
      M & & &  \\
      N & & & P \\};
  \path[-stealth]
    (m-1-1) edge node [left] {$\pi$} (m-2-1)
    (m-1-1) edge node [right] {$F\circ\pi$} (m-2-4)
    (m-2-1) edge node [below] {$F$} (m-2-4);
\end{tikzpicture}
\end{equation}
\end{theorem}
For a proof consider \cite[Theorem~4.29]{lee2003introduction}.

\begin{proposition}\label{prophomspace}
Let $ H $ be a closed subgroup of $ G $. There is a natural transitive action
\begin{align}\label{eq07t}
G\times G/H\ni(g,g'\cdot H)\mapsto(gg')\cdot H\in G/H
\end{align}
of $ G $ on $ G/H $ induced by the left multiplication on $ G $. For $ g\cdot H\in G/H $ the stabilizer is
\begin{align}
G_{g\cdot H}=\text{Conj}_{g}(H)=gHg^{-1}.
\end{align}
Moreover, if we denote the Lie algebras corresponding to $ G $ and $ H $, respectively, by $ \mathfrak{g} $ and
$ \mathfrak{h} $ the tangent map at $ e\in G $ of the map $ \operatorname{pr} $ defined in Eq. (\ref{eq03t}) is an
isomorphism and
\begin{align}\label{eq08t}
T_{e}\operatorname{pr}\colon\mathfrak{g}/\mathfrak{h}\rightarrow T_{\operatorname{pr}(e)}(G/H).
\end{align}
\end{proposition}

\begin{proof}
First, Theorem \ref{hommfd} states that $ G/H $ is a manifold. Denote the map that is defined in Eq. (\ref{eq07t}) by
$ \Phi $ and let $ g_{1},g_{2}\in G $ and $ g\cdot H\in G/H $ be arbitrary. Then
$ \Phi(e,g\cdot H)=(eg)\cdot H=g\cdot H $ and
\begin{align*}
\Phi(g_{1},\Phi(g_{2},g\cdot H))=(g_{1}(g_{2}g))\cdot H=((g_{1}g_{2})g)\cdot H=\Phi(g_{1}g_{2},g\cdot H),
\end{align*}
i.e. $ \Phi $ is a set theoretic action. Since the multiplication $ \mu\colon G\times G\ni(g_{1},g_{2})\mapsto
g_{1}g_{2}\in G $ is smooth and $ \text{pr} $ defined in Eq. (\ref{eq03t}) is smooth according to Theorem \ref{hommfd},
the map
\begin{align}
\text{pr}\circ\mu\colon G\times G\ni(g_{1},g_{2})\mapsto(g_{1}g_{2})\cdot H\in G/H
\end{align}
is smooth and the map
\begin{align}
\text{id}\times\text{pr}\colon G\times G\ni(g_{1},g_{2})\mapsto(g_{1},g_{2}\cdot H)\in G\times(G/H)
\end{align}
is a smooth surjective submersion. According to Theorem \ref{ThmSurjSubm} the smoothness of $ \Phi $
follows from the commutativity of the diagram
\begin{equation}
\begin{tikzpicture}
  \matrix (m) [matrix of math nodes,row sep=3em,column sep=0.5em,minimum width=2em]
  {
      G\times G & & & \\
      G\times(G/H) & & & G/H, \\};
  \path[-stealth]
    (m-1-1) edge node [left] {$\text{id}\times\text{pr}$} (m-2-1)
    (m-1-1) edge node [right] {$\text{pr}\circ\mu$} (m-2-4)
    (m-2-1) edge node [below] {$\Phi$} (m-2-4);
\end{tikzpicture}
\end{equation}
which holds since
\begin{align}
(\text{pr}\circ\mu)(g_{1},g_{2})=(g_{1}g_{2})\cdot H=(\Phi\circ(\text{id}\times\text{pr}))(g_{1},g_{2}),
\end{align}
for all $ g_{1},g_{2}\in G $. Thus $ \Phi $ defined in Eq.
(\ref{eq07t}) is indeed an action. It is also transitive since for arbitrary $ g_{1}\cdot H, g_{2}\cdot H\in
G/H $ the element $ g_{2}g_{1}^{-1}\in G $ satisfies $ \Phi(g_{2}g_{1}^{-1},g_{1}\cdot H)=g_{2}\cdot H $.
If we take now an arbitrary element $ g\cdot H $, the corresponding stabilizer is given by
\begin{align}
G_{g\cdot H}=\left\lbrace g'\in G~|~g\cdot H=\Phi(g',g\cdot H)=(g'g)\cdot H\right\rbrace.
\end{align}
We see that if this condition is satisfied there is an element $ h\in H $ such that $ gh=g'g $, which is
equivalent to $ g'=ghg^{-1}=\text{Conj}_{g}(h)\in\text{Conj}_{g}(H) $. Thus
$ G_{g\cdot H}\subseteq\text{Conj}_{g}(H) $. Conversely, any $ h\in H $ satisfies
\begin{align}
\Phi(ghg^{-1},g\cdot H)=(ghg^{-1}g)\cdot H=(gh)\cdot H=g\cdot H.
\end{align}
This implies $ G_{g\cdot H}=\text{Conj}_{g}(H) $. Let us come to the
infinitesimal part. Since $ \text{pr}\colon G\rightarrow G/H $ is a submersion the map $ T_{e}\text{pr}\colon
T_{e}G\rightarrow T_{\text{pr}(e)}(G/H) $ is surjective. By assumption $ T_{e}G=\mathfrak{g} $. We first show that
$ \mathfrak{h}=\ker T_{e}\text{pr} $. For any $ \xi\in\mathfrak{h} $, one has $ \exp(t\xi)\in H $ for all $
t\in\mathbb{R} $ and then $ \text{pr}(\exp(t\xi))=\text{pr}(e) $ for all $ t\in\mathbb{R} $. Thus by the chain rule
we have
\begin{align}
T_{e}\text{pr}(\xi)=T_{e}\text{pr}\left.\frac{\mathrm{d}}{\mathrm{d}t}\right|_{t=0}\exp(t\xi)
=\left.\frac{\mathrm{d}}{\mathrm{d}t}\right|_{t=0}\text{pr}
(\exp(t\xi))=\left.\frac{\mathrm{d}}{\mathrm{d}t}\right|_{t=0}\text{pr}(e)=0.
\end{align}
We showed $ \mathfrak{h}\subseteq\ker T_{e}\text{pr} $. In fact, this is sufficient since $ \text{pr}\colon
G\rightarrow G/H $ is submersive and the dimension formula implies
\begin{align*}
\dim(H)=\dim(G)-\dim(G/H)=\dim(G)-(\dim(G)-\dim\ker\text{pr})=\dim\ker\text{pr}.
\end{align*}
Then $ \mathfrak{h}=\ker T_{e}\text{pr} $ follows. By dividing by the kernel one gets the isomorphism
defined in Eq. (\ref{eq08t}).
\end{proof}

In the last theorem we have seen that there is a natural transitive action of $ G $ on any homogeneous space $ G/H $.
The next theorem gives the converse statement: if $ G $ acts transitively on $ M $ then $ M $ can be regarded as a
homogeneous space $ G/H $ where $ H $ is the stabilizer of an arbitrary point of $ M $.

\begin{theorem}\label{homspacetransact}
Let $ \Phi\colon G\times M\rightarrow M $ be a transitive action. Then for any $ x\in M $ the map
\begin{align}\label{eq10t}
\beta_{x}\colon G/G_{x}\ni g\cdot G_{x}\mapsto\Phi(g,x)\in M
\end{align}
is a diffeomorphism, which commutes with the action of $ G $ on $ G/G_{x} $. Moreover, all stabilizers are isomorphic
via
\begin{align}\label{eq09t}
\operatorname{Conj}_{g}\colon G_{x}\rightarrow G_{\Phi_{g}(x)},
\end{align}
where $ x\in M $.
\end{theorem}

\begin{proof}
Parts of the proof and the notation are inspired by \cite[Theorem~I.3.3]{onishchik1993lie}.
Let $ x\in M $ be an arbitrary point. We already mentioned that $ G_{x}\subseteq G $ is a closed subset. Then $ G/G_{x}
$ is a homogeneous space according to Theorem \ref{hommfd}. If we take any $ g\in G $ and $ h\in G_{x} $ we
know on the one hand that $ \beta_{x}(h\cdot G_{x})=\Phi(h,x)=x $ and on the other hand that
\begin{align}
\beta_{x}((gh)\cdot G_{x})=\Phi(gh,x)=\Phi(g,\Phi(h,x))=\Phi(g,x)=\beta_{x}(g\cdot G_{x}),
\end{align}
which proves that the map $ \beta_{x} $ defined in Eq. (\ref{eq10t}) is well-defined. The surjectivity of
$ \beta_{x} $ follows directly form the
transitivity of $ \Phi $, i.e. for any $ y\in M $ there is a $ g\in G $ such that $ \beta_{x}(g\cdot G_{x})=
\Phi(g,x)=y $. To prove the injectivity of $ \beta_{x} $ we take $ g\cdot G_{x},h\cdot G_{x}\in G/G_{x} $ and assume
that $ \Phi(g,x)=\beta_{x}(g\cdot G_{x})=\beta_{x}(h\cdot G_{x})=\Phi(h,x) $. This implies
\begin{align}
\Phi(g^{-1}h,x)=\Phi(g^{-1},\Phi(h,x))=\Phi(g^{-1},\Phi(g,x))=\Phi(e,x)=x,
\end{align}
i.e. $ g^{-1}h\in G_{x} $ from which $ h\in g\cdot G_{x} $ follows. This is the injectivity $ g\cdot G_{x}=
h\cdot G_{x} $. The smoothness of $ \beta_{x} $ follows similarly to the proof of Proposition \ref{prophomspace} from
Theorem \ref{ThmSurjSubm} and the commutativity of the diagram
\begin{equation}
\begin{tikzpicture}
  \matrix (m) [matrix of math nodes,row sep=3em,column sep=0.5em,minimum width=2em]
  {
      G & & & \\
      G/G_{x} & & & M,\\};
  \path[-stealth]
    (m-1-1) edge node [left] {$\text{pr}$} (m-2-1)
    (m-1-1) edge node [right] {$\Phi_{x}$} (m-2-4)
    (m-2-1) edge node [below] {$\beta_{x}$} (m-2-4);
\end{tikzpicture}
\end{equation}
since the projection is a smooth surjective submersion and $ \Phi_{x}\colon G\rightarrow M $
is smooth. The next step is to show that $ \beta_{x} $ is a bijective immersion. Then it has to be a diffeomorphism
since we have already proven its smoothness (c.f. \cite[Corollary~7.11]{lee2003introduction}). Thus consider
\begin{align}\label{eq11t}
T_{\text{pr}(e)}\beta_{x}\colon T_{\text{pr}(e)}G/G_{x}\rightarrow T_{x}M.
\end{align}
Denote the Lie algebra corresponding to $ G_{x} $ by $ \mathfrak{g}_{x} $ (remember that $ G_{x} $ is a Lie
group according to Theorem \ref{thmsubmfd} ii.)) and take $ \xi\in\mathfrak{g} $. We can
identify a equivalence class $ \left[\xi\right]\in\mathfrak{g}/\mathfrak{g}_{x} $ with
$ T_{e}\text{pr}\left[\xi\right]\in
T_{\text{pr}(e)}G/G_{x} $ according to Proposition \ref{prophomspace}. Thus if we assume $ \left[\xi\right]\in\ker
T_{\text{pr}(e)}\beta_{x} $, one has $ \exp(t\xi)\in G_{x} $ for any $ t\in\mathbb{R} $, since $ \exp(0\xi)=e
\in G_{x} $ and
\begin{align*}
\frac{d}{dt}\Phi_{x}(\exp(t\xi))&=\left.\frac{\mathrm{d}}{\mathrm{d}s}\right|_{s=0}\Phi_{\exp(t\xi)}
(\Phi_{\exp(s\xi)}(x))\\
&=T_{x}\Phi_{\exp(t\xi)}\Bigg(\left.\frac{\mathrm{d}}{\mathrm{d}s}\right|_{s=0}\Phi_{\exp(s\xi)}(x)\Bigg)\\
&=T_{x}\Phi_{\exp(t\xi)}\Bigg(\left.\frac{\mathrm{d}}{\mathrm{d}s}\right|_{s=0}\beta_{x}(\exp(s\xi)\cdot G_{x})\Bigg)\\
&=T_{x}\Phi_{\exp(t\xi)}(T_{\text{pr}(e)}\beta_{x}\left[\xi\right])\\
&=0.
\end{align*}
Thus $ \xi\in\mathfrak{g}_{x} $ for which $ \left[\xi\right]=0 $ and the map defined in (\ref{eq11t}) is injective.
For $ g\in G $, the derivative of the equality
\begin{align}
\beta_{x}(g\cdot G_{x})=\Phi(g,x)=\Phi(g,\Phi(e,x))=\Phi_{g}\circ\beta_{x}(e\cdot G_{x})
=\Phi_{g}\circ\beta_{x}\circ\ell_{g^{-1}}(g\cdot G_{x})
\end{align}
gives
\begin{align}
T_{\text{pr}(g)}\beta_{x}=T_{x}\Phi_{g}\circ T_{\text{pr}(e)}\beta_{x}\circ T_{g}\ell_{g^{-1}},
\end{align}
according to the chain rule, where $ \ell_{g^{-1}}\colon G/G_{x}\ni g'\cdot G_{x}\mapsto(g^{-1}g')\cdot G_{x}
\in G/G_{x} $. Since $ \Phi_{g} $ and $ \ell_{g^{-1}} $ are diffeomorphisms this implies the injectivity of
$ T_{\operatorname{pr}(g)}\beta_{x} $.
As we argued this is enough to show that $ \beta_{x} $ is a diffeomorphism since the involved spaces are second
countable.
Moreover, $ \beta_{x} $ is equivariant, i.e. it commutes with $ \Phi $, since for any $ g\in G $ and $ h\cdot G_{x}\in
G/G_{x} $ one has
\begin{align*}
\Phi_{g}(\beta_{x}(h\cdot G_{x}))=\Phi_{g}(\Phi_{h}(x))=\Phi_{gh}(x)=\beta_{x}((gh)\cdot G_{x})=\beta_{x}
(\tilde{\Phi}_{g}(h\cdot G_{x})),
\end{align*}
where $ \tilde{\Phi} $ denotes the action of $ G $ on $ G/G_{x} $ defined by Eq. (\ref{eq07t}). Thus $ \Phi_{g}\circ
\beta_{x}=\beta_{x}\circ\tilde{\Phi}_{g} $.
To prove that all stabilizers are isomorphic via the map defined in Eq. (\ref{eq09t}) we choose two points
$ x,y\in M $. Since $ \Phi $ is transitive there is a $ g\in G $ such that $ \Phi_{x}(g)=y $. If we take
$ g'\in G_{x} $, one has
\begin{align*}
\Phi_{\Phi_{g}(x)}(\text{Conj}_{g}(g'))&=\Phi_{\Phi_{g}(x)}(gg'g^{-1})=\Phi(gg'g^{-1},\Phi_{g}(x))=\Phi(gg',x)\\
&=\Phi(g,\Phi(g',x))=\Phi(g,x)=\Phi_{g}(x),
\end{align*}
which means that $ \text{Conj}_{g}(g')\in G_{\Phi_{g}(x)} $. Thus the map defined in Eq. (\ref{eq09t}) is well-defined.
$ \text{Conj}_{g} $ is
obviously smooth and also invertible with inverse $ \text{Conj}_{g^{-1}}\colon G_{\Phi_{g}(x)}\rightarrow G_{x} $.
Indeed, if $ g'\in G_{\Phi_{g}(x)} $ one gets
\begin{align*}
\Phi(\text{Conj}_{g^{-1}}(g'),x)&=\Phi(g^{-1}g'g,x)=\Phi(g^{-1},\Phi(g'g,x))=\Phi(g^{-1},\Phi(g',\Phi(g,x)))\\
&=\Phi(g^{-1},\Phi_{g}(x))=\Phi(e,x)=x.
\end{align*}
Then $ \text{Conj}_{g^{-1}} $ is a well-defined mapping from $ G_{\Phi_{g}(x)} $ to $ G_{x} $. Of course
$ \text{Conj}_{g}\circ\text{Conj}_{g^{-1}}=\text{id}_{G_{\Phi_{g}(x)}} $ and
$ \text{Conj}_{g^{-1}}\circ\text{Conj}_{g}=\text{id}_{G_{x}} $. Thus (\ref{eq09t}) is indeed an isomorphism
and $ G_{\Phi_{g}(x)}=G_{y} $.
\end{proof}

We want to illustrate this construction by considering three examples of transitive Lie group actions on the
$ 2 $-sphere
\begin{align}
\mathbb{S}^{2}=\left\lbrace(x_{1},x_{2},x_{3})\in\mathbb{R}^{3}~|~x_{1}^{2}+x_{2}^{2}+x_{3}^{2}=1\right\rbrace.
\end{align}

\begin{example}\label{exampleSO3}
The \textbf{special orthogonal group} in three dimensions (also called the \textbf{rotation group})
$ \operatorname{SO}(3)=\left\lbrace A\in M_{3\times 3}(\mathbb{R})~|~AA^{T}=\mathbb{1}=A^{T}A
\text{ and }\det(A)=1\right\rbrace $
is a connected compact $ 3 $-dimensional Lie group. Consider the map
\begin{align}\label{eq04t}
\Phi\colon\operatorname{SO}(3)\times\mathbb{S}^{2}\ni(A,p)\mapsto\Phi(A,p)=Ap\in\mathbb{S}^{2},
\end{align}
where the right side of Eq. (\ref{eq04t}) denotes the matrix-vector multiplication. It is well-defined, since
\begin{align*}
||Ap||^{2}=\langle Ap,Ap\rangle=\langle A^{T}Ap,p\rangle=\langle p,p\rangle=||p||^{2}=1,
\end{align*}
for any $ A\in\operatorname{SO}(3) $ and $ p\in\mathbb{S}^{2} $. The smoothness of the linear matrix-vector
multiplication induces the smoothness of $ \Phi $. Moreover, $ \Phi $ is a Lie group action of
$ \operatorname{SO}(3) $ on $ \mathbb{S}^{2} $ since for any $ A,B\in\operatorname{SO}(3) $ and
$ p\in\mathbb{S}^{2} $ one has
\begin{align*}
\Phi(A,\Phi(B,p))=A(Bp)=(AB)p=\Phi(AB,p)\text{ and }\Phi(\mathbb{1},p)=\mathbb{1}p=p,
\end{align*}
where $ \mathbb{1}\in\operatorname{SO}(3) $ denotes the identity matrix. This action is even transitive:
choose two points $ p\neq q $ on $ \mathbb{S}^{2} $. Then there is a unique plane $ E $ in $ \mathbb{R}^{3} $
such that $ 0,p,q\in E $, where $ 0=(0~0~0)^{T}\in\mathbb{R}^{3} $. Choose a vector $ 0\neq n\in\mathbb{R}^{3} $
perpendicular to $ E $. There is an angle $ \alpha\in\left[ 0,2\pi\right] $ such that a rotation of $ E $ around
$ n $ transfers $ p $ onto $ q $. After a change of coordinates we can assume that this is a rotation around the
$ z $-axis and $ p $ lies on the $ x $-axis. Then
\begin{align}
A=\begin{pmatrix}
\cos(\alpha) & -\sin(\alpha) & 0 \\
\sin(\alpha) & \cos(\alpha) & 0 \\
0 & 0 & 1
\end{pmatrix}
\in\text{SO}(3)
\end{align}
satisfies $ Ap=q $. If $ p=q\in\mathbb{S}^{2} $ one simply has $ \mathbb{1}p=q $. Altogether $ \Phi $ is transitive.
Since any $ \mathbb{1}\neq A\in\operatorname{SO}(3) $ has an eigenvalue $ 1\neq\lambda\in\mathbb{C} $, only
$ \mathbb{1} $ induces the identity diffeomorphism on $ \mathbb{S}^{2} $ and $ \Phi $ is effective in addition.
We want to compute the stabilizer at the point $ e_{1}=(1~0~0)^{T}\in\mathbb{S}^{2} $. Thus we are
searching for the elements $ A=\begin{pmatrix}
a & b & c \\
d & e & f \\
g & h & i
\end{pmatrix}\in\text{SO}(3) $ that satisfy $ Ae_{1}=e_{1} $. The matrix $ A $ preserves $ e_{1} $ if and only if
$ a=1 $ and $ d=g=0 $.
Then $ A^{T}A=\mathbb{1} $ implies $ b=c=0 $. The conditions remaining for $ e,h,f,i $
are those for $ B=\begin{pmatrix}
e & f \\
h & i
\end{pmatrix}
\in M_{2\times 2}(\mathbb{R}) $ to be an element of $ \text{SO}(2) $. This shows
\begin{align}
\text{SO}(3)_{e_{1}}=\begin{pmatrix}
1 & \vec{0}^{T} \\
\vec{0} & \text{SO}(2)
\end{pmatrix},
\end{align}
where $ \vec{0}=(0~0)^{T}\in\mathbb{R}^{2} $. Thus we can identify $ \text{SO}(3)_{e_{1}} $ with $ \text{SO}(2) $.
According to Theorem \ref{homspacetransact} the $ 2 $-sphere $ \mathbb{S}^{2} $ is a homogeneous space and there is a
diffeomorphism
\begin{align}
\text{SO}(3)/\text{SO}(2)\cong\mathbb{S}^{2}.
\end{align}
More explicit, the diffeomorphism reads
\begin{align}
\beta_{e_{1}}\colon\operatorname{SO}(3)/\operatorname{SO}(2)\ni A\cdot\begin{pmatrix}
1 & \vec{0}^{T} \\
\vec{0} & \text{SO}(2)
\end{pmatrix}\mapsto
Ae_{1}\in\mathbb{S}^{2}.
\end{align}
\end{example}

A non-compact Lie group containing the compact Lie group $ \text{SO}(3) $ is the special linear group
\begin{align}
\text{SL}(3,\mathbb{R})=\left\lbrace A\in M_{3\times 3}(\mathbb{R})~|~\det(A)=1\right\rbrace.
\end{align}
One can extend the linear action of $ \text{SO}(3) $ on $ \mathbb{S}^{2} $ to
$ \text{SL}(3,\mathbb{R}) $ as we can see in the following

\begin{example}\label{exampleSL3}
Consider the map $ \Phi\colon\text{SL}(3,\mathbb{R})\times\mathbb{S}^{2}\rightarrow\mathbb{S}^{2} $ defined by
\begin{align}\label{ActionSL3}
\Phi(A,p)=\frac{Ap}{||Ap||}.
\end{align}
The denominator of Eq. (\ref{ActionSL3}) is not zero since $ A\in\text{SL}(3,\mathbb{R}) $ is invertible and
$ 0\notin\mathbb{S}^{2} $. Moreover, $ \frac{Ap}{||Ap||}\in\mathbb{S}^{2} $ since $ ||\frac{Ap}{||Ap||}||=1 $. Thus
$ \Phi $ is well-defined. It is smooth since all involved mappings are smooth and it is a Lie group action because
for all $ x\in\mathbb{S}^{2} $ and $ A,B\in\operatorname{SL}(3,\mathbb{R}) $ one has
\begin{align*}
\Phi(\mathbb{1},x)=\frac{x}{||x||}=x\text{ and }\Phi(A,\Phi(B,x))=\Phi(A,\frac{Bx}{||Bx||})
=\frac{A\frac{Bx}{||Bx||}}{||A\frac{Bx}{||Bx||}||}=\frac{ABx}{||ABx||}=\Phi(AB,x).
\end{align*}
If we restrict $ \Phi $ to $ \text{SO}(3)\times\mathbb{S}^{2} $ the action reads
$ \Phi(A,p)=\frac{Ap}{||Ap||}=Ap $, since $ ||Ap||=1 $ as it was shown in the last example. Thus $ \Phi $ is indeed
an extension of the action in Example \ref{exampleSO3}. This shows immediately the transitivity of $ \Phi $.
It is also effective, since any $ \mathbb{1}\neq A\in\operatorname{SL}(3,\mathbb{R}) $ has an eigenvector
$ 1\neq\lambda\in\mathbb{C} $.
Consequently, $ \mathbb{S}^{2}\cong\text{SL}(3,\mathbb{R})/\text{SL}(3,\mathbb{R})_{p} $ for any
$ p\in\mathbb{S}^{2} $,
where the stabilizer $ \text{SL}(3,\mathbb{R})_{p} $ has to be $ 6 $-dimensional since $ \dim(\text{SL}(3,\mathbb{R}))
=8 $. Take for example $ e_{1}=(1~0~0)^{T}\in\mathbb{S}^{2} $. Then the corresponding stabilizer is
\begin{align*}
\operatorname{SL}(3,\mathbb{R})_{e_{1}}=\left.\left\lbrace \begin{pmatrix}
\frac{1}{\det(A)} & \alpha & \beta \\
0 & A_{11} & A_{12} \\
0 & A_{21} & A_{22}
\end{pmatrix}\in M_{3\times 3}(\mathbb{R})~\right|~
\begin{pmatrix}
A_{11} & A_{12} \\
A_{21} & A_{22}
\end{pmatrix}\in\operatorname{GL}(2,\mathbb{R}),~\alpha,\beta\in\mathbb{R}\right\rbrace.
\end{align*}
To prove this, consider an arbitrary matrix $ B=\begin{pmatrix}
B_{11} & B_{12} & B_{13} \\
B_{21} & B_{22} & B_{23} \\
B_{31} & B_{32} & B_{33}
\end{pmatrix}\in\operatorname{SL}(3,\mathbb{R}) $ and calculate $ ||Be_{1}||=\sqrt{B_{11}^{2}+B_{21}^{2}+B_{31}^{2}} $.
Then $ \frac{Be_{1}}{||Be_{1}||}=e_{1} $ implies $ B_{21}=B_{31}=0 $ and $ ||Be_{1}||=B_{11} $. The condition
$ \det(B)=1 $ forces $ \begin{pmatrix}
B_{22} & B_{23} \\
B_{32} & B_{33}
\end{pmatrix}\in\operatorname{GL}(2,\mathbb{R}) $ and $ B_{11}=\frac{1}{\det(B)} $. Counting degrees of freedom one
obtains $ \dim(\operatorname{SL}(3,\mathbb{R})_{e_{1}})=6 $. Then the diffeomorphism between
$ \operatorname{SL}(3,\mathbb{R})/\operatorname{SL}(3,\mathbb{R})_{e_{1}} $ and $ \mathbb{S}^{2} $ is
\begin{align}
\beta_{e_{1}}\colon\operatorname{SL}(3,\mathbb{R})/\operatorname{SL}(3,\mathbb{R})_{e_{1}}\ni
B\cdot\operatorname{SL}(3,\mathbb{R})_{e_{1}}\mapsto\frac{Be_{1}}{||Be_{1}||}\in\mathbb{S}^{2}.
\end{align}
\end{example}

For the next example we remind the reader of the notion of \textit{Minkowski space} which is common in
special relativity. Consider the matrix
\begin{align}\label{eq05t}
I_{1,3}=\begin{pmatrix}
1 & 0 & 0 & 0 \\
0 & -1 & 0 & 0 \\
0 & 0 & -1 & 0 \\
0 & 0 & 0 & -1
\end{pmatrix}
\in M_{4\times 4}(\mathbb{R}).
\end{align}
The space $ \mathbb{R}^{4} $ together with the \textit{Lorentz metric}
\begin{align}
\varphi_{1,3}\colon\mathbb{R}^{4}\times\mathbb{R}^{4}\ni((x^{0},...,x^{3}),(y^{0},...,y^{3}))\mapsto x^{0}y^{0}
-\sum\limits_{k=1}^{3}x^{k}y^{k}\in\mathbb{R}
\end{align}
associated to the matrix (\ref{eq05t}) is said to be the Minkowski space. The corresponding quadratic form is
\begin{align}
\varPhi_{1,3}\colon\mathbb{R}^{4}\ni(x^{0},...,x^{3})\mapsto(x^{0})^{2}-\sum\limits_{k=1}^{3}(x^{k})^{2}\in
\mathbb{R}.
\end{align}
The set $ \operatorname{O}(1,3,\mathbb{R}) $ of isometries of the Lorentz metric $ \varphi_{1,3} $ is said to be
the \textbf{Lorentz group}. Thus
$ \Lambda\in\operatorname{O}(1,3,\mathbb{R}) $ if $ \varphi_{1,3}(\Lambda x,\Lambda y)=\varphi_{1,3}(x,y) $ for all
$ x,y\in\mathbb{R}^{4} $. The set of matrices $ \Lambda\in\operatorname{O}(1,3,\mathbb{R}) $ that fulfil
$ \det(\Lambda)=1 $ is further denoted by $ \operatorname{SO}(1,3,\mathbb{R}) $. One often writes
$ x=\begin{pmatrix}
x^{0}\\
\vec{x}
\end{pmatrix}\in\mathbb{R}^{4} $ to indicate the \textit{time part} $ x^{0}\in\mathbb{R} $ and the
\textit{space part} $ \vec{x}\in\mathbb{R}^{3} $ of the vector $ x $.

\begin{example}\label{exampleSO31}
Consider the \textbf{proper orthochronous Lorentz group}
\begin{align}
\mathcal{L}_{3}^{\uparrow,+}=\left\lbrace\Lambda\in\operatorname{SO}(1,3,\mathbb{R})~|~\Lambda_{11}\geq 1\right\rbrace.
\end{align}
It is the connected component of the unit in $ \operatorname{O}(1,3,\mathbb{R}) $ and a $ 6 $-dimensional
non-compact Lie group. For vectors
$ \vec{x}\in\mathbb{S}^{2}\subseteq\mathbb{R}^{3} $ and matrices $ \Lambda\in\mathcal{L}_{3}^{\uparrow,+} $
one defines a map $ \Phi\colon\mathcal{L}_{3}^{\uparrow,+}\times\mathbb{S}^{2}\rightarrow\mathbb{S}^{2} $ by
\begin{align}
\Phi(\Lambda,\vec{x})=\frac{\operatorname{pr}\Bigg(\Lambda\begin{pmatrix}
1\\
\vec{x}
\end{pmatrix}\Bigg)}{\Bigg|\Bigg|\operatorname{pr}\Bigg(\Lambda\begin{pmatrix}
1\\
\vec{x}
\end{pmatrix}\Bigg)\Bigg|\Bigg|},
\end{align}
where $ \operatorname{pr} $ is the projection on the space part of a vector, i.e.
$ \operatorname{pr}((x^{0}~x^{1}~x^{2}~x^{3})^{T})=(x^{1}~x^{2}~x^{3})^{T} $ for all
$ (x^{0}~x^{1}~x^{2}~x^{3})^{T}\in\mathbb{R}^{4} $. The map $ \Phi $ is obviously well-defined and smooth. For any
$ \vec{x}\in\mathbb{S}^{2} $ one has
\begin{align*}
\Phi(\mathbb{1},\vec{x})=\frac{\operatorname{pr}\Bigg(\mathbb{1}\begin{pmatrix}
1\\
\vec{x}
\end{pmatrix}\Bigg)}{\Bigg|\Bigg|\operatorname{pr}\Bigg(\mathbb{1}\begin{pmatrix}
1\\
\vec{x}
\end{pmatrix}\Bigg)\Bigg|\Bigg|}=\frac{\vec{x}}{||\vec{x}||}=\vec{x},
\end{align*}
since $ ||\vec{x}||=1 $. Proving the action property and effectiveness of $ \Phi $ is a bit more involved.
We want to give another argument to prove that there is a transitive action of $ \mathcal{L}_{3}^{\uparrow,+} $
on $ \mathbb{S}^{2} $:
we calculate the Iwasawa decomposition (see Definition \ref{IwasawaDec}) of $ \mathcal{L}_{3}^{\uparrow,+} $
following \cite[Section~4.5]{GallierNotes}.
The Lie algebra corresponding to $ \mathcal{L}_{3}^{\uparrow,+} $ is
\begin{align}
\mathfrak{so}(1,3)=\left.\left\lbrace\begin{pmatrix}
0 & \vec{a}^{T} \\
\vec{a} & A
\end{pmatrix}
\in M_{4\times 4}(\mathbb{R})~\right|~\vec{a}\in\mathbb{R}^{3}, A^{T}=-A\in M_{3\times 3}(\mathbb{R})\right\rbrace,
\end{align}
what can be concluded quite easily since the Lie algebra of the special orthogonal matrices
$ \text{SO}(4) $ are the skew-symmetric matrices $ \mathfrak{so}(4) $. This can be adapted to signature $ (1,3) $.
Then the Cartan decomposition (see Definition \ref{CartanDec} of $ \mathfrak{so}(1,3) $ is
\begin{align}
\mathfrak{so}(1,3)=\mathfrak{k}\oplus\mathfrak{p},
\end{align}
where
\begin{align}
\mathfrak{k}=\left.\left\lbrace\begin{pmatrix}
0 & 0\\
0 & A
\end{pmatrix}
\in M_{4\times 4}(\mathbb{R})~\right|~A=-A^{T}\in M_{3\times 3}(\mathbb{R})\right\rbrace
\end{align}
and
\begin{align}
\mathfrak{p}=\left.\left\lbrace
\begin{pmatrix}
0 & \vec{a}^{T} \\
\vec{a} & 0_{3\times 3}
\end{pmatrix}
\in M_{4\times 4}(\mathbb{R})~\right|~\vec{a}\in\mathbb{R}^{3}\right\rbrace.
\end{align}
This can be obtained as follows: consider the Cartan
involution (see Definition \ref{CartanInv})
$ \Theta\colon\mathfrak{so}(1,3)\ni B\mapsto-B^{T}\in\mathfrak{so}(1,3) $ on $ \mathfrak{so}(1,3) $. Then
\begin{align*}
\left.\left\lbrace\begin{pmatrix}
0 & 0\\
0 & A
\end{pmatrix}
\in M_{4\times 4}(\mathbb{R})~\right|~A=-A^{T}\in M_{3\times 3}(\mathbb{R})\right\rbrace
=\left\lbrace B\in\mathfrak{so}(1,3)~|~\Theta(B)=B\right\rbrace
=\mathfrak{k}
\end{align*}
and
\begin{align*}
\left.\left\lbrace
\begin{pmatrix}
0 & \vec{a}^{T} \\
\vec{a} & 0_{3\times 3}
\end{pmatrix}
\in M_{4\times 4}(\mathbb{R})~\right|~\vec{a}\in\mathbb{R}^{3}\right\rbrace
=\left\lbrace B\in\mathfrak{so}(1,3)~|~\Theta(B)=-B\right\rbrace
=\mathfrak{p}.
\end{align*}
To get a maximal abelian subalgebra $ \mathfrak{a} $ of $ \mathfrak{p} $ one calculates the commutator of
$ A=\begin{pmatrix}
0 & \vec{a}^{T} \\
\vec{a} & 0_{3\times 3}
\end{pmatrix} $
and $ B=\begin{pmatrix}
0 & \vec{b}^{T} \\
\vec{b} & 0_{3\times 3}
\end{pmatrix}\in\mathfrak{p} $:
\begin{align*}
\left[ A,B\right]=\begin{pmatrix}
0 & 0 & 0 & 0 \\
0 & 0 & a^{1}b^{2}-a^{2}b^{1} & a^{1}b^{3}-a^{3}b^{1} \\
0 & a^{2}b^{1}-a^{1}b^{2} & 0 & a^{2}b^{3}-a^{3}b^{2} \\
0 & a^{3}b^{1}-a^{1}b^{3} & a^{3}b^{2}-a^{2}b^{3} & 0
\end{pmatrix}.
\end{align*}
Thus, it is easy to see that
\begin{align}
\mathfrak{a}=\left.\left\lbrace\begin{pmatrix}
0 & a & 0 & 0 \\
a & 0 & 0 & 0 \\
0 & 0 & 0 & 0 \\
0 & 0 & 0 & 0
\end{pmatrix}~\right|~a\in\mathbb{R}\right\rbrace
\end{align}
is a maximal abelian subalgebra of $ \mathfrak{p} $.
A nilpotent Lie subalgebra $ \mathfrak{n} $ of $ \mathfrak{so}(1,3) $ given as the sum of root spaces of a choice
of positive roots on $ \mathfrak{a} $ is given by
\begin{align}
\mathfrak{n}=\left.\left\lbrace\begin{pmatrix}
0 & 0 & b & c \\
0 & 0 & b & c \\
b & -b & 0 & 0 \\
c & -c & 0 & 0
\end{pmatrix}\in M_{4\times 4}(\mathbb{R})~\right|~b,c\in\mathbb{R}\right\rbrace.
\end{align}
Thus the Iwasawa decomposition of $ \mathfrak{so}(1,3) $ is
\begin{align}
\mathfrak{so}(1,3)=\mathfrak{k}\oplus\mathfrak{a}\oplus\mathfrak{n}.
\end{align}
If we denote the corresponding connected Lie groups by $ K, A $ and $ N $ we get the Iwasawa decomposition of
$ \mathcal{L}_{3}^{\uparrow,+} $, namely
\begin{align}
\mathcal{L}_{3}^{\uparrow,+}=KAN.
\end{align}
It is clear by definition that $ K\cong\text{SO}(3) $. Defining $ M=\text{SO}(2) $
we conclude by Example \ref{exampleSO3}
\begin{align}
\mathbb{S}^{2}\cong\text{SO}(3)/\text{SO}(2)\cong K/M\cong KAN/MAN\cong\mathcal{L}_{3}^{\uparrow,+}/MAN.
\end{align}
Since $ \left[\mathfrak{a},\mathfrak{n}\right]\subseteq\mathfrak{n} $ we know that $ \mathfrak{a}\oplus\mathfrak{n} $
is a Lie subalgebra of $ \mathfrak{so}(1,3) $. Thus $ AN $ is a subgroup of $ \mathcal{L}_{3}^{\uparrow,+} $
and we showed that $ \mathbb{S}^{2} $ can be structured as a homogeneous space via a transitive action of
$ \mathcal{L}_{3}^{\uparrow,+} $.
\end{example}

We see in Section \ref{sectiononish} that Example \ref{exampleSO3}, Example \ref{exampleSL3} and Example
\ref{exampleSO31} provide the only three (non-equivalent) ways to structure $ \mathbb{S}^{2} $ as a homogeneous space
via an effective action.

\section{The Euler Characteristic of Compact Homogeneous Spaces}\label{SecEulerPos}

In Proposition \ref{prophomspace} and Theorem \ref{homspacetransact} we gave a characterization of homogeneous
spaces by quotient spaces of Lie groups that act transitively on the homogeneous space modulo a stabilizer of the
action. There is another way to detect homogeneous spaces with a topological background.

For $ k\in\mathbb{N}_{0} $ we denote the set of smooth $ k $-differential forms on a smooth manifold $ M $ by
$ \Omega^{k}(M) $. An element $ \omega\in\Omega^{k}(M) $ is said to be \textit{closed} if the exterior derivative
of $ \omega $ is zero, i.e. $ \mathrm{d}\omega=0 $. It is called \textit{exact} if there is a $ (k-1) $-form
$ \alpha\in\Omega^{k-1}(M) $ such that $ \omega=\mathrm{d}\alpha $. Since $ \mathrm{d}^{2}=0 $ every exact
$ k $-form is closed and one can define the quotient space $ H_{\operatorname{dR}}^{k}(M) $ of all closed $ k $-forms
modulo all exact $ k $-forms which is called
the $ k $-th \textit{de Rham cohomology} of $ M $. From a Theorem by de Rham this is isomorphic to the $ k $-th
singular cohomology on $ M $ with coefficients in $ \mathbb{R} $. The direct sum of all these spaces
\begin{align}
H_{\operatorname{dR}}^{\bullet}(M)=\bigoplus_{k\in\mathbb{N}_{0}}H_{\operatorname{dR}}^{k}(M)
\end{align}
is said to be the de Rham cohomology of $ M $.

\begin{definition}[Euler characteristic]\label{DefEulerCharac}
The Euler characteristic $ \chi(M) $ of a smooth manifold $ M $ is defined as
\begin{align}\label{Euler}
\chi(M)=\sum\limits_{k=0}^{\infty}(-1)^{k}\dim(H_{\operatorname{dR}}^{k}(M)).
\end{align}
\end{definition}

\begin{remark}
In Definition \ref{DefEulerCharac} one has to argue that the series (\ref{Euler}) converges. While this is true
for compact smooth manifolds this will not be the case for arbitrary smooth manifolds. However we are mainly
interested in compact manifolds $ M $, that also fulfil $ \dim(H_{\operatorname{dR}}^{k}(M))<\infty $ for all
$ k\in\mathbb{N}_{0} $. For more informations about the Euler characteristic of a smooth
manifold consider \cite[Chapter~11]{bott2013differential}.
\end{remark}

This is our tool to detect homogeneous spaces. A homogeneous space $ M $ is compact if
for example the Lie group $ G $ that acts transitively on
$ M $ is compact. We prove that in this situation the Euler characteristic of $ M $ is non-negative, following
\cite[Section~5]{samelson1958}. We need to sum up some basic definitions and properties of Lie group representations
first.

\begin{definition}[Lie Group Representation]
A representation of a Lie group $ G $ is a finite-dimensional real vector space $ V $ together with a
Lie group homomorphism
\begin{align}
\pi\colon G\rightarrow\operatorname{GL}(V),
\end{align}
where $ \operatorname{GL}(V) $ denotes the general linear group of $ V $ consisting of the bijective linear maps
$ V\rightarrow V $ with concatenation as group multiplication.
\end{definition}

\begin{example}
Consider the \textbf{adjoint representation} $ (\text{Ad},\mathfrak{g}) $ of $ G $ on $ \mathfrak{g} $, where
\begin{align}
\text{Ad}\colon G\ni g\mapsto T_{e}\text{Conj}_{g}\in\text{GL}(\mathfrak{g}).
\end{align}
The map $ \text{Ad}^{*}\colon G\rightarrow\text{GL}(\mathfrak{g}^{*}) $ that assigns an element $ g\in G $ the
adjoint of the linear map $ \text{Ad}(g^{-1}) $ is said to be the \textbf{coadjoint representation} of $ G $ on
$ \mathfrak{g}^{*} $. For any $ k\in\mathbb{N} $ one can extend those two maps linearly to obtain representations
$ \Anti^{k}\text{Ad} $ and $ \Anti^{k}\text{Ad}^{*} $ of $ G $ on $ \Anti^{k}\mathfrak{g} $ and
$ \Anti^{k}\mathfrak{g}^{*} $, respectively.
\end{example}

Now choose an arbitrary representation $ (\pi,V) $ of $ G $. If there is a scalar product $ \langle\cdot,\cdot
\rangle $ on $ V $ such that
\begin{align}\label{OrthoRep}
\langle\pi(g)v,\pi(g)w\rangle=\langle v,w\rangle
\end{align}
for all $ g\in G $ and $ v,w\in V $ the representation $ (\pi,V) $ is said to be \textit{orthogonal}. A real
vector space can always be endowed with a scalar product. We want to
prove that there is
a scalar product on $ V $ satisfying (\ref{OrthoRep}) for a chosen representation $ (\pi,V) $ if $ G $ is
compact. The first step is to consider a special measure. Take a $ n $-dimensional manifold $ M $ and a
\textit{volume form} $ \omega $ on it, i.e.
$ \omega=w\mathrm{d}x_{1}\wedge\ldots\wedge\mathrm{d}x_{n} $ in local coordinates
$ (U,x_{i}) $. The \textit{measure} $ \mu $ corresponding to $ \omega $ is then defined by the relation
\begin{align}
\int f\mathrm{d}\mu=\int_{U} f(x_{1},\ldots,x_{n})|w|\mathrm{d}x_{1}\cdots\mathrm{d}x_{n}
\end{align}
for all $ f\in\Cinfty(U) $. One can check that this is independent of the chart $ (U,x_{i}) $. In the case of a Lie
group $ G $ we choose a non-zero element $ \tilde{\omega}\in\Anti^{n}(T_{e}G) $ and define
$ \omega=T_{e}L_{g}\tilde{\omega} $ for any $ g\in G $. If $ G $ is compact in addition we can normalize $ \omega $
such that the corresponding measure fulfils
\begin{align}
\int_{G}\mathrm{d}g=1.
\end{align}
It has been proved in \cite{HaarMeasure} that this measure is unique and bi-invariant, i.e.
\begin{align}
\int_{G}f(g)\mathrm{d}g=\int_{G}f(gh)\mathrm{d}g=\int_{G}f(hg)\mathrm{d}g
\end{align}
for all $ f\in\Cinfty(G) $ and $ h\in G $. We call it the \textit{normalized Haar measure} on $ G $. Now we come
back to orthogonal representations:

\begin{proposition}\label{PropOrthTrans}
Let $ (\pi,V) $ be a representation of a compact Lie group $ G $ and $ \langle\cdot,\cdot\rangle' $ any scalar
product on $ V $. Then
\begin{align}\label{OrthScalarProd}
\langle h_{1},h_{2}\rangle=\int_{G}\langle\pi(g)h_{1},\pi(g)h_{2}\rangle'\mathrm{d}g
\end{align}
is a scalar product on $ V $ such that Eq. (\ref{OrthoRep}) holds, i.e. $ (\pi,V) $ is a orthogonal representation
of $ G $.
\end{proposition}

\begin{proof}
Clearly Eq. (\ref{OrthScalarProd}) defines a scalar product on $ G $. It fulfils (\ref{OrthoRep}) since for
$ g',h_{1},h_{2}\in G $ one has
\begin{align*}
\langle\pi(g')h_{1},\pi(g')h_{2}\rangle&=\int_{G}\langle\pi(g)\pi(g')h_{1},\pi(g)\pi(g')h_{2}\rangle'\mathrm{d}g\\
&=\int_{G}\langle\pi(gg')h_{1},\pi(gg')h_{2}\rangle'\mathrm{d}g\\
&=\int_{G}\langle\pi(g)h_{1},\pi(g)h_{2}\rangle'\mathrm{d}g\\
&=\langle h_{1},h_{2}\rangle.
\end{align*}
\end{proof}

Also remark the notion of maximal tori:

\begin{definition}
Let $ G $ be a Lie group. Then a Lie subgroup $ T $ is said to be a
\begin{compactenum}
\item \textbf{torus} in $ G $ if $ T $ is isomorphic to a product $ \mathbb{S}^{1}\times\cdots\times\mathbb{S}^{1} $
of $ 1 $-spheres.
\item \textbf{maximal torus} in $ G $ if for any torus $ T' $ in $ G $ such that $ T\subseteq T' $ one has $ T'=T $.
\end{compactenum}
\end{definition}
These objects are of fundamental use in the theory of connected compact Lie groups. Some properties are collected
in the following

\begin{proposition}
Let $ G $ be a connected compact Lie group.
\begin{compactenum}
\item Any tours in $ G $ is contained in a maximal torus in $ G $.
\item Any element in $ G $ is contained in a maximal torus in $ G $.
\item Any two maximal tori in $ G $ are conjugated to each other, i.e. if $ T $ and $ T' $ are maximal tori
in $ G $ there is a $ g\in G $ such that
\begin{align}
gTg^{-1}=T'.
\end{align}
For this one calls the dimension of any maximal torus in $ G $ the \textbf{rank} of $ G $.
\item Any maximal torus $ T $ in $ G $ has a \textbf{generating element} $ g_{0}\in T $, i.e. for any $ g\in T $
there is a natural number $ k\in\mathbb{N}_{0} $ such that $ g=g_{0}^{k} $.
\end{compactenum}
\end{proposition}
All statements are taken from \cite[Section~2.4]{arvanitogeorgos2003introduction} ,
\cite[Chapter~IV]{brocker2013representations} and \cite[Chapter~6]{hofmann2013structure}.
Consider those references for more informations about maximal tori.

Furthermore, one has a topological result on representations of compact Lie groups. In analogy to
Definition \ref{EulerCompact} we want to define the Euler characteristic of a representation. Let $ (\pi,V) $ be a
representation of a compact Lie group $ G $ with $ n=\dim(V) $.
Then for any $ 1\leq k\leq n $ the map $ \pi $ induces a representation $ \Anti^{k}\pi $ on $ \Anti^{k}V $.
Denote by $ I_{k} $ the subspaces of $ \Anti^{k}V $ consisting of
the elements $ v $ that satisfy
$ (\Anti^{k}\pi(g))(v)=v $ for all $ g\in G $. Then define the \textit{Poincaré polynomial} $ P_{\pi} $ of
$ (\pi,V) $ by
\begin{align}\label{PoincPoly}
P_{\pi}(t)=\sum\limits_{k=1}^{n}\lambda_{k}t^{k},
\end{align}
where $ \lambda_{k}=\dim(I_{k}) $ are the \textit{Betti numbers} corresponding to $ \pi $. We know
from representation theory (c.f. \cite[Section~12.6.2]{chirikjian2011stochastic}) that the Betti numbers are obtained
by integration of the trace of the corresponding representation over the whole group, i.e.
\begin{align}
\lambda_{k}=\int_{G}\text{tr}(\Anti^{k}\pi(g))\mathrm{d}g.
\end{align}

\begin{definition}
Let $ (\pi,V) $ be a representation of a compact Lie group $ G $. The \textit{Euler characteristic} $ \chi(\pi) $
of $ (\pi,V) $ is defined by
\begin{align}
\chi(\pi)=P_{\pi}(-1).
\end{align}
\end{definition}

In the end we want to prove that the Euler characteristic of a connected compact homogeneous space is non-negative.
The next lemma is a step towards this:

\begin{lemma}\label{LemmaEulerRepresentation}
Let $ (\pi,V) $ be a representation of a connected compact Lie group $ G $. Then $ \chi(\pi)\geq 0 $. Moreover, one
has $ \chi(\pi)>0 $ if and only if for each maximal torus $ T $ in $ G $ one has $ \pi(g)v=v $ for all $ g\in G $, only
for a $ v\in T $ if $ v=0 $.
\end{lemma}

\begin{proof}
The discussion above implies
\begin{align*}
\chi(\pi)&=P_{\pi}(-1)
=\sum\limits_{k=1}^{n}\lambda_{k}(-1)^{k}
=\sum\limits_{k=1}^{n}\int_{G}\text{tr}(\Anti^{k}\pi(g))\mathrm{d}g(-1)^{k}\\
&=\int_{G}\sum\limits_{k=1}^{n}\text{tr}(\Anti^{k}\pi(g))(-1)^{k}\mathrm{d}g
\overset{(\ast)}{=}\int_{G}\det(-\pi(g)+\text{id}_{V})\mathrm{d}g,
\end{align*}
where in $ (\ast) $ we used the formula
\begin{align}
\sum\limits_{k=1}^{n}\text{tr}(\Anti^{k}\pi(g))t^{k}=\det(t\pi(g)+\text{id}_{V})
\end{align}
resulting from a principal minor argument. According to Proposition \ref{PropOrthTrans} we can assume that $ \pi(g) $
is a orthogonal transformation for all $ g\in G $. Since $ G $ is connected $ \pi(g) $ is also orientation preserving.
First assume that $ n $ is odd. Then according to \cite[Lemma~7.3.1]{gallier2012geometric} $ \pi(g) $ has an
eigenvalue $ \lambda=1 $, i.e. $ \det(-\pi(g)+\text{id}_{V})=0 $. Thus in this case $ \chi(\pi)=0 $. If $ n $ is even
one has the normal form of orthogonal matrices
\begin{align}\label{normalform}
\begin{pmatrix}
R_{1} & & & & & \\
 & \ddots & & & & \\
 & & R_{m} & & & \\
 & & & \pm 1 & & \\
 & & & & \ddots & \\
 & & & & & \pm 1
\end{pmatrix},
\end{align}
with $ R_{i}=\begin{pmatrix}
\cos(\phi_{i}) & -\sin(\phi_{i})\\
\sin(\phi_{i}) & \cos(\phi_{i})
\end{pmatrix} $ for some $ \phi_{i}\in\mathbb{R} $. If at least one $ -1 $ appears in the matrix
(\ref{normalform}) one has $ \det(-\pi(g)+\text{id}_{V})=0 $. Since $ \det(R_{i})=1 $ and $ \pi(g) $ preserves
orientation there is a even number of $ 1 $ appearing in (\ref{normalform}) and since we have
\begin{align*}
\det\Bigg(-\begin{pmatrix}
\cos(\phi_{i}) -1 & -\sin(\phi_{i})\\
\sin(\phi_{i}) & \cos(\phi_{i})-1
\end{pmatrix}\Bigg)&=(-1)^{2}(\cos^{2}(\phi_{i})-2\cos(\phi_{i})+1+\sin^{2}(\phi_{i}))\\
&=2(1-\cos(\phi_{i}))\geq 0,
\end{align*}
this implies $ \det(-\pi(g)+\text{id}_{V})\geq 0 $. Thus $ \chi(\pi)\geq 0 $ and we have proven the first
claim of the proposition. We prove both directions of the second statement by contraposition. Assume
$ \chi(\pi)=0 $. Since $ \det(-\pi(g)+\text{id}_{V})\geq 0 $, this implies $ \det(-\pi(g)+\text{id}_{V})=0 $ for all
$ g\in G $. Choose any maximal torus $ T $ in $ G $ with generating element $ g_{0}\in G $. Then of course
$ \det(-\pi(g_{0})+\text{id}_{V})=0 $ which means that $ 1 $ is an eigenvalue of $ \pi(g_{0}) $, i.e. $ \pi(g_{0}) $
fixes a non-zero vector $ v\in V $. Since $ \pi $ is a Lie group homomorphism and any other element $ g $ in $ T $ is
obtained as a power $ p\in\mathbb{N} $ of $ g_{0} $ one has $ \pi(g)v=\pi(g_{0}^{p})v=\pi(g_{0})\cdots\pi(g_{0})v=v $
and every element
of $ T $ fixes $ v $. Conversely, assume that there is a maximal torus $ T $ fixing a non-zero vector $ v\in V $, i.e.
$ \det(-\pi(g)+\text{id}_{V})=0 $ for all $ g\in T $. Now any element $ g\in G $ is conjugated to an element
$ g_{0}\in T $, i.e. there is an $ h\in G $ such that $ g=hg_{0}h^{-1} $. Via cyclic permutations of the arguments
of the determinant this gives together with the homomorphism property of $ \pi $
\begin{align}
\det(-\pi(g)+\text{id}_{V})=\det(-\pi(h^{-1})\pi(h)\pi(g_{0})+\text{id}_{V})=0.
\end{align}
This proves $ \det(-\pi(g)+\text{id}_{V})=0 $ for all $ g\in G $ and $ \chi(\pi)=0 $.
\end{proof}

The next proof deals with differential forms on a manifold that are invariant under the action of a Lie group
that acts on the manifold. We denote by $ \alpha $ a differential $ k $-form
on a smooth manifold $ M $ and we say that it is \textit{invariant} under the action $ \Phi $ of a Lie group
$ G $ on $ M $ if
\begin{align}
\Phi_{g}^{*}\alpha=\alpha\text{ for all }g\in G,
\end{align}
i.e. for all $ p\in M $, $ g\in G $ and all tangent vectors $ X_{1},\ldots,X_{k}\in T_{p}M $ one has
\begin{align}
\alpha_{\Phi(g,p)}(T_{p}\Phi_{g}~X_{1},\ldots,T_{p}\Phi_{g}~X_{k})=\alpha_{p}(X_{1},\ldots,X_{k}).
\end{align}
We use the notation
\begin{align}
\Omega^{k}(M)^{G,\Phi}=\left.\left\lbrace\alpha\in\Omega^{k}(M)~\right|~\Phi_{g}^{*}\alpha=\alpha\text{ for all }g\in G\right\rbrace
\end{align}
to denote set of $ k $-forms on $ M $ that are invariant under the action $ \Phi $. The corresponding
\textit{$ k $-th invariant de Rham cohomology} is then defined by
\begin{align}
H_{\operatorname{dR}}^{k}(M)^{G,\Phi}=\frac{\left\lbrace\alpha\in\Omega^{k}(M)^{G,\Phi}~|~\mathrm{d}\alpha=0
\right\rbrace}{\left\lbrace\mathrm{d}\beta\in\Omega^{k}(M)~|~\beta\in\Omega^{k-1}(M)^{G,\Phi}\right\rbrace},
\end{align}
and the direct sum of all these spaces is the invariant de Rham cohomology
$ H_{\operatorname{dR}}^{\bullet}(M)^{G,\Phi} $. It is interesting to compare it to
$ H_{\operatorname{dR}}^{\bullet}(M) $.

\begin{lemma}\label{LemGInvDeRham}
Let $ G $ be a connected compact Lie group that acts on a connected manifold $ M $ via $ \Phi $. Then
\begin{align}
H_{\operatorname{dR}}^{\bullet}(M)\cong H_{\operatorname{dR}}^{\bullet}(M)^{G,\Phi}.
\end{align}
\end{lemma}

\begin{proof}
Choose an element $ g\in G $. Since $ G $ is connected there is a smooth curve
$ \gamma\colon\left[0,1\right]\rightarrow G $ in $ G $ such that $ \gamma(0)=e $ and $ \gamma(1)=g $. Furthermore
one defines the smooth map
\begin{align}
\Psi_{g}\colon\left[0,1\right]\times M\ni(t,p)\mapsto\Psi_{g}(t,p)=\Phi_{\gamma(t)}(p)\in M.
\end{align}
Since $ \Psi_{g}(0,\cdot)=\Phi_{\gamma(0)}=\operatorname{id}_{M} $ and $ \Psi_{g}(1,\cdot)=\Phi_{\gamma(1)}=\Phi_{g} $,
the map $ \Psi_{g} $ is a homotopy of $ \operatorname{id}_{M} $ and $ \Phi_{g} $. By a homotopy argument
(see \cite[Satz~2.2.1]{BachelorThesis}) the two maps
\begin{align}\label{eqHomotop}
\operatorname{id}_{M}^{*}=\Phi_{g}^{*}\colon H_{\operatorname{dR}}^{\bullet}(M)\rightarrow
H_{\operatorname{dR}}^{\bullet}(M)
\end{align}
coincide on cohomology level. This construction is valid for all $ g\in G $. In a next step we consider the map
\begin{align}
\Xi\colon\Omega^{\bullet}(M)\ni\alpha\mapsto\int_{G}\Phi_{g}^{*}\alpha\mathrm{d}g
\in\Omega^{\bullet}(M)^{G,\Phi}.
\end{align}
It is well-defined, since
\begin{align*}
\Phi_{h}^{*}\int_{G}\Phi_{g}^{*}\alpha\mathrm{d}g=\int_{G}\Phi_{gh}^{*}\alpha\mathrm{d}g
=\int_{G}\Phi_{g}^{*}\alpha\mathrm{d}g,
\end{align*}
for any $ \alpha\in\Omega^{\bullet}(M) $, by the bi-invariance of the normalized Haar measure. We want to show that
\begin{align}
\hat{\Xi}\colon H_{\operatorname{dR}}^{\bullet}(M)\ni\left[\alpha\right]\mapsto
\left[\Xi(\alpha)\right]\in H_{\operatorname{dR}}^{\bullet}(M)^{G,\Phi}
\end{align}
defines a linear isomorphism. The map $ \hat{\Xi} $ is well-defined since
$ \Xi(\alpha)\in\Omega^{\bullet}(M)^{G,\Phi} $ for any $ \alpha\in\Omega^{\bullet}(M) $. Thus $ \Xi(\alpha) $
is a representative of the equivalence class $ \hat{\Xi}(\left[\alpha\right]) $ in
$ H_{\operatorname{dR}}^{\bullet}(M)^{G,\Phi} $. The linearity of $ \hat{\Xi} $ is clear. Let
$ \left[\alpha\right]\in H_{\operatorname{dR}}^{\bullet}(M)^{G,\Phi}\subseteq
H_{\operatorname{dR}}^{\bullet}(M) $ be arbitrary. Then
\begin{align*}
\left[\alpha\right]=\left[\int_{G}\Phi_{g}^{*}\alpha\mathrm{d}g\right]=\left[\Xi(\alpha)\right],
\end{align*}
which implies that $ \hat{\Xi} $ is surjective. By Eq. (\ref{eqHomotop}) one has for any
$ \left[\alpha\right]\in H_{\operatorname{dR}}^{\bullet}(M) $
\begin{align*}
\left[\alpha\right]&=\left[\int_{G}\operatorname{id}_{M}^{*}\alpha\mathrm{d}g\right]\\
&=\left[\int_{G}\Phi_{g}^{*}\alpha\mathrm{d}g\right]\\
&=\left[\Xi(\alpha)\right],
\end{align*}
i.e. $ \hat{\Xi} $ is injective. This concludes the proof.
\end{proof}
Now consider a connected homogeneous space $ M=G/H $, where $ G $ is a connected and compact Lie group.
Set $ e'=e\cdot H $. One can define a left $ H $-module structure on $ \Anti^{\bullet}T_{e'}^{*}M $ by
\begin{align}
H\times\Anti^{\bullet}T_{e'}^{*}M\ni(h,\alpha_{e'})\mapsto h\rhd\alpha_{e'}=
\Anti^{\bullet}\operatorname{Ad}_{h}^{*}(\alpha_{e'})\in\Anti^{\bullet}T_{e'}^{*}M.
\end{align}
We denote the elements that are invariant under this action by
\begin{align}
(\Anti^{\bullet}T_{e'}^{*}M)^{H}=\left\lbrace\alpha_{e'}\in\Anti^{\bullet}T_{e'}^{*}M~|~
\Anti^{\bullet}\operatorname{Ad}_{h}^{*}(\alpha_{e'})=\alpha_{e'}\text{ for all }h\in H\right\rbrace.
\end{align}
Then one can prove the following

\begin{lemma}\label{LemComplexHInv}
Let $ M=G/H $ be a connected homogeneous space such that $ G $ is connected and compact, where we denote the
transitive action of $ G $ on $ M $ by $ \Phi $. Then
\begin{align}
\Omega^{\bullet}(M)^{G,\Phi}\cong(\Anti^{\bullet}T_{e'}^{*}M)^{H}.
\end{align}
\end{lemma}

\begin{proof}
Consider the map
\begin{align}
F\colon(\Anti^{\bullet}T_{e'}^{*}M)^{H}\ni\alpha_{e'}\mapsto(M\ni p\mapsto\Phi_{g}^{*}\alpha_{e'}\in
\Anti^{\bullet}T_{p}^{*}M,\text{ where }p=\Phi_{g}(e'))\in\Omega^{\bullet}(M)^{G,\Phi}.
\end{align}
It is easy to see that $ F(\alpha_{e'})\in\Omega^{\bullet}(M)^{G,\Phi} $ for any
$ \alpha_{e'}\in(\Anti^{\bullet}T_{e'}^{*}M)^{H}. $
To prove that this map is well-defined one has to argue that $ \Phi_{g}(e')=\Phi_{h}(e') $, for any $ g,h\in G $
and $ \alpha_{e'}\in(\Anti^{\bullet}T_{e'}^{*}M)^{H} $,
implies $ \Phi_{g}^{*}\alpha_{e'}=\Phi_{h}^{*}\alpha_{e'} $ or equivalently
$ (\Phi_{g}\circ\Phi_{h^{-1}})^{*}\alpha_{e'}=\alpha_{e'} $. Since $ \Phi_{g}(e')=\Phi_{h}(e') $ for $ g,h\in G $
implies that $ g^{-1}h\in H $, it suffices to prove that $ \Phi_{h}^{*}\alpha_{e'}=\alpha_{e'} $ for all $ h\in H $
and $ \alpha_{e'}\in(\Anti^{\bullet}T_{e'}^{*}M)^{H} $. Since $ \Phi $ is transitive it is also
locally transitive, i.e. the corresponding Lie algebra action
$ \phi|_{p}\colon\mathfrak{g}\rightarrow T_{p}M $ is surjective for all $ p\in M $. Thus, for any $ v_{p}\in T_{p}M $
there is a $ \xi\in\mathfrak{g} $ such that the fundamental vector field $ \xi_{M} $ for $ \xi $ on $ M $
satisfies
\begin{align}
\xi_{M}(p)=T_{e}\Phi_{p}\xi=v_{p},
\end{align}
for any $ p\in M $. Applying this to every tensor component of an arbitrary $ X_{p}\in\Anti^{k}T_{p}M $, one obtains
an element $ \xi=\xi_{1}\wedge\cdots\wedge\xi_{k}\in\Anti^{k}\mathfrak{g} $ such that
\begin{align}\label{FundMultVF}
T_{e}\Phi_{p}\xi:=T_{e}\Phi_{p}\xi_{1}\wedge\cdots\wedge T_{e}\Phi_{p}\xi_{1}=X_{p}.
\end{align}
We denote the left hand side of Eq. (\ref{FundMultVF}) also by $ \xi_{M} $.
Let $ \alpha_{e'}\in(\Anti^{k}T_{e'}^{*}M)^{H} $, $ X_{e'}\in\Anti^{k}T_{e'}M $ and $ h\in H $ be arbitrary.
Then
\begin{align*}
(\Phi_{h}^{*}\alpha_{e'})(X_{e'})&=\alpha_{e'}(T_{e'}\Phi_{h}X_{e'})
=\langle\alpha_{e'},T_{e'}\Phi_{h}\xi_{M}(e')\rangle
=\langle\alpha_{e'},\Phi_{h^{-1}}^{*}\xi_{M}(e')\rangle\\
&=\langle\alpha_{e'},\operatorname{Ad}_{u}(\xi_{M}(e'))\rangle
=\langle\operatorname{Ad}^{*}_{u^{-1}}\alpha_{e'},X_{e'}\rangle
=\langle\alpha_{e'},X_{e'}\rangle,
\end{align*}
where we used fundamental equations of the Cartan calculus. This shows $ \Phi_{h}^{*}\alpha_{e'}=\alpha_{e'} $
for all $ h\in H $, i.e. $ F $ is well-defined. It is clear that $ F $ is linear. To show that it is injective
choose $ \alpha_{e'},\beta_{e'}\in(\Anti^{k}T_{e'}^{*}M)^{H} $ and assume $ F(\alpha_{e'})=F(\beta_{e'}) $.
Since $ e'=\Phi_{e}(e') $ this implies
\begin{align*}
\alpha_{e'}=F(\alpha_{e'})(e)=F(\beta_{e'})(e)=\beta_{e'}.
\end{align*}
To prove that $ F $ is surjective choose an arbitrary $ \alpha\in\Omega^{k}(M)^{G,\Phi} $, i.e.
$ \Phi_{g}^{*}\alpha=\alpha $ for all $ g\in G $. Let $ X_{e'}\in\Anti^{k}T_{e'}M $. Then for any $ g\in G $ one has
\begin{align*}
\alpha_{e'}(X_{e'})=\Phi_{g}^{*}\alpha_{e'}(X_{e'})=\alpha_{\Phi_{g}(e')}(\Anti^{k}T_{e'}\Phi_{g}X_{e'})
\end{align*}
and since $ T_{e'}\Phi_{g} $ is a linear isomorphism this implies for any $ p\in M $ that
\begin{align*}
\alpha_{p}(X_{p})=\alpha_{e'}(\Anti^{k}T_{p}\Phi_{g^{-1}}X_{p}),
\end{align*}
where $ g\in G $ such that $ \Phi_{g}(e')=p $. This concludes the proof.
\end{proof}

Now we can prove that the Euler characteristic of a connected compact homogeneous space is non-negative if there
is a connected compact Lie group acting transitively on it.

\begin{proposition}\label{PropEulerNonNeg}
Let $ M=G/H $ be a connected compact homogeneous space such that $ G $ is a connected compact Lie group and $ H $ a
connected subgroup of $ G $. Then $ \chi(M)\geq 0 $. Moreover, $ \chi(M)>0 $ if and only if $ H $ is a subgroup of
maximal rank, i.e. if any maximal torus in $ H $ has the same dimension as $ H $.
\end{proposition}

\begin{proof}
Parts of the proof are taken from \cite[Theorem~II]{samelson1958}.
The exterior derivative $ \mathrm{d} $ makes $ (\Anti^{\bullet}T_{e'}^{*}M)^{H} $ into a complex. This is a
consequence of $ \mathrm{d}\alpha\in(\Anti^{k+1}T_{e'}^{*}M)^{H} $ for $ \alpha\in(\Anti^{k}T_{e'}^{*}M)^{H} $,
which is true since $ \mathrm{d} $ and the pull back of any function commute. Any $ k $-cochain of this complex
can be identified with a $ k $-cochain on the complex $ (\Omega^{\bullet}(M)^{G,\Phi},\mathrm{d}) $ according
to Lemma \ref{LemComplexHInv}. Any $ k $-cochain of $ \Omega^{\bullet}(M)^{G,\Phi} $ can be written as the direct sum
of an invariant $ (k+1) $-coboundary and an invariant $ k $-cocycle. Moreover, any invariant $ k $-cocycle
can be written as the direct sum of an invariant $ k $-coboundary and an element of the $ k $-th invariant de Rham
cohomology. Altogether, this implies
\begin{align}\label{InvCochain}
C^{k}(M)^{G,\Phi}\cong B^{k+1}(M)^{G,\Phi}\oplus H_{\operatorname{dR}}^{k}(M)^{G,\Phi}\oplus B^{k}(M)^{G,\Phi},
\end{align}
where $ C^{k}(M)^{G,\Phi} $ and $ B^{k}(M)^{G,\Phi} $ denote the set of all invariant $ k $-cocycles and
$ k $-coboundaries on $ M $, respectively. Then it follows from Lemma \ref{LemGInvDeRham} that
\begin{align*}
\sum_{k=0}^{\infty}(-1)^{k}\dim((\Anti^{k}T_{e'}^{*}M)^{H})&=\sum_{k=0}^{\infty}(-1)^{k}\dim(C^{k}(M)^{G,\Phi})\\
&\overset{(\ast)}{=}\sum_{k=0}^{\infty}(-1)^{k}\dim(H_{\operatorname{dR}}^{k}(M)^{G,\Phi})\\
&=\sum_{k=0}^{\infty}(-1)^{k}\dim(H_{\operatorname{dR}}^{k}(M))\\
&=\chi(M),
\end{align*}
where we used Eq. (\ref{InvCochain}) in $ (\ast) $.
But since there is the representation $ \Anti^{\bullet}\operatorname{Ad}^{*} $ of the connected compact Lie group
$ H $ on $ \Anti^{\bullet}T_{e'}^{*}M $ Lemma \ref{LemmaEulerRepresentation} implies
\begin{align}
0\leq\chi(\Anti^{\bullet}\operatorname{Ad}^{*})=\sum_{k=0}^{\infty}(-1)^{k}\dim((\Anti^{k}T_{e'}^{*}M)^{H})=\chi(M).
\end{align}
Also the statement about the maximal torus can be deduced from Lemma \ref{LemmaEulerRepresentation}.
\end{proof}

The next step consists in extending this result to connected non-compact Lie groups. This is in general very hard to
obtain. The situation gets easier if one assumes the space $ M $ to be simply connected in addition, what will be
enough for considering the sphere $ \mathbb{S}^{2} $ as an example. Nevertheless, we cite the proofs for the
case when $ M $ is only assumed to be connected since we are also interested in the pretzel surfaces that are
not simply connected. We want to get a tool to use
Proposition \ref{PropEulerNonNeg} also in the case of non-compact Lie groups $ G $ if $ M $ is simply connected.
If $ G $ is connected and $ M $ simply connected, any
stabilizer
$ G_{x} $, where $ x\in M $, of the transitive action of $ G $ on $ M $ has to be connected since $ M=G/G_{x} $. This
means we can use the following famous theorem of Cartan for each, $ G $ and $ G_{x} $ (c.f.
\cite[Theorem~2]{IwasawaOnSomeTypes}):

\begin{theorem}\label{LemmaCartan}
Every connected Lie group $ G $ is a Cartesian product of a maximal compact Lie subgroup $ K\subseteq G $ and an
Euclidean vector space $ E $, i.e.
\begin{align}
G=K\times E.
\end{align}
\end{theorem}
Assume that $ M=G/G_{x} $ is simply connected. Then there are compact Lie subgroups $ K,L $ and Euclidean vector
spaces $ E,F $ such that $ G=K\times E $ and $ G_{x}=L\times F $ according to Theorem \ref{LemmaCartan}.
Since $ G_{x}\subseteq G $ and any two maximal compact subsets of $ G $ are conjugate to each other we can assume
$ L\subseteq K $. Now $ G/L $ can be seen as the total space of a vector bundle once with basis $ G/K $ and fibre
$ K/L $ and twice with basis $ G/G_{x} $ and fibre $ G_{x}/L $. For example $ \pi\colon G/L\rightarrow G/K $ maps
each coset of $ L $ in $ G $ to the coset of $ K $ in $ G $ which contains it. But since $ G/K\cong E $ is an
Euclidean space, a general topological result says that this vector bundle is trivial, i.e.
\begin{align}
G/L=K/L\times E.
\end{align}
The same holds for the second vector bundle, which implies that there is a global section
\begin{align}
f\colon M\cong G/G_{x}\rightarrow G/L.
\end{align}
Then some topological estimations (see \cite[Corollary~1]{Montgomery:1950}) imply the following

\begin{theorem}\label{ThmCompactSubActing}
If $ M=G/H $ is a simply connected compact homogeneous space and $ G $ a connected Lie group (probably non-compact)
there exists a connected compact Lie subgroup $ K $ of $ G $ that acts transitively on $ M $, i.e.
\begin{align}
M=G/H=K/L,
\end{align}
where $ L=K\cap H $.
\end{theorem}
An immediate consequence of the last theorem and Proposition \ref{PropEulerNonNeg} is

\begin{theorem}\label{EulerCompactSimply}
Let $ M=G/H $ be a simply connected compact homogeneous space. Then $ \chi(M)\geq 0 $. One has $ \chi(M)>0 $ if and
only if $ H $ has full rank.
\end{theorem}

There is an even more general result due to A. L. Onishchik and V. V. Gorbatsevich which shows that the assumption
of $ M $ being simply connected is not necessary (c.f. \cite[Part~II,~Chapter~5,~Theorem~1.2]{onishchik1993lie}).

\begin{theorem}\label{EulerCompact}
Let $ M=G/H $ be a connected compact homogeneous space. Then $ \chi(M)\geq 0 $. One has $ \chi(M)>0 $ if and
only if $ H $ has full rank.
\end{theorem}
The proof is proposing Tits bundles.
Remark that there is a different proof of this theorem from 2005 of G. D. Mostow that uses iterated fibrations
(c.f. \cite{mostow2005}).

\section{The Classification of Transitive Actions on $ \mathbb{S}^{2} $}\label{sectiononish}

The goal of this section is to describe all ways one can structure the $ 2 $-sphere as a homogeneous space. As
we discovered in Section \ref{sectionhomspace} it is equivalent to determine all transitive actions on $ \mathbb{S}^{2}
$. This classification has to be done with a notion of equivalent actions.

\begin{definition}
Assume there are two Lie groups $ G $ and $ G' $ acting on a manifold $ M $ via $ \Phi $ and $ \Phi' $, respectively.
These two actions are said to be \textbf{equivalent} if there is a Lie group morphism $ \phi\colon G\rightarrow G' $
such that
\begin{align}
\Phi(g,x)=\Phi'(\phi(g),x)
\end{align}
for all $ g\in G $ and $ x\in M $.
They are called \textbf{locally equivalent} if there is an open subset $ U\subseteq M $ such that $ \Phi|_{G\times U} $
and $ \Phi'|_{G\times U} $ are equivalent.
\end{definition}

It is clear by definition that equivalent actions are locally equivalent.
It is also quite intuitive that a classification makes only sense if one forces the actions to be effective.
Otherwise one could enlarge the Lie group by adding elements that simply induce the identity diffeomorphism on
the manifold. It is reasonable to classify transitive Lie groups such that every element induces a different
diffeomorphism, i.e. actions that are effective in addition. In Chapter \ref{ChapObstruction} we realize that this
assumption is no restriction at all for our purpose.
In the 1960's many classifications of transitive effective actions appeared. Pioneer work was done by A.L. Onishchik.
He considered compact homogeneous spaces and gave classifications of compact Lie groups that act
as well as for non-compact Lie groups (see e.g. \cite{Onishchikcompact,OnishchikII,OnishchikIII,onishchik1993lie}).
In \cite[Theorem~2.6]{onishchik1993lie} he classified in particular all locally effective
transitive actions of connected compact Lie groups on simply connected homogeneous spaces of rank $ 1 $. We follow a
source Onishchik also uses, namely the paper \cite{MontgomeryTransSpheres} of D. Montgomery and H. Samelson. There they
classify all connected compact Lie groups that act transitively and effectively on $ \mathbb{S}^{n} $,
$ n\in\mathbb{N} $.
We want to mention that the classification splits into even and odd $ n $. But since we are just interested in $ n=2 $
we do not see these effects. We first prove that if $ G $ acts in the way described above it already has to be simple.
Then there is just a list of Lie groups left to check. For this we need a little

\begin{lemma}\label{LemmaSO3}
Let $ G_{1},G_{2} $ be connected compact Lie groups and $ N $ a finite normal subgroup of $ \overline{G}=
G_{1}\times G_{2} $. If $ G=(G_{1}\times G_{2})/N $ acts transitively on $ \mathbb{S}^{2} $ then for one $ i=1,2 $ and
$ x\in\mathbb{S}^{2} $ the Lie group $ G_{i}/(G_{i}\cap G_{x}) $ acts transitively on $ \mathbb{S}^{2} $.
\end{lemma}

\begin{proof}
The proof and the notation is taken from \cite[Theorem~I']{MontgomeryTransSpheres}.
If $ G $ acts transitively on $ \mathbb{S}^{2} $ then of course also $ \overline{G} $ acts transitively on it. Take the
stabilizer $ \overline{G}_{x} $ of this action of any point $ x\in\mathbb{S}^{2} $. It is connected because
$ \mathbb{S}^{2} $ is simply connected and compact. Moreover, any element $ h\in\overline{G}_{x} $
can be written as $ h=h_{1}\times h_{2} $ with some $ h_{1}\in G_{1} $ and $ h_{2}\in G_{2} $ and is contained in a
maximal toral subgroup $ T\subseteq\overline{G}_{x} $. It is proven in \cite{SamelsonBeitraege1941} that
$ \overline{G} $ and
$ \overline{G}_{x} $ have the same rank, which implies that $ T\subseteq\overline{G} $, thus $ T=T_{1}\times T_{2} $,
where $ T_{1} $ and $ T_{2} $ are maximal toral subgroups of $ G_{1} $ and $ G_{2} $, respectively.
This implies that $ h_{1},h_{2}\in T
\subseteq\overline{G}_{x} $, i.e. $ \overline{G}_{x}=H_{1}\times H_{2} $ with $ H_{i}=G_{i}\cap\overline{G}_{x} $.
Thus, we have a decomposition
\begin{align}
\overline{G}/\overline{G}_{x}\cong G_{1}/H_{1}\times G_{2}/H_{2}
\end{align}
and $ G_{1}/H_{1}\times G_{2}/H_{2}\cong\overline{G}/\overline{G}_{x}\cong\mathbb{S}^{2} $. But in this case it is
known that $ G_{1}/H_{1} $ or $ G_{2}/H_{2} $ has to be a single point. Let $ G_{2}/H_{2}\cong\left\lbrace
\mathrm{pt}\right\rbrace $, i.e. $ G_{2}=H_{2} $ which means that the elements of $ G_{2} $ stabilize the point $ x $
in the action of $ \overline{G} $ on $ \mathbb{S}^{2} $. Then $ G_{1}/H_{1} $ has to be transitive on
$ \mathbb{S}^{2} $.
\end{proof}
Now, following \cite[Theorem~I]{MontgomeryTransSpheres}, we have

\begin{theorem}\label{ThmSphereSimpleCompact}
Every connected compact Lie group that acts transitively and effectively on $ \mathbb{S}^{2} $ is simple.
\end{theorem}

\begin{proof}
Assume that the connected compact Lie group $ G $ acts transitively and effectively on $ \mathbb{S}^{2} $ but is not
simple. Since $ G $ is a connected and compact Lie group there are two connected compact
subgroups $ G_{1} $ and $ G_{2} $ of $ G $ and a finite normal subgroup $ N $ of $ \overline{G}=G_{1}\times G_{2} $
such that $ G=(G_{1}\times G_{2})/N $ according to \cite[Theorem~6.19]{hofmann2013structure}. Since $ G $ acts
effectively we know that $ \overline{G} $ has to act
almost effectively, i.e. the action of $ \overline{G} $ on $ \mathbb{S}^{2} $ only has a finite set of fixed points,
namely the elements of $ N $. This implies that $ \overline{G}_{x} $ can not contain any infinite normal subgroup
of $ \overline{G} $ for any $ x\in\mathbb{S}^{2} $. But as in the proof of Lemma \ref{LemmaSO3} we conclude that
$ G_{2}\subseteq\overline{G}_{x} $ stabilizes the action of $ \overline{G} $. Now $ G_{2} $ is a non-finite normal
subgroup which gives the contradiction.
\end{proof}

\begin{theorem}\label{ThmCompActionSphere}
Any effective transitive action of a connected compact Lie group on $ \mathbb{S}^{2} $ is equivalent
to the linear action of $ \operatorname{SO}(3) $ on $ \mathbb{S}^{2} $.
\end{theorem}

The proof of Theorem \ref{ThmCompActionSphere} can be summarized as follows:
according to Theorem \ref{ThmSphereSimpleCompact} any connected compact Lie group $ G $ that acts transitively and
effectively on $ \mathbb{S}^{2} $ has to be simple. By the Cartan-Killing classification there is a finite list of
all connected compact simple Lie groups that are not locally isomorphic (c.f. \cite{CartanKillingBosshardt}).
In \cite{SamelsonBeitraege1941} the
homology rings of these simple Lie groups are connected to the homology rings of Cartesian products of spheres
of several dimensions. It comes out that for dimensional reasons the only simple Lie group satisfying these conditions
is $ \text{SO}(3) $ with stabilizer $ \text{SO}(2) $ (c.f. \cite[Theorem~II'~and~Lemma~1]{MontgomeryTransSpheres}).
Remark that the action of $ \text{SO}(3) $ on $ \mathbb{S}^{2} $ is the one we discussed in Example \ref{exampleSO3}.
In a very similar way one proves the following result for connected non-compact Lie groups:

\begin{lemma}\label{SphereSemiSimple}
Any connected Lie group that acts transitively and locally effectively on $ \mathbb{S}^{2} $ is semisimple.
\end{lemma}

Again we just want to sketch the proof:
Onishchik proves this theorem for all compact homogeneous spaces with positive Euler characteristic and
also gives a decomposition of the Lie group into simple normal divisors (c.f. \cite[Theorem~2]{OnishchikIII}). The
proof
bases on the classification of the $ \mathfrak{t} $-subalgebras and $ \mathfrak{k} $-subalgebras of all semisimple
Lie groups which was also done by Onishchik in \cite{OnishchikII}. In this paper he constructs standard
$ \mathfrak{t} $- and $ \mathfrak{k} $-subgroups and shows that any $ \mathfrak{t} $- and $ \mathfrak{k} $-subgroup
of a semisimple Lie group is conjugate to a standard one on Lie algebra level (c.f.
\cite[Theorem~3~and~Theorem~4]{OnishchikII}).

Since semisimple Lie groups are well known and classified one can conclude the next theorem (c.f.
\cite[Theorem~3]{OnishchikIII}).

\begin{theorem}\label{ThmAllSphereActions}
Any connected Lie group that acts transitively and effectively on $ \mathbb{S}^{2} $ is equivalent to
$ \operatorname{SO}(3) $, $ \operatorname{SL}(3,\mathbb{R}) $ or $ \mathcal{L}_{3}^{\uparrow,+} $.
\end{theorem}

Remark that the actions mentioned in Theorem \ref{ThmAllSphereActions} are given in Example \ref{exampleSO3}, Example
\ref{exampleSL3} and Example \ref{exampleSO31}, respectively. This is also consistent with Theorem
\ref{ThmCompactSubActing} since $ \text{SO}(3) $ is a compact subgroup
of $ \mathcal{L}_{3}^{\uparrow,+} $ and $ \text{SL}(3,\mathbb{R}) $ that acts transitively on $ \mathbb{S}^{2} $.

%% file: Chapter2.tex
\chapter{Lie Bialgebras and $ r $-Matrices}\label{chapBialgrMatrix}

We discussed the geometry of homogeneous spaces and their connection to transitive Lie group actions. In a
way, it is natural to study those Lie groups in order to learn more about homogeneous spaces. In Subsection
\ref{sectiononish} we showed that for some examples only very few Lie groups are relevant. In general, if we ask the
homogeneous space to have additional structures the situation is of particular interest and it is natural to ask
whether there is an impact on the corresponding Lie group. The converse might be even more interesting: is there a
structure on a Lie group that acts on a manifold, that forces the manifold to be a
homogeneous space? It turns out that we have to treat this on infinitesimal level, i.e. on the corresponding
Lie algebra. These ideas are made concrete in Chapter \ref{rmatrixhomogeneous}, but shall suffice as a rough
motivation right now.

The structures we mentioned are $ r $-matrices. Those objects where developed in the course of
integrable lattice systems. This is because $ r $-matrix theory gives rise to Lax equations
(see e.g. \cite{RMatrixIntegrableSystem}).
In order to motivate
the notion of $ r $-matrices we embed them in the theory of Lie bialgebras, which are themselves interesting for
several reasons: first they occur as infinitesimal objects of Poisson-Lie groups that are involved again in
integrable systems via monodromy matrices (consider for example \cite{DressingSymmetriesBabelon}). We state the
connection between Lie bialgebras and Poisson-Lie groups in Section \ref{secLieBiPoLie} without proving this
statement due to V. G. Drinfel\textquoteright d. Second, Lie bialgebras are semi-classical limits of quantum groups and
conversely any Lie bialgebra can be quantized in that way. This is the famous theorem by P. Etingof
and D. Kazhdan (c.f. \cite{etingofbialg}). In the same manner $ R $-matrices appear as quantizations of
$ r $-matrices (c.f \cite[Theorem~9.2]{etingof2002lectures}) and the quantum Yang-Baxter equation as a quantization
of the classical Yang-Baxter equation (c.f. \cite[Theorem~9.3]{etingof2002lectures}). We introduce the classical
Yang-Baxter equation as a measure for $ r $-matrices.

For Lie bialgebras we need some general Lie algebra
cohomology and in particular the Chevalley-Eilenberg complex which we develop in Section \ref{sec1thomas}
together with some well-known results in this topic. Then we are able to make a motivated definition and
classification of Lie bialgebras in Section \ref{secLieBiPoLie}. Some interesting examples of Lie bialgebras
finally lead to $ r $-matrices in Section \ref{sec2thomas}. We classify them and also consider the case
when the corresponding Lie algebra is simple what will be helpful in the later sections. As a tool we define
the Etingof-Schiffmann subalgebra in Section \ref{SectionEtingof}, a finite-dimensional Lie subalgebra
in which the $ r $-matrix is non-degenerate.

\section{Lie Algebra Representations and Cohomology}\label{sec1thomas}

Working with Lie algebras is interesting because they can act on other
structures. We say that a Lie algebra $ (\mathfrak{g},\left[\cdot,\cdot\right]) $
over a field $ \mathbb{k} $ \textit{acts} on a $ \mathbb{k} $-vector space $ V $
if for every $ x\in\mathfrak{g} $ there is a linear map $ \rho_{x}\colon V\rightarrow V $ such that the assignment
$ \rho\colon\mathfrak{g}\ni x\mapsto\rho_{x}\in\mathfrak{gl}(V) $ is a Lie algebra homomorphism, i.e. for all $
x,y\in\mathfrak{g} $ one has
\begin{align}\label{eq1thomas}
\rho _{\left[x,y\right]}=\rho_{x}\rho_{y}-\rho_{y}\rho_{x},
\end{align}
where we endowed the general linear Lie algebra $ \mathfrak{gl}(V) $ with the commutator. In this case we say that
$ (\rho,V) $ is a \textit{representation} of $ \mathfrak{g} $ or if we want to emphasize that the Lie algebra
$ \mathfrak{g} $
acts on $ V $ we call $ V $ a \textit{$ \mathfrak{g} $-module}. We also use the short notation $ x.v=\rho_{x}(v) $, for
$ x\in\mathfrak{g} $, $ v\in V $. The following is inspired by \cite{WagemannLecture}.

\begin{example}
Let $ \mathfrak{g} $ be a Lie algebra over a field $ \mathbb{k} $ of characteristic
zero. We give some examples that are of great interest in our topic.
\begin{compactenum}
\item First of all, $ \mathfrak{g} $ can act trivially
on $ \mathbb{k} $, i.e. $ x.\lambda=0\in\mathbb{k} $ for all $ x\in\mathfrak{g} $, $ \lambda\in\mathbb{k} $. In this
way $ \mathbb{k} $ becomes a $ \mathfrak{g} $-module, called the \textit{trivial module}.
\item Furthermore, $ \mathfrak{g} $ can
act on itself by the \textit{adjoint action}, i.e. $ \text{ad}_{x}(y)=\left[x,y\right] $ for $
x,y\in\mathfrak{g} $. Indeed, $ \mathfrak{g} $ becomes a $ \mathfrak{g} $-module in this way because the Lie bracket is
bilinear by definition and Eq. (\ref{eq1thomas}) is just the Jacobi identity in this case.
\item In the same fashion we can define adjoint
actions of $ \mathfrak{g} $ on tensor powers of $ \mathfrak{g} $. Explicitly, for $ x\in\mathfrak{g} $ and factorizing
tensors $ y_{1}\tensor\cdots\tensor y_{n}\in\mathfrak{g}\tensor\cdots\tensor\mathfrak{g} $ we define
\begin{align}
\text{ad}_{x}^{(n)}(y_{1}\tensor\cdots\tensor y_{n})=\sum\limits_{k=1}^{n}y_{1}\tensor\cdots\tensor\text{ad}
_{x}(y_{k})\tensor\cdots\tensor y_{n}
\end{align}
and extend this linearly to $ \mathfrak{g}\tensor\cdots\tensor\mathfrak{g} $. To prove Eq. (\ref{eq1thomas}) it is
enough to consider factorizing elements $ y_{1}\tensor\cdots\tensor y_{n}\in\mathfrak{g}\tensor
\cdots\tensor\mathfrak{g} $.  By the Jacobi identity for $ \mathfrak{g} $ we get
\begin{align*}
(\text{ad}_{x}^{(n)}&\circ\text{ad}_{y}^{(n)}-\text{ad}_{y}^{(n)}\circ\text{ad}_{x}^{(n)})(y_{1}\tensor
\cdots\tensor y_{n})\\
&=\text{ad}_{x}^{(n)}\Bigg(\sum\limits_{k=1}^{n}y_{1}\tensor\cdots\tensor\text{ad}
_{y}(y_{k})\tensor\cdots\tensor y_{n}\Bigg)\\
&-\text{ad}_{y}^{(n)}\Bigg(\sum\limits_{k=1}^{n}y_{1}
\tensor\cdots\tensor\text{ad}_{x}(y_{k})\tensor\cdots\tensor y_{n}\Bigg)\\
&=\sum\limits_{k\neq
j=1}^{n}y_{1}\tensor\cdots\tensor\text{ad}_{y}(y_{k})\tensor\cdots\tensor\text{ad}_{x}(y_{j})\tensor\cdots
\tensor y_{n}\\
&+\sum\limits_{k=1}^{n}y_{1}\tensor\cdots\tensor\text{ad}_{x}(\text{ad}_{y}(y_{k}))\tensor\cdots\tensor y_{n}\\
&-\sum\limits_{k\neq
j=1}^{n}y_{1}\tensor\cdots\tensor\text{ad}_{x}(y_{k})\tensor\cdots\tensor\text{ad}_{y}(y_{j})\tensor\cdots
\tensor y_{n}\\
&-\sum\limits_{k=1}^{n}y_{1}\tensor\cdots\tensor\text{ad}_{y}(\text{ad}_{x}(y_{k}))\tensor\cdots\tensor y_{n}\\
&=\sum\limits_{k=1}^{n}(y_{1}\tensor\cdots\tensor\text{ad}_{x}(\text{ad}_{y}(y_{k}))\tensor\cdots\tensor y_{n}
-y_{1}\tensor\cdots\tensor\text{ad}_{y}(\text{ad}_{x}(y_{k}))\tensor\cdots\tensor y_{n})\\
&=\sum\limits_{k=1}^{n}y_{1}\tensor\cdots\tensor(\text{ad}_{x}(\text{ad}_{y}(y_{k}))-\text{ad}_{y}
(\text{ad}_{x}(y_{k})))
\tensor\cdots\tensor y_{n}\\
&=\sum\limits_{k=1}^{n}y_{1}\tensor\cdots\tensor\text{ad}_{\left[x,y\right]}(y_{k})\tensor\cdots\tensor y_{n}\\
&=\text{ad}_{\left[x,y\right]}^{(n)}(y_{1}\tensor\cdots\tensor y_{n}),
\end{align*}
where in the third equation the summand $ (k,j) $ of the first sum cancels with the summand $ (j,k) $ in the third
sum. The linearity is obvious and for this $ \mathfrak{g}\tensor\cdots\tensor\mathfrak{g} $ has the structure of a
$ \mathfrak{g} $-module. If not stated differently we always view $ \mathfrak{g}\tensor\cdots\tensor\mathfrak{g} $
as this $ \mathfrak{g} $-module (see Chapter \ref{sec2thomas}).
For example for $ n=2 $ we get
\begin{align}
\text{ad}_{x}^{(2)}(y_{1}\tensor y_{2})=\left[x,y_{1}\right]\tensor y_{2}+y_{1}\tensor\left[x,y_{2}\right],
\end{align}
for $ x\in\mathfrak{g} $, $ y_{1}\tensor y_{2}\in\mathfrak{g}\tensor\mathfrak{g} $. This motivates the notations
\begin{align}
\text{ad}_{x}^{(2)}(v)=(\text{ad}_{x}\tensor 1+1\tensor\text{ad}_{x})(v)=\left[x\tensor 1+1\tensor x,v\right],
\end{align}
for $ x\in\mathfrak{g} $ , $ v\in\mathfrak{g}\tensor\mathfrak{g}, $ where
$ 1\colon\mathfrak{g}\rightarrow\mathfrak{g} $ is the identity map.
\end{compactenum}
\end{example}

Having in mind these examples we consider again some arbitrary $ \mathfrak{g} $-module $ V $, where $ \mathfrak{g} $
is a Lie algebra over a field $ \mathbb{k} $ of characteristic zero. For a natural number $ n\in\mathbb{N} $ we
denote by $ C^{n}(\mathfrak{g},V)=\text{Hom}_{\mathbb{k}}(\Anti^{n}\mathfrak{g},V) $ the vector space of
$ \mathbb{k} $-linear maps from the $ n $-th exterior power $ \Anti^{n}\mathfrak{g} $ of $ \mathfrak{g} $ to $ V $.
The elements of $ C^{n}(\mathfrak{g},V) $ are called \textit{$ n $-cochains on $ \mathfrak{g} $ with values in $ V $}
(or just $ n $-cochains if there is no confusion which Lie algebra module is meant). We define a mapping
$ \delta_{n}\colon
C^{n}(\mathfrak{g},V)\rightarrow C^{n+1}(\mathfrak{g},V) $ from the set of $ n $-cochains to the set of $ (n+1)
$-cochains by
\begin{align*}
(\delta_{n}c)(x_{1}\wedge\cdots\wedge x_{n+1})&=\sum\limits_{k=1}^{n+1}x_{k}.c(x_{1}
\wedge\cdots\wedge\hat{x_{k}}\wedge\cdots\wedge x_{n+1})\\
&+\sum\limits_{k<j}c(\left[x_{k},x_{j}\right]\wedge\cdots\wedge\hat{x_{k}}\wedge\cdots\wedge\hat{x_{j}}
\wedge\cdots\wedge x_{n+1}),
\end{align*}
for all $ c\in C^{n}(\mathfrak{g},V) $ and $ x_{1}\wedge\cdots\wedge x_{n+1}\in\Anti^{n}\mathfrak{g} $. The point
in the first sum stands of course for the representation $ \rho $ corresponding to the $ \mathfrak{g} $-module
structure of $ V $, i.e. $ x_{k}.c(x_{1}\wedge\cdots\wedge\hat{x_{k}}\wedge\cdots\wedge x_{n+1})=\rho_{x_{k}}
(c(x_{1}\wedge\cdots\wedge\hat{x_{k}}\wedge\cdots\wedge x_{n+1})) $ and $ \hat{x_{k}} $
denotes that the element $ x_{k} $
is removed in the exterior product. In particular, we are interested in the kernel of $ \delta_{n} $. We denote this
vector space by $ Z_{n}(\mathfrak{g},V)=\ker(\delta_{n}) $ and call its elements \textit{$ n $-cocycles on $
\mathfrak{g} $ with values in $ V $} or just $ n $-cocycles. This space can be compared to the image of $ \delta_{n-1}
$, i.e. the vector space $ B_{n}(\mathfrak{g},V)=\text{im}(\delta_{n-1}) $ whose elements are said to be the
\textit{$ n $-coboundaries on $ \mathfrak{g} $ with values in $ V $} or again just $ n $-coboundaries, where we set $
B_{0}(\mathfrak{g},V)=\left\lbrace 0\right\rbrace $. The reader can check that
$ \delta_{n+1}\circ\delta_{n}=0 $. We often just write $ \delta=\delta_{n} $. The
equation $ \delta^{2}=0 $ implies $ B_{n}(\mathfrak{g},V)\subseteq Z_{n}(\mathfrak{g},V)\subseteq C^{n}(\mathfrak{g},V)
$, i.e. every $ n $-coboundary is a $ n $-cocycle. Then the set
\begin{align}
H^{n}(\mathfrak{g},V)=Z_{n}(\mathfrak{g},V)/B_{n}(\mathfrak{g},V)
\end{align}
is a well-defined $ \mathbb{k} $-vector space.

\begin{definition}
Let $ (\rho,V) $ be a Lie algebra representation of $ \mathfrak{g} $. The tuple $ (C^{n}(\mathfrak{g},V),\delta_{n})
_{n\in\mathbb{N}_{0}} $
is said to be the \textbf{Chevalley-Eilenberg complex} and $ \delta_{n}=\delta $ the \textbf{Chevalley-Eilenberg
differential}. Moreover, one calls $ H^{n}(\mathfrak{g},V) $ the \textbf{$ n $-th cohomology group on $ \mathfrak{g} $
with values in $ V $}, or just the $ n $-th cohomology group.
\end{definition}

\begin{remark}
Let us consider the low order cohomologies.
\begin{compactenum}
\item For $ n=0 $
one has $ c\in H^{0}(\mathfrak{g},V)=Z_{0}(\mathfrak{g},V)\subseteq V $ if and only if
\begin{align*}
0=\delta c(x)=x.c,
\end{align*}
for all $ x\in\mathfrak{g} $. This vector space is sometimes denoted by $ V^{\mathfrak{g}} $. Thus the $ 0 $-th
cohomology group consists of the elements of $ V $ that are annihilated by all elements of $ \mathfrak{g} $ via the
corresponding $ \mathfrak{g} $-module representation, i.e.
\begin{align}
H^{0}(\mathfrak{g},V)=V^{\mathfrak{g}}\subseteq V.
\end{align}
If $ \mathfrak{g} $ acts trivially on $ \mathbb{k} $ we obtain $ H^{0}(\mathfrak{g},\mathbb{k})=
\mathbb{k} $. If $ \mathfrak{g} $ acts on $ \mathfrak{g} $ by the adjoint action we get
\begin{align}
H^{0}(\mathfrak{g},\mathfrak{g})=\left\lbrace y\in\mathfrak{g}~|~\left[x,y\right]=0\text{ for all }x\in
\mathfrak{g}\right\rbrace
\end{align}
and if $ \mathfrak{g} $ acts on $ \mathfrak{g}\tensor\mathfrak{g} $ by $ \text{ad}^{(2)} $ we see that
\begin{align}
H^{0}(\mathfrak{g},\mathfrak{g}\tensor\mathfrak{g})=
\operatorname{span}_{\mathbb{k}}\left\lbrace y_{1}\tensor y_{2}\in\mathfrak{g}\tensor\mathfrak{g}
~|~\left[y_{1},x\right]\tensor y_{2}=y_{1}\tensor\left[x,y_{2}\right]\text{ for all }x\in\mathfrak{g}\right\rbrace.
\end{align}
\item For $ n=1 $ we observe
that $ c\in Z_{1}(\mathfrak{g},V) $ if and only if for all $ x_{1}, x_{2}\in\mathfrak{g} $
\begin{align}
0=(-1)^{1+1}x_{1}.c(x_{2})+(-1)^{2+1}x_{2}.c(x_{1})+(-1)^{1+2}c(\left[x_{1},x_{2}\right])
\end{align}
or equivalently
\begin{align}
c(\left[x_{1},x_{2}\right])=x_{1}.c(x_{2})-x_{2}.c(x_{1}),
\end{align}
i.e. if and only if $ c\colon\mathfrak{g}\rightarrow V $ is a \textit{derivation}. We denote the vector space of
derivations on $ \mathfrak{g} $ with values in $ V $ by $ \text{Der}(\mathfrak{g},V) $. Then $ Z_{1}(\mathfrak{g},V)=
\text{Der}(\mathfrak{g},V) $. Moreover, $ c\in B_{1}(\mathfrak{g},V) $ if and only if there is an element $ b\in
C^{0}(\mathfrak{g},V)=V $ such that $ c(x)=\delta b(x)=x.b $ for all $ x\in\mathfrak{g} $. Such elements are said to
be the \textit{inner derivations} on $ \mathfrak{g} $ with values in $ V $ and the vector space of them is denoted
by $ \text{IDer}(\mathfrak{g},V) $. Thus $ B_{1}(\mathfrak{g},V)=\text{IDer}(\mathfrak{g},V) $ and
\begin{align}
H^{1}(\mathfrak{g},V)=\text{Der}(\mathfrak{g},V)/\text{IDer}(\mathfrak{g},V)=\text{Out}(\mathfrak{g},V),
\end{align}
where the quotient of derivations and inner derivations results in the so-called \textit{outer derivations}
$ \text{Out}(\mathfrak{g},V) $ on $ \mathfrak{g} $ with values in $ V $. In particular, if $ \mathfrak{g} $ acts
trivially on $ \mathbb{k} $ we obtain
\begin{align}
H^{1}(\mathfrak{g},\mathbb{k})=\left.\left\lbrace
c\in C^{1}(\mathfrak{g},\mathbb{k})~\right|~c(\left[x_{1},x_{2}\right])=0\text{
 for all }x_{1}, x_{2}\in\mathfrak{g}\right\rbrace/\left\lbrace 0\right\rbrace=(\mathfrak{g}/\left[\mathfrak{g},
\mathfrak{g}\right])^{*},
\end{align}
where the star denotes the dual space and $ \left[\mathfrak{g},\mathfrak{g}\right] $ the $ \mathbb{k} $-span of all
brackets $ \left[x,y\right] $, for some $ x,y\in\mathfrak{g} $.
\item For $ n=2 $ one can prove (c.f \cite[Proposition~I.3.]{neebabelian})
that
\begin{align}
H^{2}(\mathfrak{g},V)\cong\text{Ext}(\mathfrak{g},V),
\end{align}
where $ \text{Ext}(\mathfrak{g},V) $ is the set of equivalence classes of abelian extensions of $ \mathfrak{g} $ by
$ V $. If $ \mathfrak{g} $ is a simple Lie algebra over $ \mathbb{C} $ and $ \Omega\in\mathfrak{g}\tensor\mathfrak{g} $
the \textit{Casimir element} associated to a non-degenerate invariant form on $ \mathfrak{g} $ we get
\begin{align}
(\Anti^{3}\mathfrak{g})^{\mathfrak{g}}=\mathbb{C}Z,
\end{align}
where $ Z=\left[\Omega\tensor 1,1\tensor\Omega\right] $.
\end{compactenum}
\end{remark}

Many details about Lie algebra cohomology can be found in \cite{hilton2013course}, \cite{neebabelian} and
\cite{onishchik2000lie}. To end this section we want to state some famous results about Lie algebra cohomology.

\begin{theorem}\label{Whitehead}
Let $ \mathfrak{g} $ be a complex semisimple Lie algebra and $ V $ a
finite-dimensional $ \mathfrak{g} $-module. Then
$ H^{1}(\mathfrak{g},V)=H^{2}(\mathfrak{g},V)=\left\lbrace 0\right\rbrace $. If $ V^{\mathfrak{g}}=
\left\lbrace 0\right\rbrace $ we even get $ H^{n}(\mathfrak{g},V)=\left\lbrace 0\right\rbrace $ for all
$ n\in\mathbb{N}_{0} $.
\end{theorem}

\begin{proof}
The
first statement is known as \textbf{Whitehead's Lemma} and can be found in \cite{varadarajan1984lie}. The second
statement is a consequence of this lemma (c.f. \cite{hochschildserrecohomo}).
\end{proof}

\begin{remark}
In \cite{Deninger1988} it is proven that if $ \mathfrak{g} $
is a nilpotent Lie algebra of dimension $ n\in\mathbb{N} $ the dimension of $ H^{k}(\mathfrak{g},\mathbb{k}) $ is at
least $ 2 $ for $ 1\leq k\leq n-1 $. This gives a helpful intuition that nilpotent Lie algebras have ``big'' Lie
algebra cohomologies while the Lie algebra cohomologies of semisimple Lie algebras are ``small''.
Finally, for a semisimple Lie algebras there is a splitting
of $ H^{n}(\mathfrak{g},V) $ into $ 1 $-dimensional modules with respect to a root decomposition of $ \mathfrak{g} $.
This is discussed in \cite{kostantcohomborel}.
\end{remark}

\section{Lie Bialgebras and Poisson-Lie Groups}\label{secLieBiPoLie}

Lie bialgebras and Poisson-Lie groups are well-studied objects. We
refer to the first section of the notes \cite{kosmann2004integrability} of Y. Kosmann-Schwarzbach for a introduction
to Lie bialgebras and to the fourth section to examine Poisson-Lie groups. Further sources are
\cite[Chapter~1]{chari1995guide}, \cite[Chapter~2]{etingof2002lectures}, \cite[Chapter~11]{laurent2012poisson},
\cite[Chapter~8]{majid2000foundations} and \cite[Chapter~10]{vaisman1994lectures}.

Consider a Lie algebra $ \mathfrak{g} $ over a field $ \mathbb{k} $ of characteristic zero. We want to equip $ 
\mathfrak{g} $ with some additional structure. As motivated from the last section we regard a linear map
\begin{align}
\gamma\colon\mathfrak{g}\rightarrow\mathfrak{g}\tensor\mathfrak{g}
\end{align}
that is a $ 1 $-cocycle on $ \mathfrak{g} $ with values in $ \mathfrak{g}\tensor\mathfrak{g} $, where we view
$ \mathfrak{g}\tensor\mathfrak{g} $ as a $ \mathfrak{g} $-module via $ \text{ad}^{(2)} $ as usual, i.e.
$ \gamma\in Z_{1}(\mathfrak{g},\mathfrak{g}\tensor\mathfrak{g}) $. It is common to use the short notation
\begin{align}\label{eq8thomas}
\gamma(x)=\sum x_{\left[1\right]}\tensor x_{\left[2\right]}=x_{\left[1\right]}\tensor x_{\left[2\right]}
\end{align}
for $ x\in\mathfrak{g} $ to represent the element $ \gamma(x)\in\mathfrak{g}\tensor\mathfrak{g} $. The
\textit{cocycle condition} reads
\begin{align}
\delta\gamma=0
\end{align}
or equivalently
\begin{align}\label{eq5thomas}
\gamma(\left[x,y\right])=\text{ad}_{x}^{(2)}(\gamma(y))-\text{ad}_{y}^{(2)}(\gamma(x)),
\end{align}
for all $ x,y\in\mathfrak{g} $. If we require $ \gamma $ to satisfy in addition
\begin{align}\label{eq3thomas}
x_{\left[1\right]}\tensor x_{\left[2\right]}=-x_{\left[2\right]}\tensor x_{\left[1\right]}
\end{align}
and
\begin{align}\label{eq4thomas}
\text{Alt}((\gamma\tensor 1)\gamma(x))=0
\end{align}
for all $ x\in\mathfrak{g} $, where $ \text{Alt} $ is the $ \mathbb{k} $-linear function $ \bigotimes^{3}\mathfrak{g}
\rightarrow\bigotimes^{3}\mathfrak{g} $ that satisfies on factorizing tensors
$ \text{Alt}(x\tensor y\tensor z)=x\tensor y\tensor z+y\tensor z\tensor x
+z\tensor x\tensor y $, the Lie algebra $ \mathfrak{g} $ becomes also a Lie coalgebra, or equivalently the dual
$ \mathfrak{g}^{*} $ of $ \mathfrak{g} $ gets a Lie algebra. The Lie bracket on $ \mathfrak{g}^{*} $ is given by the
left transpose of $ \gamma $, i.e.
\begin{align}
^{t}\gamma\colon\mathfrak{g}^{*}\tensor\mathfrak{g}^{*}\rightarrow\mathfrak{g}^{*}
\end{align}
is a Lie bracket on $ \mathfrak{g}^{*} $. Indeed we can define a bracket $ \left[\cdot,\cdot\right]_{\mathfrak{g}^{*}}
\colon\mathfrak{g}^{*}\tensor\mathfrak{g}^{*}\rightarrow\mathfrak{g}^{*} $ on $ \mathfrak{g}^{*} $ by
\begin{align}\label{eq2thomas}
\left[\xi,\eta\right]_{\mathfrak{g}^{*}}={}^t\gamma(\xi\tensor\eta),
\end{align}
for all $ \xi,\eta\in\mathfrak{g}^{*} $. If we introduce the non-degenerate dual pairing $ \langle\cdot,\cdot\rangle
\colon\bigotimes^{n}\mathfrak{g}^{*}\times\bigotimes^{n}\mathfrak{g}\rightarrow\mathbb{k} $ we can prove that
Eq. (\ref{eq2thomas}) defines a Lie bracket on $ \mathfrak{g}^{*} $ if and only if $ \gamma $ satisfies conditions
(\ref{eq3thomas}) and (\ref{eq4thomas}). Of course condition (\ref{eq3thomas}) can immediately be absorbed if we
require $ \gamma $ to be a map $ \mathfrak{g}\rightarrow\mathfrak{g}\wedge\mathfrak{g} $. The skew-symmetry of
$ \left[\cdot,\cdot\right]_{\mathfrak{g}^{*}} $ follows from condition (\ref{eq3thomas}) because for all
$ x\in\mathfrak{g} $, $ \xi,\eta\in\mathfrak{g}^{*} $ one has
\begin{align*}
\langle\left[\xi,\eta\right]_{\mathfrak{g}^{*}},x\rangle&=\langle^t\gamma(\xi\tensor\eta),x\rangle=\langle
\xi\tensor\eta,\gamma(x)\rangle\\
&=\langle\xi\tensor\eta,x_{\left[1\right]}\tensor x_{\left[2\right]}\rangle=\langle\eta\tensor\xi,
x_{\left[2\right]}\tensor x_{\left[1\right]}\rangle\\
&=-\langle\eta\tensor\xi,x_{\left[1\right]}\tensor x_{\left[2\right]}\rangle=-\langle^t\gamma
(\eta,\xi),x\rangle\\
&=\langle-\left[\eta,\xi\right]_{\mathfrak{g}^{*}},x\rangle.
\end{align*}
This implies $ \left[\xi,\eta\right]_{\mathfrak{g}^{*}}=-\left[\eta,\xi\right]_{\mathfrak{g}^{*}} $ because the
dual pairing is non-degenerate. Arranging the above equations in another order we see that the skew-symmetry of
$ \left[\cdot,\cdot\right]_{\mathfrak{g}^{*}} $ implies condition (\ref{eq3thomas}). Moreover, the Jacobi identity of
$ \left[\cdot,\cdot\right]_{\mathfrak{g}^{*}} $ is equivalent to condition (\ref{eq4thomas}), the so called
\textit{coJacobi identity}. To prove this, observe that for $ x\in\mathfrak{g} $ and
$ \xi,\eta,\zeta\in\mathfrak{g}^{*} $ one has
\begin{align*}
\langle\big[\big[\xi,\eta\big]_{\mathfrak{g}^{*}},\zeta\big]_{\mathfrak{g}^{*}},x\rangle&=
\langle^t\gamma(^t\gamma(\xi\tensor\eta)\tensor\zeta),x\rangle=\langle^t\gamma(\xi\tensor\eta)\tensor\zeta,
\gamma(x)\rangle\\
&=\langle\xi\tensor\eta\tensor\zeta,(\gamma\tensor 1)\circ\gamma(x)\rangle.
\end{align*}
Thus every term of the coJacobi identity corresponds to one term of the Jacobi identity of $ \left[\cdot,\cdot\right]
_{\mathfrak{g}^{*}} $ and the equivalence is proved. After recognizing the powerful impact of the above considerations
we give the well-motivated definition of a Lie bialgebra.

\begin{definition}[Lie Bialgebra]
A Lie bialgebra $ (\mathfrak{g},\left[\cdot,\cdot\right],\gamma) $ is a Lie algebra $ (\mathfrak{g},
\left[\cdot,\cdot\right]) $ together with a map
\begin{align}
\gamma\colon\mathfrak{g}\rightarrow\mathfrak{g}\wedge\mathfrak{g}
\end{align}
satisfying the coJacobi identity (\ref{eq4thomas}) and the cocycle identity (\ref{eq5thomas}). The cocycle $ \gamma $
is called the Lie bialgebra structure of $ (\mathfrak{g},\left[\cdot,\cdot\right],\gamma) $.
\end{definition}
We already proved the following proposition in the above lines.

\begin{proposition}\label{pro1thomas}
A Lie algebra $ (\mathfrak{g},\left[\cdot,\cdot\right]) $ endowed with a $ 1 $-cocycle $ \gamma\colon\mathfrak{g}
\rightarrow\mathfrak{g}\tensor\mathfrak{g} $
is a Lie bialgebra $ (\mathfrak{g},\left[\cdot,\cdot\right],
\gamma) $ if and only if
\begin{align}
\left[\xi,\eta\right]_{\mathfrak{g}^{*}}={}^t\gamma(\xi\tensor\eta),
\end{align}
where $ \xi,\eta\in\mathfrak{g}^{*} $, defines a Lie bracket on $ \mathfrak{g}^{*} $.
\end{proposition}
Every new object goes hand in hand with its morphisms and substructures.

\begin{definition}
Let $ (\mathfrak{g},\left[\cdot,\cdot\right]_{\mathfrak{g}},\gamma_{\mathfrak{g}}) $ and
$ (\mathfrak{k},\left[\cdot,\cdot\right]_{\mathfrak{k}},\gamma_{\mathfrak{k}}) $ be two Lie bialgebras.
\begin{compactenum}
\item A Lie algebra morphism $ \phi\colon\mathfrak{g}\rightarrow\mathfrak{k} $ is said to be a \textbf{morphism of Lie
bialgebras} if it respects the cocycles, i.e.
\begin{align}
(\phi\tensor\phi)(\gamma_{\mathfrak{g}}(x))=\gamma_{\mathfrak{k}}(\phi(x)),
\end{align}
for all $ x\in\mathfrak{g} $. The condition being a Lie algebra morphism reads
\begin{align}
\phi(\left[x,y\right]_{\mathfrak{g}})=\left[\phi(x),\phi(y)\right]_{\mathfrak{k}},
\end{align}
for all $ x,y\in\mathfrak{g} $.
\item A Lie subalgebra $ \mathfrak{h} $ of $ \mathfrak{g} $ is said to be a \textbf{Lie subbialgebra} if the
cocycle on $ \mathfrak{g} $ restricts to a map
\begin{align}
\gamma_{\mathfrak{g}}\colon\mathfrak{h}\rightarrow\mathfrak{h}\wedge\mathfrak{h},
\end{align}
i.e. $ \gamma_{\mathfrak{g}}(\mathfrak{h})\subseteq\mathfrak{h}\wedge\mathfrak{h} $. The condition being a Lie
subalgebra of $ \mathfrak{g} $ reads
$ \left[\mathfrak{h},\mathfrak{h}\right]_{\mathfrak{g}}\subseteq\mathfrak{h} $.
\item A Lie ideal $ \mathfrak{h} $ of $ \mathfrak{g} $ is said to be \textbf{Lie coideal} if
\begin{align}
\gamma_{\mathfrak{g}}(\mathfrak{h})\subseteq\mathfrak{g}\tensor\mathfrak{h}+\mathfrak{h}\tensor\mathfrak{g}.
\end{align}
The condition being a Lie ideal of $ \mathfrak{g} $ is that $ \mathfrak{h} $ is a left ideal of $ \mathfrak{g} $
with respect to the Lie bracket (or equivalently a right ideal or an ideal), i.e. $ \left[\mathfrak{h},\mathfrak{g}
\right]_{\mathfrak{g}}\subseteq\mathfrak{h} $.
\end{compactenum}
\end{definition}

Remark that if $ \mathfrak{h} $ is a Lie ideal of $ (\mathfrak{g},\left[\cdot,\cdot\right]) $ the quotient Lie algebra
$ (\mathfrak{g}/\mathfrak{h},\left[\cdot,\cdot\right]) $ inherits a Lie bialgebra structure from $ \mathfrak{g} $
if and only if $ \mathfrak{h} $ is a Lie coideal. We already realized in Proposition \ref{pro1thomas} that the cocycle
$ \gamma\colon\mathfrak{g}\rightarrow\mathfrak{g}\wedge\mathfrak{g} $ on $ \mathfrak{g} $ that satisfies the coJacobi
identity is equivalent to a Lie bracket on the dual $ \mathfrak{g}^{*} $ of $ \mathfrak{g} $. We expand this result
in the sense that if $ \gamma $ satisfies the above conditions and the cocylce condition, not only $ \mathfrak{g} $ but
also $ \mathfrak{g}^{*} $ has the structure of a Lie bialgebra. Thus the notion of Lie bialgebra is totally self dual.

\begin{proposition}\label{pro2thomas}
For a finite-dimensional Lie bialgebra $ (\mathfrak{g},\left[\cdot,\cdot\right],\gamma) $ consider the dual maps
\begin{align}
^{t}\left[\cdot,\cdot\right]\colon\mathfrak{g}^{*}\rightarrow\mathfrak{g}^{*}\wedge\mathfrak{g}^{*}\text{   and   }
^{t}\gamma\colon\mathfrak{g}^{*}\wedge\mathfrak{g}^{*}\rightarrow\mathfrak{g}^{*}
\end{align}
of $ \left[\cdot,\cdot\right] $ and $ \gamma $. Then $ (\mathfrak{g}^{*},{}^{t}\gamma,{}^{t}\left[\cdot,\cdot\right]) $
is a Lie bialgebra, where $ ^{t}\gamma $ is the Lie bracket and $ ^{t}\left[\cdot,\cdot\right] $ the cocycle on
$ \mathfrak{g}^{*} $.
\end{proposition}

\begin{proof}
We already proved that $ (\mathfrak{g}^{*},{}^t\gamma) $ is a Lie algebra if and only if $ \gamma $ satisfies
conditions (\ref{eq3thomas}) and (\ref{eq4thomas}). The only thing left to prove the proposition is that
$ ^{t}\left[\cdot,\cdot\right] $ is a cocycle. The restriction to finite-dimensional Lie bialgebras is necessary
if we want $ ^{t}(^{t}\left[\cdot,\cdot\right]) $ to define the Lie bracket $ \left[\cdot,\cdot\right] $ on
$ (\mathfrak{g}^{*})^{*}=\mathfrak{g} $ again. Let us compute for $ x,y\in\mathfrak{g} $ and $ \xi,\eta\in
\mathfrak{g}^{*} $
\begin{align*}
\langle^{t}\left[\cdot,\cdot\right](^{t}\gamma(\xi\wedge\eta)),x\wedge y\rangle&=\langle^{t}\gamma(\xi\wedge\eta)
,\left[x,y\right]\rangle\\
&=\langle\xi\wedge\eta,\gamma(\left[x,y\right])\rangle\\
&=\langle\xi\wedge\eta,(\text{ad}_{x}^{(2)}(\gamma(y))-\text{ad}_{y}^{(2)}(\gamma(x)))\rangle\\
&=\langle\xi\wedge\eta,(\left[x,y_{\left[1\right]}\right]\wedge y_{\left[2\right]}+y_{\left[1\right]}\wedge
\left[x,y_{\left[2\right]}\right]\\
&-\left[y,x_{\left[1\right]}\right]\wedge x_{\left[2\right]}-x_{\left[1\right]}\wedge
\left[y,x_{\left[2\right]}\right])\rangle\\
&=\langle(^{t}\left[\cdot,\cdot\right](\xi)\wedge\eta-\xi\wedge(^{t}\left[\cdot,\cdot\right](\eta))),x\wedge
y_{\left[1\right]}\wedge y_{\left[2\right]}\rangle\\
&-\langle(^{t}\left[\cdot,\cdot\right](\xi)\wedge\eta-\xi\wedge(^{t}\left[\cdot,\cdot\right])(\eta)),
x_{\left[1\right]}\wedge x_{\left[2\right]}\wedge y\rangle\\
&=\langle(-\xi_{\left[1\right]}\wedge\eta\wedge\xi_{\left[2\right]}+\eta_{\left[1\right]}\wedge\xi\wedge\eta_
{\left[2\right]}),x\wedge\gamma(y)\rangle\\
&-\langle(\eta\wedge\xi_{\left[1\right]}\wedge\xi_{\left[2\right]}-\xi\wedge\eta_{\left[1\right]}\wedge
\eta_{\left[2\right]}),\gamma(x)\wedge y\rangle\\
&=\langle(\left[\xi,\eta_{\left[1\right]}\right]_{\mathfrak{g}^{*}}\wedge\eta_{\left[2\right]}+
\eta_{\left[1\right]}\wedge\left[\xi,\eta_{\left[2\right]}\right]_{\mathfrak{g}^{*}}\\
&-(\left[\eta,\xi_{\left[1\right]}\right]_{\mathfrak{g}^{*}}\wedge\xi_{\left[2\right]}+
\xi_{\left[1\right]}\wedge\left[\eta,\xi_{\left[2\right]}\right]_{\mathfrak{g}^{*}})),x\wedge y\rangle\\
&=\langle((\text{ad}_{\xi}^{*})^{(2)}(^{t}\left[\cdot,\cdot\right](\eta))-(\text{ad}_{\eta}^{*})^{(2)}(^{t}
\left[\cdot,\cdot\right](\xi))),x\wedge y\rangle.
\end{align*}
The claim follows because the dual pairing is non-degenerate. Then $ ^{t}\left[\cdot,\cdot\right] $ is a cocycle on 
$ \mathfrak{g}^{*} $ if $ \gamma $ is a cocylce on $ \mathfrak{g} $. Using the equations above in different order one
can check that if $ ^{t}\left[\cdot,\cdot\right] $ is a cocycle on $ \mathfrak{g}^{*} $ also $ \gamma $ is a cocylce on
$ \mathfrak{g} $ and in the finite-dimensional case these are indeed equivalent conditions. A nice pictorial proof of
this statement can be found in \cite[Proposition~2.2]{etingof2002lectures}.
\end{proof}
Now we can introduce a first example of a Lie bialgebra.

\begin{example}\label{ExampleXYBialg}
Consider a noncommutative Lie algebra $ (\mathfrak{g},\left[\cdot,\cdot\right]) $ of dimension $ 2 $ over a field
$ \mathbb{k} $ of characteristic zero and take a basis $ \left\lbrace X,Y\right\rbrace $ of
$ \mathfrak{g} $ such that
\begin{align}\label{bracketXY}
\left[X,Y\right]=X.
\end{align}
By the bilinearity and antisymmetry of $ \left[\cdot,\cdot\right] $ condition (\ref{bracketXY}) determines the
Lie bracket completely. Define a map $ \gamma\colon\mathfrak{g}\rightarrow\mathfrak{g}\wedge\mathfrak{g} $ by
$ \gamma(X)=0 $, $ \gamma(Y)=-X\wedge Y $ and linear expansion and prove that this structures
$ (\mathfrak{g},\left[\cdot,\cdot\right]) $ as a Lie bialgebra. Of course it is enough to prove this just for
the basis elements. Indeed, the cocycle condition holds:
\begin{align*}
\text{ad}_{X}^{(2)}(\gamma(Y))-\text{ad}_{Y}^{(2)}(\gamma(X))&=\text{ad}_{X}^{(2)}(-X\wedge Y)-0\\
&=\text{ad}_{X}^{(2)}(Y\tensor X)-\text{ad}_{X}^{(2)}(X\tensor Y)\\
&=\left[ X,Y\right]\tensor X+Y\tensor\left[ X,X\right]-\left[ X,X\right]\tensor Y-X\tensor\left[ X,Y\right]\\
&=X\tensor X-X\tensor X\\
&=0\\
&=\gamma(X)\\
&=\gamma(\left[ X,Y\right]).
\end{align*}
We also check the coJacobi identity for $ X $
\begin{align*}
\text{Alt}((\gamma\tensor 1)\gamma(X))=\text{Alt}(0)=0
\end{align*}
and for $ Y $
\begin{align*}
\text{Alt}((\gamma\tensor 1)\gamma(Y))&=\text{Alt}((\gamma\tensor 1)(Y\tensor X-X\tensor Y))\\
&=\text{Alt}((Y\tensor X-X\tensor Y)\tensor X)-0\\
&=Y\tensor X\tensor X+X\tensor Y\tensor X+X\tensor X\tensor Y\\
&-X\tensor Y\tensor X-X\tensor X\tensor Y-Y\tensor X\tensor X\\
&=0.
\end{align*}
Thus $ (\mathfrak{g},\left[\cdot,\cdot\right],\gamma) $ is a Lie bialgebra.
\end{example}

There are many more examples of Lie bialgebras. Some of them can be found in
\cite[Section~2.2.2]{etingof2002lectures}, \cite[Section~1]{kosmann2004integrability}
or \cite[Section~8]{majid2000foundations}.
A very important motivation for Lie bialgebras comes from the famous
theorem of Sophus Lie, stating that the categories of simply connected Lie groups and finite-dimensional
Lie algebras are equivalent. This gives the \textit{Lie functor}
\begin{align}\label{eq7thomas}
\text{Lie}(G)=\mathfrak{g},
\end{align}
that assigns to every simply connected Lie group $ G $ its infinitesimal counterpart, i.e. a finite-dimensional
Lie algebra $ \mathfrak{g} $. Our question is now: what is the global counterpart of a Lie bialgebra?
To answer this we give a very short repetition on Poisson manifolds. For a conceptual introduction consider for
example \cite{laurent2012poisson}, \cite{vaisman1994lectures} or \cite{waldmannbuch1}.
A smooth manifold $ M $ together with a \textit{Poisson bracket}
\begin{align}\label{PoissonBracket}
\left\lbrace\cdot,\cdot\right\rbrace\colon\Cinfty(M)\times\Cinfty(M)\rightarrow\Cinfty(M)
\end{align}
for its smooth real-valued functions $ \Cinfty(M) $ is said to be a \textit{Poisson manifold}. The axioms that the
$ \mathbb{k} $-bilinear map in Eq. (\ref{PoissonBracket}) has to satisfy to define a Poisson bracket are
skew-symmetry, Jacobi identity
\begin{align}
\left\lbrace f,\left\lbrace g,h\right\rbrace\right\rbrace=
\left\lbrace\left\lbrace f,g\right\rbrace,h\right\rbrace+\left\lbrace g,\left\lbrace f,h\right\rbrace\right\rbrace
\end{align}
and the Leibniz rule
\begin{align}
\left\lbrace f,gh\right\rbrace=\left\lbrace f,g\right\rbrace h+g\left\lbrace f,h\right\rbrace
\end{align}
for all $ f,g,h\in\Cinfty(M) $. The smooth functions on $ M $ are viewed as a \textit{Poisson algebra} via its
commutative associative algebra structure of pointwise multiplication together with a Poisson bracket. Such Poisson
brackets are in one-to-one correspondence with bivector fields $ \pi\in\Gamma^{\infty}(\Lambda^{2}TM) $ on $ M $ via
\begin{align}\label{PoissonTensor}
\left\lbrace f,g\right\rbrace=\pi(\mathrm{d}f,\mathrm{d}g)
\end{align}
for all $ f,g\in\Cinfty(M) $, where $ \mathrm{d} $ denotes the exterior derivative. For a Poisson manifold
$ (M,\left\lbrace\cdot,\cdot\right\rbrace) $ we call
$ \pi $ the corresponding \textit{Poisson bivector} or \textit{Poisson structure}.
A smooth map $ \Phi\colon(M_{1},
\left\lbrace\cdot,\cdot\right\rbrace_{1})\rightarrow(M_{2},\left\lbrace\cdot,\cdot\right\rbrace_{2}) $ between
two Poisson manifolds $ (M_{1},\left\lbrace\cdot,\cdot\right\rbrace_{1}) $ and
$ (M_{2},\left\lbrace\cdot,\cdot\right\rbrace_{2}) $ is said to be a \textit{Poisson map} if for all
$ f,g\in\Cinfty(M_{2}) $ one has
\begin{align}
\Phi^{*}\left\lbrace f,g\right\rbrace_{2}=\left\lbrace\Phi^{*}f,\Phi^{*}g\right\rbrace_{1}
\end{align}
where $ \Phi^{*} $ denotes the \textit{pull back} of $ \Phi $, i.e. $ \Phi^{*}f=f\circ\Phi $ for all
$ f\in\Cinfty(M_{2}) $.
A very important class of examples of Poisson manifolds is given by symplectic manifolds.
Remember that a \textit{symplectic manifold} is a smooth manifold $ M $ together with a
closed non-degenerated $ 2 $-form $ \omega\in\Gamma^{\infty}(\Anti^{2}T^{*}M) $, called the
\textit{symplectic structure}. In fact, for any symplectic
structure $ \omega\in\Gamma^{\infty}(\Anti^{2}T^{*}M) $ we can always define a Poisson bracket by
\begin{align}\label{SymplPoissonBracket}
\left\lbrace f,g\right\rbrace=\omega(X_{f},X_{g})
\end{align}
for all $ f,g\in\Cinfty(M) $. The objects $ X_{f} $ and $ X_{g} $
in Eq. (\ref{SymplPoissonBracket}) are said to be the \textit{Hamiltonian vector fields} of $ f $ and $ g $,
respectively and they are defined by
\begin{align}
X_{f}=(\mathrm{d}f)^{\sharp},
\end{align}
where $ \sharp\colon T^{*}M\rightarrow TM $ is the \textit{musical isomorphism}. Thus another way of writing Eq.
(\ref{SymplPoissonBracket}) is
\begin{align}\label{eq06thomas}
\left\lbrace f,g\right\rbrace=X_{g}f,
\end{align}
for $ f,g\in\Cinfty(M) $.
It is known that for any $ p\in M $ the map
\begin{align}\label{eq20t}
\tilde{\pi}\colon T_{p}^{*}M\ni\alpha_{p}\mapsto\pi_{p}(\cdot,\alpha_{p})\in T_{p}M
\end{align}
is surjective if $ (M,\omega) $ is a symplectic manifold and that
\begin{align}\label{SymplPoisBiVect}
\pi(\mathrm{d}f,\mathrm{d}g)=\omega(X_{f},X_{g})
\end{align}
holds for all $ f,g\in\Cinfty(M) $.

An interesting structure occurs if we combine the concept of Poisson manifolds with the
notion of Lie groups.

\begin{definition}[Poisson-Lie Group]
A Poisson manifold $ (G,\left\lbrace\cdot,\cdot\right\rbrace) $ that is also a Lie group is said to be a
Poisson-Lie group if the group multiplication $ \mathrm{m}\colon G\times G\rightarrow G $ is a Poisson map.
\end{definition}

\begin{remark}
Let $ (G,\left\lbrace\cdot,\cdot\right\rbrace) $ be a Poisson-Lie group and denote by
$ \lambda_{u}\colon G\ni v\mapsto uv\in G $ and $ \rho_{u}\colon G\ni v\mapsto vu\in G $ the left and
right multiplication on $ G $ with some element $ u\in G $, respectively. Then the group multiplication $ m $
is a Poisson map, i.e.
\begin{align}\label{eq6thomas}
\left\lbrace f,g\right\rbrace(uv)=\left\lbrace f,g\right\rbrace(\mathrm{m}(u\times v))=\left\lbrace f\circ\lambda_{u},
g\circ\lambda_{u}\right\rbrace(v)+\left\lbrace f\circ\rho_{v},g\circ\rho_{v}\right\rbrace(u),
\end{align}
where $ f,g\in\Cinfty(G) $ and $ u,v\in G $.
There is also an equivalent condition to (\ref{eq6thomas}) for the corresponding Poisson bivector
$ \pi\in\Gamma^{\infty}(\Anti^{2}TG) $, namely
\begin{align}
\pi(uv)=(T_{u}\rho_{v}\tensor T_{u}\rho_{v})\pi(u)+(T_{v}\lambda_{u}\tensor
T_{v}\lambda_{u})\pi(v),
\end{align}
for all $ u,v\in G $. In particular, we see that if $ G $ is a Poisson-Lie group with identity element $ e\in G $ we
get
\begin{align}\label{eq9thomas}
\pi(e)=0.
\end{align}
\end{remark}

After this recap we come back to our question:
a result by Drinfel\textquoteright d shows the connection between Poisson-Lie groups and Lie bialgebras in analogy to
the Lie functor given in Eq. (\ref{eq7thomas}).

\begin{theorem}[Drinfel\textquoteright d]\label{TheoremDrinFunctor}
There is a functor
\begin{align}
\operatorname{Drin}\colon G\mapsto\operatorname{Lie}(G)
\end{align}
between the categories of simply connected Poisson-Lie groups and
finite-dimensional Lie bialgebras and this functor gives an equivalence of these categories.
\end{theorem}

Thus, one can see Lie bialgebras in a bigger context. The proof is not necessary for our purpose and
we just refer to \cite[Theorem~2.2]{etingof2002lectures}. For instance, this theorem motivates the notion of the
\textit{dual} $ G^{*} $ of a Poisson-Lie group $ G $ , where $ \mathfrak{g}^{*}=\text{Drin}(G^{*}) $ is the dual of
$ \mathfrak{g}=\text{Drin}(G) $. In the next section we see some special kinds of Lie bialgebras. They have
global counterparts in conformity with the functor of the last theorem. Consider
\cite[Section~9.4]{etingof2002lectures} for the definitions of these so called \textit{coboundary Hopf algebras}
and the connection to their infinitesimal objects.

\section{Coboundary Lie Bialgebras and the Classical Yang-Baxter Equation}
\label{sec2thomas}

Here we give a nice connection of the last two sections. For convenient literature we refer to
\cite[Chapter~2]{chari1995guide}, \cite[Chapter~3]{etingof2002lectures} and \cite[Section~2]{kosmann2004integrability}.
Consider Lie bialgebras $ (\mathfrak{g},
\left[\cdot,\cdot\right],\gamma) $ with a cocycle $ \gamma $ on $ \mathfrak{g} $ with values in $ \mathfrak{g}\tensor
\mathfrak{g} $. How does the situation change when $ \gamma $ is also a coboundary on $ \mathfrak{g} $ with values in $
\mathfrak{g}\tensor\mathfrak{g} $, i.e. if there is an element $ r\in C^{0}(\mathfrak{g},\mathfrak{g}\tensor
\mathfrak{g})=\mathfrak{g}\tensor\mathfrak{g} $ such that $ \gamma=\delta r $? And what conditions has an arbitrary
element $ r\in\mathfrak{g}\tensor\mathfrak{g} $ to obey to define a cocycle $ \gamma=\delta r $ such that $
(\mathfrak{g},\left[\cdot,\cdot\right],\gamma) $ becomes a Lie bialgebra? Because $ \delta^{2}=0 $ the element $ \delta
r $ will automatically be a cocycle, but conditions (\ref{eq3thomas}) and (\ref{eq4thomas}) do not have to be valid
in the general case. A powerful instrument to measure if the coJacobi identity is also given
is the so called \textit{classical Yang-Baxter map} $ \text{CYB}\colon\mathfrak{g}\tensor\mathfrak{g}\rightarrow
\mathfrak{g}\tensor\mathfrak{g}\tensor\mathfrak{g}=\mathfrak{g}^{\tensor 3} $ defined by
\begin{align}\label{CYBmap}
\text{CYB}(r)=\left[ r_{12},r_{13}\right]+\left[ r_{12},r_{23}\right]+\left[ r_{13},r_{23}\right],
\end{align}
for all elements $ r=\sum r_{1}\tensor r_{2}=r_{1}\tensor r_{2}\in\mathfrak{g}\tensor\mathfrak{g} $ (we use a notation
very similar to (\ref{eq8thomas})), where we defined
\begin{align*}
&r_{12}=r_{1}\tensor r_{2}\tensor 1,\\
&r_{13}=r_{1}\tensor 1\tensor r_{2},\\
&r_{23}=1\tensor r_{1}\tensor r_{2}\in\mathcal{U}(\mathfrak{g})^{\tensor 3}
\end{align*}
on the universal enveloping algebra $ \mathcal{U}(\mathfrak{g}) $ of $ \mathfrak{g} $.
The brackets in (\ref{CYBmap}) are intuitive abbreviations defined by
\begin{align*}
\left[ r_{12},r_{13}\right]&=\left[r_{1},r_{1}'\right]\tensor r_{2}\tensor r_{2}',\\
\left[ r_{12},r_{23}\right]&=r_{1}\tensor\left[r_{2},r_{1}'\right]\tensor r_{2}',\\
\left[ r_{13},r_{23}\right]&=r_{1}\tensor r_{1}'\tensor\left[r_{2},r_{2}'\right],
\end{align*}
where we denoted the second summation by $ r=r_{1}'\tensor r_{2}' $.
With this map we are able to determine whenever $ r $ gives rise to a Lie bialgebra via the above construction.

\begin{theorem}\label{the1thomas}
Given a Lie algebra $ (\mathfrak{g},\left[\cdot,\cdot\right]) $ and an element $ r\in\mathfrak{g}\tensor\mathfrak{g} $
that is skew-symmetric. The triple $ (\mathfrak{g},\left[\cdot,\cdot\right],
\delta r) $ is a Lie bialgbra if and only if $ \operatorname{CYB}(r)\in\mathfrak{g}^{\tensor 3} $ is
a $ \mathfrak{g} $-invariant element. The condition for $ \operatorname{CYB}(r) $ to be $ \mathfrak{g} $-invariant
reads
\begin{align}
\operatorname{ad}_{x}^{(3)}(\operatorname{CYB}(r))=\left[x,\operatorname{CYB}(r)\right]=0,
\end{align}
for all $ x\in\mathfrak{g} $.
\end{theorem}

\begin{proof}
These lines are inspired by \cite[Proposition~on~page~17]{kosmann2004integrability}. A geometric proof can be found in
\cite[Theorem~3.1]{etingof2002lectures}. Our aim is to prove for any $ x\in\mathfrak{g} $ the identity
\begin{align}\label{FormulaOfHorror}
\text{Alt}((\gamma\tensor 1)\gamma(x))=-\text{ad}_{x}^{(3)}(\text{CYB}(r))
\end{align}
by only using the skew-symmetry and Jacobi identity of $ \left[\cdot,\cdot\right] $, the skew-symmetry of $ r $ and
\begin{align}
\gamma(x)=\delta r(x)=\text{ad}_{x}^{(2)}r.
\end{align}
We start writing out the left side of Eq. (\ref{FormulaOfHorror}):
\begin{align*}
\gamma(x)=\text{ad}_{x}^{(2)}r=\left[x,r_{1}\right]\tensor r_{2}+r_{1}\tensor\left[x,r_{2}\right]
\end{align*}
implies
\begin{align*}
(\gamma\tensor 1)\gamma(x)&=\gamma(\left[x,r_{1}\right])\tensor r_{2}+\gamma(r_{1})\tensor\left[x,r_{2}\right]\\
&=\text{ad}_{\left[x,r_{1}\right]}^{(2)}r\tensor r_{2}+\text{ad}_{r_{1}}^{(2)}r\tensor\left[x,r_{2}\right]\\
&=\left[\left[x,r_{1}\right],r_{1}'\right]\tensor r_{2}'\tensor r_{2}
+r_{1}'\tensor\left[\left[x,r_{1}\right],r_{2}'\right]\tensor r_{2}\\
&+\left[r_{1},r_{1}'\right]\tensor r_{2}'\tensor\left[x,r_{2}\right]
+r_{1}'\tensor\left[r_{1},r_{2}'\right]\tensor\left[x,r_{2}\right].
\end{align*}
To separate the two summations over the tensor components of $ r $ we denoted the second one by
$ r_{1}'\tensor r_{2}' $.
After applying $ \text{Alt} $ we obtain twelve terms. The skew-symmetry of
$ \left[\cdot,\cdot\right] $ and $ r $ imply
\begin{align*}
\text{Alt}((\gamma\tensor 1)\gamma(x))&=\left[\left[x,r_{1}\right],r_{1}'\right]\tensor r_{2}'\tensor r_{2}
-r_{1}\tensor\left[\left[x,r_{2}\right],r_{1}'\right]\tensor r_{2}'
+r_{1}'\tensor r_{1}\tensor\left[\left[x,r_{2}\right],r_{2}'\right]\\
&+r_{1}'\tensor\left[\left[x,r_{1}\right],r_{2}'\right]\tensor r_{2}
-r_{1}\tensor r_{1}'\tensor\left[\left[x,r_{2}\right],r_{2}'\right]
-\left[\left[x,r_{1}\right],r_{1}'\right]\tensor r_{2}\tensor r_{2}'\\
&-\left[r_{1}',r_{1}\right]\tensor r_{2}'\tensor\left[x,r_{2}\right]
+\left[x,r_{2}\right]\tensor\left[r_{1},r_{1}'\right]\tensor r_{2}'
+r_{1}'\tensor\left[x,r_{2}\right]\tensor\left[r_{2}',r_{1}\right]\\
&-r_{1}'\tensor\left[r_{2}',r_{1}\right]\tensor\left[x,r_{2}\right]
-\left[x,r_{1}\right]\tensor r_{1}'\tensor\left[r_{2},r_{2}'\right]
-\left[r_{1},r_{1}'\right]\tensor\left[x,r_{2}\right]\tensor r_{2}'.
\end{align*}
We know that
\begin{align*}
\text{CYB}(r)=\left[r_{1},r_{1}'\right]\tensor r_{2}\tensor r_{2}'+r_{1}\tensor\left[r_{2},r_{1}'\right]\tensor r_{2}'
+r_{1}\tensor r_{1}'\tensor\left[r_{2},r_{2}'\right],
\end{align*}
so the right hand side of Eq. (\ref{FormulaOfHorror}) reads
\begin{align*}
\text{ad}_{x}^{(3)}(\text{CYB}(r))&=\left[x,\left[r_{1},r_{1}'\right]\right]\tensor r_{2}\tensor r_{2}'
+\left[r_{1},r_{1}'\right]\tensor\left[x,r_{2}\right]\tensor r_{2}'
+\left[r_{1},r_{1}'\right]\tensor r_{2}\tensor\left[x,r_{2}'\right]\\
&+\left[x,r_{1}\right]\tensor\left[r_{2},r_{1}'\right]\tensor r_{2}'
+r_{1}\tensor\left[x,\left[r_{2},r_{1}'\right]\right]\tensor r_{2}'
+r_{1}\tensor\left[r_{2},r_{1}'\right]\tensor\left[x,r_{2}'\right]\\
&+\left[x,r_{1}\right]\tensor r_{1}'\tensor\left[r_{2},r_{2}'\right]
+r_{1}\tensor\left[x,r_{1}'\right]\tensor\left[r_{2},r_{2}'\right]
+r_{1}\tensor r_{1}'\tensor\left[x,\left[r_{2},r_{2}'\right]\right].
\end{align*}
We apply the Jacobi identity to the terms with double Lie bracket and get
\begin{align*}
\text{ad}_{x}^{(3)}(\text{CYB}(r))&=\left[\left[x,r_{1}\right],r_{1}'\right]\tensor r_{2}\tensor r_{2}'
-\left[\left[x,r_{1}'\right],r_{1}\right]\tensor r_{2}\tensor r_{2}'\\
&+r_{1}\tensor\left[\left[x,r_{2}\right],r_{1}'\right]\tensor r_{2}'
-r_{1}\tensor\left[\left[x,r_{2}'\right],r_{1}\right]\tensor r_{2}'\\
&+r_{1}\tensor r_{1}'\tensor\left[\left[x,r_{2}\right],r_{2}'\right]
-r_{1}\tensor r_{1}'\tensor\left[\left[x,r_{2}'\right],r_{2}\right]\\
&+\left[r_{1},r_{1}'\right]\tensor r_{2}\tensor\left[x,r_{2}'\right]
-\left[x,r_{2}\right]\tensor\left[r_{1},r_{1}'\right]\tensor r_{2}'
-r_{1}\tensor\left[x,r_{2}'\right]\tensor\left[r_{2},r_{1}'\right]\\
&+r_{1}\tensor\left[r_{2},r_{1}'\right]\tensor\left[x,r_{2}'\right]
+\left[x,r_{1}\right]\tensor r_{1}'\tensor\left[r_{2},r_{2}'\right]
+\left[r_{1},r_{1}'\right]\tensor\left[x,r_{2}\right]\tensor r_{2}'.
\end{align*}
These are exactly the twelve terms that we obtained on the left side times minus one. We only have to interchange the
summation over the tensor components of $ r $ in some terms, i.e. we have to switch $ r_{1}\tensor r_{2} $ and
$ r_{1}'\tensor r_{2}' $, which is allowed since these summations are all finite. Formula (\ref{FormulaOfHorror}) is
proved. Finally, the element $ \delta r $ is a coboundary, i.e.
it fulfils the cocycle condition $ \delta^{2}r=0 $. Thus
$ (\mathfrak{g},\left[\cdot,\cdot\right],\delta r) $ is a Lie bialgebra if and only if $ \delta r $ satisfies
the coJacobi identity $ \text{Alt}((\delta r\tensor 1)\delta r(x))=0 $ for all $ x\in\mathfrak{g} $. Now that
Eq. (\ref{FormulaOfHorror}) holds this is the case if and only if $ \text{ad}_{x}^{(3)}(\text{CYB}(r))=0 $ for all
$ x\in\mathfrak{g} $. This concludes the proof.
\end{proof}

\begin{definition}[Coboundary Lie Bialgebra]\label{def8thomas}
The triple $ (\mathfrak{g},\left[\cdot,\cdot\right],r) $, where $ r\in\mathfrak{g}\tensor\mathfrak{g} $, is said to be
a coboundary Lie bialgebra if $ (\mathfrak{g},\left[\cdot,\cdot\right],\delta r) $ is a Lie bialgebra.
In this case $ r $ is called a
coboundary structure of $ (\mathfrak{g},\left[\cdot,\cdot\right],\delta r) $.
\end{definition}

\begin{remark}
By
Theorem \ref{the1thomas} $ (\mathfrak{g},\left[\cdot,\cdot\right],r) $ is a coboundary Lie bialgebra if and only if $
r $ is skew-symmetric and $ \operatorname{ad}_{x}^{(3)}(\text{CYB}(r))=0 $.
It is clear from the definition that a Lie bialgebra $ (\mathfrak{g},\left[\cdot,\cdot\right],\gamma) $ can have more
than one coboundary structure. If $ r\in\mathfrak{g}\tensor\mathfrak{g} $ is a coboundary structure of this coboundary
Lie bialgebra and $ \alpha\in\mathfrak{g}\wedge\mathfrak{g} $ is $ \mathfrak{g} $-invariant, i.e. $ \delta\alpha=0 $,
also $ r'=r+\alpha $ is a coboundary structure of $ (\mathfrak{g},\left[\cdot,\cdot\right],\gamma) $ because $ \delta
r'=\delta r+0=\gamma $. If on the other hand $ r $ and $ r' $ are two coboundary structures of
$ (\mathfrak{g},\left[\cdot,\cdot\right],\gamma) $ there is $ 0=\gamma-\gamma=\delta r-\delta r'=\delta(r-r') $,
i.e. there is an $ \mathfrak{g} $-invariant element $ \alpha\in\mathfrak{g}\wedge\mathfrak{g} $ such that
$ r'=r+\alpha $.
This  gives a one-to-one correspondence between coboundary structures $ r' $ that are different
from $ r $ and $ \mathfrak{g} $-invariant elements $ \alpha\in\mathfrak{g}\wedge\mathfrak{g} $ denoted by
$ (\Anti^{2}\mathfrak{g})^{\mathfrak{g}} $. 
\end{remark}

Following the spirit of this thesis we define the corresponding morphisms and substructures.

\begin{definition}
Consider two coboundary Lie bialgebras $ (\mathfrak{g},\left[\cdot,\cdot\right],r) $ and $ (\mathfrak{g}',\left
[\cdot,\cdot\right]',r') $.
\begin{compactenum}
\item A morphism of Lie bialgebras $ \phi\colon(\mathfrak{g},\left[\cdot,\cdot\right],r)\rightarrow(\mathfrak{g}',
\left[\cdot,\cdot\right]',r') $ is said to be a \textbf{morphism of coboundary Lie bialgebras} if $ \phi $ respects
the coboundary structures, i.e.
\begin{align}
(\phi\tensor\phi)(r)=r'.
\end{align}
\item We call $ (\mathfrak{h},\left[\cdot,\cdot\right],r) $ a \textbf{coboundary Lie subbialgebra} of $ (\mathfrak{g},
\left[\cdot,\cdot\right],r) $ if $ \mathfrak{h}\subseteq\mathfrak{g} $ is a Lie subbialgebra such that $ r\in\bigwedge
^{2}\mathfrak{h} $.
\end{compactenum}
\end{definition}

What follows are several examples of types of coboundary  structures of a Lie bialgebra. We start with the most
general one.

\begin{definition}[Quasitriangular Lie Bialgebra]\label{def1thomas}
A triple $ (\mathfrak{g},\left[\cdot,\cdot\right],r) $, where $ r\in\mathfrak{g}\tensor\mathfrak{g} $, is called a
quasitriangular Lie bialgebra if the following three conditions are satisfied,
\begin{compactenum}
\item $ (\mathfrak{g},\left[\cdot,\cdot\right],\delta r) $ is a Lie bialgebra,
\item $ \operatorname{CYB}(r)=0 $,
\item $ r+\sigma(r) $ is $ \mathfrak{g} $-invariant, i.e. $ \operatorname{ad}_{x}^{(2)}(r+\sigma(r))=0 $ for all
$ x\in\mathfrak{g} $, where $ \sigma\colon\mathfrak{g}\tensor\mathfrak{g}\rightarrow\mathfrak{g}\tensor\mathfrak{g} $
is the map that flips the tensor components.
\end{compactenum}
We call $ r\in\mathfrak{g}\tensor\mathfrak{g} $ a quasitriangular structure for a Lie bialgebra $
(\mathfrak{g},\left[\cdot,\cdot\right],\gamma) $ if $ \gamma=\delta r $ and conditions ii) and iii) are satisfied.
\end{definition}

Quasitriangular Lie bialgebras are examples of coboundary Lie bialgebras. Indeed if $ (\mathfrak{g},
\left[\cdot,\cdot\right],r) $ is a quasitriangular Lie bialgebra define $ r'=\frac{1}{2}(r-\sigma(r)) $ to be the
skew-symmetric part of $ r $. Because of condition iii.) in Definition \ref{def1thomas} $ \delta r=\delta r-
\frac{1}{2}\delta(r+\sigma(r))=\delta r' $. The skew-symmetric element $ r' $ satisfies
\begin{align}
\text{ad}_{x}^{(3)}(\text{CYB}(r'))=0
\end{align}
and, for this reason it is a coboundary structure of $ (\mathfrak{g},\left[\cdot,\cdot\right],\delta r) $. Conversely,
one can check
that a coboundary Lie bialgebra $ (\mathfrak{g},\left[\cdot,\cdot\right],r') $ has a quasitriangular structure if and
only if there is an element $ T\in(\Sym^{2}\mathfrak{g})^{\mathfrak{g}} $ such that
\begin{align}
\text{CYB}(r')=\frac{1}{4}\left[ T_{12},T_{23}\right].
\end{align}
In this case the quasitriangular structure can be chosen by $ r=r'+\frac{1}{2}T $
(c.f. \cite[Remark~page~31]{etingof2002lectures}).
A more specific type of coboundary Lie bialgebras is defined below. In the following chapters
we are especially interested in this kind of Lie bialgebra.

\begin{definition}[Triangular Lie Bialgebra]\label{def4thomas}
A coboundary Lie bialgebra $ (\mathfrak{g},\left[\cdot,\cdot\right],r) $ is said to be a triangular Lie
bialgebra if $ \operatorname{CYB}(r)=0 $. Moreover an element $ r\in\mathfrak{g}\tensor\mathfrak{g} $ is said to be a
triangular structure on a Lie bialgebra $ (\mathfrak{g},\left[\cdot,\cdot\right],\gamma) $ if the following
three conditions are satisfied,
\begin{compactenum}
\item $ \delta r=\gamma $,
\item $ \operatorname{CYB}(r)=0 $,
\item $ r-\sigma(r)=0 $, i.e. $ r $ is skew-symmetric or $ r\in\Anti^{2}\mathfrak{g} $.
\end{compactenum}
In this case $ r $ is said to be a \textbf{$ r $-matrix} of $ \mathfrak{g} $. Condition ii.) is called the
\textbf{classical Yang-Baxter equation}.
\end{definition}

It follows from the definition above that triangular Lie bialgebras, and as a consequence also $ r $-matrices, are in
one-to-one correspondence to solutions of the classical Yang-Baxter equation in $ \Anti^{2}\mathfrak{g} $.
Furthermore, every triangular structure is quasitriangular and every triangular Lie bialgebra has also the
structure of a quasitriangular Lie bialgebra. As a first example of triangular Lie bialgebras and $ r $-matrices we
go back to Example \ref{ExampleXYBialg}. There we endowed any two-dimensional noncommutative Lie algebra with a Lie
bialgebra structure.

\begin{example}
Consider the Lie algebra $ (\mathfrak{g},\left[\cdot,\cdot\right]) $ over a field $ \mathbb{k} $ of characteristic
zero with basis elements $ X,Y\in\mathfrak{g} $ that satisfy $ \left[X,Y\right]=X $. Then define
\begin{align}
r=X\wedge Y=X\tensor Y-Y\tensor X.
\end{align}
We want to prove that this is an $ r $-matrix on $ (\mathfrak{g},\left[\cdot,\cdot\right]) $. The skew-symmetry is
evident. Moreover $ r $ is a solution of the classical Yang-Baxter equation on $ \mathfrak{g} $:
\begin{align*}
\text{CYB}(r)&=\left[r_{12},r_{13}\right]+\left[r_{12},r_{23}\right]+\left[r_{13},r_{23}\right]\\
&=\left[X,X\right]\tensor Y\tensor Y-\left[X,Y\right]\tensor Y\tensor X-\left[Y,X\right]\tensor X\tensor Y
+\left[Y,Y\right]\tensor X\tensor X\\
&+X\tensor\left[Y,X\right]\tensor Y-X\tensor\left[Y,Y\right]\tensor X-Y\tensor\left[X,X\right]\tensor Y
+Y\tensor\left[X,Y\right]\tensor X\\
&+X\tensor X\tensor\left[Y,Y\right]-X\tensor Y\tensor\left[Y,X\right]-Y\tensor X\tensor\left[X,Y\right]
+Y\tensor Y\tensor\left[X,X\right]\\
&=-X\tensor Y\tensor X+X\tensor X\tensor Y-X\tensor X\tensor Y+Y\tensor X\tensor X\\
&+X\tensor Y\tensor X-Y\tensor X\tensor X\\
&=0.
\end{align*}
We recover the Lie bialgebra structure $ \gamma $ given in Example \ref{ExampleXYBialg} because
\begin{align*}
\delta r(X)&=\text{ad}_{X}^{(2)}(r)=\text{ad}_{X}^{(2)}(X\tensor Y)-\text{ad}_{X}^{(2)}(Y\tensor X)\\
&=\left[X,X\right]\tensor Y+X\tensor\left[X,Y\right]-\left[X,Y\right]\tensor X-Y\tensor\left[X,X\right]\\
&=X\tensor X-X\tensor X\\
&=0
\end{align*}
and
\begin{align*}
\delta r(Y)&=\text{ad}_{Y}^{(2)}(r)=\text{ad}_{Y}^{(2)}(X\tensor Y)-\text{ad}_{Y}^{(2)}(Y\tensor X)\\
&=\left[Y,X\right]\tensor Y+X\tensor\left[Y,Y\right]-\left[Y,Y\right]\tensor X-Y\tensor\left[Y,X\right]\\
&=Y\tensor X-X\tensor Y\\
&=-X\wedge Y,
\end{align*}
i.e. $ \gamma=\delta r $. Thus $ (\mathfrak{g},\left[\cdot,\cdot\right],r) $ is a triangular Lie bialgebra,
$ r $ is a triangular structure on $ (\mathfrak{g},\left[\cdot,\cdot\right],\gamma) $ and a $ r $-matrix
on $ \mathfrak{g} $.
\end{example}

\begin{remark}
If $ \phi\colon\mathfrak{g}\rightarrow\mathfrak{a} $ is a
Lie algebra homomorphism to another Lie algebra $ \mathfrak{a} $ and $ r $ a triangular structures on $ \mathfrak{g} $,
then $ (\phi\tensor\phi)r $ is a triangular structure on $ \mathfrak{a} $.
We prove this in Lemma \ref{LemRMatrixHom}.
This is not true for an arbitrary
coboundary structure. For example the image of a quasitriangular structure under a Lie algebra homomorphism does not
need to be a quasitriangular structure (see e.g. \cite[page~29]{etingof2002lectures}).
\end{remark}

As promised before we pass from the infinitesimal world of Lie bialgebras to the global world of Poisson-Lie Groups,
inspired by Theorem \ref{TheoremDrinFunctor}.

\begin{proposition}
Let $ r\in\mathfrak{g}\wedge\mathfrak{g} $ be an $ r $-matrix on $ \mathfrak{g} $. The elements $ \pi_{r}^{\lambda} $
and $ \pi_{r}^{\rho} $ defined for $ x\in\mathfrak{g} $ by
\begin{align}
\pi_{r}^{\lambda}(x)=(T_{e}\lambda_{x}\tensor T_{e}\lambda_{x})r
\end{align}
and
\begin{align}
\pi_{r}^{\rho}(x)=(T_{e}\rho_{x}\tensor T_{e}\rho_{x})r
\end{align}
are Poisson bivectors.
Moreover,
\begin{align}
\pi=\pi_{r}^{\rho}-\pi_{r}^{\lambda}
\end{align}
is a Poisson-Lie structure on $ G $ and it coincides with the one obtained in Theorem \ref{TheoremDrinFunctor}.
\end{proposition}

For a proof consider \cite[Proposition~3.1]{etingof2002lectures} and \cite[page~48]{kosmann2004integrability}.
We see that if $ r\neq 0 $ it follows that $ \pi_{r}^{\lambda}(e)=\pi_{r}^{\rho}(e)=r\neq 0 $. According to
Eq. (\ref{eq9thomas}) $ \pi_{r}^{\lambda} $ and $ \pi_{r}^{\rho} $ are not Poisson-Lie structures while
$ \pi=\pi_{r}^{\rho}-\pi_{r}^{\lambda} $ is, according to the last theorem.

\section{The Etingof-Schiffmann Subalgebra}\label{SectionEtingof}

In this section we discuss non-degeneracy of $ r $-matrices.
Remark that there is a one-to-one correspondence between non-degenerate $ r $-matrices on $ \mathfrak{g} $ and
non-degenerate $ 2 $-cocycles on $ \mathfrak{g} $ corresponding to the trivial representation of $ \mathfrak{g} $
on $ \mathbb{C} $ (c.f. \cite[Proposition~3.3]{etingof2002lectures}). In other words, there is a symplectic
structure on $ \mathfrak{g} $ if and only if there is a non-degenerate $ r $-matrix on $ \mathfrak{g} $. Thus
the classification of non-degenerate $ r $-matrices gets easier for small $ H^{2}(\mathfrak{g},\mathbb{C}) $.

\begin{proposition}\label{PropRmatrixSimpleDeg}
Let $ \mathfrak{g} $ be a complex semisimple Lie algebra. Then there are no non-degenerate triangular structures on
$ \mathfrak{g} $.
\end{proposition}

\begin{proof}
We follow \cite[Proposition~5.2]{etingof2002lectures}. Assume there is a non-degenerate $ r $-matrix
$ r\in\Anti^{2}\mathfrak{g} $ and consider the corresponding symplectic form
$ \omega\in\Anti^{2}\mathfrak{g}^{*} $. Now $ \mathfrak{g} $ is semisimple, then by Whitehead's Lemma
(see Theorem \ref{Whitehead}) one has
$ H^{2}(\mathfrak{g},\mathbb{C})=\left\lbrace 0\right\rbrace $, i.e. we find a $ 1 $-cocylce $ f\in\mathfrak{g}^{*} $
such that
\begin{align}
\omega(x,y)=\delta f(x,y)=f(\left[x,y\right])
\end{align}
for any $ x,y\in\mathfrak{g} $. On the other hand, the Killing form $ \kappa\in\bigwedge^{2}\mathfrak{g}^{*} $
is non-degenerate on any semisimple Lie algebra. One can define a flat map
\begin{align}
\flat\colon\mathfrak{g}\ni x\mapsto\kappa(x,\cdot)\in\mathfrak{g}^{*}
\end{align}
with respect to $ \kappa $ with inverse map $ \sharp\colon\mathfrak{g}^{*}\rightarrow\mathfrak{g} $. Let
$ x,y\in\mathfrak{g} $ be arbitrary and define $ z=-f^{\sharp} $. Then
\begin{align*}
\omega(x,y)=f(\left[x,y\right])=\kappa(f^{\sharp},\left[x,y\right])=\kappa(\left[x,y\right],z)
\overset{(\ast)}{=}\kappa(x,\left[y,z\right]),
\end{align*}
where we use the associativity of $ \kappa $ (see Proposition \ref{PropKilling} ii.)) in $ (\ast) $. If we set
$ y=z $ one has $ \omega(x,z)=\kappa(x,\left[z,z\right])=0 $ for all $ x\in\mathfrak{g} $ and $ \omega $
is degenerate, which is a contradiction.
\end{proof}

However, one can always find a Lie subalgebra $ \mathfrak{h}\subseteq\mathfrak{g} $
such that a $ r $-matrix on $ \mathfrak{g} $ is non-degenerate viewed as an element of $ \mathfrak{h}\wedge
\mathfrak{h} $. For this we need the classical Yang-Baxter equation in coordinate form: let $ \mathfrak{g} $ be a
finite-dimensional Lie algebra with basis $ e_{1},\ldots,e_{n}\in\mathfrak{g} $. Since
$ \left[ e_{i},e_{j}\right]\in\mathfrak{g} $ there are unique numbers $ C_{ij}^{k}\in\mathbb{k} $, called the
\textit{structure constants}, such that
\begin{align}
\left[ e_{i},e_{j}\right]=\sum\limits_{k=1}^{n}C_{ij}e_{k},
\end{align}
for all $ i,j\in\left\lbrace 1,\ldots,n\right\rbrace $. Now let $ r=\frac{1}{2}\sum\limits_{i,j=1}^{n}r^{ij}e_{i}\wedge
e_{j}\in\mathfrak{g}\wedge\mathfrak{g} $ be a $ r $-matrix.

\begin{lemma}
The classical Yang-Baxter equation $ \operatorname{CYB}(r)=0 $ is equivalent to the equations
\begin{align}\label{eq16t}
\sum\limits_{j,k=1}^{n}(r^{jm}r^{kl}C_{jk}^{i}+r^{ij}r^{mk}C_{jk}^{l}+r^{ij}r^{kl}C_{jk}^{m})=0,
\end{align}
where $ i,m,l\in\left\lbrace 1,\ldots,n\right\rbrace $.
\end{lemma}

\begin{proof}
Since $ r^{ij}=-r^{ji}\in\mathbb{k} $ one has $ r=\sum\limits_{i,j=1}^{n}r^{ij}e_{i}\tensor e_{j} $. Then
\begin{align*}
\left[ r_{12},r_{23}\right]&=\left[\sum\limits_{i,j=1}^{n}r^{ij}e_{i}\tensor e_{j}\tensor 1,
\sum\limits_{k,l=1}^{n}r^{kl} 1\tensor e_{k}\tensor e_{l} \right]\\
&=\sum\limits_{i,j,k,l=1}^{n}r^{ij}r^{kl}\left[ e_{i}\tensor e_{j}\tensor 1,
1\tensor e_{k}\tensor e_{l} \right]\\
&=\sum\limits_{i,j,k,l=1}^{n}r^{ij}r^{kl}e_{i}\tensor\left[ e_{j},e_{k}\right]\tensor e_{l}\\
&=\sum\limits_{i,j,k,l,m=1}^{n}r^{ij}r^{kl}C_{jk}^{m}e_{i}\tensor e_{m}\tensor e_{l}.
\end{align*}
Similarly one gets
\begin{align*}
\left[ r_{13},r_{23}\right]=\sum\limits_{i,j,k,l,m=1}^{n}r^{ij}r^{mk}C_{jk}^{l}e_{i}\tensor e_{m}\tensor e_{l},\\
\left[ r_{12},r_{13}\right]=\sum\limits_{i,j,k,l,m=1}^{n}r^{jm}r^{kl}C_{jk}^{i}e_{i}\tensor e_{m}\tensor e_{l}.
\end{align*}
Thus
\begin{align*}
\text{CYB}(r)&=\left[ r_{12},r_{13}\right]+\left[ r_{13},r_{23}\right]+\left[ r_{12},r_{23}\right]\\
&=\sum\limits_{i,j,k,l,m=1}^{n}(r^{jm}r^{kl}C_{jk}^{i}+r^{ij}r^{mk}C_{jk}^{l}
+r^{ij}r^{kl}C_{jk}^{m})e_{i}\tensor e_{m}\tensor e_{l}
\end{align*}
and the classical Yang-Baxter equation $ \operatorname{CYB}(r)=0 $ is equivalent to
\begin{align*}
\sum\limits_{j,k=1}^{n}(r^{jm}r^{kl}C_{jk}^{i}+r^{ij}r^{mk}C_{jk}^{l}+r^{ij}r^{kl}C_{jk}^{m})=0
\text{ for all }i,m,l\in\left\lbrace 1,\ldots,n\right\rbrace.
\end{align*}
\end{proof}
This result allows the following

\begin{proposition}\label{PropEtingofSub}
Let $ r\in\mathfrak{g}\wedge\mathfrak{g} $ be a $ r $-matrix. Then the set
\begin{align}\label{eq15t}
\mathfrak{h}_{r}=\left\lbrace (f\tensor 1)r\in\mathfrak{g}~|~f\in\mathfrak{g}^{*}\right\rbrace
\end{align}
is a Lie subalgebra of $ \mathfrak{g} $ and $ r\in\mathfrak{h}_{r}\wedge\mathfrak{h}_{r}\subseteq\mathfrak{g}\wedge
\mathfrak{g} $ is non-degenerate in $ \mathfrak{h}_{r}\wedge\mathfrak{h}_{r} $.
\end{proposition}

\begin{proof}
We follow \cite[Section~3.5]{etingof2002lectures}.
The isomorphism $ \mathbb{k}\tensor\mathfrak{g}\cong\mathfrak{g} $ implies $ \mathfrak{h}_{r}
\subseteq\mathfrak{g} $. Moreover, $ \mathfrak{h}_{r} $ is a Lie subalgebra of $ \mathfrak{g} $ since, if we take
$ f,g\in\mathfrak{g}^{*} $, then
\begin{align*}
\sum\limits_{k=1}^{n}\text{ad}_{(f\tensor 1)r}e_{k}&=\sum\limits_{k=1}^{n}\left[ (f\tensor 1)r,
e_{k}\right]
=\sum\limits_{k=1}^{n}\left[ \sum\limits_{i,j=1}^{n}r^{ij}f(e_{i})e_{j},e_{k}\right]\\
&=\sum\limits_{i,j,k=1}^{n}r^{ij}f(e_{i})\left[ e_{j},e_{k}\right]
=\sum\limits_{i,j,k,l=1}^{n}r^{ij}f(e_{i})C_{jk}^{l}e_{l}.
\end{align*}
Similarly for $ g $. Then we have
\begin{align*}
\left[ (f\tensor 1)r,(g\tensor 1)r\right]&=\left[\sum\limits_{i,j=1}^{n}r^{ij}f(e_{i})e_{j},
\sum\limits_{k,l=1}^{n}r^{kl}g(e_{k})e_{l}\right]\\
&=\sum\limits_{i,j,k,l=1}^{n}r^{ij}r^{kl}f(e_{i})g(e_{k})\left[e_{j},e_{l}\right]\\
&=\sum\limits_{i,j,k,l,m=1}^{n}r^{ij}r^{kl}C_{jl}^{m}f(e_{i})g(e_{k})e_{m}\\
&\overset{(\ast)}{=}\sum\limits_{i,j,k,l,m=1}^{n}f(e_{i})g(e_{k})(-r^{jk}C_{jl}^{i}-r^{ij}C_{jl}^{k})r^{lm}e_{m}\\
&=\sum\limits_{i,j,k,l,m=1}^{n}(f(r^{kj}g(e_{k})C_{jl}^{i}e_{i})-g(r^{ij}f(e_{i})C_{jl}^{k}e_{k}))r^{lm}e_{m}\\
&=\sum\limits_{l,m=1}^{n}(f(\text{ad}_{(g\tensor 1)r}e_{l})-g(\text{ad}_{(f\tensor 1)r}e_{l}))
r^{lm}e_{m}\\
&=\sum\limits_{l,m=1}^{n}((f\circ\text{ad}_{(g\tensor 1)r})\tensor 1-
(g\circ\text{ad}_{(f\tensor 1)r})\tensor 1)r^{lm}e_{l}\tensor e_{m}\\
&=(h\tensor 1)r,
\end{align*}
where we used (\ref{eq16t}) (look at the second term and change $ m\rightarrow k, k\rightarrow l, l\rightarrow m $)
in $ (\ast) $ and defined
\begin{align*}
h=(f\circ\text{ad}_{(g\tensor 1)r}-g\circ\text{ad}_{(f\tensor 1)r})\in\mathfrak{g}^{*}.
\end{align*}
The next step is to show that $ r $ is an element of $ \mathfrak{h}_{r}\wedge\mathfrak{h}_{r} $. Since we already
proved that
$ \mathfrak{h}_{r}\subseteq\mathfrak{g} $ is a Lie subalgebra we can choose a basis
$ e_{1},\ldots,e_{k}\in\mathfrak{g} $ of $ \mathfrak{h}_{r} $, where $ k\in\left\lbrace 1,\ldots,n\right\rbrace $,
and complete this to a basis $ e_{1},\ldots,e_{n}\in\mathfrak{g} $ of $ \mathfrak{g} $.
Thus $ \mathfrak{h}_{r}=\text{span}_{\mathbb{k}}\left\lbrace
e_{1},\ldots,e_{k}\right\rbrace $. For an arbitrary $ f\in\mathfrak{g}^{*} $ this implies
\begin{align}\label{eq17t}
(f\tensor 1)r=\sum\limits_{i,j=1}^{n}r^{ij}f(e_{i})e_{j}
\in\text{span}_{\mathbb{k}}\left\lbrace e_{1},\ldots,e_{k}\right\rbrace.
\end{align}
Since $ f\in\mathfrak{g}^{*} $ was arbitrary there must be $ r^{ij}=0 $ for $ j\notin\left\lbrace 1,\ldots,k
\right\rbrace $. By the skew-symmetry of $ r^{ij} $ this implies $ r^{ij}=0 $ for $ i\notin\left\lbrace 1,\ldots,k
\right\rbrace $ or $ j\notin\left\lbrace 1,\ldots,k\right\rbrace $. Then indeed
\begin{align}
r=\frac{1}{2}\sum\limits_{i,j=1}^{k}r^{ij}e_{i}\wedge e_{j}\in\mathfrak{h}_{r}\wedge\mathfrak{h}_{r}.
\end{align}
We see that Eq. (\ref{eq17t}) actually reads $ \sum\limits_{i,j=1}^{k}r^{ij}f(e_{i})e_{j}
\in\text{span}_{\mathbb{k}}\left\lbrace e_{1},\ldots,e_{k}\right\rbrace $, and we have
\begin{align}
\text{span}_{\mathbb{k}}\left\lbrace e_{1},\ldots,e_{k}\right\rbrace=\mathfrak{h}_{r}
\subseteq\text{span}_{\mathbb{k}}\left\lbrace
\sum\limits_{j=1}^{k}r^{1j}e_{j},\ldots,\sum\limits_{j=1}^{k}r^{kj}e_{j}\right\rbrace\subseteq
\text{span}_{\mathbb{k}}\left\lbrace e_{1},\ldots,e_{k}\right\rbrace.
\end{align}
In other words, the matrix $ (r^{ij})_{ij}\in M_{k\times k}(\mathbb{k}) $ has full rank. For this $ r $ is
non-degenerate viewed as an element of $ \mathfrak{h}_{r}\wedge\mathfrak{h}_{r} $.
\end{proof}

Following a discussion \cite{JonasDiscussion} with J. Schnitzer, we do the following

\begin{definition}[Etingof-Schiffmann Subalgebra]
The Lie subalgebra $ \mathfrak{h}_{r} $ defined in Eq. (\ref{eq15t}) is called
\textit{Etingof-Schiffmann subalgebra} for a $ r
$-matrix $ r\in\mathfrak{g}\wedge\mathfrak{g} $. The corresponding connected Lie group $ H_{r} $ is said to be the
\textit{Etingof-Schiffmann subgroup} for $ r $.
\end{definition}

In Chapter
\ref{ChapObstruction} this is our tool to produce obstructions to twist star products on the sphere
$ \mathbb{S}^{2} $.

%% file: Chapter3.tex
\chapter{Twist Deformation}\label{ChapTwistDef}

While the last chapter was dedicated to the understanding of Lie algebras with additional structure, also
with some geometric aspects, we start by discussing pure algebraic features of general algebras here.
The ideas and notions of Hopf algebras are developed in Appendix \ref{appendixHopf} and appear here the
first time. A introduction to this topic is given in \cite[Chapter~1]{majid2000foundations}. Also consider
\cite{chari1995guide}, \cite{milnor1961structure} and \cite[Chapter~4]{underwood2011introduction}.
In a way, Hopf algebras take over the role of the Lie group, like the algebra does for the manifold. By this
we mean the following: in Chapter \ref{chapTransAct} we talked about Lie groups acting on manifolds and some of the
involved properties. Now we define a process such that a Hopf algebra can act on an algebra. Of course there is no
smooth structure. Instead, the algebraic relations of a Lie group action are taken as the new axioms. In
this case the algebra is said to be a (left) Hopf algebra module. For our purpose this is still not enough, thus we
force the module structure to be compatible with the algebra multiplication in addition. This leads to (left)
Hopf algebra module algebras. But here comes the interesting question: what happens if we deform the Hopf algebra?
Is this possible in a way such that the deformed Hopf algebra is still a Hopf algebra? And is there even a
natural way to deform any (left) module algebra of this Hopf algebra such that the deformed algebra is a (left)
module algebra for the deformed Hopf algebra? For a general deformation there is no reason for this to be the case,
but for twist deformations there are positive answers to all of these questions and we give them in Section
\ref{secHopfModuleAlg}. We start by defining twists on a Hopf algebra and prove that they deform the Hopf algebra
and any (left) module algebra of this Hopf algebra in the way we asked for. In a parallel move we evolve
this theory for the inverse of a twist. This might sound exaggerated and quite trivial, but there are two reasons
why we insist on this: first, the defining properties and deformations change slightly, thus, one has to be careful.
And second, in literature both notions occur and it is reasonable to understand
them both. As a last word to this topic we also prove that both definitions and the resulting deformations
are indeed equivalent.

Since we are interested in star products we have to modify our definition of twist, i.e. we have
to define twists on formal power series of universal enveloping algebras. Section \ref{secStarProdTwist}
starts with a short repetition on universal enveloping algebras, their Hopf algebra structure and formal
power series. We also discuss the problem that the resulting object is no Hopf algebra anymore in general.
In this special situation we want to apply the results of Section \ref{secHopfModuleAlg}. Instead of Hopf algebra
twists we take those of the formal power series of universal enveloping algebras and instead of arbitrary algebras
we take the algebra of smooth real-valued functions on a Poisson manifold with pointwise multiplication. If the
deformed multiplication on the smooth functions is a star product on the manifold, this star product is said to be
a twist star product. Moreover, one says that an arbitrary star product on a Poisson manifold can be induced by
a twist if there is a twist on the formal power series of any universal enveloping algebra such that the
deformation of the pointwise multiplication via this twist gives the star product. Thus at the end of this chapter we
arrive at our objects of interest: star products and in particular the question whether they can or can not be induced
by a twist. As a last observation we prove that a twist on the formal power series of a universal enveloping algebra
always induces a $ r $-matrix on the underlying Lie algebra. We stress that this is no one-to-one correspondence.
This is the most important connection to Chapter \ref{chapBialgrMatrix}. In the next chapter we finally see the
connection to homogeneous spaces.

\section{Twists and left Hopf Algebra Module Algebras}\label{secHopfModuleAlg}

We need the notion of Hopf algebras and refer to Appendix \ref{appendixHopf} for a conceptional motivation of the
involved axioms (see Definition \ref{hopfalgdef}). For instance, we use Sweedler's notation
\begin{align}
\Delta(\xi)=\xi_{(1)}\tensor\xi_{(2)}
\end{align}
for the coproduct of a Hopf algebra element $ \xi $, which is also discussed in detail in this appendix
(consider Remark \ref{RemarkSweedler}). Do not mix this
up with the short notation $ \mathcal{F}=\mathcal{F}_{1}\tensor\mathcal{F}_{2} $ of an arbitrary element
$ \mathcal{F} $ of the tensor product of a Hopf algebra. Let
$ (H,+,\cdot,\eta,\Delta,\epsilon,S,\mathbb{k}) $ be a Hopf algebra and $ \mathcal{F}\in H\tensor H $. We introduce
some very useful abbreviations
\begin{align*}
\mathcal{F}_{12}&=\mathcal{F}\tensor 1_{H},\\
\mathcal{F}_{23}&=1_{H}\tensor\mathcal{F},\\
\mathcal{F}_{13}&=\mathcal{F}_{1}\tensor 1_{H}\tensor\mathcal{F}_{2},\\
\mathcal{F}_{21}&=\mathcal{F}_{_{2}}\tensor\mathcal{F}_{1}\tensor 1_{H},\\
\mathcal{F}_{32}&=1_{H}\tensor\mathcal{F}_{2}\tensor\mathcal{F}_{1},\\
\mathcal{F}_{31}&=\mathcal{F}_{2}\tensor 1_{H}\tensor\mathcal{F}_{1}.
\end{align*}
All of them are elements in $ H\tensor H\tensor H=H^{\tensor 3} $.
Further, we denote by $ 1 $ the identity map $ H\rightarrow H $ and by $ 1_{H} $ the unit of $ H $.
We define the central element in twist deformation, namely a twisting element, following
\cite[Section~3.1]{schenkeldrintwistdef} and also adopt the notation that is used there. Similar references are
\cite[Section~7.8]{chari1995guide}, \cite[Section~9.5]{etingof2002lectures}, and
\cite[Section~1]{Giaquintobialgebraactions}.

\begin{definition}[Drinfel\textquoteright d Twist]\label{DefTwist}
Let $ (H,+,\cdot,\eta,\Delta,\epsilon,S,\mathbb{k}) $ be a Hopf algebra. An invertible element $ \mathcal{F}\in
H\tensor H $ is said to be a Drinfel\textquoteright d twist or twist if the following two conditions are satisfied,
\begin{compactenum}
\item $ \mathcal{F}_{12}(\Delta\tensor 1)\mathcal{F}=\mathcal{F}_{23}(1\tensor\Delta)\mathcal{F} $,
\item $ (\epsilon\tensor 1)\mathcal{F}=(1\tensor\epsilon)\mathcal{F}=1_{H} $.
\end{compactenum}
The first condition is said to be the 2-cocycle condition and the second one the normalization
property.
\end{definition}

\begin{remark}\label{RemarkJTwist1}
There is also a kind of $ 2 $-cocylce condition
\begin{align}\label{JConvention1}
((\Delta\tensor 1)\mathcal{F}^{-1})\cdot\mathcal{F}_{12}^{-1}=((1\tensor\Delta)\mathcal{F}^{-1})
\cdot\mathcal{F}_{23}^{-1}
\end{align}
for $ \mathcal{F}^{-1} $, since $ \mathcal{F} $ is invertible. It is equivalent to Definition \ref{DefTwist}
i.), since we know that inverting reverses the order of elements and that the inverse element of $ (\Delta\tensor
1)\mathcal{F} $ in $ H^{\tensor 3} $ is $ (\Delta\tensor 1)\mathcal{F}^{-1} $ and so on.
Furthermore, $ \mathcal{F}^{-1} $ fulfils the normalization property
\begin{align}\label{JConvention2}
(\epsilon\tensor 1)\mathcal{F}^{-1}=(1\tensor\epsilon)\mathcal{F}^{-1}=1_{H},
\end{align}
since $ (\epsilon\tensor 1)\mathcal{F}^{-1}=(\epsilon\tensor 1)((\epsilon\tensor 1)\mathcal{F}\mathcal{F}^{-1})
=1_{H} $ according to Definition \ref{DefTwist} ii.). Later on we see that there is a Hopf algebra structure on $ H $
such that $ \mathcal{F}^{-1} $ is a twist in the sense of Definition \ref{DefTwist}.
\end{remark}

The primary use of a twist $ \mathcal{F} $ is to deform the Hopf algebra structure. If we consider
$ \Delta^{\mathcal{F}}\colon H\rightarrow H\tensor H $ defined by
\begin{align}\label{deformedcoproduct}
\Delta^{\mathcal{F}}(\xi)=\mathcal{F}\cdot\Delta(\xi)\cdot\mathcal{F}^{-1}
\end{align}
and $ S^{\mathcal{F}}\colon H\rightarrow H $ defined by
\begin{align}\label{deformedantipode}
S^{\mathcal{F}}(\xi)=U\cdot S(\xi)\cdot U^{-1},
\end{align}
where $ \xi\in H $, $ \mathrm{m}\colon H\tensor H\ni(\xi\tensor\zeta)\mapsto\xi\cdot\zeta\in H $ and
$ U=\mathrm{m}((1\tensor S)\mathcal{F}) $. Now we can state the following

\begin{theorem}\label{theoremdeformedhopf}
Let $ (H,+,\cdot,\eta,\Delta,\epsilon,S,\mathbb{k}) $ be a Hopf algebra and $ \mathcal{F} $ a twist on $ H $.
Then one gets another Hopf algebra $ (H,+,\cdot,\eta,\Delta^{\mathcal{F}},\epsilon,S^{\mathcal{F}},\mathbb{k}) $,
with the deformed coproduct and antipode defined in Eq. (\ref{deformedcoproduct}) and Eq. (\ref{deformedantipode}),
respectively.
\end{theorem}

\begin{proof}
The proof we give is a mixture of the ones in \cite[Proposition~9.6]{etingof2002lectures} and
\cite[Theorem~2.3.4]{majid2000foundations}.
We split the proof into three parts. First we prove that $ (H,+,\Delta^{\mathcal{F}},\epsilon,\mathbb{k}) $ is a
coalgebra (see Definition \ref{DefCoAlgebra}). Since the algebra structure is not changed the second step is to
check if the algebra and the coalgebra structure are compatible in the sense of a bialgebra (consider Definition
\ref{DefBiAlgebra}).
Finally,
$ S^{\mathcal{F}} $ has to be an antipode to finish the proof. Thus the first
question is whether the diagrams (\ref{eq48thomas}) commute. To prove that $ \Delta^{\mathcal{F}} $ satisfies the
coassociativity condition we need the coassociativity of $ \Delta $, the $ 2 $-cocycle condition of $ \mathcal{F} $
and the $ 2 $-cocycle condition of $ \mathcal{F}^{-1} $. Furthermore, we use that $ \Delta $ is an algebra map.
Let $ \xi\in H $. Then
\begin{align*}
(\Delta^{\mathcal{F}}\tensor 1)(\Delta^{\mathcal{F}}(\xi))&=\mathcal{F}_{12}\cdot((\Delta\tensor 1)
(\mathcal{F}\cdot\Delta(\xi)\cdot\mathcal{F}^{-1}))\cdot\mathcal{F}_{12}^{-1}\\
&=\mathcal{F}_{12}\cdot((\Delta\tensor 1)\mathcal{F})\cdot((\Delta\tensor 1)\Delta(\xi))\cdot((\Delta\tensor
1)\mathcal{F}^{-1})\cdot\mathcal{F}_{12}^{-1}\\
&=\mathcal{F}_{23}\cdot((1\tensor\Delta)\mathcal{F})\cdot((1\tensor\Delta)\Delta(\xi))\cdot((1\tensor\Delta)
\mathcal{F}^{-1})\cdot\mathcal{F}_{23}^{-1}\\
&=\mathcal{F}_{23}\cdot((1\tensor\Delta)(\mathcal{F}\cdot\Delta(\xi)\cdot\mathcal{F}^{-1}))\cdot\mathcal{F}_{23}^{-1}\\
&=(1\tensor\Delta^{\mathcal{F}})\Delta^{\mathcal{F}}(\xi).
\end{align*}
Then the normalization property of $ \mathcal{F} $ and $ \mathcal{F}^{-1} $ together with the fact that
$ \epsilon $ is an algebra map implies the second (and third) coalgebra axiom: for $ \xi\in H $ we calculate
\begin{align*}
(\epsilon\tensor 1)(\Delta^{\mathcal{F}}(\xi))&=((\epsilon\tensor 1)(\mathcal{F}))\cdot((\epsilon\tensor 1)
(\Delta(\xi)))\cdot((\epsilon\tensor 1)(\mathcal{F}^{-1}))\\
&=1_{H}\\
&=(1\tensor\epsilon)(\Delta^{\mathcal{F}}(\xi)).
\end{align*}
This means that $ (H,+,\Delta^{\mathcal{F}},\epsilon,\mathbb{k}) $ is a coalgebra and the first step is done.
According to Definition \ref{DefBiAlgebra} $ (H,+,\cdot,\eta,\Delta^{\mathcal{F}},\epsilon,\mathbb{k}) $ is a
bialgebra if $ \Delta^{\mathcal{F}} $ and $ \epsilon $ are algebra maps. Thus it remains to check this property
for $ \Delta^{\mathcal{F}} $. Let $ \xi,\zeta\in H $. Then
\begin{align*}
\Delta^{\mathcal{F}}(\xi\cdot\zeta)&=\mathcal{F}\cdot\Delta(\xi\cdot\zeta)\cdot\mathcal{F}^{-1}\\
&=\mathcal{F}\cdot\Delta(\xi)\cdot\mathcal{F}^{-1}\cdot\mathcal{F}\cdot\Delta(\zeta)\cdot\mathcal{F}^{-1}\\
&=\Delta^{\mathcal{F}}(\xi)\cdot\Delta^{\mathcal{F}}(\zeta),
\end{align*}
since $ \Delta $ is an algebra map. Also
$ \Delta^{\mathcal{F}}(1_{H})=\mathcal{F}\cdot\Delta(1_{H})\cdot\mathcal{F}^{-1}=1_{H}\tensor 1_{H} $ holds.
The third and last step is to check that $ S^{\mathcal{F}} $ is an
antipode of the bialgebra $ (H,+,\cdot,\eta,\Delta^{\mathcal{F}},\epsilon,\mathbb{k}) $. First of all, remark that
$ U=\mathrm{m}((1\tensor S)\mathcal{F})=\mathcal{F}_{1}S(\mathcal{F}_{2}) $ is invertible since $ \mathcal{F} $ is
invertible and $ U^{-1}:=\mathrm{m}((S\tensor 1)\mathcal{F}^{-1})=S(\mathcal{F}_{1}^{-1})\mathcal{F}_{2}^{-1} $
satisfies
\begin{align*}
UU^{-1}&=\mathcal{F}_{1}S(\mathcal{F}_{2})S(\mathcal{F}_{1}^{-1})\mathcal{F}_{2}^{-1}\\
&=\underbrace{\mathcal{F}_{1}^{-1}\epsilon(\mathcal{F}_{2}^{-1})}_{=1_{H}}\mathcal{F}_{1}S(\mathcal{F}_{2})
S(\mathcal{F}_{1}^{-1})\mathcal{F}_{2}^{-1}\\
&=\mathcal{F}_{1}^{-1}\mathcal{F}_{1}S(\mathcal{F}_{2})S(\mathcal{F}_{1}^{-1})S((\mathcal{F}_{2}^{-1})_{(1)})
(\mathcal{F}_{2}^{-1})_{(2)}\mathcal{F}_{2}^{-1}\\
&=S((\mathcal{F}_{2}^{-1})_{(1)}\mathcal{F}_{1}^{-1}\mathcal{F}_{2})
(\mathcal{F}_{2}^{-1})_{(2)}\mathcal{F}_{2}^{-1}\\
&=\mathrm{m}((\mathrm{m}\tensor 1)\circ(1\tensor S\tensor 1)
(((1\tensor\Delta)\mathcal{F}^{-1})\cdot\mathcal{F}_{23}^{-1}\cdot\mathcal{F}_{12}))\\
&=\mathrm{m}((\mathrm{m}\tensor 1)\circ(1\tensor S\tensor 1)((\Delta\tensor 1)\mathcal{F}^{-1}))\\
&=(\mathcal{F}_{1}^{-1})_{(1)}S((\mathcal{F}_{1}^{-1})_{(2)})\mathcal{F}_{2}^{-1}\\
&=\epsilon(\mathcal{F}_{1}^{-1})\mathcal{F}_{2}^{-1}\\
&=1_{H},
\end{align*}
where we used several times the normalization property of $ \mathcal{F}^{-1} $ and the antipode axiom for $ S $ and
$ \Delta $ as well as the $ 2 $-cocylce condition for $ \mathcal{F}^{-1} $. Similarly
\begin{align*}
U^{-1}U&=S(\mathcal{F}_{1}^{-1})\mathcal{F}_{2}^{-1}\mathcal{F}_{1}S(\mathcal{F}_{2})\\
&=S(\underbrace{\mathcal{F}_{1}\epsilon(\mathcal{F}_{2})}_{=1_{H}})S(\mathcal{F}_{1}^{-1})\mathcal{F}_{2}^{-1}
\mathcal{F}_{1}S(\mathcal{F}_{2})\\
&=S(\mathcal{F}_{1})S(\mathcal{F}_{1}^{-1})\mathcal{F}_{2}^{-1}\mathcal{F}_{1}\epsilon(\mathcal{F}_{2})
S(\mathcal{F}_{2})\\
&=S(\mathcal{F}_{1}^{-1}\mathcal{F}_{1})\mathcal{F}_{2}^{-1}\mathcal{F}_{1}(\mathcal{F}_{2})_{(1)}
S((\mathcal{F}_{2})_{(2)})S(\mathcal{F}_{2})\\
&=\mathrm{m}((\mathrm{m}\tensor 1)\circ(S\tensor 1\tensor S)
(1\tensor\mathcal{F}_{2}^{-1}\mathcal{F}_{1}(\mathcal{F}_{2})_{(1)}
\tensor\mathcal{F}_{2}(\mathcal{F}_{2})_{(2)}))\\
&=\mathrm{m}((\mathrm{m}\tensor 1)\circ(S\tensor 1\tensor S)
(\mathcal{F}_{12}^{-1}\mathcal{F}_{23}(1\tensor\Delta)\mathcal{F}))\\
&=\mathrm{m}((\mathrm{m}\tensor 1)\circ(S\tensor 1\tensor S)((\Delta\tensor 1)\mathcal{F}))\\
&=S((\mathcal{F}_{1})_{(1)})(\mathcal{F}_{1})_{(2)}S(\mathcal{F}_{2})\\
&=\epsilon(\mathcal{F}_{1})S(\mathcal{F}_{2})\\
&=S(\epsilon(\mathcal{F}_{1})\mathcal{F}_{2})\\
&=S(1_{H})\\
&=1_{H}.
\end{align*}
Now the antipode axiom for $ S^{\mathcal{F}} $ and $ \Delta^{\mathcal{F}} $ is not hard to verify. Let $ \xi\in H $.
Then
\begin{align*}
\mathrm{m}((S^{\mathcal{F}}\tensor 1)\Delta^{\mathcal{F}}(\xi))&=\mathrm{m}(US(\mathcal{F}_{1}
\xi_{(1)}\mathcal{F}_{1}^{-1})U^{-1}
\tensor\mathcal{F}_{2}\xi_{(2)}\mathcal{F}_{2}^{-1})\\
&=US(\mathcal{F}_{1}^{-1})S(\xi_{(1)})S(\mathcal{F}_{1})S(\mathcal{F}_{1}^{-1})\mathcal{F}_{2}^{-1}
\mathcal{F}_{2}\xi_{(2)}\mathcal{F}_{2}^{-1}\\
&=US(\mathcal{F}_{1}^{-1})S(\xi_{(1)})\xi_{(2)}\mathcal{F}_{2}^{-1}\\
&=US(\mathcal{F}_{1}^{-1})\epsilon(\xi)\mathcal{F}_{2}^{-1}\\
&=\epsilon(\xi)UU^{-1}\\
&=\eta(\epsilon(\xi)),
\end{align*}
again by the antialgebra map property and antipode property of $ S $ and similarly
\begin{align*}
\mathrm{m}((1\tensor S^{\mathcal{F}})\Delta^{\mathcal{F}}(\xi))&=\mathrm{m}(\mathcal{F}_{1}\xi_{(1)}
\mathcal{F}_{1}^{-1}\tensor
US(\mathcal{F}_{2}\xi_{(2)}\mathcal{F}_{2}^{-1})U^{-1})\\
&=\mathcal{F}_{1}\xi_{(1)}\mathcal{F}_{1}^{-1}\mathcal{F}_{1}S(\mathcal{F}_{2})S(\mathcal{F}_{2}^{-1})
S(\xi_{(2)})S(\mathcal{F}_{2})U^{-1}\\
&=\mathcal{F}_{1}\xi_{(1)}\underbrace{S(\mathcal{F}_{2}^{-1}\mathcal{F}_{2})}_{=1_{H}}
S(\xi_{(2)})S(\mathcal{F}_{2})U^{-1}\\
&=\mathcal{F}_{1}\xi_{(1)}S(\xi_{(2)})S(\mathcal{F}_{2})U^{-1}\\
&=\epsilon(\xi)\underbrace{\mathcal{F}_{1}S(\mathcal{F}_{2})}_{=U}U^{-1}\\
&=\eta(\epsilon(\xi)).
\end{align*}
Then $ S^{\mathcal{F}} $ is an antipode compatible to the bialgebra structure and
$ (H,+,\cdot,\eta,\Delta^{\mathcal{F}},\epsilon,S^{\mathcal{F}},\mathbb{k}) $ is a Hopf algebra.
\end{proof}

Assume that $ (H,+,\cdot,\eta,\Delta,\epsilon,S,\mathbb{k}) $ is a cocommutative Hopf algebra, i.e. for $ \xi\in H $
one has
\begin{align}
\xi_{(2)}\tensor\xi_{(1)}=\xi_{(1)}\tensor\xi_{(2)}
\end{align}
or, equivalently, the commutativity of the following diagram
\begin{equation}
\begin{tikzpicture}
  \matrix (m) [matrix of math nodes,row sep=3em,column sep=0.5em,minimum width=2em]
  {
      H & & & H\tensor H \\
       & & & H\tensor H, \\};
  \path[-stealth]
    (m-1-1) edge node [above] {$\Delta$} (m-1-4)
    (m-1-4) edge node [right] {$\sigma$} (m-2-4)
    (m-1-1) edge node [below] {$\Delta$} (m-2-4);
\end{tikzpicture}
\end{equation}
where $ \sigma\colon H\tensor H\rightarrow H\tensor H $ denotes the isomorphism flipping the tensor components called 
the \textit{braiding isomorphism}. The deformed coproduct is not cocommutative in general, because for
$ \xi\in H $ one has
\begin{align*}
\sigma(\Delta^{\mathcal{F}}(\xi))=\sigma(\mathcal{F}\cdot(\xi_{(1)}\tensor\xi_{(2)})\cdot\mathcal{F}^{-1})
=\mathcal{F}_{2}\xi_{(2)}\mathcal{F}_{2}^{-1}\tensor\mathcal{F}_{1}\xi_{(1)}\mathcal{F}_{1}^{-1},
\end{align*}
which in general does not coincide with
\begin{align*}
\Delta^{\mathcal{F}}(\xi)=\mathcal{F}\cdot\Delta(\xi)\cdot\mathcal{F}^{-1}=\mathcal{F}_{1}\xi_{(1)}\mathcal{F}_{1}^{-1}
\tensor\mathcal{F}_{2}\xi_{(2)}\mathcal{F}_{2}^{-1}.
\end{align*}
In this sense we can say that the twist $ \mathcal{F} $ sometimes also induces a quantization of
the Hopf algebra.

\begin{remark}\label{RemarkJTwist2}
We discovered in Remark \ref{RemarkJTwist1} that conditions (\ref{JConvention1}) and (\ref{JConvention2}) for
an invertible element $ \mathcal{F}^{-1}\in H\tensor H $ are equivalent to the twist conditions for
$ (\mathcal{F}^{-1})^{-1}=\mathcal{F} $ stated in Definition \ref{DefTwist}. Let us look a bit closer at these
conditions for $ \mathcal{F}^{-1} $. Assume that $ \mathcal{F} $ is a twist. Then according to Theorem
\ref{theoremdeformedhopf} there is the Hopf algebra structure
$ (H,+,\cdot,\eta,\Delta^{\mathcal{F}},\epsilon,S^{\mathcal{F}},\mathbb{k}) $. In terms of this structure the twist
conditions for $ \mathcal{F}^{-1} $ have the form
\begin{align*}
\mathcal{F}_{12}^{-1}\cdot((\Delta^{\mathcal{F}}\tensor 1)\mathcal{F}^{-1})&=\mathcal{F}_{12}^{-1}\cdot(
\mathcal{F}\Delta(\mathcal{F}_{1}^{-1})\mathcal{F}^{-1}\tensor\mathcal{F}_{2}^{-1})\\
&=\Delta(\mathcal{F}_{1}^{-1})\mathcal{F}^{-1}\tensor\mathcal{F}_{2}^{-1}\\
&=((\Delta\tensor 1)\mathcal{F}^{-1})\cdot\mathcal{F}_{12}^{-1}\\
&=((1\tensor\Delta)\mathcal{F}^{-1})\cdot\mathcal{F}_{23}^{-1}\\
&=\mathcal{F}_{1}^{-1}\tensor\Delta(\mathcal{F}_{2}^{-1})\cdot\mathcal{F}^{-1}\\
&=\mathcal{F}_{1}^{-1}\tensor\mathcal{F}^{-1}\Delta^{\mathcal{F}}(\mathcal{F}_{2}^{-1})\\
&=\mathcal{F}_{23}^{-1}\cdot((1\tensor\Delta^{\mathcal{F}})\mathcal{F}^{-1})
\end{align*}
and $ (\epsilon\tensor 1)\mathcal{F}^{-1}=(1\tensor\epsilon)\mathcal{F}^{-1}=1_{H} $. Thus $ \mathcal{F}^{-1} $ is
a twist on $ (H,+,\cdot,\eta,\Delta^{\mathcal{F}},\epsilon,S^{\mathcal{F}},\mathbb{k}) $ if $ \mathcal{F} $ is a
twist on $ (H,+,\cdot,\eta,\Delta,\epsilon,S,\mathbb{k}) $. Arranging the above equations in another order we see
that also the converse is true. Then $ \mathcal{F} $ is a twist on
$ (H,+,\cdot,\eta,\Delta,\epsilon,S,\mathbb{k}) $ if and only if $ \mathcal{F}^{-1} $ is a twist on
$ (H,+,\cdot,\eta,\Delta^{\mathcal{F}},\epsilon,S^{\mathcal{F}},\mathbb{k}) $ and the convention chosen in
Definition \ref{DefTwist} is equivalent to the axiom system
\begin{align}
((\Delta\tensor 1)J\cdot J_{12}&=((1\tensor\Delta)J)\cdot J_{23},\\
(\epsilon\tensor 1)J&=(1\tensor\epsilon)J=1_{H},
\end{align}
where we defined $ J=\mathcal{F}^{-1} $. The procedure of twisting is even involutive, i.e. twisting the
twisted structure gives back the original one. This can be seen quite easyly, since for $ \xi\in H $ one has
\begin{align*}
(\Delta^{\mathcal{F}})^{\mathcal{F}^{-1}}(\xi)=\mathcal{F}^{-1}\Delta^{\mathcal{F}}(\xi)(\mathcal{F}^{-1})^{-1}
=\mathcal{F}^{-1}\mathcal{F}\Delta(\xi)\mathcal{F}^{-1}\mathcal{F}
=\Delta(\xi)
=(\Delta^{\mathcal{F}^{-1}})^{\mathcal{F}}(\xi)
\end{align*}
and
\begin{align*}
(S^{\mathcal{F}})^{\mathcal{F}^{-1}}(\xi)&=\mathcal{F}_{1}^{-1}S(\mathcal{F}_{2}^{-1})S^{\mathcal{F}}(\xi)
S(\mathcal{F}_{1})\mathcal{F}_{2}\\
&=\mathcal{F}_{1}^{-1}S(\mathcal{F}_{2}^{-1})\mathcal{F}_{1}\underbrace{\mathcal{F}_{1}^{-1}\epsilon
(\mathcal{F}_{1}^{-1})}_{=1_{H}}S(\mathcal{F}_{2})S(\xi)
S(\mathcal{F}_{1}^{-1})\mathcal{F}_{2}^{-1}\underbrace{\mathcal{F}_{2}\epsilon(\mathcal{F}_{2})}_{=1_{H}}
S(\mathcal{F}_{1})\mathcal{F}_{2}\\
&=\mathcal{F}_{1}^{-1}\epsilon(\mathcal{F}_{1}^{-1})S(\mathcal{F}_{2}\mathcal{F}_{2}^{-1})S(\xi)
S(\mathcal{F}_{1}\mathcal{F}_{1}^{-1})\mathcal{F}_{2}\epsilon(\mathcal{F}_{2})\\
&=S(\xi)\\
&=\mathcal{F}_{1}S(\mathcal{F}_{2})\mathcal{F}_{1}^{-1}\mathcal{F}_{1}\epsilon(\mathcal{F}_{1})
S(\mathcal{F}_{2}^{-1})S(\xi)
S(\mathcal{F}_{1})\mathcal{F}_{2}\mathcal{F}_{2}^{-1}\epsilon(\mathcal{F}_{2}^{-1})
S(\mathcal{F}_{1}^{-1})\mathcal{F}_{2}^{-1}\\
&=\mathcal{F}_{1}S(\mathcal{F}_{2})\mathcal{F}_{1}^{-1}S(\mathcal{F}_{2}^{-1})S(\xi)
S(\mathcal{F}_{1})\mathcal{F}_{2}S(\mathcal{F}_{1}^{-1})\mathcal{F}_{2}^{-1}\\
&=\mathcal{F}_{1}S(\mathcal{F}_{2})S^{\mathcal{F}^{-1}}(\xi)S(\mathcal{F}_{1}^{-1})\mathcal{F}_{2}^{-1}\\
&=(S^{\mathcal{F}^{-1}})^{\mathcal{F}}(\xi).
\end{align*}
This is called \textit{twisting back}.
Thus it is not surprising that both conventions appear in the literature.
We want to follow the $ \mathcal{F} $-convention as it is done in \cite{schenkeldrintwistdef}, but remark that the
$ J $-convention is also very common. Consider for example the inspiring book \cite{etingof2002lectures} by
P. Etingof and O. Schiffman.
\end{remark}

In a next step we show that a twist $ \mathcal{F} $ on $ (H,+,\cdot,\eta,\Delta,\epsilon,S,\mathbb{k}) $ also induces
a deformation (quantization) on any left $ H $-module algebra. Here and in the following we use the convention to
denote $ (H,+,\cdot,\eta,\Delta,\epsilon,S,\mathbb{k}) $ shortly by $ H $ and the deformed Hopf algebra
$ (H,+,\cdot,\eta,\Delta^{\mathcal{F}},\epsilon,S^{\mathcal{F}},\mathbb{k}) $ by $ H^{\mathcal{F}} $.

\begin{definition}[Left Hopf Algebra Module Algebra]\label{DefModAlg}
Let $ (H,+,\cdot,\eta,\Delta,\epsilon,S,\mathbb{k}) $ be a Hopf algebra and $ (\algebra{A},\cdot,+,\mathbb{k}) $ an
algebra (see Definition \ref{defalgebra}). Then $ (\algebra{A},\cdot,+,\mathbb{k}) $ is said to be a
left $ H $-module algebra if the following two conditions are satisfied.
\begin{compactenum}
\item $ (\algebra{A},\cdot,+,\mathbb{k}) $ is a left $ H $-module, i.e. there is a $ \mathbb{k} $-linear map
$ \rhd\colon H\tensor\algebra{A}\rightarrow\algebra{A} $ such that for all $ \xi,\zeta\in H $ and $ a\in\algebra{A} $
one has
\begin{align}\label{HModuleProp1}
(\xi\zeta)\rhd a=\xi\rhd(\zeta\rhd a)
\end{align}
and
\begin{align}\label{HModuleProp2}
1_{H}\rhd a=a.
\end{align}
\item The $ H $-module structure of $ (\algebra{A},\cdot,+,\mathbb{k}) $ respects the algebra multiplication, i.e.
for all $ \xi\in H $ and $ a,b\in\algebra{A} $ one has
\begin{align}\label{HModuleProp3}
\xi\rhd(ab)=(\xi_{(1)}\rhd a)(\xi_{(2)}\rhd b)
\end{align}
and
\begin{align}\label{HModuleProp4}
\xi\rhd 1_{\algebra{A}}=\epsilon(\xi)1_{\algebra{A}}.
\end{align}
\end{compactenum}
\end{definition}

If $ \mathcal{F} $ is a twist on $ H $ we get from Theorem \ref{theoremdeformedhopf} the deformed Hopf algebra
$ H^{\mathcal{F}} $.
The question we are interested in is the following: is there a way to structure any left
$ H $-module algebra $ (\algebra{A},\cdot,+,\mathbb{k}) $ such that it becomes a left $ H^{\mathcal{F}} $-module?
Is this possible in a way such that $ \algebra{A} $ stays an algebra? And finally, is the new module structure
compatible with the new multiplication? The next theorem gives a positive answer to all those questions.

\begin{theorem}\label{TheoremTwistedStar}
Let $ (H,+,\cdot,\eta,\Delta,\epsilon,S,\mathbb{k}) $ be a Hopf algebra and $ \mathcal{F} $ a twist on $ H $.
Moreover, consider a left $ H $-module algebra $ (\algebra{A},\cdot,+,\mathbb{k}) $ and define a map
$ \star\colon\algebra{A}\times\algebra{A}\rightarrow\algebra{A} $ for any $ a,b\in\algebra{A} $ by
\begin{align}\label{twistedstar}
a\star b=m(\mathcal{F}^{-1}\rhd(a\tensor b))=(\mathcal{F}_{1}^{-1}\rhd a)(\mathcal{F}_{2}^{-1}\rhd b),
\end{align}
where $ \mathrm{m}\colon\algebra{A}\tensor\algebra{A}\rightarrow\algebra{A} $ is the map
$ \mathrm{m}(a\tensor b)=a\cdot b=ab $ and
$ \rhd $ the left $ H $-module action of $ H $ on $ (\algebra{A},\cdot,+,\mathbb{k}) $. Remark that we
extended the action to an action of $ H\tensor H $ on $ \algebra{A}\tensor\algebra{A} $ in the obvious way without
changing the notation. Then $ (\algebra{A},\star,+,\mathbb{k}) $ is a left $ H^{\mathcal{F}} $-module algebra.
\end{theorem}

\begin{proof}
We have to show that
$ (\algebra{A},\star,+,\mathbb{k}) $ (which we will shortly denote by $ \algebra{A}^{\mathcal{F}} $ in the following)
is an algebra. Moreover, we have to verify axioms i.) and ii.) of Definition \ref{DefModAlg},
having in mind that $ \algebra{A} $ is a left $ H $-module algebra.
First of all, $ H^{\mathcal{F}} $ is a Hopf algebra according to Theorem \ref{theoremdeformedhopf}. Next we prove that
$ \algebra{A}^{\mathcal{F}} $ is an algebra. Let $ a,b,c\in\algebra{A} $ and calculate
\begin{align*}
(a\star b)\star c&=((\mathcal{F}_{1}^{-1}\rhd a)\cdot(\mathcal{F}_{2}^{-1}\rhd b))\star c\\
&=\mathcal{F}_{1}^{-1}\rhd((\mathcal{F}_{1}^{-1}\rhd a)\cdot(\mathcal{F}_{2}^{-1}\rhd b))(\mathcal{F}_{2}^{-1}\rhd c)\\
&=((\mathcal{F}_{1}^{-1})_{(1)}\rhd(\mathcal{F}_{1}^{-1}\rhd a))\cdot
((\mathcal{F}_{1}^{-1})_{(2)}\rhd(\mathcal{F}_{2}^{-1}\rhd b))\cdot(\mathcal{F}_{2}^{-1}\rhd c)\\
&=(((\mathcal{F}_{1}^{-1})_{(1)}\mathcal{F}_{1}^{-1})\rhd a)\cdot
(((\mathcal{F}_{1}^{-1})_{(2)}\mathcal{F}_{2}^{-1})\rhd b)\cdot(\mathcal{F}_{2}^{-1}\rhd c)\\
&=\mathrm{m}((\mathrm{m}\tensor 1)\circ((\mathcal{F}_{1}^{-1})_{(1)}\mathcal{F}_{1}^{-1}\tensor
(\mathcal{F}_{1}^{-1})_{(2)}\mathcal{F}_{2}^{-1}\tensor\mathcal{F}_{2}^{-1})\rhd(a\tensor b\tensor c))\\
&=\mathrm{m}((\mathrm{m}\tensor 1)\circ
(((\Delta\tensor 1)\mathcal{F}^{-1})\cdot\mathcal{F}_{12}^{-1})\rhd(a\tensor b\tensor c))\\
&=\mathrm{m}((\mathrm{m}\tensor 1)\circ
(((1\tensor\Delta)\mathcal{F}^{-1})\cdot\mathcal{F}_{23}^{-1})\rhd(a\tensor b\tensor c))\\
&=\mathrm{m}((\mathrm{m}\tensor 1)\circ
(\mathcal{F}_{1}^{-1}\tensor(\mathcal{F}_{2}^{-1})_{(1)}\mathcal{F}_{1}^{-1}\tensor
(\mathcal{F}_{2}^{-1})_{(2)}\mathcal{F}_{2}^{-1})\rhd(a\tensor b\tensor c))\\
&=(\mathcal{F}_{1}^{-1}\rhd a)\cdot(((\mathcal{F}_{2}^{-1})_{(1)}\mathcal{F}_{1}^{-1})\rhd b)\cdot
(((\mathcal{F}_{2}^{-1})_{(2)}\mathcal{F}_{2}^{-1})\rhd c)\\
&=(\mathcal{F}_{1}^{-1}\rhd a)\cdot((\mathcal{F}_{2}^{-1})_{(1)}\rhd(\mathcal{F}_{1}^{-1}\rhd b))\cdot
((\mathcal{F}_{2}^{-1})_{(2)}\rhd(\mathcal{F}_{2}^{-1}\rhd c))\\
&=(\mathcal{F}_{1}^{-1}\rhd a)\cdot(\mathcal{F}_{2}^{-1}\rhd((\mathcal{F}_{1}^{-1}\rhd b)\cdot
(\mathcal{F}_{2}^{-1}\rhd c)))\\
&=(\mathcal{F}_{1}^{-1}\rhd a)\cdot(\mathcal{F}_{2}^{-1}\rhd(b\star c))\\
&=a\star(b\star c),
\end{align*}
where we used the $ 2 $-cocycle property (\ref{JConvention1}) of $ \mathcal{F}^{-1} $ and the left $ H $-module
properties (\ref{HModuleProp1}) and (\ref{HModuleProp3}) of $ \algebra{A} $ several times. This shows that the
product $ \star $ is associative. There is even the same unit element $ 1_{\algebra{A}} $ for $ \star $. To see this,
take $ a\in\algebra{A} $ and use properties (\ref{HModuleProp2}), (\ref{HModuleProp4}) and the normalization
property (\ref{JConvention2}) of $ \mathcal{F}^{-1} $ to get
\begin{align*}
a\star 1_{\algebra{A}}&=(\mathcal{F}_{1}^{-1}\rhd a)\cdot(\mathcal{F}_{2}^{-1}\rhd 1_{\algebra{A}})\\
&=(\mathcal{F}_{1}^{-1}\rhd a)\cdot\epsilon(\mathcal{F}_{2}^{-1})1_{\algebra{A}}\\
&=(\mathcal{F}_{1}^{-1}\epsilon(\mathcal{F}_{2}^{-1}))\rhd a\\
&=1_{H}\rhd a\\
&=a\\
&=(\epsilon(\mathcal{F}_{1}^{-1})\mathcal{F}_{2}^{-1})\rhd a\\
&=(\epsilon(\mathcal{F}_{1}^{-1})1_{\algebra{A}})\cdot(\mathcal{F}_{2}^{-1}\rhd a)\\
&=(\mathcal{F}_{1}^{-1}\rhd 1_{\algebra{A}})\cdot(\mathcal{F}_{2}^{-1}\rhd a)\\
&=1_{\algebra{A}}\star a.
\end{align*}
Remark that also the $ \mathbb{k} $-linearity of $ \rhd $ was needed to prove the above formula.
This is enough to show that $ \algebra{A}^{\mathcal{F}} $ is an algebra, since $ \rhd $ is linear.
The only problematic property of Definition \ref{DefModAlg} we have to check is (\ref{HModuleProp3}).
Let $ a,b\in\algebra{A} $ and $ \xi\in H $. Then
\begin{align*}
\xi\rhd(a\star b)&=\xi\rhd((\mathcal{F}_{1}^{-1}\rhd a)\cdot(\mathcal{F}_{2}^{-1}\rhd b))\\
&=(\xi_{(1)}\rhd(\mathcal{F}_{1}^{-1}\rhd a))\cdot(\xi_{(2)}\rhd(\mathcal{F}_{2}^{-1}\rhd b))\\
&=\mathrm{m}((\xi_{(1)}\mathcal{F}_{1}^{-1}\tensor\xi_{(2)}\mathcal{F}_{2}^{-1})\rhd(a\tensor b))\\
&=\mathrm{m}((\Delta(\xi)\tensor\mathcal{F}^{-1})\rhd(a\tensor b))\\
&=\mathrm{m}((\mathcal{F}^{-1}\cdot\Delta^{\mathcal{F}}(\xi))\rhd(a\tensor b))\\
&=\mathrm{m}(\mathcal{F}^{-1}\rhd(\Delta^{\mathcal{F}}(\xi)\rhd(a\tensor b)))\\
&=(\xi_{(1)}^{\mathcal{F}}\rhd a)\star(\xi_{(2)}^{\mathcal{F}}\rhd b),
\end{align*}
where we use the short notation $ \Delta^{\mathcal{F}}(\xi)=\xi_{(1)}^{\mathcal{F}}\tensor\xi_{(2)}^{\mathcal{F}} $.
We applied Eq. (\ref{HModuleProp1}) and Eq. (\ref{HModuleProp3}) several times as well as the definition
(\ref{deformedcoproduct}) of the deformed coproduct $ \Delta^{\mathcal{F}} $. This proves the result.
\end{proof}

This theorem tells us that we can deform the product of a left $ H $-module algebra to obtain a left module algebra
of the deformed Hopf algebra. Assume that $ (\algebra{A},\cdot,+,\mathbb{k}) $ is a commutative left $ H $-module
algebra, i.e. $ ab=ba $ for all $ a,b\in\algebra{A} $. Then the deformed algebra $ (\algebra{A},\star,+,\mathbb{k}) $
(which we may shortly denote by $ \algebra{A}^{\mathcal{F}} $ as in the last proof) is not commutative
again in general, since for $ a,b\in\algebra{A} $ the term
\begin{align*}
a\star b=(\mathcal{F}_{1}^{-1}\rhd a)(\mathcal{F}_{2}^{-1}\rhd b)
\end{align*}
may not coincide with
\begin{align*}
b\star a=(\mathcal{F}_{1}^{-1}\rhd b)(\mathcal{F}_{2}^{-1}\rhd a).
\end{align*}
Thus $ \mathcal{F} $ will not only deform but also quantize $ \algebra{A} $ in a way which is compatible to
the deformation (quantization) of $ H $.

\section{Twist Star Products}\label{secStarProdTwist}

We are only interested in the case when $ H $ is a universal enveloping algebra $ \mathcal{U}(\mathfrak{g}) $ of a
(finite-dimensional) Lie algebra $ \mathfrak{g} $ over $ \mathbb{k} $. As a short repetition, the \textit{universal
enveloping algebra} of $ \mathfrak{g} $ is the tensor algebra of $ \mathfrak{g} $ modulo the relation
\begin{align}
\left[\xi,\zeta\right]=\xi\zeta-\zeta\xi,
\end{align}
for all $ \xi,\zeta\in\mathfrak{g} $. One structures $ \mathcal{U}(\mathfrak{g}) $ as a Hopf algebra by setting
\begin{align*}
\Delta_{\mathcal{U}(\mathfrak{g})}(\xi)&=\xi\tensor 1_{\mathbb{k}}+1_{\mathbb{k}}\tensor\xi,\\
\epsilon_{\mathcal{U}(\mathfrak{g})}(\xi)&=0,\\
S_{\mathcal{U}(\mathfrak{g})}(\xi)&=-\xi,
\end{align*}
where $ \xi\in\mathfrak{g} $. This extends to algebra (anti) homomorphisms and yields a cocommutative Hopf algebra.
In a general Hopf algebra $ H $ one calls an element $ \xi\in H $ that
satisfies $ \Delta(\xi)=\xi\tensor 1_{H}+1_{H}\tensor\xi $ \textit{primitive}.
We want to pass to formal power series thus we have to recap them first.

Let $ \mathbb{k} $ be a field of characteristic zero and define $ K=\mathbb{k}\left[\left[\hbar\right]\right] $ to be
the space of all sequences $ a=(a_{0},a_{1},a_{2},...) $ in $ \mathbb{k} $ that can be written as a formal power series
\begin{align}
a=\sum\limits_{r=0}^{\infty}\hbar^{r}a_{r}
\end{align}
with $ a_{r}\in\mathbb{k} $ and \textit{Planck\textquoteright s constant} $ \hbar $. This space has the structure
of an associative commutative unital ring if we declare a product on $ K $ by
\begin{align}\label{eq18thomas}
ab
=\Bigg(\sum\limits_{r=0}^{\infty}\hbar^{r}a_{r}\Bigg)\Bigg(\sum\limits_{r=0}^{\infty}\hbar^{r}b_{r}\Bigg)
=\sum\limits_{r=0}^{\infty}\hbar^{r}\Bigg(\sum\limits_{s=0}^{r}a_{s}b_{r-s}\Bigg)
\end{align}
for all $ a,b\in K $. If $ V $ is a $ \mathbb{k} $-vector space, the formal power series
$ V[[\hbar]] $ form a $ K $-module by defining for $ a\in K $ and $ 
v\in V[[\hbar]] $
\begin{align}
av
=\Bigg(\sum\limits_{r=0}^{\infty}\hbar^{r}a_{r}\Bigg)\Bigg(\sum\limits_{r=0}^{\infty}\hbar^{r}v_{r}\Bigg)
=\sum\limits_{r=0}^{\infty}\hbar^{r}\Bigg(\sum\limits_{s=0}^{r}a_{s}v_{r-s}\Bigg)
\end{align}
with the same multiplication (\ref{eq18thomas}) for elements in $ V[[\hbar]] $. The next step is to define a
topology such that $ V[[\hbar]] $ is complete. So choose a $ C>1 $ and define a valuation on $ K $ by
\begin{align}\label{eq19thomas}
\parallel a_{n}\hbar^{n}+a_{n+1}\hbar^{n+1}+...\parallel=C^{-n},
\end{align}
where $ (a_{1},a_{2},...) $ is an arbitrary element of $ K $ such that $ a_{n} $ is the smallest coefficient that
is not equal to zero. For the sequence $ (0,0,...) $ that is constant zero one defines the valuation to be $ 0 $.
One can check that this is indeed a well-defined valuation on $ K $ called the \textit{$ \hbar $-adic valuation}.
The topology induced by the $ \hbar $-adic valuation is said to be the \textit{$ \hbar $-adic topology} and one can
check that $ K $ is complete with respect to this topology. Moreover, if we extend the valuation to $ V[[\hbar]] $, also
$ V[[\hbar]] $ is complete with respect to the $ \hbar $-adic topology (c.f.
\cite[Section~1.1.1]{etingof2002lectures} and \cite[Section~6.2.1]{waldmannbuch1}).

Coming back to the universal enveloping algebra $ \mathcal{U}(\mathfrak{g}) $ we can now pass to
$ \mathcal{U}(\mathfrak{g})[[\hbar]] $. By extending
$ \Delta_{\mathcal{U}(\mathfrak{g})}, \epsilon_{\mathcal{U}(\mathfrak{g})} $ and $ S_{\mathcal{U}(\mathfrak{g})} $
$ \mathbb{k}[[\hbar]] $-linearly one gets a Hopf algebra structure
\begin{align*}
(\mathcal{U}(\mathfrak{g})[[\hbar]],+,\cdot,\Delta,\epsilon,S,\mathbb{k})
\end{align*}
on $ \mathcal{U}(\mathfrak{g})[[\hbar]] $, too. Only the comultiplication $ \Delta\colon
\mathcal{U}(\mathfrak{g})[[\hbar]]\rightarrow\mathcal{U}(\mathfrak{g})[[\hbar]]\tensor\mathcal{U}(\mathfrak{g})
[[\hbar]] $ has to be modified. To assure it is well-defined we have to take the completion
$ (\mathcal{U}(\mathfrak{g})\tensor\mathcal{U}(\mathfrak{g}))[[\hbar]] $ of $ \mathcal{U}(\mathfrak{g})[[\hbar]]
\tensor\mathcal{U}(\mathfrak{g})[[\hbar]] $, i.e. we have to assume
\begin{align}
\Delta\colon\mathcal{U}(\mathfrak{g})[[\hbar]]\rightarrow(\mathcal{U}(\mathfrak{g})\tensor\mathcal{U}(\mathfrak{g}))
[[\hbar]].
\end{align}
By the same argumentation twists on $ \mathcal{U}(\mathfrak{g})[[\hbar]] $ are defined as elements of
$ (\mathcal{U}(\mathfrak{g})\tensor\mathcal{U}(\mathfrak{g}))[[\hbar]] $. In complete analogy to the last section
we define a twist for our special ``Hopf algebra''.

\begin{definition}\label{DefTwistUEA}
Let $ \mathfrak{g} $ be a Lie algebra over $ \mathbb{k} $. We call an element $ \mathcal{F}\in(\mathcal{U}
(\mathfrak{g})\tensor\mathcal{U}(\mathfrak{g}))[[\hbar]] $ a \textbf{twist} on $ \mathcal{U}(\mathfrak{g})[[\hbar]] $
if the following three conditions are satisfied,
\begin{compactenum}
\item $ \mathcal{F}_{12}(\Delta\tensor 1)\mathcal{F}=\mathcal{F}_{23}(1\tensor\Delta)\mathcal{F} $,
\item $ (\epsilon\tensor 1)\mathcal{F}=(1\tensor\epsilon)\mathcal{F}=1_{\mathcal{U}(\mathfrak{g}))[[\hbar]]} $,
\item $ \mathcal{F}=1_{\mathcal{U}(\mathfrak{g})[[\hbar]]}\tensor 1_{\mathcal{U}(\mathfrak{g})[[\hbar]]}
\mod{\hbar} $.
\end{compactenum}
\end{definition}

\begin{remark}
One motivation for the additional axiom iii.) in Definition \ref{DefTwistUEA} is the following: if we consider a left
$ \mathcal{U}(\mathfrak{g}))[[\hbar]] $-module algebra $ \algebra{A} $, one gets a deformed product on
$ \algebra{A} $ via Eq. (\ref{twistedstar}) by Theorem \ref{TheoremTwistedStar}. Axiom iii.) assures that $ \star $
coincides with the original multiplication $ \cdot $ on $ \algebra{A} $ in zero order of $ \hbar $. If we think of
quantization this is the classical limit that gives back the classical structure.
\end{remark}

The main motivation to consider \textit{twist products} $ \star $, defined in Eq. (\ref{twistedstar}), is that
in many cases $ \star $ denotes a star product. We do not give a motivated introduction into the rich field
of star products and formal deformation quantization, but refer to \cite{delorme2012noncommutative},
\cite[Chapter~3]{esposito2014formality} and \cite[Chapter~6]{waldmannbuch1}. For this thesis it suffices to see
star products as products on formal power series of smooth functions of a manifold. We are mainly interested in
their existence but not in their properties. Some nice examples of star products can be found in
\cite{Minkowskistar}. Remark that Kontsevich proved that there is a star product on
any Poisson manifold (c.f. \cite{kontsevich1997deformation}).

	\begin{definition}[Star Product]\label{def3thomas}
	A star product on a Poisson manifold $ (M,\pi) $ is a $ \mathbb{R}[[\hbar]] $-bilinear map
		\begin{align}
		\star\colon\Cinfty(M)[[\hbar]]\times\Cinfty(M)[[\hbar]]\rightarrow\Cinfty(M)[[\hbar]]
		\end{align}
	of the form
		\begin{align}
		f\star g=\sum\limits_{r=0}^{\infty}\hbar^{r}C_{r}(f,g),
		\end{align}
	for $ f,g\in\Cinfty(M)[[\hbar]] $, where $ C_{r}\colon\Cinfty(M)[[\hbar]]\times\Cinfty(M)[[\hbar]]
	\rightarrow\Cinfty(M)[[\hbar]] $ are $
	\mathbb{R}[[\hbar]] $-bilinear bidifferential operators such that the following four conditions are satisfied,
		\begin{compactenum}
		\item $ \star $ is associative,
		\item it holds for $ f,g\in\Cinfty(M)[[\hbar]] $ that $ C_{0}(f,g)=fg $,
		\item it holds for $ f,g\in\Cinfty(M)[[\hbar]] $ that $ C_{1}(f,g)-C_{1}(g,f)=\left\lbrace f,g\right\rbrace $,
		where the Poisson bracket $ \left\lbrace\cdot,\cdot\right\rbrace $ is defined in Eq. (\ref{PoissonTensor}),
		\item it holds for $ f\in\Cinfty(M)[[\hbar]] $ that $ 1\star f=f=f\star 1 $, where $ 1 $ denotes the
		function that is constant $ 1 $.
		\end{compactenum}
	\end{definition}
Condition ii.) of Definition \ref{def3thomas} guarantees, that the classical limit $ \hbar\rightarrow 0 $ gives back the
commutative pointwise multiplication of functions, while
condition iii.) can be interpreted as follows: the star product $ \star $
\textit{deforms} the Poisson bracket $ \left\lbrace\cdot,\cdot\right\rbrace $ (or equivalently the Poisson
bivector $ \pi $ corresponding to $ \left\lbrace\cdot,\cdot\right\rbrace $).
This is the correspondence principle
we mentioned in the introduction. We only omitted a prefactor. If $ \star $ is a star product on a symplectic
manifold $ (M,\omega) $ one says that $ \star $ deforms the symplectic structure $ \omega $ of $ M $. This is
nearby thinking of Eq. (\ref{SymplPoissonBracket}).

Inspired by Theorem \ref{TheoremTwistedStar} we have the
following

\begin{definition}[Twist Star Product]
Let $ (M,\pi) $ be a Poisson manifold and $ \star $ a star product on it. Then $ \star $ is said to be a
star product induced by a twist or a twist star product if there is a Lie algebra $ \mathfrak{g} $
over $ \mathbb{R} $ together with a twist $ \mathcal{F} $ on $ \mathcal{U}(\mathfrak{g})[[\hbar]] $ such that
$ \Cinfty(M)[[\hbar]] $ is a left $ \mathcal{U}(\mathfrak{g})[[\hbar]] $-module algebra via $ \rhd $ and for all
$ f,g\in\Cinfty(M)[[\hbar]] $ one has
\begin{align}\label{TwistInducedStar}
f\star g=\mathrm{m}(\mathcal{F}^{-1}\rhd(f\tensor g))=(\mathcal{F}_{1}^{-1}\rhd f)(\mathcal{F}_{2}^{-1}\rhd g),
\end{align}
where $ \mathrm{m} $ denotes the multiplication $ \Cinfty(M)[[\hbar]]\tensor\Cinfty(M)[[\hbar]]
\rightarrow\Cinfty(M)[[\hbar]] $.
\end{definition}

There is a connection between twists on $ \mathcal{U}(\mathfrak{g})[[\hbar]] $ and $ r $-matrices on $ \mathfrak{g} $.
The correspondence principle states that the first order of the commutator of the star product gives the Poisson
bracket of the manifold. In a very similar way a twist can be seen as a deformation of a $ r $-matrix.
Conversely, the skew-symmetrization of the first order of a twist gives a $ r $-matrix. This is part of the following

\begin{proposition}\label{Proprmatrixfromtwist}
Let $ \mathfrak{g} $ be a Lie algebra over $ \mathbb{R} $ and
\begin{align}
\mathcal{F}=\sum\limits_{i=0}^{\infty}\hbar^{i}F_{i}
\end{align}
a twist on $ \mathcal{U}(\mathfrak{g}))[[\hbar]] $. Then
\begin{align}\label{rmatrixfromtwist}
r:=F_{1}^{-1}-\sigma(F_{1}^{-1})\in\mathfrak{g}\wedge\mathfrak{g}
\end{align}
is a $ r $-matrix on $ \mathfrak{g} $, where $ \sigma $ denotes the braiding isomorphism.
\end{proposition}

\begin{proof}
According to Definition \ref{DefTwistUEA} iii.) $ F_{0}=1_{\mathcal{U}(\mathfrak{g})[[\hbar]]}
\tensor 1_{\mathcal{U}(\mathfrak{g})[[\hbar]]}=:\mathbb{1} $ is an invertible
element. It is known that in this case $ \mathcal{F} $ is invertible with inverse
\begin{align}
\mathcal{F}^{-1}=\sum\limits_{i=0}^{\infty}\hbar^{i}J_{i},
\end{align}
where $ J_{0}=\mathbb{1} $ and
$ J_{i}=-\sum\limits_{j=1}^{i}F_{j}J_{i-j} $ for $ i>0 $. To choose $ J $ for the coefficients is natural
having in mind Remark \ref{RemarkJTwist2}. Then $ r $ given in Eq. (\ref{rmatrixfromtwist}) is a well-defined object
$ r=J_{1}-\sigma(J_{1}) $ in $ \mathcal{U}(\mathfrak{g})\tensor\mathcal{U}(\mathfrak{g}) $.
Let us write down the $ 2 $-cocycle condition (\ref{JConvention1}) for $ \mathcal{F}^{-1} $ by using the
multiplication rules of formal power series:
\begin{align}
\sum\limits_{i=0}^{\infty}\hbar^{i}\sum\limits_{j=0}^{i}((\Delta\tensor 1)J_{j})\cdot(J_{i-j}\tensor\mathbb{1})
=\sum\limits_{i=0}^{\infty}\hbar^{i}\sum\limits_{j=0}^{i}((1\tensor\Delta)J_{j})\cdot(\mathbb{1}\tensor J_{i-j}).
\end{align}
The first order in $ \hbar $ reads
\begin{align*}
J_{1}\tensor\mathbb{1}+(\Delta\tensor 1)J_{1}&=((\Delta\tensor 1)J_{0})\cdot(J_{1}\tensor\mathbb{1})
+((\Delta\tensor 1)J_{1})\cdot(J_{0}\tensor\mathbb{1})\\
&=((1\tensor\Delta)J_{0})\cdot(\mathbb{1}\tensor J_{1})+((1\tensor\Delta)J_{1})\cdot(\mathbb{1}\tensor J_{0})\\
&=\mathbb{1}\tensor J_{1}+(1\tensor\Delta)J_{1}.
\end{align*}
Since $ \sigma $ is an isomorphism also $ \sigma(\mathcal{F}^{-1}) $ fulfils Eq. (\ref{JConvention1}) and the same
calculation implies
\begin{align}\label{subtract}
\sigma(J_{1})\tensor\mathbb{1}+(\Delta\tensor 1)\sigma(J_{1})=\mathbb{1}\tensor\sigma(J_{1})+(1+\Delta)\sigma(J_{1}).
\end{align}
Subtracting Eq. (\ref{subtract}) from $ J_{1}\tensor\mathbb{1}+(\Delta\tensor 1)J_{1}
=\mathbb{1}\tensor J_{1}+(1\tensor\Delta)J_{1} $ we get a condition on $ r $, namely
\begin{align}\label{rmatrixhelp}
r\tensor\mathbb{1}+(\Delta\tensor 1)r=\mathbb{1}\tensor r+(1\tensor\Delta)r.
\end{align}
We first show that $ r $ is an element of $ \mathfrak{g}\wedge\mathfrak{g} $. The skew-symmetry is clear, since
$ \sigma(r)=\sigma(J_{1})-\sigma^{2}(J_{1})=-r $. The question is whether $ r $ lies in the wedge product of
$ \mathfrak{g} $, i.e. if the first and second tensor component of $ r $ are primitive elements. With the help of Eq.
(\ref{rmatrixhelp}) and the skew-symmetry of $ r $ we obtain
\begin{align*}
(\Delta\tensor 1)r-r_{23}-r_{13}&=r_{23}-r_{12}+(1+\Delta)r-r_{23}-r_{13}\\
&=(1+\Delta)r-r_{12}-r_{13}\\
&=-(1+\Delta)\sigma(r)-r_{12}-r_{13}\\
&=-r_{2}\tensor(r_{1})_{(1)}\tensor(r_{1})_{(2)}-r_{12}-r_{13}.
\end{align*}
Permuting the tensor factors as $ (1,2,3)\rightarrow(2,3,1) $ gives no minus sign, since $ r $ is skew-symmetric.
Thus
\begin{align*}
(\Delta\tensor 1)r-r_{23}-r_{13}&=-(r_{1})_{(1)}\tensor(r_{1})_{(2)}\tensor r_{2}-r_{32}-r_{31}\\
&=-((\Delta+1)r-r_{23}-r_{13}).
\end{align*}
But this implies $ \Delta(r_{1})\tensor r_{2}=(\Delta+1)r=r_{23}+r_{13}=\mathbb{1}\tensor r_{1}\tensor r_{2}
+r_{1}\tensor\mathbb{1}\tensor r_{2} $, i.e.
\begin{align}
\Delta(r_{1})=\mathbb{1}\tensor r_{1}+r_{1}\tensor\mathbb{1}.
\end{align}
Thus $ r_{1} $ is a primitive element and for this $ r_{1}\in\mathfrak{g} $. Since $ r $ is skew-symmetric there
must be $ r_{2}\in\mathfrak{g} $. This can also be seen by considering again Eq. (\ref{rmatrixhelp}) and using
\begin{align}\label{rmatrixhelp2}
(\Delta\tensor 1)r=r_{23}+r_{13}.
\end{align}
Then the condition for $ r_{2} $ reads analogous to Eq. (\ref{rmatrixhelp2})
\begin{align}\label{rmatrixhelp3}
(1\tensor\Delta)r=r_{13}+r_{12}.
\end{align}
Finally $ r\in\mathfrak{g}\wedge\mathfrak{g} $. The only thing left to prove is
that $ \text{CYB}(r)=0 $. First remember that the Lie bracket on $ \mathcal{U}(\mathfrak{g}) $ is the
commutator. Then we obtain with Eq. (\ref{rmatrixhelp2}) and the skew-symmetry of $ r $
\begin{align*}
\text{Alt}((r\tensor\mathbb{1})(\Delta\tensor 1)r)&=\text{Alt}(r_{12}(r_{23}+r_{13}))\\
&=\text{Alt}(r_{1}\tensor r_{2}r_{1}'\tensor r_{2}'+r_{1}r_{1}'\tensor r_{2}\tensor r_{2}')\\
&=r_{1}\tensor r_{2}r_{1}'\tensor r_{2}'+r_{1}r_{1}'\tensor r_{2}\tensor r_{2}'
+r_{2}'\tensor r_{1}\tensor r_{2}r_{1}'+r_{2}'\tensor r_{1}r_{1}'\tensor r_{2}\\
&+r_{2}r_{1}'\tensor r_{2}'\tensor r_{1}+r_{2}\tensor r_{2}'\tensor r_{1}r_{1}'\\
&=r_{1}r_{1}'\tensor r_{2}\tensor r_{2}'-r_{1}'r_{1}\tensor r_{2}\tensor r_{2}'
+r_{1}\tensor r_{2}r_{1}'\tensor r_{2}'-r_{1}\tensor r_{1}'r_{2}\tensor r_{2}'\\
&+r_{2}'\tensor r_{1}\tensor r_{2}r_{1}'-r_{2}'\tensor r_{1}\tensor r_{1}'r_{2}\\
&=\left[ r_{1},r_{1}'\right]\tensor r_{2}\tensor r_{2}'+r_{1}\tensor\left[ r_{2},r_{1}'\right]\tensor r_{2}'
+r_{1}\tensor r_{1}'\tensor\left[ r_{2},r_{2}'\right]\\
&=\left[ r_{12},r_{13}\right]+\left[ r_{12},r_{23}\right]+\left[ r_{13},r_{23}\right]\\
&=\text{CYB}(r),
\end{align*}
where we changed to summation several times. Now let us define the $ \mathbb{k} $-linear maps
\begin{align}
\text{Alt}_{1}\colon\mathcal{U}(\mathfrak{g})^{\tensor 3}\rightarrow\mathcal{U}(\mathfrak{g})^{\tensor 3}
\end{align}
and
\begin{align}
\text{Alt}_{2}\colon\mathcal{U}(\mathfrak{g})^{\tensor 3}\rightarrow\mathcal{U}(\mathfrak{g})^{\tensor 3}
\end{align}
on factorizing tensors $ a\tensor b\tensor c $ by $ \text{Alt}_{1}(a\tensor b\tensor c)=c\tensor a\tensor b $
and $ \text{Alt}_{2}=\text{Alt}_{1}^{2} $. Clearly one has $ \text{Alt}_{1}^{3}=1 $.
With this we can calculate
\begin{align*}
((\mathcal{F}-\sigma(\mathcal{F}))\tensor\mathbb{1})&\cdot(\Delta\tensor 1)(\mathcal{F}-\sigma(\mathcal{F}))\\
&=(\mathcal{F}\tensor\mathbb{1})\cdot(\Delta\tensor 1)\mathcal{F}
-(\sigma(\mathcal{F})\tensor\mathbb{1})\cdot(\Delta\tensor 1)\mathcal{F}\\
&-(\mathcal{F}\tensor\mathbb{1})\cdot(\Delta\tensor 1)\sigma(\mathcal{F})
+(\sigma(\mathcal{F})\tensor\mathbb{1})\cdot(\Delta\tensor 1)\sigma(\mathcal{F})\\
&=(\mathcal{F}\tensor\mathbb{1})\cdot(\Delta\tensor 1)\mathcal{F}
-\text{Alt}_{1}((\mathbb{1}\tensor\sigma(\mathcal{F}))\cdot(1\tensor\Delta)\sigma(\mathcal{F}))\\
&-\text{Alt}_{1}((\mathbb{1}\tensor\mathcal{F})\cdot(1\tensor\Delta)\mathcal{F})
+(\sigma(\mathcal{F})\tensor\mathbb{1})\cdot(\Delta\tensor 1)\sigma(\mathcal{F})\\
&=(\mathcal{F}\tensor\mathbb{1})\cdot(\Delta\tensor 1)\mathcal{F}
-\text{Alt}_{1}((\sigma(\mathcal{F})\tensor\mathbb{1})\cdot(\Delta\tensor 1)\sigma(\mathcal{F}))\\
&-\text{Alt}_{1}((\mathcal{F}\tensor\mathbb{1})\cdot(\Delta\tensor 1)\mathcal{F})
+(\sigma(\mathcal{F})\tensor\mathbb{1})\cdot(\Delta\tensor 1)\sigma(\mathcal{F})\\
&=\mathcal{A}-\text{Alt}_{1}(\mathcal{A})+\mathcal{B}-\text{Alt}_{1}(\mathcal{B}),
\end{align*}
where we used the $ 2 $-cocycle condition for $ \mathcal{F} $ and $ \sigma(\mathcal{F}) $ and defined
$ \mathcal{A}=(\mathcal{F}\tensor\mathbb{1})\cdot(\Delta\tensor 1)\mathcal{F} $ and
$ \mathcal{B}=(\sigma(\mathcal{F})\tensor\mathbb{1})\cdot(\Delta\tensor 1)\sigma(\mathcal{F}) $.
In this form the next step is quite easy:
\begin{align*}
\text{Alt}(((\mathcal{F}-\sigma(\mathcal{F}))\tensor\mathbb{1})&\cdot(\Delta\tensor 1)
(\mathcal{F}-\sigma(\mathcal{F})))\\
&=(1+\text{Alt}_{1}+\text{Alt}_{2})
(((\mathcal{F}-\sigma(\mathcal{F}))\tensor\mathbb{1})\cdot
(\Delta\tensor 1)(\mathcal{F}-\sigma(\mathcal{F})))\\
&=(1+\text{Alt}_{1}+\text{Alt}_{2})(\mathcal{A}-\text{Alt}_{1}(\mathcal{A})+\mathcal{B}
-\text{Alt}_{1}(\mathcal{B}))\\
&=\mathcal{A}-\text{Alt}_{1}(\mathcal{A})+\mathcal{B}-\text{Alt}_{1}(\mathcal{B})\\
&+\text{Alt}_{1}(\mathcal{A})-\text{Alt}_{2}(\mathcal{A})+\text{Alt}_{1}(\mathcal{B})-\text{Alt}_{2}(\mathcal{B})\\
&+\text{Alt}_{2}(\mathcal{A})-\mathcal{A}+\text{Alt}_{2}(\mathcal{B})-\mathcal{B}\\
&=0.
\end{align*}
This equation has to be satisfied in every order of $ \hbar $. In particular, for order $ \hbar^{2} $ we get
\begin{align}
0=\text{Alt}((r\tensor\mathbb{1})(\Delta\tensor 1)r)=\text{CYB}(r)
\end{align}
by the above calculation. This concludes the proof.
\end{proof}

\begin{corollary}\label{CorStarRmatrix}
Let $ \star $ be a twist star product on a Poisson manifold $ (M,\pi) $. Then there is a $ r $-matrix on a
Lie algebra such that the formal power series of the corresponding universal enveloping algebra structure
$ \Cinfty(M)[[\hbar]] $ as a left module algebra.
\end{corollary}

We connected the existence of a twist star product to solutions of the classical Yang-Baxter equation in the
Lie algebra that bases the left module algebra structure of $ \Cinfty(M)[[\hbar]] $.

%% file: Chapter4.tex
\chapter{$ r $-Matrices that lead to Homogeneous Structures}\label{rmatrixhomogeneous}

In this chapter we prove one of the main results of the thesis. Here we connect all
previous chapters and justify their study. In a nutshell we show the following: if the Poisson bivector of a
symplectic manifold is the image of a $ r $-matrix under a Lie algebra action, the manifold has to be a homogeneous
space. It is necessary to demand that the involved Lie algebra action integrates to a Lie group action. The
manifolds we later consider are compact thus the integration condition is fulfilled naturally.
Then an important step is to argue that this Lie group action is already locally transitive in this situation. Here
the symplectic character of the manifold is essential. Remark that we develop parallel to the proceeding analogous
results for the Etingof-Schiffmann subgroup
of the $ r $-matrix. It is surprising as well as of enormous use that one can restrict the Lie
algebra to the Lie subalgebra in which the $ r $-matrix is non-degenerate and still gets the desired result.
Thinking of Section \ref{SecEulerPos} and Section \ref{sectiononish} there are many
candidates for obstructions now: the higher pretzel surfaces are not homogeneous spaces and the $ 2 $-sphere can only be
structured as a homogeneous space in very few ways. The question is whether all the requested assumptions are
given for these examples. And moreover: why should it be desirable to find obstructions to Poisson bivectors that
are images of $ r $-matrices through Lie algebra homomorphisms? While Section \ref{SecSpecialBiVector} is fully
dedicated to the proof
of the essential theorem, Section \ref{SecTwistedStar} answers the latest question. The main motivation to consider such
Poisson bivectors is the following: they occur naturally in the presence of twist star products. We have
already shown in Proposition \ref{Proprmatrixfromtwist} that every twist on the formal power series of a universal
enveloping algebra induces a $ r $-matrix on the corresponding Lie algebra. If the twist star
product deforms the symplectic structure of the manifold, the corresponding Poisson bivector is
of the requested form. Thus we do not only give obstructions to special Poisson bivectors but to twist star
products. Then this chapter ends with the statement that any connected compact symplectic manifold endowed with
a twist star product can be structured as a homogeneous space. Following the way of the Etingof-Schiffmann subgroup
we even know that there is a Lie group that acts transitively on the symplectic manifold such that
the $ r $-matrix is non-degenerate in the corresponding Lie algebra.

\section{Symplectic Manifolds with a special Poisson Bivector}\label{SecSpecialBiVector}

Suppose that $ (M,\omega) $ is a symplectic manifold and assume that there is a $ r $-matrix
\begin{align}
r=\frac{1}{2}\sum\limits_{i,j=1}^{n}r^{ij}e_{i}\wedge e_{j}\in\mathfrak{g}\wedge\mathfrak{g}
\end{align}
on a Lie algebra $ \mathfrak{g} $ with basis $ \left\lbrace e_{1},\ldots,e_{n}\right\rbrace $ and a Lie algebra action
\begin{align}
\phi\colon\mathfrak{g}\rightarrow\Gamma^{\infty}(TM)
\end{align}
of $ \mathfrak{g} $ on $ M $ such that for any $ p\in M $ one has
\begin{align}\label{eq18t}
\pi_{p}=\frac{1}{2}\sum\limits_{i,j=1}^{n}r^{ij}\phi(e_{i})_{p}\wedge\phi(e_{j})_{p}.
\end{align}
If $ \phi $ integrates to a Lie group action $ \Phi $, this group action is locally transitive. More precisely,

\begin{lemma}\label{hilfslemma1}
Let $ (M,\omega) $ be a symplectic manifold with corresponding Poisson bivector $ \pi $ defined in Eq.
(\ref{SymplPoisBiVect}). Assume there is a $ r $-matrix $ r\in\mathfrak{g}\wedge\mathfrak{g} $ and a Lie algebra action
$ \phi\colon\mathfrak{g}\rightarrow\Gamma^{\infty}(TM) $ of $ \mathfrak{g} $ on $ M $ which
integrates to a Lie group action $ \Phi\colon G\times M\rightarrow M $ of $ G $ on $ M $ such
that Eq. (\ref{eq18t}) holds. Then $ \Phi $ is locally transitive.
Furthermore, the restriction $ \Phi|_{H_{r}} $ of $ \Phi $ to $ H_{r}\times M $, where $ H_{r} $ denotes the
Etingof-Schiffmann subgroup corresponding to $ r $, is locally transitive.
\end{lemma}

\begin{proof}
Fix $ p\in M $ and choose an arbitrary $ v_{p}\in T_{p}M $. Since $ \tilde{\pi}\colon T_{p}^{*}M\rightarrow T_{p}M $
defined in Eq. (\ref{eq20t}) is surjective there is a $ \alpha_{p}
\in T_{p}^{*}M $ such that
\begin{align}
v_{p}=-\pi_{p}(\cdot,\alpha_{p})=\sum\limits_{i,j=1}^{n}r^{ij}\alpha_{p}(\phi(e_{i})_{p})\phi(e_{j})_{p}.
\end{align}
But this means that
\begin{align}
v_{p}\in\text{span}_{\mathbb{k}}\left\lbrace \sum\limits_{j=1}^{n}r^{1j}\phi(e_{j})_{p},\ldots,
\sum\limits_{j=1}^{n}r^{nj}\phi(e_{j})_{p}\right\rbrace\subseteq\text{span}_{\mathbb{k}}\left\lbrace
\phi(e_{1})_{p},\ldots,\phi(e_{n})_{p}\right\rbrace.
\end{align}
Since $ v_{p}\in T_{p}M $ was arbitrary the map
\begin{align*}
\phi|_{p}\colon\mathfrak{g}\ni\xi\mapsto\phi(\xi)_{p}\in T_{p}M
\end{align*}
has to be surjective. For the last statement consider the
Etingof-Schiffmann subalgebra $ \mathfrak{h}_{r} $ corresponding to $ r $. Like in the proof of
Proposition \ref{PropEtingofSub} we can choose a basis $ \left\lbrace e_{1},\ldots,e_{k}\right\rbrace $ of
$ \mathfrak{h}_{r}\subseteq\mathfrak{g} $ and complete this to a basis
$ \left\lbrace e_{1},\ldots,e_{n}\right\rbrace $ of $ \mathfrak{g} $, such that
\begin{align}
r=\frac{1}{2}\sum\limits_{i,j=1}^{n}r^{ij}e_{i}\wedge e_{j}=\frac{1}{2}\sum\limits_{i,j=1}^{k}r^{ij}e_{i}\wedge e_{j}.
\end{align}
But this implies
\begin{align}
\pi_{p}=\frac{1}{2}\sum\limits_{i,j=1}^{n}r^{ij}\phi(e_{i})_{p}\wedge\phi(e_{j})_{p}
=\frac{1}{2}\sum\limits_{i,j=1}^{k}r^{ij}\phi(e_{i})_{p}\wedge\phi(e_{j})_{p},
\end{align}
i.e. the restriction $ \phi_{\mathfrak{h}_{r}}|_{p} $ of $ \phi|_{p} $ to $ \mathfrak{h}_{r} $ still satisfies the
transitivity condition on $ T_{p}M $ for every $ p\in M $. Since $ \phi $ was integrable this is also true for the
restriction to a Lie subgroup. This gives $ \Phi|_{H_{r}}\colon H_{r}\times M\rightarrow M $, the restriction of
$ \Phi $ to $ H_{r}\times M $, which
is locally transitive and we can restrict ourselves without loss of generality to the Etingof-Schiffmann subalgebra in
which $ r $ is non-degenerate.
\end{proof}

Our aim is to show that the setting of the last lemma already implies that $ M $ is a homogeneous space. For this
we need another

\begin{lemma}\label{hilfslemma2}
Let $ \Phi\colon G\times M\rightarrow M $ be a locally transitive Lie group action on a connected manifold $ M $.
Then there is only one orbit of $ \Phi $ and it coincides with $ M $.
\end{lemma}

\begin{proof}
Let $ \phi $ be the Lie algebra action corresponding to the locally transitive Lie group action $ \Phi $.
Fix $ p\in M $ and denote by $ \mathfrak{g}_{p} $ the Lie algebra corresponding to $ G_{p} $.
Remark that this is possible since $ G_{p}\subseteq G $ is a Lie group according to Theorem \ref{thmsubmfd} ii.).
The map $ \beta_{p}\colon G/G_{p}\ni g\cdot G_{p}\mapsto\Phi(g,p)\in M $ is a diffeomorphism
if $ \Phi $ is transitive according to Theorem \ref{homspacetransact}. Our aim is to show that $ \beta_{p} $ is also
a diffeomorphism if $ \Phi $ is just locally transitive. We first prove that $ \beta_{p} $ is a local diffeomorphism.
We have already shown in Theorem \ref{homspacetransact} that $ \beta_{p} $ is well-defined and smooth.
Thus we can differentiate $ \beta_{p} $ at $ e\cdot G_{p}\in G/G_{p} $ and obtain
\begin{align}
T_{e\cdot G_{p}}\beta_{p}\colon\mathfrak{g}/\mathfrak{g}_{p}\rightarrow T_{\beta_{p}(e\cdot G_{p})}M=T_{p}M.
\end{align}
Then the map
\begin{align}\label{eq23t}
\tilde{\phi}\colon\mathfrak{g}/\mathfrak{g}_{p}\ni\left[\xi\right]\mapsto\tilde{\phi}(\left[\xi\right])=
T_{e\cdot G_{p}}\beta_{p}\left[\xi\right]=\phi(\xi)_{p}\in T_{p}M
\end{align}
is well-defined. Let $ v_{p}\in T_{p}M $ be arbitrary. Since $ \phi|_{p}\colon\mathfrak{g}\rightarrow T_{p}M $
is surjective there is a $ \xi\in
\mathfrak{g} $ such that $ \phi(\xi)_{p}=v_{p} $. Then $ \tilde{\phi}(\left[\xi\right])=\phi(\xi)_{p}=v_{p} $ and
$ \tilde{\phi} $ is surjective. To show that $ \tilde{\phi} $ is also injective take
$ \left[\xi\right]\in\ker\tilde{\phi} $. For all $ t\in\mathbb{R} $ one gets
\begin{align*}
\frac{\mathrm{d}}{\mathrm{d}t}\beta_{p}(\exp(t\xi)\cdot G_{p})&=\frac{\mathrm{d}}{\mathrm{d}t}\Phi_{\exp(t\xi)}(p)\\
&=\left.\frac{\mathrm{d}}{\mathrm{d}s}\right|_{s=0}\Phi_{\exp(t\xi)}(\Phi_{\exp(s\xi)}(p))\\
&=T_{p}\Phi_{\exp(t\xi)}\left.\frac{\mathrm{d}}{\mathrm{d}s}\right|_{s=0}\Phi_{\exp(s\xi)}(p)\\
&=T_{p}\Phi_{\exp(t\xi)}\left.\frac{\mathrm{d}}{\mathrm{d}s}\right|_{s=0}\beta_{p}(\exp(s\xi)\cdot G_{p})\\
&=T_{p}\Phi_{\exp(t\xi)}\tilde{\phi}(\left[\xi\right])\\
&=0.
\end{align*}
Since $ \beta_{p}(\exp(0\cdot\xi)\cdot G_{p})=\Phi(e,p)=p $ the last calculation implies that $ \beta_{p}(\exp(t\xi)
\cdot G_{p})=p $ for all $ t\in\mathbb{R} $. This shows $ \exp(t\xi)\in G_{p} $ for all $ t\in\mathbb{R} $ and for
this $ \xi\in\mathfrak{g}_{p} $, thus $ \left[\xi\right]=0 $ and $ \tilde{\phi} $ is injective.
Then for any $ g\in G $
\begin{align}
\beta_{x}(g\cdot G_{p})=\Phi(g,p)=\Phi(g,\Phi(e,p))=\Phi_{g}\circ\beta_{p}(e\cdot G_{p})
=\Phi_{g}\circ\beta_{p}\circ\ell_{g^{-1}}(g\cdot G_{p})
\end{align}
implies
\begin{align}
T_{g\cdot G_{p}}\beta_{p}=T_{p}\Phi_{g}\circ T_{e\cdot G_{p}}\beta_{p}\circ T_{g}\ell_{g^{-1}},
\end{align}
where $ \ell_{g^{-1}}\colon G/G_{p}\ni g'\cdot G_{p}\mapsto(g^{-1}g')\cdot G_{p}\in G/G_{p} $. Since $ \Phi_{g} $
and $ \ell_{g}^{-1} $ are diffeomorphisms also $ T_{g\cdot G_{p}}\beta_{p} $ is bijective.
Thus $ \beta_{p} $ is a local
diffeomorphism for all $ p\in M $, i.e. for any $ g\cdot G_{p}\in G/G_{p} $ there is an open neighbourhood $ U
\subseteq G/G_{p} $ of $ g\cdot G_{p} $ such that $ \beta_{p}(U)\subseteq M $ is open. Choose any point
$ q\in G\cdot p $ of the orbit of $ p $ under $ \Phi $. Then for every $ g\in\Phi_{p}^{-1}(\left\lbrace q
\right\rbrace)\subseteq G $ there is an open neighbourhood $ U\subseteq G/G_{p} $ such that $ \beta_{p}(U)\subseteq
G\cdot p\subseteq M $ is open. Since $ q\in\beta_{p}(U) $ we have found an open neighbourhood of $ q $ that
is contained in $ G\cdot p $ and the orbit $ G\cdot p $ is open. Let again $ p\in M $ be arbitrary and assume that there
is a point $ q\in M\setminus
\left\lbrace G\cdot p\right\rbrace $. Since $ M $ is always partitioned by the orbits of any group action one
has $ G\cdot p\cap G\cdot q=\emptyset $, since otherwise $ G\cdot p=G\cdot q $ what is not possible because
$ q\notin G\cdot p $. By this we get a partition of $ M $ by open sets what gives a contradiction if $ M $ is
connected. This concludes the proof.
\end{proof}

Now we just collect our results and conclude one of the main theorems of this thesis.

\begin{theorem}\textup{\textbf{\cite{NoTwistPaper}.}}\label{mastertheorem}
Let $ (M,\omega) $ be a connected symplectic manifold with corresponding Poisson bivector $ \pi $,
$ r\in\mathfrak{g}\wedge\mathfrak{g} $ a $ r $-matrix and $
\phi\colon\mathfrak{g}\rightarrow\Gamma^{\infty}(TM) $ a Lie algebra action of $ \mathfrak{g} $ on $ M $ that
integrates to a Lie group action $ \Phi\colon G\times M\rightarrow M $ of $ G $ on $ M $ such that for any $ p\in M $
one has
\begin{align}\label{PoissonBiVectorRMatrix}
\pi_{p}=\frac{1}{2}\sum\limits_{i,j=1}^{n}r^{ij}\phi(e_{i})_{p}\wedge\phi(e_{j})_{p}.
\end{align}
Then $ M $ can be structured as a homogeneous space via $ \Phi $.
Moreover, the Etingof-Schiffmann subgroup $ H_{r} $ acts transitively on $ M $.
\end{theorem}

\begin{proof}
This is a direct consequence of Lemma \ref{hilfslemma1} and Lemma \ref{hilfslemma2}.
\end{proof}

For this, a geometric information provides topological impact on the manifold: if the Poisson bivector is the
image of a $ r $-matrix under an integrable Lie algebra action, the manifold has to be a homogeneous space. For
instance, the Euler characteristic of such a compact space has to be non-negative according to Theorem
\ref{EulerCompact}. We make use of this in the next chapter to get obstructions on the higher pretzel surfaces.
The great benefit of the second statement of Theorem \ref{mastertheorem} is that the $ r $-matrix is non-degenerate
in the Lie algebra of a Lie
group that acts transitively on $ M $. This means we have more information on an infinitesimal level. Later on this
helps us to find obstructions in the case of $ M=\mathbb{S}^{2} $.

It is clear that one can apply Theorem \ref{mastertheorem} to any leaf of a
symplectic foliation of a Poisson manifold. Remember that for any Poisson manifold $ (M,\pi) $ the
image of the sharp map
\begin{align}
\sharp\colon T^{*}M\rightarrow TM
\end{align}
is a smooth distribution. By this we mean that $ \text{im}(\sharp) $ is a set of linear subspaces
$ \text{im}(\sharp_{x}) $ of $ T_{x}M $, where $ x\in M $, such that each is spanned by a finite number of smooth
vector fields $ X_{1},...,X_{k}\in\text{im}(\sharp) $, i.e. $ \text{im}(\sharp_{x})=\text{span}\left\lbrace
X_{1}(x),...,X_{k}(x)\right\rbrace $. One can also imagine $ \text{im}(\sharp_{x}) $ as the set
\begin{align}
\left\lbrace v\in T_{x}M~|~\exists f\in\Cinfty(M)\text{ such that }X_{f}(x)=v\right\rbrace,
\end{align}
i.e. the set of tangent vectors which are images of Hamiltonian vector fields at $ x $. Now $ \text{im}(\sharp) $
integrates to a foliation of $ M $, that is a cover of $ M $ by connected subsets which are said to be the leaves of
this foliation. Moreover, the Poisson bivector $ \pi $ induces symplectic Poisson bivectors on these leaves and for this
we obtain a foliation of $ M $ called a \textit{symplectic foliation}. The leaves are maximal symplectic submanifolds of
$ M $ which are said to be the \textit{symplectic leaves} of $ M $. All these facts can be found in
\cite[Chapter~2]{vaisman1994lectures}.
Thus we can modify our last result in the following way:

\begin{corollary}\label{CorSymplLeaves}
Let $ (M,\pi) $ be a Poisson manifold. Any symplectic leaf $ (S,\pi^{S}) $ of $ M $ whose Poisson bivector
is the image of a $ r $-matrix under an integrable Lie algebra action
$ \phi\colon\mathfrak{g}\rightarrow\Gamma^{\infty}(TS) $ is a homogeneous space.
Moreover, the Etingof-Schiffmann subgroup of this $ r $-matrix acts transitively on $ S $.
\end{corollary}

\section{Structures induced by Twist Star Products}\label{SecTwistedStar}

As an a posteriori motivation to consider Poisson bivectors of the form (\ref{eq18t}) we examine star products that are
induced by a twist. In Corollary \ref{CorStarRmatrix} we have seen the connection between twist star products and
$ r $-matrices. It is also interesting to consider the
Poisson bracket and the Poisson bivector which are deformed by a twist star product.

\begin{proposition}\label{PropStructuresForTwStarProd}
Let $ (M,\left\lbrace\cdot,\cdot\right\rbrace) $ be a Poisson manifold and $ \star $ a star product on $ M $ deforming
the Poisson bracket $ \left\lbrace\cdot,\cdot\right\rbrace $. Assume that $ \star $ is induced by a twist
$ \mathcal{F}
=\sum\limits_{i=0}^{\infty}\hbar^{i}F_{i}\in(\mathcal{U}(\mathfrak{g})\tensor\mathcal{U}(\mathfrak{g}))[[\hbar]] $
and denote the left module action of $ \mathcal{U}(\mathfrak{g})[[\hbar]] $ on $ \Cinfty(M)[[\hbar]] $ by $ \rhd $. Then
\begin{compactenum}
\item there is a $ r $-matrix
\begin{align}\label{RMATRIX}
r=\sigma(F_{1})-F_{1}
\end{align}
on $ \mathfrak{g} $,
\item the Poisson bracket reads
\begin{align}
\left\lbrace f,g\right\rbrace=m(r\rhd(f\tensor g))
\end{align}
for all $ f,g\in\Cinfty(M)[[\hbar]] $, where
$ m\colon\Cinfty(M)[[\hbar]]\tensor\Cinfty(M)[[\hbar]]\ni(f\tensor g)\mapsto fg\in\Cinfty(M)[[\hbar]] $,
\item there is a Lie algebra action $ \phi $ of $ \mathfrak{g} $ on $ M $ defined for any $ \xi\in\mathfrak{g} $ by
\begin{align}\label{LieDerivative}
\mathcal{L}_{-\phi(\xi)}=\xi\rhd\colon\Cinfty(M)[[\hbar]]\rightarrow\Cinfty(M)[[\hbar]],
\end{align}
where $ \mathcal{L} $ denotes the Lie derivative,
\item the Poisson bivector $ \pi $ corresponding to $ \left\lbrace\cdot,\cdot\right\rbrace $ reads
\begin{align}
\pi_{p}=\frac{1}{2}\sum\limits_{i,j=1}^{n}r^{ij}\phi(e_{i})_{p}\wedge\phi(e_{j})_{p}
\end{align}
for any $ p\in M $, where $ e_{1},\ldots,e_{n} $ is a basis of $ \mathfrak{g} $ and
$ r=\frac{1}{2}\sum\limits_{i,j=1}^{n}r^{ij}e_{i}\wedge e_{j} $.
\end{compactenum}
\end{proposition}

\begin{proof}
For any $ f,g\in\Cinfty(M)[[\hbar]] $ one has
\begin{align*}
\sum\limits_{i=0}^{\infty}\hbar^{i}C_{i}(f,g)&=f\star g\\
&=\mathrm{m}(\mathcal{F}^{-1}\rhd(f\tensor g))\\
&=\mathrm{m}(\sum\limits_{i=0}^{\infty}\hbar^{i}F_{i}^{-1}\rhd(f\tensor g))\\
&=\sum\limits_{i=0}^{\infty}\hbar^{i}\mathrm{m}(((F_{i}^{-1})_{1}\rhd f)\tensor((F_{i}^{-1})_{2}\rhd g))\\
&=\sum\limits_{i=0}^{\infty}\hbar^{i}((F_{i}^{-1})_{1}\rhd f)\cdot((F_{i}^{-1})_{2}\rhd g).
\end{align*}
In particular, $ C_{1}(f,g)=((F_{1}^{-1})_{1}\rhd f)\cdot((F_{1}^{-1})_{2}\rhd g)
=\mathrm{m}(F_{1}^{-1}\rhd(f\tensor g)) $.
Then the
corresponding Poisson bracket satisfies for all $ f,g\in\Cinfty(M)[[\hbar]] $
\begin{align*}
\left\lbrace f,g\right\rbrace&=C_{1}(f,g)-C_{1}(g,f)\\
&=\mathrm{m}(F_{1}^{-1}\rhd(f\tensor g))-m(F_{1}^{-1}\rhd(g\tensor f))\\
&=\mathrm{m}((F_{1}^{-1}-\sigma(F_{1}^{-1}))\rhd(f\tensor g))\\
&=\mathrm{m}(r\rhd(f\tensor g)),
\end{align*}
since pointwise multiplication is
commutative, where $ r $ is the $ r $-matrix on $ \mathfrak{g} $ defined in Eq. (\ref{rmatrixfromtwist}).
Because $ F_{1}^{-1}=-F_{1} $ we proved the first two statements. Now define a map
\begin{align}\label{LieAlgAction}
\phi\colon\mathfrak{g}\rightarrow\Gamma^{\infty}(TM)
\end{align}
via Eq. (\ref{LieDerivative}).
Then $ \phi $ is a Lie algebra action of $ \mathfrak{g} $ on $ M $. To see this, choose $ \xi,\zeta\in\mathfrak{g} $,
$ f\in\Cinfty(M)[[\hbar]] $ and calculate
\begin{align*}
\mathcal{L}_{\left[\phi(\xi),\phi(\zeta)\right])}f=\left[\mathcal{L}_{\phi(\xi)},\mathcal{L}_{\phi(\zeta)}\right](f)
=\left[-\xi\rhd,-\zeta\rhd\right](f)
=\left[\xi,\zeta\right]\rhd f
=\mathcal{L}_{-\phi(\left[\xi,\zeta\right])}f,
\end{align*}
which implies $ \left[\phi(\xi),\phi(\zeta)\right]=-\phi(\left[\xi,\zeta\right]) $, i.e. $ \phi $ is a Lie algebra
anti-homomorphism. This proves the third part.
If we assume
$ \mathfrak{g} $ to be finite-dimensional with basis $ e_{1},\ldots,e_{n}\in\mathfrak{g} $ one has
\begin{align}
r=\frac{1}{2}\sum\limits_{i,j=1}^{n}r^{ij}e_{i}\wedge e_{j}=\sum\limits_{i,j=1}^{n}r^{ij}e_{i}\tensor e_{j},
\end{align}
and for this
\begin{align}\label{PoissonBracketRmatrix}
\left\lbrace f,g\right\rbrace=\frac{1}{2}\sum\limits_{i,j=1}^{n}r^{ij}(e_{i}\rhd f)(e_{j}\rhd g),
\end{align}
for all $ f,g\in\Cinfty(M)[[\hbar]] $.
As a consequence the Poisson bivector reads
\begin{align}
\pi_{p}=\frac{1}{2}\sum\limits_{i,j=1}^{n}r^{ij}\phi(e_{i})_{p}\wedge\phi(e_{j})_{p},
\end{align}
for any $ p\in M $. This concludes the proof.
\end{proof}

We can use Theorem \ref{mastertheorem} to state a necessary condition under which a star
product $ \star $ on a symplectic manifold $ (M,\omega) $ can be induced by a twist: if $ \star $ deforms the
Poisson bivector $ \pi $ corresponding to $ \omega $, it can only be induced by a twist $ \mathcal{F} $ if the
diagram
\begin{equation}
\begin{tikzpicture}
  \matrix (m) [matrix of math nodes,row sep=3em,column sep=3em,minimum width=2em]
  {
      \mathcal{F}^{-1} & \star \\
      r & \pi \\};
  \path[-stealth]
    (m-1-1) edge node [above] {$\rhd$} (m-1-2)
    (m-1-1) edge node [left] {$ 1 $-st} (m-2-1)
    (m-1-2) edge node [right] {$ 1 $-st} (m-2-2)
    (m-2-1) edge node [below] {$\rhd$} (m-2-2);
\end{tikzpicture},
\end{equation}
commutes. This results in the following

\begin{theorem}\label{ThmIntegrabelStar}
Let $ (M,\omega) $ be a connected symplectic manifold endowed with a twist star product such that the Lie algebra action
defined in Eq. (\ref{LieDerivative}) integrates to a Lie group action. Then $ M $ can be structured as a
homogeneous space and the Etingof-Schiffmann subgroup corresponding to the $ r $-matrix defined in Eq.
(\ref{RMATRIX}) acts transitively on $ M $.
\end{theorem}

\begin{proof}
This follows directly from Proposition \ref{PropStructuresForTwStarProd} and Theorem \ref{mastertheorem}.
\end{proof}
In particular, the requirements of Theorem \ref{ThmIntegrabelStar} are fulfilled if the symplectic manifold is compact.

\begin{corollary}\label{CompactTwistHomSp}
Let $ (M,\omega) $ be a connected compact symplectic manifold endowed with a twist star product. Then $ M $ can be
structured as a homogeneous space. Moreover, the Etingof-Schiffmann subgroup corresponding to the $ r $-matrix
defined in Eq. (\ref{RMATRIX}) acts transitively on $ M $.
\end{corollary}

\begin{proof}
Since any flow of a vector field on a compact manifold is complete, a theorem of Palais (c.f.
\cite[Theorem~6.5]{michortopics})
implies that any Lie algebra action on a connected compact manifold integrates to a Lie group action.
\end{proof}

%% file: Chapter5.tex
\chapter{Obstructions and Examples of Twist Star Products}\label{ChapObstruction}

In this chapter we gather all previous notions and results to prove or disprove the existence of twist star
products deforming the symplectic structure on connected compact manifolds of dimension $ 2 $. One can
generalize these results easily to arbitrary Poisson manifolds that inherit symplectic leaves of these types.
It turns out that the examples that allow existence of twist star products deforming the symplectic structure
are quite rare. In fact there is only one example of such a twist star product on the $ 2 $-torus and obstructions
in all other cases.

Consider a symplectic manifold $ (M,\omega) $ of dimension $ \dim M=2n\in\mathbb{N} $. Since
$ \omega\in\Gamma^{\infty}(\Anti^{2}T^{*}M) $ is a non-degenerate $ 2 $-form on $ M $ the corresponding
\textit{Liouville form}
\begin{align}
\Omega=\underbrace{\omega\wedge...\wedge\omega}_{n-\text{times}}\in\Gamma^{\infty}(\Anti^{2n}T^{*}M)
\end{align}
is a non-degenerate $ 2n $-form on $ M $, i.e. $ M $ is orientable. For $ n=1 $ the converse is also true:
if $ M $ is a orientable $ 2 $-dimensional manifold, there is a non-degenerate $ 2 $-form
$ \Omega=\omega\in\Gamma^{\infty}(\Anti^{2}T^{*}M) $ on $ M $. But since there are no $ 3 $-forms on a
$ 2 $-dimensional manifold $ \mathrm{d}\omega=0 $, i.e. $ (M,\omega) $ is a symplectic manifold. Thus a
$ 2 $-dimensional manifold is symplectic if and only if it is orientable. There is a well-known classification of
connected
orientable compact $ 2 $-dimensional manifold: two such manifolds are diffeomorphic if and only if they have the same
genus $ g\in\mathbb{N}_{0} $ (c.f. \cite[Section~12.1.5]{lamotke2006riemannsche}). We denote by $ \mathrm{T}(g) $
an arbitrary representative of a connected orientable compact $ 2 $-dimensional manifold of genus $ g $
and call it the \textit{pretzel surface} of genus $ g $. The topological invariant
$ g $ encodes the number of holes. Clearly $ \mathrm{T}(0)\cong\mathbb{S}^{2} $ is
diffeomorphic to the $ 2 $-sphere and $ \mathrm{T}(1)\cong\mathbb{T}^{2} $ is diffeomorphic to the $ 2 $-torus.
According to our discussion, the manifolds $ \mathrm{T}(g) $ are symplectic and on the other hand any $ 2 $-dimensional
connected compact symplectic manifold is diffeomorphic to $ \mathrm{T}(g) $ for a unique $ g\in\mathbb{N}_{0} $. We are
also able to construct a generic symplectic form on $ \mathrm{T}(g) $. Since $ \mathrm{T}(g) $ is a manifold it is
Hausdorff and second countable. Together with the property of being connected this implies the existence of a
\textit{Riemannian metric} $ m\in\Gamma^{\infty}(S^{2}T^{*}\mathrm{T}(g)) $ on $ \mathrm{T}(g) $, i.e. for any
$ p\in\mathrm{T}(g) $
\begin{align}
m_{p}\colon T_{p}\mathrm{T}(g)\times T_{p}\mathrm{T}(g)\rightarrow\mathbb{R}
\end{align}
is a positive definite symmetric bilinear form on $ T_{p}\mathrm{T}(g) $
(c.f. \cite[Proposition~5.5.11]{abraham1993manifolds}).
In a positively oriented local chart $ (U,x=(x_{1},x_{2})) $ of $ \mathrm{T}(g) $ one has
\begin{align}
m|_{U}=\sum\limits_{i,j=1}^{2}\frac{1}{2}m_{ij}dx^{i}\vee dx^{j}.
\end{align}
Then define
\begin{align}\label{SymplecticFromLocal}
\omega|_{U}=\sum\limits_{i,j=1}^{2}\frac{1}{2}\sqrt{\det(m_{k\ell})_{k,\ell=1,2}}dx^{i}\wedge dx^{j}
=\sqrt{m_{11}m_{22}-m_{12}m_{21}}dx^{1}\wedge dx^{2}
\end{align}
which gives a $ 2 $-form $ \omega\in\Gamma^{\infty}(\Anti^{2}T^{*}\mathrm{T}(g)) $ on $ \mathrm{T}(g) $
by gluing together (\ref{SymplecticFromLocal}) on an atlas of $ \mathrm{T}(g) $.

\section{The higher Pretzel Surfaces}

According to \cite[Section~12.1.4]{lamotke2006riemannsche} the Euler characteristic of $ \mathrm{T}(g) $ is
\begin{align}
\chi(\mathrm{T}(g))=2-2g
\end{align}
for any $ g\in\mathbb{N}_{0} $. In particular, one has $ \chi(\mathrm{T}(g))<0 $ for $ g>1 $. Then by Theorem
\ref{EulerCompact} one can not structure $ \mathrm{T}(g) $ as a homogeneous space if $ g>1 $.
Corollary \ref{CompactTwistHomSp} immediately implies the following

\begin{theorem}\label{ThmPretzel}
Let $ g>1 $. There is no twist star product on $ \mathrm{T}(g) $ deforming a symplectic structure of
$ \mathrm{T}(g) $.
\end{theorem}

On the other hand, there is a star product on any $ \mathrm{T}(g) $ obtained by the Fedosov construction since
$ \mathrm{T}(g) $ is symplectic (c.f. \cite{fedosov1994}). Thus there are star products deforming the symplectic
structure on $ \mathrm{T}(g) $ which can not be induced by a twist and we have infinitely many examples.
This gives a partial proof of the main theorem stated in the introduction.

\section{The $ 2 $-Torus}

Surprisingly, there is a twist star product on the $ 2 $-torus $ \mathbb{T}^{2} $. Consider $ \mathbb{R}^{2} $ with
coordinates $ (x_{1},x_{2}) $. If we choose the zero Lie algebra $ \mathfrak{g} $ over $ \mathbb{R}^{2} $ then
$ \mathcal{U}(\mathfrak{g})=\Sym^{\bullet}\mathbb{R}^{2} $ can be endowed with a Hopf algebra structure
$ \Delta(x)=x\tensor 1+1\tensor x $, $ \epsilon(x)=0 $ and $ S(x)=-x $ for all $ x\in\mathbb{R}^{2} $. As described
in Section \ref{secStarProdTwist} this can be extended to a Hopf algebra $ \Sym^{\bullet}\mathbb{R}^{2}[[\hbar]] $. Now define
\begin{align}
\mathcal{F}=\exp(-\mathrm{i}\hbar\frac{\partial}{\partial x_{1}}\wedge\frac{\partial}{\partial x_{2}})\in
(\Sym^{\bullet}\mathbb{R}^{2}\tensor\Sym^{\bullet}\mathbb{R}^{2})[[\hbar]].
\end{align}
This satisfies the twist conditions of Definition \ref{DefTwistUEA}. Denote by $ \rhd $ the standard action of
$ \mathbb{R}^{2} $ on $ \mathbb{T}^{2} $. Then
\begin{align}
f\star g=\mathrm{m}\circ\mathcal{F}^{-1}\rhd(f\tensor g)=\mathrm{m}
(\exp(\mathrm{i}\hbar\frac{\partial}{\partial x_{1}}\wedge\frac{\partial}{\partial x_{2}})\rhd(f\tensor g))
\end{align}
is the well-known Weyl-Moyal star product on $ \mathbb{T}^{2} $, where $ f,g\in\Cinfty(\mathbb{T}^{2})[[\hbar]] $.
It deforms the symplectic structure on $ \mathbb{T}^{2} $.
More information can be found in \cite{Moskaliuk:2012zz} and \cite{FluxSzabo}.

\section{The $ 2 $-Sphere}

Finally, we consider $ \mathrm{T}(0) $, i.e. the sphere $ \mathbb{S}^{2} $. It is also a connected compact
symplectic
manifold: take for example the symplectic form $ \omega\in\Gamma^{\infty}(\Anti^{2}T^{*}\mathbb{S}^{2}) $ defined
for any $ x\in\mathbb{S}^{2} $ and $ v,w\in T_{x}\mathbb{S}^{2} $ by
\begin{align}
\omega_{x}(v,w)=x\cdot(v\times w),
\end{align}
where $ \cdot $ denotes the standard scalar product on $ \mathbb{R}^{3} $ and $ \times $ the vector product. While
the closedness of $ \omega $ is trivially fulfilled, also the non-degeneracy can be concluded quite easily since
$ T_{x}\mathbb{S}^{2} $ is the plain of vectors that are perpendicular to $ x $ in $ \mathbb{R}^{3} $.
Unfortunately, the trick used in Theorem \ref{ThmPretzel} does not apply here since $ \chi(\mathbb{S}^{2})=2>0 $. But
we have seen in Section \ref{sectiononish} that there are only very few connected (non-equivalent) Lie groups that act
transitively
and effectively on $ \mathbb{S}^{2} $. The question is now the following: are the additional assumptions of being
connected and effectiveness to specific to produce general obstructions? Anyway, the connected Lie groups that
act transitively and effectively give only a few possibilities to structure $ \mathbb{S}^{2} $ as a homogeneous space.
But luckily for any transitive action on $ \mathbb{S}^{2} $ one can construct a connected
one that acts effective in addition, such that the $ r $-matrix survives the transformation. As an immediate consequence
of Corollary \ref{CompactTwistHomSp} we can assume that the Lie group that acts transitively is connected: if
there is a twist star product on a compact symplectic manifold the Etingof-Schiffmann subgroup (a connected Lie group)
corresponding to the $ r $-matrix induced by the twist acts transitively on the manifold and the $ r $-matrix
is non-degenerate in the corresponding Lie algebra. We are interested in obstructions, thus this is exactly the
implication we want since the existence of this connected Lie group acting transitively on $ \mathbb{S}^{2} $
produces a contradiction in the end.

We now prove that $ r $-matrices fit well in our categorial frame.


\begin{lemma}\label{LemRMatrixHom}
Let
$ \phi\colon(\mathfrak{g},\left[\cdot,\cdot\right]_{\mathfrak{g}})\rightarrow
(\mathfrak{h},\left[\cdot,\cdot\right]_{\mathfrak{h}}) $ be a Lie algebra homomorphism. If
$ r\in\mathfrak{g}\tensor\mathfrak{g} $ is a $ r $-matrix then $ \tilde{r}=(\phi\tensor\phi)r $ is a $ r $-matrix
on $ \mathfrak{h} $.
\end{lemma}

\begin{proof}
This is a simple calculation:
\begin{align*}
\text{CYB}(\tilde{r})&=\left[\phi(r_{1}),\phi(r_{1}')\right]_{\mathfrak{h}}\tensor\phi(r_{2})\tensor\phi(r_{2}')
+\phi(r_{1})\tensor\left[\phi(r_{2}),\phi(r_{1}')\right]_{\mathfrak{h}}\tensor\phi(r_{2}')\\
&+\phi(r_{1})\tensor\phi(r_{1}')\tensor\left[\phi(r_{2}),\phi(r_{2}')\right]_{\mathfrak{h}}\\
&=\phi(\left[r_{1},r_{1}'\right]_{\mathfrak{g}})\tensor\phi(r_{2})\tensor\phi(r_{2}')
+\phi(r_{1})\tensor\phi(\left[r_{2},r_{1}'\right]_{\mathfrak{g}})\tensor\phi(r_{2}')\\
&+\phi(r_{1})\tensor\phi(r_{1}')\tensor\phi(\left[r_{2},r_{2}'\right]_{\mathfrak{g}})\\
&=(\phi\tensor\phi\tensor\phi)\text{CYB}(r)
\end{align*}
by the Lie algebra homomorphism property of $ \phi $. Then, if $ r $ is a $ r $-matrix one has
$ \text{CYB}(\tilde{r})=(\phi\tensor\phi\tensor\phi)0=0 $. Also the skew-symmetry of $ \tilde{r} $ is clear since
$ \phi $ is linear.
\end{proof}

We pass from transitive to transitive and effective actions. With the help of the last lemma we do not loose
the $ r $-matrix.

\begin{proposition}\label{PropTransToEffect}
Let $ G $ be a Lie group acting on a manifold $ M $ via $ \Phi $ and define
\begin{align}
N=\left\lbrace g\in G~|~\Phi_{g}=\operatorname{id}_{M}\right\rbrace.
\end{align}
If $ r $ is a $ r $-matrix on the corresponding Lie algebra $ \mathfrak{g} $ then there is a $ r $-matrix
on the Lie algebra $ \mathfrak{g}/\mathfrak{n} $ corresponding to $ G/N $.
If the action $ \Phi\colon G\times M\rightarrow M $ is transitive, so is the effective action
\begin{align}\label{QuotientAction}
\Psi\colon G/N\times M\ni(g\cdot N,x)\mapsto\Phi(g,x)\in M.
\end{align}
\end{proposition}

\begin{proof}
We already know that $ N $ is a Lie subgroup of $ G $. To see that it is a normal subgroup choose $ g\in G $ and
$ n\in N $ and calculate
\begin{align*}
\Phi_{gng^{-1}}=\Phi_{g}\circ\Phi_{n}\circ\Phi_{g}^{-1}=\Phi_{g}\circ\Phi_{g}^{-1}=\text{id}_{M},
\end{align*}
which shows that $ gng^{-1}\in N $. Thus $ gNg^{-1}\subseteq N $ for all $ g\in G $ and $ N $ is normal.
It is well-known (c.f. \cite[Proposition~5.7.8]{rudolph2012differential}) that the quotient $ G/H $ of a Lie group
$ G $ and a normal subgroup $ H\subseteq G $ is again a Lie group and there is a Lie group homomorphism
$ \pi\colon G\rightarrow G/H $. Consequently, we find a Lie group homomorphism $ \pi\colon G\rightarrow G/N $.
Then its tangent map $ \phi\colon\mathfrak{g}\rightarrow\mathfrak{g}/\mathfrak{n} $ at $ e\in G $ is a Lie
algebra homomorphism. Thus the first claim of the proposition follows from Lemma \ref{LemRMatrixHom}.
For the second one, we first check that $ \Psi $ defined in Eq. (\ref{QuotientAction}) is well-defined. Choose
elements $ g,g'\in G $ such that $ g\cdot N=g'\cdot N $, i.e. there is a $ h\in N $ such that $ g'=gh $. Then
\begin{align*}
\Psi(g'\cdot N,x)=\Phi(g',x)=\Phi(g,\Phi(h,x))=\Phi(g,x)=\Psi(g\cdot N,x)
\end{align*}
for all $ x\in M $. Also $ \Psi(h\cdot N,x)=\Phi(h,x)=x $ for all $ h\in N $. The smoothness as well as
the action properties of $ \Psi $ are induced by $ \Phi $. It is effective by construction since the kernel $ N $
is eliminated.
\end{proof}
Thus for our observations we can require without loss of generality that the Lie group acting transitively on the
homogeneous space is connected and acts effectively in addition: if there is a
$ r $-matrix on $ \mathfrak{g} $ there also has to be one on $ \mathfrak{g}/\mathfrak{n} $. For the sphere
$ \mathbb{S}^{2} $ we have classified all such actions in Section \ref{sectiononish} up the equivalence. Also this
equivalence can be fixed.

\begin{proposition}\label{LastProposition}
Let $ G $ be a Lie group acting on $ M $ via $ \Phi $. If there is a $ r $-matrix on $ \mathfrak{g} $
one can construct a $ r $-matrix on the Lie algebra of any Lie group that acts in a to $ \Phi $ equivalent way on
$ M $.
\end{proposition}

\begin{proof}
Let $ \Phi\colon G\times M\rightarrow M $ and $ \Phi'\colon G'\times M\rightarrow M $ be two equivalent actions, i.e.
there is a Lie group homomorphism $ \phi\colon G\rightarrow G' $ such that $ \Phi(g,x)=\Phi'(\phi(g),x) $ for all
$ g\in G $, $ x\in M $. The Lie algebra action corresponding to $ \phi $ maps a $ r $-matrix on $ \mathfrak{g} $ to
a $ r $-matrix on $ \mathfrak{g}' $ according to Lemma \ref{LemRMatrixHom}.
\end{proof}

Finally, we can prove the main theorem stated in the introduction.

\begin{theorem}\textup{\textbf{\cite{NoTwistPaper}.}}\label{LastTheorem}
There is no twist star product on $ \mathbb{S}^{2} $ deforming any symplectic structure of $ \mathbb{S}^{2} $.
\end{theorem}

\begin{proof}
Assume there is a twist star product $ \star=\sum\limits_{i=0}^{\infty}\hbar^{i}C_{i}(\cdot,\cdot) $ on
$ \mathbb{S}^{2} $ such that
\begin{align}
\left\lbrace f,g\right\rbrace=C_{1}(f,g)-C_{1}(g,f)
\end{align}
for all $ f,g\in\Cinfty(\mathbb{S}^{2})[[\hbar]] $. According to Proposition \ref{PropStructuresForTwStarProd}
there is a $ r $-matrix $ r'\in\mathfrak{g}'\wedge\mathfrak{g}' $ such that $ \left\lbrace f,g\right\rbrace
=m(r'\rhd(f\tensor g)) $ for all $ f,g\in\Cinfty(\mathbb{S}^{2})[[\hbar]] $
and the corresponding Lie group $ G' $ acts transitively on $ \mathbb{S}^{2} $. We proved in
Proposition \ref{PropTransToEffect} that this implies the existence of a Lie group $ G $ which acts transitively
and effectively on $ \mathbb{S}^{2} $ and the existence of a $ r $-matrix $ r $ on the corresponding Lie algebra
$ \mathfrak{g} $.
Moreover, Corollary \ref{CompactTwistHomSp} states that the Etingof-Schiffmann subgroup $ H_{r} $ corresponding to $ r $
acts transitively on $ \mathbb{S}^{2} $. This action is the restriction of the action of $ G $ to $ H_{r}\subseteq G $,
which means that $ H_{r} $ acts effectively in addition. Since $ H_{r} $ is connected, it has to be a semisimple
Lie group according to Lemma \ref{SphereSemiSimple}. Consequently, the Etingof-Schiffmann subalgebra
$ \mathfrak{h}_{r} $ has to be a semisimple Lie algebra, in which
$ r\in\mathfrak{h}_{r}\wedge\mathfrak{h}_{r} $ is non-degenerate (consider Proposition \ref{PropEtingofSub}).
This is a contradiction to Proposition \ref{PropRmatrixSimpleDeg}, since there we
proved that there are no non-degenerate $ r $-matrices on semisimple Lie algebras.
\end{proof}

Remark that we did not use the concrete form of all non-equivalent connected Lie groups that act transitively and
effectively on $ \mathbb{S}^{2} $ (see Theorem \ref{ThmAllSphereActions}) in Theorem \ref{LastTheorem}.
We also did not use Proposition \ref{LastProposition} but only that the acting Lie groups
are semisimple. We showed the existence of a connected Lie group, which is not semisimple, that acts
transitively and effectively on $ \mathbb{S}^{2} $ if the $ 2 $-sphere inherits a twist star product.
Again, this result can be applied to any symplectic foliation of a Poisson manifold. In literature there are
deformations of the $ q $-deformed sphere (consider e.q.
\cite{WatamuraFuzzy,SteinackerFuzzyII,KurkcuogluFuzzy,SteinackerFuzzyI}). These result are not contradictory
to Theorem \ref{LastTheorem} since the corresponding Poisson structures are degenerate. But in this case
the $ q $-deformed sphere is not rotationally invariant, which is unfavorable if one wants to describe the
symmetries of a physical system.

%% file: Appendix.tex
\appendix
\chapter{From Groups to Hopf Algebras}\label{appendixHopf}

This appendix is meant to clarify notations, collect the most basic definitions of algebraic structures we use all
the time, sometimes without referring to them, and prove some results that are of fundamental use. By decoding
the corresponding axioms in commutative diagrams the process of dualizing objects is very nearby: it is
immediately clear how the dual structures are defined since one just has to reverse arrows. By demanding
compatibility to the dual notion we arrive at bialgebras. An interesting statement is that there are two equivalent
definitions that are dual to each other. Then there is just one ingredient missing to define Hopf algebras, i.e.
the antipode. It has some nice properties that we are going to prove. Of course we have a categorical
approach in mind, thus we always define objects with their corresponding morphisms. We basically
follow the first chapter of the book \cite{majid2000foundations}, but also refer to \cite{radford2011hopf}.

We start with arbitrary groups and how they can interact with other sets. A \textit{group} is a non-empty set together
with an associative multiplication and a unit element such that every element is invertible. The group is said
to be \textit{abelian} if the multiplication is commutative. A notion of groups that interact
with other sets is

\begin{definition}[Group Action]
Let $ (G,\cdot) $ be a group and $ M $ an arbitrary set. We say that $ (G,\cdot) $ acts from the left on $
M $ if for every element $ a\in G $ there is a map $ M\rightarrow M $ denoted by $ M\ni m\mapsto a\rhd m\in M $ such
that for all $ a,b\in G $ and $ m\in M $ one has
\begin{align}
(a\cdot b)\rhd m=a\rhd(b\rhd m)
\end{align}
and
\begin{align}
e\rhd m=m,
\end{align}
where $ e $ denotes the unit of $ G $. The axioms of a right action $ \lhd\colon
M\ni m\mapsto m\lhd a\in M $ read
\begin{align}
m\lhd(a\cdot b)=(m\lhd a)\lhd b\text{   and   }m\lhd e=m
\end{align}
for all $ a,b\in G $ and $ m\in M $.
\end{definition}

We often omit to write $ \rhd $ or $ \lhd $. As a first application we define vector spaces via group actions.
This point of view gives much structural background to some axioms.

\begin{definition}[Vector Space]
Let $ (\mathbb{k},\cdot,+) $ be a field and $ (V,+) $ an abelian group, where we denoted the addition on $ V $ also by
$ + $. The triple $ (V,+,\mathbb{k}) $ is said to be a vector space over the field $ \mathbb{k} $ if
$ (\mathbb{k},\cdot) $ acts on $ (V,+) $ via $ \rhd $, i.e.
\begin{align}
(\lambda\cdot\mu)\rhd v=\lambda\rhd(\mu\rhd v)\text{   and   }1\rhd v=v
\end{align}
for all $ \lambda,\mu\in\mathbb{k} $ and $ v\in V $, where $ 1 $ denotes the unit of $ (\mathbb{k},\cdot) $ and we set
$ 0\rhd v=0_{V} $ for all $ v\in V $, where $ 0 $ denotes the unit of $ (\mathbb{k},+) $ and $ 0_{V} $ the unit of
$ (V,+) $. Moreover, these structures have to be compatible, i.e.
\begin{align}
\lambda\rhd(v+w)=\lambda\rhd v+\lambda\rhd w\text{   and   }(\lambda+\mu)\rhd v=\lambda\rhd v+\mu\rhd v
\end{align}
for all $ \lambda,\mu\in\mathbb{k} $ and $ v,w\in V $. The map $ \rhd $, called scalar multiplication, is omitted in
most cases.
\end{definition}

We do not want to give an introduction to linear algebra, but recall the definitions of the tensor product and the
dual space of a vector space since they are used in this thesis. In the following $ \mathbb{k} $ always
denotes a field.

\begin{example}
We start with the free abelian group on $ V\times W $. This is the set of finite strings of tuples in $ V\times W $
modulo the relation
\begin{align}
(a,b)+(c,d)=(c,d)+(a,b),
\end{align}
where $ + $ denotes the concatenation of finite strings. We further
quotient out for $ \lambda\in\mathbb{k} $ and $ (a,b),(c,d)\in V\times W $ all relations of the form
\begin{align}\label{eq35thomas}
(\lambda a,b)&=(a,\lambda b)\\
(a+c,b)&=(a,b)+(c,b)\\
(a,b+d)&=(a,b)+(a,d).
\end{align}
The set we obtain is called the \textbf{tensor product} of $ V $ and $ W $ and is denoted by $ V\tensor W $. We define
a $ \mathbb{k} $-action on the abelian group $ (V\tensor W,+) $ for $ \lambda\in\mathbb{k} $ and $ (a,b)\in V\times W $
by
\begin{align}\label{eq36thomas}
\lambda(a,b)=(\lambda a,b)
\end{align}
which equals to $ (a,\lambda b) $ in $ V\tensor W $ according to Eq. (\ref{eq35thomas}) and extend this linearly to
$ V\tensor W $. Usually, one denotes tuples $ (a,b)\in V\tensor W $ by $ a\tensor b $. Indeed Eq. (\ref{eq36thomas})
defines a left action on $ V\tensor W $, since for $ \lambda,\mu\in\mathbb{k} $ and $ a\tensor b\in V\tensor W $
one has
\begin{align*}
(\lambda\mu)(a\tensor b)&=(\lambda\mu a)\tensor b=\lambda(\mu a)\tensor b\\
&=\lambda(\mu(a\tensor b))
\end{align*}
and $ 1(a\tensor b)=(1a)\tensor b=a\tensor b $. While for another $ c\tensor d\in V\tensor W $ we have
\begin{align*}
\lambda(a\tensor b+c\tensor d)=\lambda(a\tensor b)+\lambda(c\tensor d)
\end{align*}
by the linear extension and
\begin{align*}
(\lambda+\mu)(a\tensor b)&=((\lambda+\mu)a)\tensor b=(\lambda a+\mu a)\tensor b\\
&=(\lambda a)\tensor b+(\mu a)\tensor b=\lambda(a\tensor b)+\mu(a\tensor b)
\end{align*}
and thus the triple $ (V\tensor W,+,\mathbb{k}) $ is a $ \mathbb{k} $-vector space.
\end{example}

Remark that there is also a definition of tensor products via a universal property, consider e.g.
\cite{surhone2010tensor}. As an example we want to discuss the $ \mathbb{k} $-vector space $ \mathbb{k}\tensor V $
for some $ \mathbb{k} $-vector space $ V $ and realize that it is isomorphic to $ V $ itself. This identification is
necessary for the definitions of algebras and coalgebras.

\begin{lemma}\label{lem1thomas}
Let $ V $ be a vector space over a field $ \mathbb{k} $. Then $ \mathbb{k}\tensor V\cong V\cong V\tensor\mathbb{k} $.
\end{lemma}

\begin{proof}
The isomorphisms we are searching for are
\begin{align}\label{eq37thomas}
\mathbb{k}\tensor V\ni\lambda\tensor v\mapsto\lambda v\in V
\end{align}
and
\begin{align}\label{eq38thomas}
V\tensor\mathbb{k}\ni v\tensor\lambda\mapsto\lambda v\in V,
\end{align}
where the right hand side denotes the $ \mathbb{k} $-action on the $ \mathbb{k} $-vector space $ V $.
We prove that the map defined in Eq.
(\ref{eq37thomas}) is an isomorphism. The proof for the second isomorphism works exactly the same. The map defined
in Eq. (\ref{eq37thomas}) is surjective, since any $ v\in V $ can be produced via $ 1\tensor v\mapsto 1v=v $.
For injectivity assume there are $ \lambda\tensor v,\mu\tensor w\in\mathbb{k}\tensor V $ such that
$ \lambda v=\mu w $. If $ \lambda=0 $ then $ \mu w=0 $, i.e. $ \mu=0 $ or $ w=0 $ and for this $ \mu\tensor w=0=
\lambda\tensor v $ because
\begin{align*}
0\tensor w=(0\cdot 0)\tensor w=0\tensor(0\cdot w)=0\tensor 0=0.
\end{align*}
Thus assume that $ \lambda\neq 0 $. Then $ v=\lambda^{-1}\mu w $ and
\begin{align*}
\lambda\tensor v&=\lambda\tensor(\lambda^{-1}\mu w)=(\lambda\lambda^{-1}\mu)\tensor w\\
&=\mu\tensor w.
\end{align*}
This completes the proof.
\end{proof}
Now we discuss the last vector space example.

\begin{example}
The set $ V^{*} $ of all $ \mathbb{k} $-linear maps $ V\rightarrow\mathbb{k} $ is called the \textbf{dual vector space}
of the $ \mathbb{k} $-vector space $ (V,+,\mathbb{k}) $. Thus an element $ \phi\in V^{*} $ is a map $ \phi\colon V
\rightarrow\mathbb{k} $ such that for all $ \lambda,\mu\in\mathbb{k} $ and $ v,w\in V $
\begin{align}
\phi(\lambda v+\mu w)=\lambda\phi(v)+\mu\phi(w).
\end{align}
Then $ V^{*} $ becomes a $ \mathbb{k} $-vector space structure if we define an addition on $ V^{*} $ pointwise for two
elements $ \phi,\psi\in V^{*} $ by
\begin{align}
(\phi+\psi)(v)=\phi(v)+\psi(v),
\end{align}
for all $ v\in V $, where the right hand side denotes the addition in $ \mathbb{k} $ and if we define an action of
$ \mathbb{k} $ on $ V^{*} $ pointwise for $ \lambda\in\mathbb{k} $ and $ \phi\in V^{*} $ by
\begin{align}
(\lambda\phi)(v)=\phi(\lambda v),
\end{align}
where the argument of $ \phi $ in the right hand side is the $ \mathbb{k} $-action of $ \lambda $ on $ v $. The vector
space axioms are easy to check. Note that there is a way to identify an element $ v\in V $ with an element in
the dual space $ V^{**} $ of $ V^{*} $ by defining $ v(\phi):=\phi(v)\in\mathbb{k} $ for any $ \phi\in V^{*} $.
We proved
\begin{align}\label{eq39thomas}
V\subseteq V^{**}.
\end{align}
Moreover, if $ \phi\in V^{*} $ and $ \psi\in W^{*} $ we define a map $ \phi\tensor\psi\in V^{*}\tensor W^{*} $ by
$ (\phi\tensor\psi)(v\tensor w)=\phi(v)\tensor\psi(w) $ for $ v\in V, w\in W $. Remark that by Lemma \ref{lem1thomas}
one has $ \phi(v)\tensor\psi(w)\in\mathbb{k}\tensor\mathbb{k}\cong\mathbb{k} $ thus $ \phi\tensor\psi\colon
V\tensor W\rightarrow\mathbb{k} $. We calculate for $ \lambda,\mu\in
\mathbb{k} $ and $ v\tensor w,x\tensor y\in V\tensor W $
\begin{align*}
(\phi\tensor\psi)(\lambda(v\tensor w)+\mu(x\tensor y))&=(\phi\tensor\psi)((\lambda v)\tensor w)+(\phi\tensor\psi)
((\mu x)\tensor y)\\
&=\phi(\lambda v)\tensor\psi(w)+\phi(\mu x)\tensor\psi(y)\\
&=(\lambda\phi(v))\tensor\psi(w)+(\mu\phi(x))\tensor\psi(y)\\
&=\lambda(\phi(v)\tensor\psi(w))+\mu(\phi(x)\tensor\psi(y))\\
&=\lambda((\phi\tensor\psi)(v\tensor w))+\mu((\phi\tensor\psi)(x\tensor y))
\end{align*}
what shows $ \phi\tensor\psi\in(V\tensor W)^{*} $. Thus we proved
\begin{align}\label{eq40thomas}
V^{*}\tensor W^{*}\subseteq(V\tensor W)^{*}.
\end{align}
In finite dimensions one can check that Eq. (\ref{eq39thomas}) and Eq. (\ref{eq40thomas}) are equalities.
\end{example}

We are now able to define a very central object in Hopf algebra theory, first
in terms of rings, fields and actions: an algebra.

\begin{definition}[Algebra]\label{defalgebra}
We call $ (\algebra{A},\cdot,+,\mathbb{k}) $ an algebra over a field $ \mathbb{k} $ if the following three
conditions are satisfied,
\begin{compactenum}
\item the triple $ (\algebra{A},\cdot,\mathbb{k}) $ is a unital ring,
\item the triple $ (\algebra{A},+,\mathbb{k}) $ is a vector space,
\item the scalar multiplication (that exists according to ii.)) is compatible with both, the multiplication and the
addition, i.e. for all $ \lambda\in\mathbb{k} $ and $ a,b\in\algebra{A} $ one has
\begin{align}\label{eq41thomas}
\lambda(a\cdot b)=(\lambda a)\cdot b=a\cdot(\lambda b)
\end{align}
and
\begin{align}\label{eq42thomas}
\lambda(a+b)=\lambda a+\lambda b
\end{align}
while the second condition always holds since $ (\algebra{A},+,\mathbb{k}) $ is a vector space according to ii.).
\end{compactenum}
\end{definition}

We also want to give an alternative definition of an algebra. Remark that conditions (\ref{eq41thomas}) and 
(\ref{eq42thomas}) can be formulated as follows: there is a $ \mathbb{k} $-linear map
\begin{align}\label{eq43thomas}
\cdot\colon\algebra{A}\tensor\algebra{A}\rightarrow\algebra{A}
\end{align}
replacing the multiplication on $ \algebra{A} $. If we assume in addition that $ (\algebra{A},+,\mathbb{k}) $ is a
vector space the only conditions for $ (\algebra{A},\cdot,+,\mathbb{k}) $ to be an algebra are conditions on the new
multiplication defined in Eq. (\ref{eq43thomas}). More precisely, the map defined in Eq. (\ref{eq43thomas})
has to obey for all $ a,b,c\in\algebra{A} $
\begin{align}\label{eq45thomas}
(\cdot\circ(\cdot\tensor 1))(a\tensor b\tensor c)=(\cdot\circ(1\tensor\cdot))((a\tensor b\tensor c)),
\end{align}
where $ 1\colon\algebra{A}\rightarrow\algebra{A} $ denotes the identity map on $ \algebra{A} $, i.e. $ \cdot $ has to
be associative and there has to be an element $ 1_{\algebra{A}}\in\algebra{A} $ such that
\begin{align}\label{eq44thomas}
\cdot(a\tensor 1_{\algebra{A}})=a=\cdot(1_{\algebra{A}}\tensor a),
\end{align}
i.e. there has to be a unit $ 1_{\algebra{A}} $ on $ \algebra{A} $. Another way to write condition (\ref{eq44thomas})
is obtained in the following way. Define for any $ a\in\algebra{A} $ a map
\begin{align}
\eta_{a}\colon\mathbb{k}\rightarrow\algebra{A}
\end{align}
by $ \eta_{a}(\lambda)=\lambda a $, where the right hand side denotes the scalar multiplication. Finally, set
$ \eta=\eta_{1_{\algebra{A}}} $ what means
\begin{align}\label{eq47thomas}
\eta(\lambda)=\lambda 1_{\algebra{A}}.
\end{align}
Thus a condition recovering (\ref{eq44thomas}) is
\begin{align}\label{eq44'thomas}
(\cdot\circ(\eta\tensor 1))(\lambda\tensor a)=\lambda a=(\cdot\circ(1\tensor\eta))(a\tensor\lambda),
\end{align}
for $ \lambda\in\mathbb{k} $ and $ a\in\algebra{A} $.
Conditions (\ref{eq45thomas}) and (\ref{eq44'thomas}) are equivalent to the commutativity of the following
diagrams, respectively.

\begin{equation}\label{eq46thomas}
\begin{tikzpicture}
  \matrix (m) [matrix of math nodes,row sep=3em,column sep=4em,minimum width=2em]
  {
     \algebra{A}\tensor\algebra{A}\tensor\algebra{A} & \algebra{A}\tensor\algebra{A} \\
     \algebra{A}\tensor\algebra{A} & \algebra{A} \\};
  \path[-stealth]
    (m-1-1) edge node [left] {$\cdot\tensor 1$} (m-2-1)
            edge node [above] {$1\tensor\cdot$} (m-1-2)
    (m-2-1) edge node [above] {$\cdot$} (m-2-2)
    (m-1-2) edge node [right] {$\cdot$} (m-2-2);
\end{tikzpicture}
\begin{tikzpicture}
  \matrix (m) [matrix of math nodes,row sep=3em,column sep=0.5em,minimum width=2em]
  {
     \algebra{A}\tensor\algebra{A} &  \\
     \mathbb{k}\tensor\algebra{A} & \algebra{A} \\};
  \path[-stealth]
    (m-2-1) edge node [left] {$\eta\tensor 1$} (m-1-1)
    (m-1-1) edge node [above] {$\cdot$} (m-2-2);
    \draw[double equal sign distance] (m-2-1) -- (m-2-2);
\end{tikzpicture}
\begin{tikzpicture}
  \matrix (m) [matrix of math nodes,row sep=3em,column sep=0.5em,minimum width=2em]
  {
     \algebra{A}\tensor\algebra{A} &  \\
     \algebra{A}\tensor\mathbb{k} & \algebra{A} \\};
  \path[-stealth]
    (m-2-1) edge node [left] {$1 \tensor\eta$} (m-1-1)
    (m-1-1) edge node [above] {$\cdot$} (m-2-2);
    \draw[double equal sign distance] (m-2-1) -- (m-2-2);
\end{tikzpicture}
\end{equation}

Remark that in the second and third diagram of (\ref{eq46thomas}) the equality sign denotes the isomorphisms
defined in Eq.
(\ref{eq37thomas}) and Eq. (\ref{eq38thomas}) of Lemma \ref{lem1thomas}, respectively. In the above lines we proved the
following

\begin{proposition}
$ (\algebra{A},\cdot,+,\mathbb{k}) $ is an algebra over a field $ \mathbb{k} $ if and only if $ (\algebra{A},
+,\mathbb{k}) $ is a $ \mathbb{k} $-vector space and the following two conditions are satisfied.
\begin{compactenum}
\item There is a $ \mathbb{k} $-linear map $ \cdot\colon\algebra{A}\tensor\algebra{A}\rightarrow\algebra{A} $, called
\textbf{multiplication} on $ \algebra{A} $ that is \textbf{associative}, i.e. the first diagram of (\ref{eq46thomas})
commutes.
\item There is the $ \mathbb{k} $-linear map $ \eta $ defined in Eq. (\ref{eq47thomas}) such that the second and third
diagram of (\ref{eq46thomas}) commute.
\end{compactenum}
\end{proposition}

We reformulated the definition of an algebra via the commuting diagrams (\ref{eq46thomas}) because
we want to \textit{dualize} our objects and
notions. The proceeding
is now as follows: simply reverse arrows and demand objects to fill the gaps of the maps that do not fit any more
in the diagrams.

\begin{definition}[Coalgebra]\label{DefCoAlgebra}
We call $ (\algebra{C},+,\Delta,\epsilon,\mathbb{k}) $ a coalgebra over a field $ \mathbb{k} $ if there is a
linear coproduct
\begin{align}
\Delta\colon\algebra{C}\rightarrow\algebra{C}\tensor\algebra{C}
\end{align}
which is coassociative and a linear counit
\begin{align}
\epsilon\colon\algebra{C}\rightarrow\mathbb{k}.
\end{align}
The coassociativity axiom and the axioms a counit has to satisfy are denoted in the commutativity of the three diagrams
of (\ref{eq48thomas}), respectively.
\end{definition}

\begin{equation}\label{eq48thomas}
\begin{tikzpicture}
  \matrix (m) [matrix of math nodes,row sep=3em,column sep=4em,minimum width=2em]
  {
     \algebra{C}\tensor\algebra{C}\tensor\algebra{C} & \algebra{C}\tensor\algebra{C} \\
     \algebra{C}\tensor\algebra{C} & \algebra{C} \\};
  \path[-stealth]
    (m-2-2) edge node [right] {$\Delta$} (m-1-2)
            edge node [above] {$\Delta$} (m-2-1)
    (m-2-1) edge node [left] {$\Delta\tensor 1$} (m-1-1)
    (m-1-2) edge node [above] {$1\tensor\Delta$} (m-1-1);
\end{tikzpicture}
\begin{tikzpicture}
  \matrix (m) [matrix of math nodes,row sep=3em,column sep=0.5em,minimum width=2em]
  {
     \algebra{C}\tensor\algebra{C} &  \\
     \mathbb{k}\tensor\algebra{C} & \algebra{C} \\};
  \path[-stealth]
    (m-1-1) edge node [left] {$\epsilon\tensor 1$} (m-2-1)
    (m-2-2) edge node [above] {$\Delta$} (m-1-1);
    \draw[double equal sign distance] (m-2-1) -- (m-2-2);
\end{tikzpicture}
\begin{tikzpicture}
  \matrix (m) [matrix of math nodes,row sep=3em,column sep=0.5em,minimum width=2em]
  {
     \algebra{C}\tensor\algebra{C} &  \\
     \algebra{C}\tensor\mathbb{k} & \algebra{C} \\};
  \path[-stealth]
    (m-1-1) edge node [left] {$1\tensor\epsilon$} (m-2-1)
    (m-2-2) edge node [above] {$\Delta$} (m-1-1);
    \draw[double equal sign distance] (m-2-1) -- (m-2-2);
\end{tikzpicture}
\end{equation}

The axioms of a coalgebra depicted in the commutativity of the diagrams (\ref{eq48thomas}) read in formulas
\begin{align}\label{eq49thomas}
(\Delta\tensor 1)\circ\Delta=(1\tensor\Delta)\circ\Delta
\end{align}
and
\begin{align}\label{eq50thomas}
(\epsilon\tensor 1)\circ\Delta=1=(1\tensor\epsilon)\circ\Delta,
\end{align}
where $ 1\colon\algebra{C}\rightarrow\algebra{C} $ denotes the identity on $ \algebra{C} $.

\begin{remark}\label{RemarkSweedler}
One could wonder what part of $ \Delta(c)\in\algebra{C}\tensor\algebra{C} $ for a $ c\in\algebra{C} $ is applied to
$ \Delta,\epsilon $ or $ 1 $ in Eq. (\ref{eq49thomas}) and Eq. (\ref{eq50thomas}), respectively. We want to discuss
this problem. $ \Delta(c)\in\algebra{C}\tensor\algebra{C} $ is a linear combination of elements $ c_{1}\tensor
c_{2}\in\algebra{C}\tensor\algebra{C} $ for $ c_{1},c_{2}\in\algebra{C} $, i.e. $ \Delta(c)=\sum
\limits_{i=1}^{n}c_{1_{i}}\tensor c_{2_{i}} $ for a finite number $ n\in\mathbb{N} $. Having this in mind we use the
short notation (called \textbf{Sweedler\textquoteright s notation})
\begin{align}\label{eq51thomas}
\Delta(c)=\sum c_{(1)}\tensor c_{(2)}=c_{(1)}\tensor c_{(2)},
\end{align}
thus sometimes we even omit the sum. Then equation (\ref{eq49thomas}) reads
\begin{align}\label{eq52thomas}
c_{(1)(1)}\tensor c_{(1)(2)}\tensor c_{(2)}=c_{(1)}\tensor c_{(2)(1)}\tensor c_{(2)(2)},
\end{align}
where we used Eq. (\ref{eq51thomas}) again for the elements $ \Delta(c_{(1)})\in\algebra{C}\tensor\algebra{C} $ and
$ \Delta(c_{(2)})\in\algebra{C}\tensor\algebra{C} $, respectively, while equation (\ref{eq50thomas}) reads
\begin{align}
\epsilon(c_{(1)})c_{(2)}=c=\epsilon(c_{(2)})c_{(1)},
\end{align}
where we again used the isomorphisms stated in the proof of Lemma \ref{lem1thomas}. Motivated by equation
(\ref{eq52thomas}) we moreover write
\begin{align}\label{eq53thomas}
c_{(1)}\tensor c_{(2)}\tensor c_{(3)}=c_{(1)(1)}\tensor c_{(1)(2)}\tensor c_{(2)}=c_{(1)}\tensor c_{(2)(1)}\tensor
c_{(2)(2)}.
\end{align}
More general we define elements of the form
$ c_{(1)}\tensor c_{(2)}\tensor\cdots\tensor c_{(r)}\in\algebra{C}\tensor\cdots\tensor\algebra{C} $
for any $ r\in\mathbb{N} $ and $ c\in\algebra{C} $ by
\begin{align*}
c_{(1)}\tensor c_{(2)}\tensor\cdots\tensor c_{(r)}&=(\Delta\tensor\underbrace{1\tensor\cdots\tensor 1}
_{(r-1)\text{-times}})\\
&\circ(\Delta\tensor\underbrace{1\tensor\cdots\tensor 1}_{(r-2)\text{-times}})\\
&\vdots\\
&\circ\Delta(c).
\end{align*}
For $ r=1 $ and $ r=2 $ we get back the definitions in Eq. (\ref{eq51thomas}) and Eq. (\ref{eq53thomas}),
respectively. As in this
two cases also permutations of $ \Delta $ and $ 1 $ in the defining equations lead to $ c_{(1)}\tensor c_{(2)}\tensor
\cdots\tensor c_{(r)} $ as a consequence of the coassociativity of
$ \Delta $. Let us prove this for $ r=3 $.
\begin{align*}
(\Delta\tensor 1\tensor 1)\circ(\Delta\tensor 1)\circ\Delta(c)&=(\Delta\tensor 1\tensor 1)\circ(\Delta\tensor
1)(c_{(1)}\tensor c_{(2)})\\
&=(\Delta\tensor 1\tensor 1)(c_{(1)(1)}\tensor c_{(1)(2)}\tensor c_{(2)})\\
&=c_{(1)(1)(1)}\tensor c_{(1)(1)(2)}\tensor c_{(1)(2)}\tensor c_{(2)}\\
&=c_{(1)(1)}\tensor c_{(1)(2)(1)}\tensor c_{(1)(2)(2)}\tensor c_{(2)}\\
&=(1\tensor\Delta\tensor 1)\circ(\Delta\tensor 1)\circ\Delta(c),
\end{align*}
where we used the coassociativity in the fourth equation for the element $ c_{(1)} $. Similarly for the other
permutations. By induction this result is true for any $ r\in\mathbb{N} $. As another short notation that we sometimes
use, define
\begin{align}
\Delta^{r-1}(c)=c_{(1)}\tensor c_{(2)}\tensor\cdots\tensor c_{(r)},
\end{align}
where we set $ \Delta^{0}=\Delta $. While $ \Delta $ adds another $ \algebra{C} $, the counit $ \epsilon $ reduces its
argument about one $ \algebra{C} $. More precise,
\begin{align*}
(\epsilon\tensor\underbrace{1\tensor\cdots\tensor 1}_{r\text{-times}})\circ\Delta^{r}(c)&=(1\tensor\epsilon\tensor 1
\tensor\cdots\tensor 1)\circ\Delta^{r}(c)\\
&=\vdots\\
&=(\underbrace{1\tensor\cdots\tensor 1}_{r\text{-times}}\tensor\epsilon)\circ\Delta^{r}(c)\\
&=c_{(1)}\tensor c_{(2)}\tensor\cdots\tensor c_{(r)}\\
&=\Delta^{r-1}(c).
\end{align*}
We prove this again only for $ r=3 $.
\begin{align*}
(\epsilon\tensor 1\tensor 1\tensor 1)\circ\Delta^{3}(c)&=(\epsilon\tensor 1\tensor 1\tensor 1)(c_{(1)(1)}\tensor
c_{(1)(2)}\tensor c_{(2)}\tensor c_{(3)})\\
&=\epsilon(c_{(1)(1)})c_{(1)(2)}\tensor c_{(2)}\tensor c_{(3)}\\
&=c_{(1)}\tensor c_{(2)}\tensor c_{(3)},
\end{align*}
where we used the counit axiom for the element $ c_{(1)} $ in the last equation. The other permutations work
analogous.
\end{remark}

Now after defining algebras and coalgebras we come to their morphisms.

\begin{definition}[Algebra and Coalgebra Map]
A $ \mathbb{k} $-linear map $ f\colon\algebra{A}\rightarrow\algebra{B} $ between two algebras $ \algebra{A} $ and
$ \algebra{B} $ over $ \mathbb{k} $ is said to be an algebra map if it satisfies for all $ a,b\in\algebra{A} $
\begin{align}
f(ab)=f(a)f(b)\text{   and   }f(1_{\algebra{A}})=1_{\algebra{B}},
\end{align}
where $ 1_{\algebra{A}} $ and $ 1_{\algebra{B}} $ denote the units of $ \algebra{A} $ and $ \algebra{B} $,
respectively.
A $ \mathbb{k} $-linear map $ f\colon\algebra{C}\rightarrow\algebra{D} $ between two coalgebras $ \algebra{C} $ and
$ \algebra{D} $ over $ \mathbb{k} $ is said to be a coalgebra map if it satisfies
\begin{align}
(f\tensor f)\circ\Delta_{\algebra{C}}=\Delta_{\algebra{D}}\circ f\text{   and   }\epsilon_{\algebra{D}}\circ f=
\epsilon_{\algebra{C}}.
\end{align}
\end{definition}

\begin{remark}
One can structure the tensor product of two (co)algebras as a (co)algebra again. Explicitly, we define for
two algebras $ \algebra{A} $ and $ \algebra{B} $ a product on the vector space $ \algebra{A}\tensor\algebra{B} $ by
\begin{align}\label{eq57thomas}
(a\tensor c)\cdot(b\tensor d):=(ab\tensor cd),
\end{align}
for all $ a,b\in\algebra{A} $ and $ c,d\in\algebra{B} $ and extend this linearly to $ \algebra{A}\tensor\algebra{B} $.
Because of the associativity of the product on $ \algebra{A} $ and $ \algebra{B} $, respectively, we get for
$ a,b,x\in\algebra{A} $, $ c,d,y\in\algebra{B} $
\begin{align*}
((a\tensor c)\cdot(b\tensor d))\cdot(x\tensor y)&=(ab\tensor cd)\cdot(x\tensor y)\\
&=(ab)x\tensor(cd)y\\
&=a(bx)\tensor c(dy)\\
&=(a\tensor c)\cdot((b\tensor d)\cdot(x\tensor y)),
\end{align*}
which is the associativity of $ \cdot $. The unit of $ \algebra{A}\tensor\algebra{B} $ is the tensor product of
the units on $ \algebra{A} $ and $ \algebra{B} $, i.e. $ 1_{\algebra{A}\tensor\algebra{B}}=1_{\algebra{A}}\tensor
1_{\algebra{B}} $ which can be seen via
\begin{align*}
1_{\algebra{A}\tensor\algebra{B}}\cdot(a\tensor b)=1_{\algebra{A}}a\tensor 1_{\algebra{B}}b=a\tensor b=
(a\tensor b)\cdot 1_{\algebra{A}\tensor\algebra{B}},
\end{align*}
for all $ a\in\algebra{A} $, $ b\in\algebra{B} $. For two coalgebras $ \algebra{C} $ and $ \algebra{D} $ we define the
coproduct on $ \algebra{C}\tensor\algebra{D} $ for all $ c\in\algebra{C} $, $ d\in\algebra{D} $ by
\begin{align}
\Delta(c\tensor d)=(c_{(1)}\tensor d_{(1)})\tensor(c_{(2)}\tensor d_{(2)})
\end{align}
and extend this linearly to $ \algebra{C}\tensor\algebra{D} $. We used again the notation introduced in Eq.
(\ref{eq51thomas}). The coassociativity is satisfied because for $ c\in\algebra{C} $ and $ d\in\algebra{D} $ one has
\begin{align*}
(\Delta\tensor 1)\circ\Delta(c\tensor d)&=(\Delta\tensor 1)(c_{(1)}\tensor d_{(1)})\tensor(c_{(2)}\tensor d_{(2)})\\
&=(c_{(1)(1)}\tensor d_{(1)(1)})\tensor(c_{(1)(2)}\tensor d_{(1)(2)})\tensor(c_{(2)}\tensor d_{(2)})\\
&=(c_{(1)}\tensor d_{(1)})\tensor(c_{(2)(1)}\tensor d_{(2)(1)})\tensor(c_{(2)(2)}\tensor d_{(2)(2)})\\
&=(1\tensor\Delta)\circ\Delta(c\tensor d)
\end{align*}
by the coassocitivity in $ \algebra{C} $ and $ \algebra{D} $, respectively. The $ 1 $ in the equation above
denotes the identity map on $ \algebra{C}\tensor\algebra{D} $. The counit of $ \algebra{C}\tensor\algebra{D} $ is the
tensor product of the counits of $ \algebra{C} $ and $ \algebra{D} $, i.e. $ \epsilon=\epsilon_{\algebra{C}}\tensor
\epsilon_{\algebra{D}} $, since for $ c\in\algebra{C} $, $ d\in\algebra{D} $ one has
\begin{align*}
(\epsilon\tensor 1)\circ\Delta(c\tensor d)&=(\epsilon\tensor 1)(c_{(1)}\tensor d_{(1)})\tensor
(c_{(2)}\tensor d_{(2)})\\
&=(\epsilon_{\algebra{C}}(c_{(1)})\tensor\epsilon_{\algebra{D}}(d_{(1)}))\tensor(c_{(2)}\tensor d_{(2)})\\
&=c\tensor d,
\end{align*}
where the last equation follows from the counit property on $ \algebra{C} $ and $ \algebra{D} $, respectively.
The other counit axiom follows similarly.
\end{remark}

The following should give an argument why on calls the process of passing from an
algebra to a coalgebra (and back) \textit{dualizing}.

\begin{proposition}
Consider a coalgebra $ \algebra{C} $ and the adjoint maps
\begin{align}
\cdot\colon\algebra{C}^{*}\tensor\algebra{C}^{*}\rightarrow\algebra{C}^{*}
\end{align}
and
\begin{align}
\eta\colon\mathbb{k}\rightarrow\algebra{C}^{*}
\end{align}
of the coproduct and counit, respectively. More explicit, one has for $ \phi,\psi\in\algebra{C}^{*} $, $ c\in
\algebra{C} $ and $ \lambda\in\mathbb{k} $
\begin{align}\label{eq54thomas}
(\phi\cdot\psi)(c)=(\phi\tensor\psi)\circ\Delta(c)
\end{align}
and
\begin{align}\label{eq55thomas}
\eta(\lambda)(c)=\lambda\epsilon(c).
\end{align}
The product $ \cdot $ and the unit $ \eta $ make the dual space $ \algebra{C}^{*} $ into an algebra.
The unit element
is $ \epsilon $. Conversely, if $ \algebra{A} $ is a finite-dimensional algebra, $ \algebra{A}^{*} $ gets a
coalgebra structure in the same fashion.
\end{proposition}

\begin{proof}
If $ \algebra{C} $ is a coalgebra Eq. (\ref{eq54thomas}) defines an associative product on $ \algebra{C}^{*} $
because for $ \phi,\psi,\chi\in\algebra{C}^{*} $ and $ c\in\algebra{C} $ one has
\begin{align*}
((\phi\cdot\psi)\cdot\chi)(c)&=(\phi\cdot\psi)(c_{(1)})\chi(c_{(2)})\\
&=\phi(c_{(1)(1)})\psi(c_{(1)(2)})\chi(c_{(2)})\\
&=(\phi\tensor\psi\tensor\chi)\circ(\Delta\tensor 1)\circ\Delta(c)\\
&=(\phi\tensor\psi\tensor\chi)\circ(1\tensor\Delta)\circ\Delta(c)\\
&=(\phi\cdot(\psi\cdot\chi))(c),
\end{align*}
where we used the coassociativity of $ \Delta $ on $ \algebra{C} $ in the fourth equation. The unit axiom of
the map defined in Eq.
(\ref{eq55thomas}) also follows for $ \lambda\in\mathbb{k} $, $ \phi\in\algebra{C}^{*} $:
\begin{align*}
\cdot((\eta\tensor 1)(\lambda\tensor\phi))(c)&=\cdot(\eta(\lambda)\tensor\phi)(c)\\
&=\lambda\epsilon(c_{(1)})\phi(c_{(2)})\\
&=(\lambda\epsilon\tensor\phi)\circ\Delta(c)\\
&=(\lambda 1_{\mathbb{k}}\tensor\phi)\circ(\epsilon\tensor 1)\circ\Delta(c)\\
&=\lambda\phi(c),
\end{align*}
where the last equation follows from the counit property of $ \epsilon $ on $ \algebra{C} $. The other unit axiom
follows similarly. Conversely, if $ \algebra{A} $ is an algebra we can define a coproduct and a counit via Eq.
(\ref{eq54thomas}) and Eq. (\ref{eq55thomas}) where one has to read the equations from right to left. This is possible
since in the finite-dimensional situation we can identify $ \algebra{A} $ with the dual space of $ \algebra{A}^{*} $.
Thus $ (a\tensor b)\circ\Delta(\phi)=(a\cdot b)(\phi) $ and $ \lambda\epsilon(\phi)=\eta(\lambda)(\phi) $ for
$ a,b\in\algebra{A}=\algebra{A}^{**} $, $ \lambda\in\mathbb{k} $ and $ \phi\in\algebra{A}^{*} $. Then the same
equations above reordered give the coassociativity of $ \Delta $ on $ \algebra{A}^{*} $ by the associativity of
$ \cdot $ on $ \algebra{A} $ and the counit axiom of $ \epsilon $ on $ \algebra{A}^{*} $ from the unit axiom of
$ \eta $ on $ \algebra{A} $. This concludes the proof.
\end{proof}

The situation gets interesting if we combine both, the features of an algebra and a coalgebra, in one space.
Moreover, the algebra and coalgebra structure should be compatible. This leads to the next

\begin{definition}[Bialgebra]\label{DefBiAlgebra}
We call $ (H,+,\cdot,\eta,\Delta,\epsilon,\mathbb{k}) $ a bialgebra over a field $ \mathbb{k} $ if
$ (H,+,\mathbb{k}) $ is an algebra and a coalgebra over $ \mathbb{k} $ such that the respective structures are
compatible, i.e. $ \Delta\colon H\rightarrow H\tensor H $ and $ \epsilon\colon H\rightarrow\mathbb{k} $ are algebra
maps. Explicitly, this compatibility reads for $ g,h\in H $
\begin{align}
\Delta(gh)&=\Delta(g)\Delta(h),\\
\Delta(1_{H})&=1_{H}\tensor 1_{H},\\
\epsilon(gh)&=\epsilon(g)\epsilon(h),\\
\epsilon(1_{H})&=1_{\mathbb{k}}.
\end{align}
\end{definition}
Bialgebras are self dual objects. To illustrate this we prove a characterization of bialgebras via their
algebra structure. Remark that bialgebras are defined in therms of their coalgebra structure.

\begin{proposition}
An algebra and coalgebra
$ (H,+,\cdot,\eta,\Delta,\epsilon,\mathbb{k}) $ is a bialgebra if and only if $ \cdot\colon H\tensor H\rightarrow H $
and $ \eta\colon\mathbb{k}\rightarrow H $ are coalgebra maps, i.e.
\begin{align}\label{eq56thomas}
(\cdot\tensor\cdot)\circ\Delta_{H\tensor H}&=\Delta_{H}\circ\cdot,\\
\epsilon_{H}\circ\cdot&=\epsilon_{H\tensor H},\\
(\eta\tensor\eta)\circ\Delta_{\mathbb{k}}&=\Delta_{H}\circ\eta,\\
\epsilon_{H}\circ\eta&=\epsilon_{\mathbb{k}}.
\end{align}
\end{proposition}

\begin{proof}
Assume that $ (H,+,\cdot,\eta,\Delta,\epsilon,\mathbb{k}) $ is a bialgebra. It is just a matter of calculation to
verify Eq. (\ref{eq56thomas}) and the equations following. Thus let $ g,h\in H $ and $ \lambda\in\mathbb{k} $.
\begin{align*}
(\cdot\tensor\cdot)\circ\Delta_{H\tensor H}(g\tensor h)&=(\cdot\tensor\cdot)(g_{(1)}\tensor h_{(1)}\tensor
g_{(2)}\tensor h_{(2)})\\
&=(g_{(1)}\cdot h_{(1)})\tensor(g_{(2)}\cdot h_{(2)})\\
&=(g_{(1)}\tensor g_{(2)})\cdot(h_{(1)}\tensor h_{(2)})\\
&=\Delta_{H}(g)\cdot\Delta_{H}(h)\\
&=\Delta_{H}\circ\cdot(g\tensor h),
\end{align*}
where $ \cdot $ in the third line denotes the product on the tensor product $ H\tensor H $ defined in Eq.
(\ref{eq57thomas}).
\begin{align*}
\epsilon_{H}\circ\cdot(g\tensor h)&=\epsilon_{H}(g\cdot h)\\
&=\epsilon_{H}(g)\cdot\epsilon_{H}(h)\\
&=\epsilon_{H}(g)\tensor\epsilon_{H}(h)\\
&=\epsilon_{H\tensor H}(g\tensor h),
\end{align*}
where the second equation holds since $ \epsilon_{H} $ is an algebra map, the third equation is due to one isomorphism
of Lemma \ref{lem1thomas} and the last equation is the definition of $ \epsilon_{H\tensor H} $.
\begin{align*}
(\eta\tensor\eta)\circ\Delta_{\mathbb{k}}(\lambda)&=(\eta\tensor\eta)(\lambda_{(1)}\tensor\lambda_{(2)})\\
&=\eta(\lambda_{(1)})\tensor\eta(\lambda_{(2)})\\
&=(\lambda_{(1)}1_{H})\tensor(\lambda_{(2)}1_{H})\\
&=\lambda_{(1)}\lambda_{(2)}1_{H}\tensor 1_{H}\\
&=\lambda_{(1)}\lambda_{(2)}\Delta_{H}(1_{H})\\
&=\Delta_{H}((\lambda_{(1)}\tensor\lambda_{(2)})1_{H})\\
&=\Delta_{H}(\eta(\lambda)),
\end{align*}
where the equation next to the last follows from the fact that $ \Delta_{H} $ is $ \mathbb{k} $-linear and again by an
isomorphism of Lemma \ref{lem1thomas}. Finally,
\begin{align*}
\epsilon_{H}\circ\eta(\lambda)=\epsilon_{H}(\lambda 1_{H})=\lambda\epsilon_{H}(1_{H})=\lambda 1_{\mathbb{k}}=
\epsilon_{\mathbb{k}}(\lambda),
\end{align*}
where the second equation follows because $ \epsilon_{H} $ is $ \mathbb{k} $-linear, the third equation holds because
$ \epsilon_{H} $ is an algebra map and the last equation is the definition of $ \epsilon_{\mathbb{k}} $. Conversely,
assume that equation (\ref{eq56thomas}) and the three equations following hold.
The same equations of the first part of this proof reordered prove that $ \Delta $ and $ \epsilon $ are algebra
maps. Explicitly, we calculate for $ g,h\in H $
\begin{align*}
\Delta_{H}(g\cdot h)&=\Delta_{H}\circ\cdot(g\tensor h)\\
&=(\cdot\tensor\cdot)\circ\Delta_{H\tensor H}(g\tensor h)\\
&=(\cdot\tensor\cdot)(g_{(1)}\tensor h_{(1)}\tensor g_{(2)}\tensor h_{(2)})\\
&=(g_{(1)}\cdot h_{(1)})\tensor(g_{(2)}\cdot h_{(2)})\\
&=(g_{(1)}\tensor g_{(2)})\cdot(h_{(1)}\tensor h_{(2)})\\
&=\Delta_{H}(g)\cdot\Delta_{H}(h),
\end{align*}
where the second equation holds since $ \cdot $ is a coalgebra map and we used again the multiplication
on the tensor product $ H\tensor H $.
\begin{align*}
1_{H}\tensor 1_{H}=\eta(1_{\mathbb{k}})\tensor\eta(1_{\mathbb{k}})
=(\eta\tensor\eta)\circ\Delta_{\mathbb{k}}(1_{\mathbb{k}})
=\Delta_{H}(\eta(1_{\mathbb{k}}))
=\Delta_{H}(1_{H}),
\end{align*}
where we used that $ \Delta_{\mathbb{k}}(1_{\mathbb{k}})=1_{\mathbb{k}}\tensor 1_{\mathbb{k}} $
and that $ \eta $ is a coalgebra map.
\begin{align*}
\epsilon_{H}(g\cdot h)=\epsilon_{H}\circ\cdot(g\tensor h)=\epsilon_{H\tensor H}(g\tensor h)=\epsilon_{H}(g)\tensor
\epsilon_{H}(h)=\epsilon_{H}(g)\cdot\epsilon_{H}(h),
\end{align*}
where we used that $ \cdot $ is a coalgebra map, the equation next to the last equation is the definition of
$ \epsilon_{H\tensor H} $ and the last equation follows from Lemma \ref{lem1thomas}. Remark that condition $
\epsilon(1_{H})=1_{\mathbb{k}} $ is always satisfied since $ \mathbb{k} $ is a field. This concludes the proof.
\end{proof}

Now we can introduce the definition of a \textit{Hopf algebra}. Essentially, a Hopf algebra is a bialgebra over
a field (or sometimes even over a commutative ring) with one extra structure.

\begin{definition}[Hopf Algebra]\label{hopfalgdef}
We call $ (H,+,\cdot,\eta,\Delta,\epsilon,S,\mathbb{k}) $ a Hopf algebra over a field $ \mathbb{k} $ if
it is a bialgebra over $ \mathbb{k} $ and there is a $ \mathbb{k} $-linear
map
\begin{align}
S\colon H\rightarrow H
\end{align}
called an antipode of $ H $ satisfying
\begin{align}
\cdot(S\tensor 1)\circ\Delta=\eta\circ\epsilon=\cdot(1\tensor S)\circ\Delta.
\end{align}
In other words, an algebra and coalgebra $ H $ is a Hopf algebra if and only if the diagrams of (\ref{eq58thomas}) and
(\ref{eq59thomas}) commute, where (\ref{eq58thomas}) encodes the compatibility axioms of a bialgebra and
(\ref{eq59thomas}) the antipode axiom.
\end{definition}

\begin{equation}\label{eq58thomas}
\begin{tikzpicture}
  \matrix (m) [matrix of math nodes,row sep=3em,column sep=4em,minimum width=2em]
  {
     H\tensor H & H & H\tensor H \\
     H\tensor H\tensor H\tensor H &  & H\tensor H\tensor H\tensor H \\};
  \path[-stealth]
    (m-1-1) edge node [above] {$\cdot$} (m-1-2)
            edge node [left] {$\Delta\tensor\Delta$} (m-2-1)
    (m-1-2) edge node [above] {$\Delta$} (m-1-3)
    (m-2-3) edge node [right] {$\cdot\tensor\cdot$} (m-1-3)
    (m-2-1) edge node [above] {$1\tensor\sigma\tensor 1$} (m-2-3);
\end{tikzpicture}
\begin{tikzpicture}
  \matrix (m) [matrix of math nodes,row sep=3em,column sep=0.5em,minimum width=2em]
  {
     H & \mathbb{k} \\
     H\tensor H & \\};
  \path[-stealth]
    (m-2-1) edge node [left] {$\cdot$} (m-1-1)
            edge node [right] {$\epsilon\tensor\epsilon$} (m-1-2)
    (m-1-1) edge node [above] {$\epsilon$} (m-1-2);
\end{tikzpicture}
\begin{tikzpicture}
  \matrix (m) [matrix of math nodes,row sep=3em,column sep=0.5em,minimum width=2em]
  {
     H & \mathbb{k} \\
     H\tensor H & \\};
  \path[-stealth]
    (m-1-1) edge node [left] {$\Delta$} (m-2-1)
    (m-1-2) edge node [right] {$\eta\tensor\eta$} (m-2-1)
            edge node [above] {$\eta$} (m-1-1);
\end{tikzpicture}
\end{equation}

Here $ \sigma\colon H\tensor H\ni g\tensor h\mapsto h\tensor g\in H\tensor H $ denotes the tensor product twist.

\begin{equation}\label{eq59thomas}
\begin{tikzpicture}
  \matrix (m) [matrix of math nodes,row sep=3em,column sep=4em,minimum width=2em]
  {
     H & \mathbb{k} & H \\
     H\tensor H &  & H\tensor H \\};
  \path[-stealth]
    (m-1-1) edge node [above] {$\epsilon$} (m-1-2)
            edge node [left] {$\Delta$} (m-2-1)
    (m-1-2) edge node [above] {$\eta$} (m-1-3)
    (m-2-3) edge node [right] {$\cdot$} (m-1-3)
    (m-2-1) edge node [above] {$1\tensor S, S\tensor 1$} (m-2-3);
\end{tikzpicture}
\end{equation}

The next proposition gives some properties of the antipode of a Hopf algebra. On the one hand they are useful if
one wants to calculate examples on the other hand they are useful to clarify if a map is the antipode of a Hopf
algebra.

\begin{proposition}\label{pro9thomas}
The antipode $ S $ of a Hopf algebra $ H $ is unique. Moreover, $ S $ is an antialgebra map, i.e. for all $ g,h\in H $
we have
\begin{align}\label{eq67thomas}
S(g\cdot h)&=S(h)\cdot S(g)\\
S(1)&=1
\end{align}
and $ S $ is an anticoalgebra map, i.e. for all $ g\in H $ one has
\begin{align}
(S\tensor S)\circ\Delta(g)&=\sigma\circ\Delta\circ S(g)\\
\epsilon\circ S(g)&=\epsilon(g).
\end{align}
\end{proposition}

\begin{proof}
We follow \cite[Proposition~1.3.1]{majid2000foundations}.
Assume there are two antipodes $ S $ and $ S' $ of a bialgebra $ H $. We first need a useful identity, namely
we calculate for $ g\in H $
\begin{align*}
g_{(1)}S(g_{(2)})=\cdot(1\tensor S)\circ\Delta(g)=\eta\circ\epsilon(g)=\epsilon(g)1_{H},
\end{align*}
where we used the antipode axiom. Thus
\begin{align}\label{eq60thomas}
g_{(1)}S(g_{(2)})=\epsilon(g)1_{H}
\end{align}
and in the same fashion
\begin{align}\label{eq61thomas}
S(g_{(1)})g_{(2)}=\epsilon(g)1_{H}
\end{align}
if we use the second antipode with $ 1 $ and $ S $ changed. This equations are remarkable because they show that
$ S(g_{(1)}) $ is somehow an inverse to $ g_{(2)} $ (or $ g_{(1)} $ and $ g_{(2)} $ exchanged) and equations
(\ref{eq60thomas}) and (\ref{eq61thomas}) remind in a way of cancelling $ gg^{-1} $ and $ g^{-1}g $ to the unit element
of a group. Consider the proof of Proposition 1.3.1 in \cite{majid2000foundations} for more details to this idea.
Now it is just a matter of calculation. Let $ g\in H $, then
\begin{align*}
S'(g)&=S'(g_{(1)})\epsilon(g_{(2)})\\ 
&=S'(g_{(1)})g_{(2)(1)}S(g_{(2)(2)})\\ 
&=S'(g_{(1)(1)})g_{(1)(2)}S(g_{(2)})\\ 
&=(\cdot(S'\tensor 1)\circ\Delta(g_{(1)}))S(g_{(2)})\\
&=(\eta\circ\epsilon(g_{(1)}))S(g_{(2)})\\ 
&=\epsilon(g_{(1)})S(g_{(2)})\\
&=S(g) 
\end{align*}
This proves $ S'=S $ what implies that the antipode of a Hopf algebra is unique. The next step is to prove that
$ S $ is an antialgebra map. It is easy to get
\begin{align*}
S(1_{H})=1_{H}S(1_{H})=(1_{H})_{(1)}S((1_{H})_{(2)})=\epsilon(1_{H})1_{H}=1_{H}.
\end{align*}
This is again just Eq. (\ref{eq60thomas}) with $ \Delta(1_{H})=(1_{H})_{(1)}\tensor(1_{H})_{(1)}=1_{H}\tensor 1_{H} $
what is true since $ \Delta $ is an algebra map and in the last equation we used $ \epsilon(1_{H})=1_{H} $ what is
true since $ \epsilon $ is an algebra map. The other property is proved step by step. First of all, the antipode axiom
for an element $ gh $ with $ g,h\in H $ reads
\begin{align}\label{eq62thomas}
S(g_{(1)}h_{(1)})g_{(2)}h_{(2)}=\epsilon(gh)1_{H}=\epsilon(g)\epsilon(h)1_{H},
\end{align}
since $ \epsilon $ is an algebra map. If we replace $ h $ with $ h_{(1)} $, where of course $ \Delta(h)=h_{(1)}
\tensor h_{(2)} $ and add $ h_{(2)} $ as a tensor factor in equation (\ref{eq62thomas}) we obtain
\begin{align}\label{eq63thomas}
S(g_{(1)}h_{(1)(1)})g_{(2)}h_{(1)(2)}\tensor h_{(2)}=\epsilon(g)\epsilon(h_{(1)})1_{H}\tensor h_{(2)}=\epsilon(g)
1_{H}\tensor h,
\end{align}
where we used $ (\epsilon\tensor 1)\circ\Delta(h)=h $ which holds by the counit axiom. If we apply $ S $ to the second
tensor factor of (\ref{eq63thomas}) multiply both tensor factors via $ \cdot $, Eq. (\ref{eq63thomas}) reads
\begin{align}\label{eq64thomas}
S(g_{(1)}h_{(1)})g_{(2)}h_{(2)}S(h_{(3)})=\epsilon(g)S(h)
\end{align}
by the associativity of $ \cdot $. Together with Eq. (\ref{eq60thomas}) and the counit axiom we see that Eq.
(\ref{eq64thomas}) transforms to
\begin{align*}
\epsilon(g)S(h)&=S(g_{(1)}h_{(1)})g_{(2)}h_{(2)}S(h_{(3)})\\
&=S(g_{(1)}h_{(1)})g_{(2)}\epsilon(h_{(2)})\\
&=S(g_{(1)}h)g_{(2)}
\end{align*}
If we replace $ g $ by $ g_{(1)} $ and add $ g_{(2)} $ on both sides as a tensor factor the last equation is
equivalent to
\begin{align}\label{eq65thomas}
S(g_{(1)(2)}h)g_{(1)(2)}\tensor g_{(2)}=\epsilon(g_{(1)})S(h)\tensor g_{(2)}=S(h)\tensor g,
\end{align}
where we used the counit axiom. Finally, we apply $ S $ to the second tensor factor of Eq.
(\ref{eq65thomas}) and multiply on both sides of Eq. (\ref{eq65thomas}) via $ \cdot $ and obtain by the associativity
of $ \cdot $, Eq. (\ref{eq60thomas}) and the counit axiom
\begin{align*}
S(h)S(g)=S(g_{(1)}h)h_{(2)(1)}S(h_{(2)(2)})=S(g_{(1)}h)\epsilon(g_{(2)})=S(gh).
\end{align*}
Thus $ S $ is an antialgebra map. The anticoalgebra map properties are a bit easier to check. First of all, for
$ g\in H $ one has
\begin{align*}
\epsilon(S(g))&=\epsilon(S(g_{(1)}))\epsilon(g_{(2)})=\epsilon(S(g_{(1)})g_{(2)})\\
&=\epsilon(\epsilon(g)1_{H})=\epsilon(g)\epsilon(1_{H})\\
&=\epsilon(h),
\end{align*}
where we used the linearity of $ \epsilon $ together with Eq. (\ref{eq61thomas}) and that $ \epsilon $ is an algebra
map. Finally, we calculate for $ g\in H $
\begin{align*}
\sigma\circ(S\tensor S)\circ\Delta(g)&=S(g_{(2)})\tensor S(g_{(1)})\\
&=(S(g_{(1)}))_{(1)}g_{(2)(1)}S(g_{(4)})\tensor(S(g_{(1)}))_{(2)}g_{(2)(2)}S(g_{(3)})\\
&=(S(g_{(1)}))_{(1)}g_{(2)}S(g_{(3)})\tensor(S(g_{(1)}))_{(2)}\\
&=S(g)_{(1)}\tensor S(g)_{(1)}\\
&=\Delta\circ S(g),
\end{align*}
using the same techniques as above. This implies $ (S\tensor S)\circ\Delta=\sigma\circ\Delta\circ S $.
\end{proof}

Remark that if $ H $ and $ G $ are Hopf algebras with antipodes $ S $ and $ S' $ respectively, also the tensor product
$ H\tensor G $ is a Hopf algebra. The bialgebra structure is given by the usual tensor bialgebra structure, while the
unique antipode is $ S\tensor S' $. We did not yet define the corresponding morphisms:

\begin{definition}[Hopf Algebra Map]
A map $ f\colon H\rightarrow G $ between two Hopf algebras $ H $ and $ G $ which is an algebra map and a coalgebra map
is said to be a Hopf algebra map if it satisfies
\begin{align}
S'\circ f=f\circ S,
\end{align}
where $ S $ and $ S' $ are the antipodes of $ H $ and $ G $, respectively.
\end{definition}

\chapter{Semisimple Lie Algebras and Iwasawa Decomposition}\label{appSemiIwasawa}

This appendix treats several topics in Lie theory. We develop the notion of reductive Lie algebras and stress
that the difficulty in their classification shifts to the classification of simple Lie algebras. We do not
prove the classification, but provide some notions and techniques that are interesting for their own. The radical of
the Killing form controls the abelian ideals of the Lie algebra and is an indicator for semisimple Lie algebras.
Then the eigenspaces of a Cartan involution give a first useful decomposition of a semisimple Lie algebra:
the Cartan decomposition.
By considering root systems we can even refine this description and obtain a Iwasawa decomposition. All these
decompositions can also be done on Lie group level.

\section{The Classification of Reductive Lie Algebras}

Here we summarize some very basic facts about Lie algebras. It is not about giving a full description of the
theory, but to come to the definition of a semisimple Lie algebra as fast as possible. For an introduction to
Lie algebras consider e.g. \cite{hall2003lie} or \cite{humphreys1972introduction}. We follow
\cite[Chapter~2]{KacMoodyLecture2009}. Let $ \mathbb{k} $ be a field.

\begin{definition}[Lie Algebra]
A $ \mathbb{k} $-vector space $ \mathfrak{g} $ endowed with a $ \mathbb{k} $-bilinear map
\begin{align}\label{LieBracket}
\left[\cdot,\cdot\right]\colon\mathfrak{g}\times\mathfrak{g}\rightarrow\mathfrak{g}
\end{align}
which satisfies $ \left[x,x\right]=0 $ for all $ x\in\mathfrak{g} $ and the \textbf{Jacobi identity}
\begin{align}
\left[x,\left[y,z\right]\right]+\left[z,\left[x,y\right]\right]+\left[y,\left[z,x\right]\right]=0,
\end{align}
for all $ x,y,z\in\mathfrak{g} $, is said to be a Lie algebra over $ \mathbb{k} $. The map defined in Eq. 
(\ref{LieBracket}) is said to be the \textbf{Lie bracket} of $ \mathfrak{g} $ and one often writes $ (\mathfrak{g},
\left[\cdot,\cdot\right]) $ to denote the Lie algebra together with its bracket.
\end{definition}
By using the short notation
\begin{align}
\left[\mathfrak{h},\mathfrak{l}\right]=\text{span}_{\mathbb{k}}\left\lbrace\left[ h,\ell\right]~|~h\in\mathfrak{h},
\ell\in\mathfrak{l}\right\rbrace,
\end{align}
where $ \mathfrak{h} $ and $ \mathfrak{l} $ are arbitrary subsets of a Lie algebra $ \mathfrak{g} $ over
$ \mathbb{k} $, one can define some substructures and important types of Lie algebras:

\begin{definition}
Consider a Lie algebra $ (\mathfrak{g},\left[\cdot,\cdot\right]) $ over $ \mathbb{k} $.
\begin{compactenum}
\item A subset $ \mathfrak{h} $ of $ \mathfrak{g} $ is said to be
\begin{compactenum}[a.)]
\item a \textbf{Lie subalgebra} of $ \mathfrak{g} $ if
$ \left[\mathfrak{h},\mathfrak{h}\right]\subseteq\mathfrak{h} $.
\item an \textbf{ideal} of $ \mathfrak{g} $ if $ \left[\mathfrak{g},\mathfrak{h}\right]\subseteq\mathfrak{h} $.
\item a \textbf{proper ideal} of $ \mathfrak{g} $ if it is an ideal of $ \mathfrak{g} $ but not equal to
$ \left\lbrace 0\right\rbrace $ or $ \mathfrak{g} $.
\end{compactenum}
\item $ \mathfrak{g} $ is said to be
\begin{compactenum}[a.)]
\item \textbf{abelian} if $ \left[\mathfrak{g},\mathfrak{g}\right]=\left\lbrace 0\right\rbrace $.
\item \textbf{simple} if $ \mathfrak{g} $ is not abelian and there is no proper ideal contained in $ \mathfrak{g} $.
\item \textbf{semisimple} if $ \mathfrak{g} $ can be written as a direct sum of simple Lie algebras.
\item \textbf{reductive} if $ \mathfrak{g} $ can be written as a direct sum of simple and abelian Lie algebras.
\end{compactenum}
\end{compactenum}
\end{definition}
It is clear from the definition that every proper ideal is an ideal and every ideal is a Lie subalgebra.
$ \left\lbrace 0\right\rbrace $ and $ \mathfrak{g} $ are generic ideals of $ \mathfrak{g} $. Moreover,
every abelian or simple or semisimple Lie algebra is reductive and every simple Lie algebra is semisimple.
The Lie groups $ \operatorname{SO}(3) $, $ \operatorname{SL}(3,\mathbb{R}) $ and $ \mathcal{L}_{3}^{\uparrow,+} $
we consider in Section \ref{sectiononish} are semisimple (c.f. \cite[Section~0.8.3]{rosenfeld2013geometry}).

\begin{remark}
Let $ n\in\mathbb{N} $. One can prove that any $ n $-dimensional abelian Lie algebra is isomorphic to
\begin{align}
\bigoplus_{k=1}^{n}\mathfrak{u}(1),
\end{align}
where $ \mathfrak{u}(1) $ is a $ 1 $-dimensional vector space endowed with the Lie bracket that is identically
zero. A more difficult approach is needed to classify reductive and for this semisimple and simple
Lie algebras. This is called the \textbf{Cartan-Killing Classification}. For a proof that uses Dynkin diagrams
consider \cite{CartanKillingBosshardt}.
\end{remark}

\section{Semisimple Lie Algebras}


We start by defining a very central element in the theory of semisimple Lie algebras and follow \cite{procesi2007lie},
but also refer to \cite[Chapter~2]{CartanIwasawaJana}. The eigenspaces of a Cartan involution of a Lie algebra provide
a useful decomposition. For the definition we need the Killing form of a Lie algebra. Later we see that
the non-degeneracy of this object characterizes complex semisimple Lie algebras. For a Lie algebra
$ (\mathfrak{g},\left[\cdot,\cdot\right]) $ and an element $ x\in\mathfrak{g} $ we denote the \textit{adjoint
endomorphism} on $ \mathfrak{g} $ corresponding to $ x $ by
\begin{align}
\text{ad}(x)\colon\mathfrak{g}\ni y\mapsto\left[x,y\right]\in\mathfrak{g}.
\end{align}

\begin{definition}[Killing Form]
Let $ (\mathfrak{g},\left[\cdot,\cdot\right]) $ be a Lie algebra over a field $ \mathbb{k} $. The map
\begin{align}\label{KillingForm}
\kappa\colon\mathfrak{g}\times\mathfrak{g}\ni(x,y)\mapsto\operatorname{tr}(\operatorname{ad}(x)
\operatorname{ad}(y))\in\mathbb{k}
\end{align}
is said to be the Killing form on $ \mathfrak{g} $.
\end{definition}
Already for an arbitrary Lie algebra $ \kappa $ has interesting properties:

\begin{proposition}\label{PropKilling}
Consider a Lie algebra $ (\mathfrak{g},\left[\cdot,\cdot\right]) $ over a field $ \mathbb{k} $. Then the following
statements hold.
\begin{compactenum}
\item The map $ \kappa $ defined in Eq. (\ref{KillingForm}) is a well-defined symmetric bilinear form.
\item $ \kappa $ is \textbf{associative}, i.e. for all $ x,y,z\in\mathfrak{g} $ the equation
\begin{align}
\kappa(\left[x,y\right],z)=\kappa(x,\left[y,z\right])
\end{align}
holds.
\item $ \kappa $ is invariant under any automorphism $ \rho $ of $ \mathfrak{g} $, i.e. for all $ x,y\in\mathfrak{g} $
one has
\begin{align}
\kappa(\rho(x),\rho(y))=\kappa(x,y).
\end{align}
\item The \textbf{radical}
\begin{align}
S=\left\lbrace x\in\mathfrak{g}~|~\kappa(x,y)=0\text{ for all }y\in\mathfrak{g}\right\rbrace
\end{align}
is an ideal of $ \mathfrak{g} $.
\end{compactenum}
\end{proposition}

\begin{proof}
Since $ \text{ad}(x)\text{ad}(y) $ is an endomorphism which is mapped to $ \mathbb{k} $ by the trace
for all $ x,y\in\mathfrak{g} $ the map $ \kappa $ is well-defined. By the linearity of the trace and of the map
\begin{align}
\text{ad}\colon\mathfrak{g}\ni x\mapsto\text{ad}(x)\in\text{End}(\mathfrak{g})
\end{align}
the bilinearity of $ \kappa $ is clear. The symmetry of $ \kappa $ is also clear by the cyclic permutation property
of the trace. Now the second part is an easy calculation. For $ x,y,z\in\mathfrak{g} $ one has
\begin{align*}
\kappa(\left[x,y\right],z)&=\text{tr}(\text{ad}(\left[x,y\right])\text{ad}(z))
=\text{tr}(\left[\text{ad}(x),\text{ad}(y)\right]\text{ad}(z))\\
&=\text{tr}(\text{ad}(x)\text{ad}(y)\text{ad}(z)-\text{ad}(y)\text{ad}(x)\text{ad}(z))\\
&=\text{tr}(\text{ad}(x)\text{ad}(y)\text{ad}(z)-\text{ad}(x)\text{ad}(z)\text{ad}(y))\\
&=\text{tr}(\text{ad}(x)\left[\text{ad}(y),\text{ad}(z)\right])
=\text{tr}(\text{ad}(z)\text{ad}(\left[y,z\right]))\\
&=\kappa(x,\left[y,z\right]),
\end{align*}
where we used again cyclic permutations of $ \text{tr} $ and two times the formula
$ \text{ad}(\left[x,y\right])=\left[\text{ad}(x),\text{ad}(y)\right] $ which holds because $ \text{ad} $ is a Lie
algebra homomorphism. For the third statement take $ x,y,z\in\mathfrak{g} $ such that $ z=\rho(y) $. Since $ \rho $ is
invertible, the equation $ \rho(\left[x,y\right])=\left[\rho(x),\rho(y)\right] $ is
equivalent to
\begin{align}
\text{ad}(\rho(x))(z)=(\rho\circ\text{ad}(x)\circ\rho^{-1})(z).
\end{align}
This implies
\begin{align*}
\kappa(\rho(x),\rho(y))=\text{tr}(\text{ad}(\rho(x))\text{ad}(\rho(y)))
=\text{tr}(\rho\circ(\text{ad}(x)\text{ad}(y))\circ\rho^{-1})=\kappa(x,y).
\end{align*}
The last statement is a consequence of ii.). To prove this, take arbitrary elements $ x\in S $ and
$ y\in\mathfrak{g} $. Then $ \left[x,y\right]\in S $ because
\begin{align*}
\kappa(\left[x,y\right],z)=\kappa(x,\underbrace{\left[y,z\right]}_{\in\mathfrak{g}})=0
\end{align*}
for all $ z\in\mathfrak{g} $ which proves the claim.
\end{proof}

In particular, the last statement of this proposition is interesting for semisimple Lie algebras since there is
always an ideal $ S $ in $ \mathfrak{g} $, namely the radical of $ \kappa $. If $ S=\mathfrak{g} $ the Lie algebra
would be rather boring, thus one traceable claim to hope for a semisimple Lie algebra should be
$ S=\left\lbrace 0\right\rbrace $. Since $ \kappa $ is symmetric according to i.) this is equivalent to the
postulate that $ \kappa $ is non-degenerate in $ \mathfrak{g} $. Our aim is to prove that this is not
only a necessary but also a sufficient condition to characterize complex semisimple Lie algebras. We
first prove that $ S $ controls all non-proper abelian ideals.

\begin{lemma}\label{LemmaAbId}
Let $ (\mathfrak{g},\left[\cdot,\cdot\right]) $ be a Lie algebra over a field $ \mathbb{k} $ and $ J $ an abelian
ideal of $ \mathfrak{g} $. Then $ J\subseteq S $.
\end{lemma}

\begin{proof}
Let $ x\in J $ and $ y\in\mathfrak{g} $ be arbitrary. Our first step is to show that the endomorphism
\begin{align}\label{LemmaAbIdEndo}
\text{ad}(x)\text{ad}(y)
\end{align}
is nilpotent. More explicit, we prove $ (\text{ad}(x)\text{ad}(y))^{2}=0 $. First of all
$ (\text{ad}(x)\text{ad}(y))(z)\in J $ for all $ z\in\mathfrak{g} $ since
\begin{align}\label{LemmaAbId1}
(\text{ad}(x)\text{ad}(y))(z)=\big[x,\underbrace{\left[y,z\right]}_{\in\mathfrak{g}}\big].
\end{align}
Similarly we get for $ z\in J $
\begin{align}\label{LemmaAbId2}
(\text{ad}(x)\text{ad}(y))(z)\subseteq\text{ad}(x)(J)=\left\lbrace 0\right\rbrace,
\end{align}
where the last equation follows since $ J $ is abelian and $ x\in J $. Equation (\ref{LemmaAbId2}) combined with
equation (\ref{LemmaAbId1}) gives $ (\text{ad}(x)\text{ad}(y))^{2}=0 $. Thus (\ref{LemmaAbIdEndo}) is indeed
nilpotent and for this traceless. We conclude
\begin{align*}
\kappa(x,y)=\text{tr}(\text{ad}(x)\text{ad}(y))=0
\end{align*}
which shows $ J\subseteq S $.
\end{proof}

For the proof of the next theorem we need another characterization of semisimple Lie algebras taken from
\cite[Theorem~6.5]{berndt2007representations}: a Lie algebra is semisimple if and only if it has no non-zero abelian
ideal.

\begin{theorem}\label{TheoremSemiKillingNonDeg}
Consider a finite-dimensional complex Lie algebra $ (\mathfrak{g},\left[\cdot,\cdot\right]) $. It is semisimple
if and only if the Killing from $ \kappa $ is non-degenerate on $ \mathfrak{g} $.
\end{theorem}

\begin{proof}
As argued before we prove that $ \mathfrak{g} $ has a non-zero abelian ideals if and only if
$ S\neq\left\lbrace 0\right\rbrace $.
If $ \mathfrak{g} $ has a non-zero abelian ideal $ J $ than $ \left\lbrace 0\right\rbrace\neq J\subseteq S $
according to Lemma \ref{LemmaAbId}. Assume conversely $ S\neq\left\lbrace 0\right\rbrace $. Then
\textbf{Cartan's Criterion} (c.f. \cite[Corollary~3.3.14]{abbaspour2007basic}) implies that $ S $ is a proper ideal of
$ \mathfrak{g} $ since $ \kappa(x,y)=0 $ for all $ x,y\in S $. Thus $ \mathfrak{g} $ is not semisimple.
\end{proof}

In the following we often view a semisimple Lie algebra together with its Killing form. An object which is
quite related to the Killing form is a Cartan involution.

\begin{definition}[Cartan Involution]\label{CartanInv}
Let $ (\mathfrak{g},\left[\cdot,\cdot\right]) $ be a semisimple Lie algebra over a field $ \mathbb{k} $. If there is
an involution $ \Theta $ on $ \mathfrak{g} $, i.e. a automorphism of $ \mathfrak{g} $ that satisfies $ \Theta^{2}
=1 $, such that
\begin{align}
B_{\Theta}(x,y)=-\kappa(x,\Theta(y))
\end{align}
is a positive definite bilinear form, where $ x,y\in\mathfrak{g} $, then $ \Theta $ is said to be a
Cartan involution on $ \mathfrak{g} $.
\end{definition}

One could ask in which situation such a Cartan involution exists and when it is unique. There is an answer for
real semisimple Lie algebras:

\begin{theorem}
There is a Cartan involution on a real semisimple Lie algebra which is unique up to inner automorphisms.
\end{theorem}

A proof can be found in \cite[Theorem~3]{CartanIwasawaJana}. It uses Theorem \ref{TheoremSemiKillingNonDeg}
and the existence of a compact real form of a finite-dimensional semisimple Lie algebra. We want to say some words
to the last point. A real Lie algebra $ \mathfrak{g}_{0} $ is said to be a \textit{real form} of a complex
Lie algebra $ \mathfrak{g} $ if one has
\begin{align}\label{RealForm}
\mathfrak{g}\cong\mathbb{C}\tensor_{\mathbb{R}}\mathfrak{g}_{0}.
\end{align}
Conversely, a complex Lie algebra $ \mathfrak{g} $ is said to be a \textit{complexification} of a real Lie algebra
$ \mathfrak{g}_{0} $ if Eq. (\ref{RealForm}) holds. A famous theorem of Cartan says that if $ \mathfrak{g} $ is a
complex semisimple Lie algebra there is always a real form $ \mathfrak{g}_{0} $ of $ \mathfrak{g} $ and
$ \mathfrak{g}_{0} $ is a compact Lie algebra (c.f. \cite{richardson1968}).

Now consider a real semisimple Lie algebra $ \mathfrak{g} $ with a Cartan involution $ \Theta $. Since it is an
involution $ \Theta $ has exactly two eigenvalues $ \pm 1 $. The corresponding eigenspaces are commonly
denoted by $ \mathfrak{k} $ and $ \mathfrak{p} $.

\begin{definition}[Cartan Decomposition]\label{CartanDec}
The decomposition
\begin{align}\label{CartanDecomp}
\mathfrak{g}=\mathfrak{k}\oplus\mathfrak{p}
\end{align}
of a real semisimple Lie algebra $ \mathfrak{g} $ is called the Cartan decomposition of $ \mathfrak{g} $
corresponding to $ \Theta $.
\end{definition}

\begin{proposition}
The eigenspaces $ \mathfrak{k} $ and $ \mathfrak{p} $ of the Cartan decomposition (\ref{CartanDecomp}) satisfy
$ \left[\mathfrak{k},\mathfrak{k}\right]\subseteq\mathfrak{k} $,
$ \left[\mathfrak{p},\mathfrak{k}\right]\subseteq\mathfrak{p} $ and
$ \left[\mathfrak{p},\mathfrak{p}\right]\subseteq\mathfrak{k} $.
\end{proposition}

\begin{proof}
This is clear since $ \Theta $ is an automorphism. For example for $ k\in\mathfrak{k} $ and $ p\in\mathfrak{p} $
one has
\begin{align*}
\Theta(\left[k,p\right])=\left[\Theta(k),\Theta(p)\right]=\left[k,-p\right]=-\left[k,p\right].
\end{align*}
\end{proof}

\section{Root Space Decomposition and Iwasawa Decomposition}

To get in touch with Iwasawa decompositions we use roots and Cartan subalgebras and follow
\cite[Chapter~3]{KacMoodyLecture2009}. There are many equivalent definitions of Cartan subalgebras, but we
characterize them
via semisimple elements. Recall that for an element $ x $ of a Lie algebra $ (\mathfrak{g},\left[\cdot,\cdot\right]) $
the endomorphism $ \text{ad}(x) $ is diagonizable if there is a basis $ \left\lbrace x_{1},x_{2},\ldots\right\rbrace $
of $ \mathfrak{g} $ such that $ \left[x,x_{k}\right] $ is proportional to $ x_{k} $ for every $ k=1,2,\ldots $.
Moreover, an abelian subalgebra $ \mathfrak{h} $ of $ \mathfrak{g} $ is said to be \textit{maximal abelian} if there
is no abelian subalgebra $ \mathfrak{h}' $ of $ \mathfrak{g} $ such that $ \mathfrak{h}\subsetneq\mathfrak{h}'
\subsetneq\mathfrak{g} $.

\begin{definition}[Cartan Subalgebra]
Let $ (\mathfrak{g},\left[\cdot,\cdot\right]) $ be a Lie algebra.
\begin{compactenum}
\item An element $ x\in\mathfrak{g} $ is said to be semisimple if $ \operatorname{ad}(x) $ is diagonizable.
\item If $ \mathfrak{g} $ is semisimple, a maximal abelian subalgebra $ \mathfrak{h} $ of $ \mathfrak{g} $ is said
to be a Cartan subalgebra of $ \mathfrak{g} $ if it consists of semisimple elements.
\end{compactenum}
\end{definition}

Can one guarantee the existence of a Cartan subalgebra? This is discussed in the next

\begin{remark}\label{RemCartanBasis}
There are semisimple elements if the field $ \mathbb{k} $ corresponding to the semisimple Lie algebra
$ (\mathfrak{g},\left[\cdot,\cdot\right]) $ is algebraically closed. This is clear because the condition of
$ \text{ad}(x) $ being diagonizable for an $ x\in\mathfrak{g} $ reduces to solutions $ \lambda\in\mathbb{k} $
of the equation
\begin{align}
\det(M-\lambda\mathbb{1})=0,
\end{align}
where the entries of $ M $ are the structure constants $ M_{ij} $ which satisfy
$ \left[x,x_{i}\right]=\sum_{j}M_{ij}x_{j} $. Thus it is possible to choose a maximal set of
linearly independent semisimple elements $ H_{i}\in\mathfrak{g} $ that satisfy
\begin{align}\label{CartanBasisCommutes}
\left[H_{i},H_{j}\right]=0.
\end{align}
The subset spanned by these elements is a Cartan subalgebra $ \mathfrak{h} $ of $ \mathfrak{g} $.
In general there are many Cartan subalgebras of $ \mathfrak{g} $, but one can show that they have all the
same dimension, which is said to be the \textbf{rank} of $ \mathfrak{g} $.
\end{remark}

Consider again a semisimple Lie algebra $ (\mathfrak{g},\left[\cdot,\cdot\right]) $ over an algebraically closed field
$ \mathbb{k} $ and choose a Cartan subalgebra $ \mathfrak{h} $ with basis $ \left\lbrace H_{i}\right\rbrace $ as
it was done in Remark \ref{RemCartanBasis}. All elements of $ \mathfrak{h} $ are semisimple, thus $ \text{ad}(x) $
is diagonizable for all $ x\in\mathfrak{h} $ and since Eq. (\ref{CartanBasisCommutes}) holds, all endomorphisms
$ \text{ad}(x) $, for $ x\in\mathfrak{h} $, are simultaneously diagonizable. That means there is a basis
$ \left\lbrace x_{1},x_{2},\ldots\right\rbrace $ of $ \mathfrak{g} $ such that
\begin{align}\label{RootsCondition}
(\text{ad}(x))(x_{i})=\left[x,x_{i}\right]=\alpha_{x_{i}}(x)x_{i}
\end{align}
for $ i=1,2,\ldots $, where $ \alpha_{x_{i}}\in\mathfrak{h}^{*} $ is the linear function $ \mathfrak{h}\rightarrow
\mathbb{k} $ that assigns an element $ x\in\mathfrak{h} $ the eigenvalue $ \alpha_{x_{i}}(x) $ of $ \text{ad}(x) $
corresponding to the eigenvector $ x_{i} $.

\begin{definition}[Root System]
Let $ (\mathfrak{g},\left[\cdot,\cdot\right]) $ be a semisimple Lie algebra over an algebraically closed field
$ \mathbb{k} $ and $ \left\lbrace x_{1},x_{2},\ldots\right\rbrace $ a basis of $ \mathfrak{g} $ such that Eq.
(\ref{RootsCondition}) holds. Then $ \alpha_{x_{1}},\alpha_{x_{2}},\ldots $ is said to be a weight system
and $ \alpha_{x_{i}} $ a weight of the adjoint representation of the Cartan subalgebra $ \mathfrak{h} $.
The elements $ \alpha_{x_{i}}\neq 0 $ are said to be the \textit{roots} of $ \mathfrak{g} $ and the set of roots
$ \Delta $ is said to be a root system of $ \mathfrak{g} $. We often omit the index $ x_{i} $.
\end{definition}

\begin{remark}
One can prove that the elements $ x\in\mathfrak{g} $ which are weights but no roots of $ \mathfrak{g} $ are
elements of $ \mathfrak{h} $. Conversely it is clear that no root $ x\in\mathfrak{h} $ can be a weight since
$ \left[h,x\right]=0 $ for all $ h\in\mathfrak{h} $. Thus if we define for every root $ \alpha\in\Delta $ the space
\begin{align}
\mathfrak{g}_{\alpha}=\left\lbrace x\in\mathfrak{g}~|~\left[h,x\right]=\alpha(h)x\text{ for all }
h\in\mathfrak{h}\right\rbrace
\end{align}
there is a decomposition
\begin{align}\label{RootDecomposition}
\mathfrak{g}=\mathfrak{h}\oplus\bigoplus_{\alpha\in\Delta}\mathfrak{g}_{\alpha}
\end{align}
of $ \mathfrak{g} $. It is called the \textbf{root space decomposition} of $ \mathfrak{g} $ corresponding to
the Cartan subalgebra $ \mathfrak{h} $. The number $ \dim(\mathfrak{g}_{\alpha}) $ is called the
\textbf{multiplicity} of the root $ \alpha $.
\end{remark}

The joint eigenspaces $ \mathfrak{g}_{\alpha} $ have some properties that are easy to prove by calculation.

\begin{proposition}
The following statements hold for a root decomposition (\ref{RootDecomposition}). Let $ \alpha,\beta\in
\mathfrak{h}^{*} $. Then
\begin{compactenum}
\item $ \left[\mathfrak{g}_{\alpha},\mathfrak{g}_{\beta}\right]\subseteq\mathfrak{g}_{\alpha+\beta} $,
\item if $ \alpha\neq 0 $ and $ x\in\mathfrak{g}_{\alpha} $ then $ \operatorname{ad}(x) $ is a nilpotent endomorphism,
\item if $ \alpha+\beta\neq 0 $ then $ \kappa(\mathfrak{g}_{\alpha},\mathfrak{g}_{\beta})=
\left\lbrace 0\right\rbrace $.
\end{compactenum}
\end{proposition}

Let us come back to the Cartan decomposition (\ref{CartanDecomp}) of a semisimple Lie algebra. We consider a root space
decomposition of the eigenspace $ \mathfrak{p} $ corresponding to the
$ -1 $ eigenvalue of a
Cartan involution $ \Theta $. The Cartan subalgebra of $ \mathfrak{p} $ shall now be denoted by $ \mathfrak{a} $.
We are not interested in all roots $ \Delta $ but only in the positive roots $ \Delta^{+} $. Choose a basis
$ \left\lbrace\alpha_{1},\alpha_{2},\ldots\right\rbrace $ of $ \Delta $. A root $ \alpha\in\Delta $ is said to be
\textit{positive} if the first non-zero coefficient in the basis representation is positive. Then we define
\begin{align}
\mathfrak{n}=\bigoplus_{\alpha\in\Delta^{+}}\mathfrak{g}_{\alpha}.
\end{align}

\begin{definition}[Iwasawa Decomposition]\label{IwasawaDec}
Let $ (\mathfrak{g},\left[\cdot,\cdot\right]) $ be a real semisimple Lie algebra. The decomposition
\begin{align}\label{IwasawaDecomposition}
\mathfrak{g}=\mathfrak{k}\oplus\mathfrak{a}\oplus\mathfrak{n}
\end{align}
is said to be a Iwasawa decomposition of $ \mathfrak{g} $.
The Iwasawa decomposition at Lie group level is
\begin{align}
G=KAN,
\end{align}
where $ G,K,A $ and $ N $ denote the connected Lie groups corresponding to $ \mathfrak{g},\mathfrak{k},\mathfrak{a} $
and $ \mathfrak{n} $ respectively.
\end{definition}

%% file: LastPage.tex
\chapter*{}

\begin{center} 

  \null\vfill
  {\large Erklärung\\}
  \vspace{2cm}
  {\large Hiermit erkläre ich, dass ich die vorliegende Arbeit selbständig verfasst und keine anderen als die angegebenen Quellen und Hilfsmittel verwendet und die Arbeit keiner anderen Prüfungsbehörde unter Erlangung eines akademischen Grades vorgelegt habe.\\}
  
\end{center}
\vspace{2cm}
Würzburg, den 06.06.2016\\
\flushright Thomas Weber

  \vspace{2cm}
  
  \vfill